\documentclass[11pt,a4paper]{amsart}
\usepackage{newtxtext}

\usepackage{amssymb,amsmath}
\usepackage[alphabetic]{amsrefs}
\PassOptionsToPackage{hyphens}{url}\usepackage{hyperref}
\usepackage{mathtools}
\usepackage{amsthm}
\usepackage[utf8]{inputenc}
\usepackage{latexsym}
\usepackage{chngpage}
\usepackage{subcaption}
\usepackage{multirow}
\usepackage{xurl}

\usepackage{graphicx,color}
\usepackage[all]{xy}
\usepackage{tikz-cd}
\usetikzlibrary{patterns}

\usepackage{accents}
\usepackage{mathrsfs}
\usepackage[mathscr]{eucal}
\usepackage{bm}
\usepackage{bbm}
\usepackage[bbgreekl]{mathbbol}
\DeclareSymbolFontAlphabet{\mathbb}{AMSb}
\DeclareSymbolFontAlphabet{\mathbbl}{bbold}
\usepackage{scalerel}

\usepackage{fullpage}

\hypersetup{
    colorlinks=true,
    citecolor=blue,
    linkcolor=blue,
    urlcolor=cyan,
    breaklinks=true,
}

\usepackage{lipsum}
\usepackage{xpatch}
\xpatchcmd{\paragraph}{\normalfont}{{\normalfont\itshape}}{}{}

\newcounter{intro}
\counterwithin{equation}{section}

\makeatletter
\let\c@equation\c@figure
\makeatother

\newcounter{copy}
\makeatletter
\renewcommand{\thecopy}{\ifnum0=\c@section\arabic{copy}\else\thesection.\arabic{copy}'\fi}
\makeatother

\BibSpec{book}{%
    +{}  {\PrintPrimary}                {transition}
    +{.} { \textit}                     {title}
    +{.} { }                            {part}
    +{:} { \textit}                     {subtitle}
    +{,} { \PrintEdition}               {edition}
    +{}  { \PrintEditorsB}              {editor}
    +{,} { \PrintTranslatorsC}          {translator}
    +{,} { \PrintContributions}         {contribution}
    +{,} { }                            {series}
    +{,} { \voltext}                    {volume}
    +{,} { }                            {publisher}
    +{,} { }                            {organization}
    +{,} { }                            {address}
    +{,} { }                            {status}
    +{,} { \PrintISBNs}                 {isbn}
    +{}  { \parenthesize}               {language}
    +{}  { \PrintTranslation}           {translation}
    +; { \PrintReprint}               {reprint}
    +{,} { \PrintDate}                  {date}
    +{.} { }                            {note}
    +{.} {}                             {transition}
}
\BibSpec{article}{%
    +{}  {\PrintAuthors}                {author}
    +{,} { \textit}                     {title}
    +{.} { }                            {part}
    +{:} { \textit}                     {subtitle}
    +{,} { \PrintContributions}         {contribution}
    +{.} { \PrintPartials}              {partial}
    +{,} { }                            {journal}
    +{}  { \textbf}                     {volume}
    +{}  { \PrintDatePV}                {date}
    +{,} { \issuetext}                  {number}
    +{,} { \eprintpages}                {pages}
    +{,} { }                            {status}
    +{}  { \parenthesize}               {language}
    +{}  { \PrintTranslation}           {translation}
    +; { \PrintReprint}               {reprint}
    +{.} { }                            {note}
    +{.} {}                             {transition}
}
\BibSpec{collection.article}{%
    +{}  {\PrintAuthors}                {author}
    +{,} { \textit}                     {title}
    +{.} { }                            {part}
    +{:} { \textit}                     {subtitle}
    +{,} { \PrintContributions}         {contribution}
    +{,} { \PrintConference}            {conference}
    +{}  {\PrintBook}                   {book}
    +{,} { }                            {booktitle}
    +{,} { \PrintDateB}                 {date}
    +{,} { pp.~}                        {pages}
    +{,} { }                            {publisher}
    +{,} { }                            {organization}
    +{,} { }                            {address}
    +{,} { }                            {status}
    +{,} { \eprint}        {eprint}
    +{}  { \parenthesize}               {language}
    +{}  { \PrintTranslation}           {translation}
    +; { \PrintReprint}               {reprint}
    +{.} { }                            {note}
    +{.} {}                             {transition}
}

\BibSpec{misc}{
  +{}{\PrintAuthors}  {author}
  +{,}{ \textit}     {title}
  +{}{ (}             {date}
  +{),}{ }             {note}
  +{.}{}              {transition}
}

\usepackage{cleveref}
\theoremstyle{definition}
\newtheorem{defn}[equation]{Definition}

\theoremstyle{plain}
\newtheorem{thm}[equation]{Theorem}
\newtheorem{introthm}[intro]{Theorem}
\newtheorem{prp}[equation]{Proposition}
\newtheorem{lem}[equation]{Lemma}
\newtheorem{cor}[equation]{Corollary}
\newtheorem{conj}[equation]{Conjecture}
\theoremstyle{remark}
\newtheorem{rmk}[equation]{Remark}
\newtheorem{exmp}[equation]{Example}

\crefname{defn}{Definition}{Definitions}
\crefname{notn}{Notation}{Notations}
\crefname{assmp}{Assumption}{Assumptions}
\crefname{thm}{Theorem}{Theorems}
\crefname{introthm}{Theorem}{Theorems}
\crefname{prp}{Proposition}{Propositions}
\crefname{lem}{Lemma}{Lemmas}
\crefname{cor}{Corollary}{Corollaries}
\crefname{conj}{Conjecture}{Conjectures}
\crefname{rmk}{Remark}{Remarks}
\crefname{exmp}{Example}{Examples}
\crefname{section}{Section}{Sections}
\crefname{subsection}{Subsection}{Subsections}
\crefname{para}{}{}
\crefname{appendix}{Appendix}{Appendices}
\crefname{subappendix}{Appendix}{Appendices}
\crefname{table}{Table}{Tables}
\crefname{figure}{Figure}{Figures}

\newcommand{\vvert}{\| \kern-0.23ex | }
\newcommand{\bigvvert}{\big\| \kern-0.23ex \big| }
\newcommand{\Bigvvert}{\Big\| \kern-0.23ex \Big| }

\makeatletter
\newcommand{\dbloverline}[1]{\overline{\dbl@overline{#1}}}
\newcommand{\dbl@overline}[1]{\mathpalette\dbl@@overline{#1}}
\newcommand{\dbl@@overline}[2]{%
  \begingroup
  \sbox\z@{$\m@th#1\overline{#2}$}%
  \ht\z@=\dimexpr\ht\z@-2\dbl@adjust{#1}\relax
  \box\z@
  \ifx#1\scriptstyle\kern-\scriptspace\else
  \ifx#1\scriptscriptstyle\kern-\scriptspace\fi\fi
  \endgroup
}
\newcommand{\dbl@adjust}[1]{%
  \fontdimen8
  \ifx#1\displaystyle\textfont\else
  \ifx#1\textstyle\textfont\else
  \ifx#1\scriptstyle\scriptfont\else
  \scriptscriptfont\fi\fi\fi 3
}
\makeatother

\newcommand{\bA}{\mathbb{A}}
\newcommand{\bB}{\mathbb{B}}
\newcommand{\bC}{\mathbb{C}}
\newcommand{\bD}{\mathbb{D}}
\newcommand{\bE}{\mathbb{E}}

\newcommand{\bN}{\mathbb{N}}

\newcommand{\bP}{\mathbb{P}}
\newcommand{\bQ}{\mathbb{Q}}
\newcommand{\bR}{\mathbb{R}}

\newcommand{\bT}{\mathbb{T}}

\newcommand{\bV}{\mathbb{V}}

\newcommand{\bZ}{\mathbb{Z}}
\newcommand{\cA}{\mathcal{A}}
\newcommand{\cB}{\mathcal{B}}

\newcommand{\cF}{\mathcal{F}}
\newcommand{\cG}{\mathcal{G}}

\newcommand{\cI}{\mathcal{I}}
\newcommand{\cJ}{\mathcal{J}}

\newcommand{\cM}{\mathcal{M}}

\newcommand{\cO}{\mathcal{O}}
\newcommand{\cP}{\mathcal{P}}

\newcommand{\cR}{\mathcal{R}}
\newcommand{\cS}{\mathcal{S}}

\newcommand{\cU}{\mathcal{U}}

\newcommand{\cZ}{\mathcal{Z}}

\newcommand{\fC}{\mathfrak{C}}
\newcommand{\fD}{\mathfrak{D}}

\newcommand{\fH}{\mathfrak{H}}
\newcommand{\fI}{\mathfrak{I}}
\newcommand{\fJ}{\mathfrak{J}}

\newcommand{\fL}{\mathfrak{L}}

\newcommand{\fR}{\mathfrak{R}}
\newcommand{\fS}{\mathfrak{S}}
\newcommand{\fT}{\mathfrak{T}}

\newcommand{\fa}{\mathfrak{a}}

\newcommand{\fc}{\mathfrak{c}}
\newcommand{\fd}{\mathfrak{d}}

\newcommand{\fm}{\mathfrak{m}}

\newcommand{\sfE}{\mathsf{E}}

\newcommand{\sfG}{\mathsf{G}}
\newcommand{\sfH}{\mathsf{H}}

\newcommand{\sfP}{\mathsf{P}}

\newcommand{\sfT}{\mathsf{T}}

\newcommand{\sfV}{\mathsf{V}}
\newcommand{\sfW}{\mathsf{W}}
\newcommand{\sfX}{\mathsf{X}}

\newcommand{\sfh}{\mathsf{h}}

\newcommand{\sfs}{\mathsf{s}}

\newcommand{\sfu}{\mathsf{u}}

\newcommand{\sfx}{\mathsf{x}}
\newcommand{\sfy}{\mathsf{y}}

\newcommand{\sfad}{\mathsf{ad}}

\newcommand{\sA}{\mathscr{A}}

\newcommand{\sD}{\mathscr{D}}
\newcommand{\sE}{\mathscr{E}}
\newcommand{\sF}{\mathscr{F}}
\newcommand{\sG}{\mathscr{G}}
\newcommand{\sH}{\mathscr{H}}

\newcommand{\sK}{\mathscr{K}}

\newcommand{\sM}{\mathscr{M}}
\newcommand{\sN}{\mathscr{N}}

\newcommand{\sU}{\mathscr{U}}
\newcommand{\sV}{\mathscr{V}}

\newcommand{\sX}{\mathscr{X}}
\newcommand{\sY}{\mathscr{Y}}

\newcommand{\rmd}{\mathrm{d}}

\newcommand{\op}{\mathrm{op}}

\newcommand{\blank}{\text{\textvisiblespace}}
\renewcommand{\Im}{\mathrm{Im} \hspace{0.1em}}
\newcommand{\pt}{\mathrm{pt}}
\newcommand{\loc}{\mathrm{loc}}
\newcommand{\id}{\mathrm{id}}
\newcommand{\ev}{\mathrm{ev}}

\newcommand{\fin}{\mathrm{fin}}
\newcommand{\Sp}{\mathsf{Sp}}

\DeclareMathOperator*{\Rep}{\mathsf{Rep}}

\DeclareMathOperator{\hotimes}{\hat{\otimes }}
\DeclareMathOperator{\ch}{\mathrm{ch}} 
\DeclareMathOperator{\Ker}{\mathrm{Ker}}

\DeclareMathOperator{\Map}{\mathrm{Map}}
\DeclareMathOperator{\Hom}{\mathrm{Hom}}
\DeclareMathOperator{\Aut}{\mathrm{Aut}}

\DeclareMathOperator{\Ad}{\mathrm{Ad}}

\DeclareMathOperator{\diag}{\mathrm{diag}}
\DeclareMathOperator{\diam}{\mathrm{diam}}
\DeclareMathOperator*{\colim}{\mathrm{colim}}

\DeclareMathOperator{\Stab}{Stab}

\newcommand{\pr}{\mathrm{pr}}

\DeclareMathOperator{\flip}{\mathtt{f}\hspace{-0.2ex}\mathtt{l}}

\DeclareMathOperator{\sHom}{\mathscr{H}\mathrm{om}}
\newcommand{\Or}{\mathsf{Or}}
\newcommand{\Man}{\mathsf{Man}}
\DeclareMathOperator{\Sing}{\mathsf{Sing}}
\newcommand{\kTop}{\mathsf{kTop}}

\newcommand{\CW}{\mathsf{CW}}
\newcommand{\Set}{\mathsf{Set}}
\newcommand{\sSet}{\mathsf{sSet}}
\DeclareMathOperator{\Sh}{\mathsf{Sh}}
\DeclareMathOperator{\PSh}{\mathsf{PSh}}

\newcommand{\ECW}{\mathsf{ECW}}

\newcommand{\ffH}{\mathfrak{f}\mathfrak{H}}

\newcommand{\bsfH}{\boldsymbol{\mathsf{H}}}

\newcommand{\sGP}{\mathscr{GP}}
\newcommand{\sfGP}{\mathit{f}\mathscr{GP}}
\newcommand{\sIP}{\mathscr{IP}}
\newcommand{\bsIP}{\boldsymbol{\mathscr{IP}}}
\newcommand{\IP}{\mathit{IP}}
\newcommand{\rIP}{\mathrm{IP}}
\newcommand{\sfIP}{\mathit{f}\mathscr{IP}}
\newcommand{\fIP}{\mathit{fIP}}
\newcommand{\rfIP}{\mathrm{fIP}}
\newcommand{\bsfIP}{\boldsymbol{\mathit{f}\mathscr{IP}}}
\usepackage{scalerel}

\DeclareMathOperator*{\hocolim}{\mathrm{ho}\mathchar`-\mathrm{colim}}

\newcommand{\piloc}{\pi\mathchar`-\mathrm{loc}}
\newcommand{\alg}{\mathrm{alg}}
\newcommand{\al}{\mathrm{al}}
\newcommand{\Sq}{\mathit{Sq}}
\newcommand{\gap}{\delta }
\newcommand{\gapone}{1 }

\newcommand{\KO}{\mathrm{KO}}
\newcommand{\DW}{\mathrm{DW}}

\newcommand{\pst}{{}_\star}
\newcommand{\midbar}{\, | \, }
\newcommand{\fDer}{\mathfrak{D}}
\newcommand{\fFDer}{\mathfrak{fD}}
\newcommand{\fLDer}{\mathfrak{LD}}

\makeatletter
\newcommand{\doublewidetilde}[1]{{%
  \mathpalette\double@widetilde{#1}%
}}
\newcommand{\double@widetilde}[2]{%
  \sbox\z@{$\m@th#1\widetilde{#2}$}%
  \ht\z@=0.56\ht\z@
  \widetilde{\box\z@}%
}
\makeatother

\title[Stable homotopy theory of invertible gapped quantum spin systems I]{Stable homotopy theory of invertible gapped quantum spin systems I: Kitaev's \texorpdfstring{$\Omega$}{Omega}-spectrum}
\author{Yosuke Kubota}
\address{Graduate School of Science, Kyoto University, Kitashirakawa Oiwake-cho, Sakyo-ku, Kyoto 606-8502, Japan}
\email{ykubota@math.kyoto-u.ac.jp}
\date{}
\begin{document}
\begin{abstract}
We provide a mathematical realization of a conjecture by Kitaev, on the basis of the operator-algebraic formulation of infinite quantum spin systems. 
Our main results are threefold. 
First, we construct an $\Omega$-spectrum $\IP_*$ whose homotopy groups are isomorphic to the smooth homotopy group of invertible gapped quantum systems on Euclidean spaces. 
Second, we develop a model for the homology theory associated with the $\Omega$-spectrum $\IP_*$, describing it in terms of the space of quantum systems placed on an arbitrary subspace of a Euclidean space. 
This involves introducing the concept of localization flow, a semi-infinite path of quantum systems with decaying interaction range, inspired by Yu's localization C*-algebra in coarse index theory. 
Third, we incorporate spatial symmetries given by a crystallographic group $\Gamma $ and define the $\Omega$-spectrum $\IP_*^\Gamma$ of $\Gamma$-invariant invertible phases. 
We propose a strategy for computing the homotopy group $\pi_n(\IP_d^\Gamma )$ that uses the Davis--L\"{u}ck assembly map and its description by invertible gapped localization flow. 
In particular, we show that the assembly map is split injective, and hence $\pi_n(\IP_d^\Gamma)$ contains a computable direct summand.
\end{abstract}
\maketitle

\tableofcontents 

\section{Introduction}
In his lectures \cites{kitaevTopologicalClassificationManybody2011,kitaevClassificationShortrangeEntangled2013,kitaevHomotopytheoreticApproachSPT2015,kitaevTopologicalQuantumPhases2019}, A.~Kitaev proposed that the set of $d$-dimensional invertible gapped systems $\{\IP_d\}_{d\in \bZ_{\geq 0}}$ would form an $\Omega$-spectrum. 
A consequence of this conjecture is that the homotopy sets $\rIP^d(\sM)\coloneqq [\sM,\IP_d]$ form a generalized cohomology functor (as references of stable homotopy theory, we refer to \cites{rudyakThomSpectraOrientability1998,switzerAlgebraicTopologyHomotopy2002}).
Here, a quantum system represented by a gapped Hamiltonian $\sfH$ is said to be invertible if there is another Hamiltonian $\check{\sfH}$ such that the composite system $\sfH \boxtimes \check{\sfH}$ is homotopic to the trivial (i.e., a fixed choice of atomic) Hamiltonian.
This was motivated by a successful classification of topological insulators (free fermion phases) by real and complex K-theories \cites{kitaevPeriodicTableTopological2009,schnyderClassificationTopologicalInsulators2008,schnyderClassificationTopologicalInsulators2009}. 
Its  contribution is the study of the symmetry-protected topological (SPT) phase, i.e., the $0$-th homotopy group of the set of invertible gapped quantum systems invariant under an on-site group action. 
For example, the $\bZ/2$ topological index of the AKLT and Haldane phases of quantum spin chains \cites{haldaneContinuumDynamics1D1983,affleckValenceBondGround1988} has been understood as an instance of a topological invariant taking values in $\mathrm{H}^{2}(G \, ;\bT) \cong \mathrm{H}^{3}(G \, ;\bZ)$ due to recent advances \cites{guTensorentanglementfilteringRenormalizationApproach2009,pollmannEntanglementSpectrumTopological2010,pollmannSymmetryProtectionTopological2012,ogataMathbb_2Index2020,tasakiGroundStateS12025}. 
Various approaches in the last decade have studied this idea. The following is an incomplete list of existing research.
\begin{enumerate}
    \item A classification of bosonic and fermionic SPT phases via group (super-)cohomology has been considered in 
\cites{chenLocalUnitaryTransformation2010,chenClassificationGappedSymmetric2011,chenSymmetryprotectedTopologicalOrders2012,chenSymmetryProtectedTopological2013,guSymmetryprotectedTopologicalOrders2014,wangCompleteClassificationSymmetryprotected2018,wangConstructionClassificationSymmetryprotected2020,aasenCharacterizationClassificationFermionic2022,barkeshliClassification2+1DInvertible2022}. 
    \item From the viewpoint of field theory, an invertible quantum system is thought to give a character of the bordism group \cites{kapustinBosonicTopologicalInsulators2014,kapustinSymmetryProtectedTopological2014,kapustinFermionicSymmetryProtected2015}. This idea is made more sophisticated in terms of Anderson duality to the Freed--Hopkins ansatz \cites{freedShortrangeEntanglementInvertible2014,freedInvertiblePhasesMatter2020,debrayInvertiblePhasesMixed2021,freedReflectionPositivityInvertible2021,yamashitaDifferentialModelsAnderson2023,gradyDeformationClassesInvertible2023}. 
    \item Recently, some research groups have achieved results in more functional-analytic formulation of quantum spin systems. We refer to
\cites{ogataClassificationGappedHamiltonians2019,ogataMathbb_2Index2020,ogataClassificationPureStates2021,ogataValuedIndexSymmetryprotected2021,ogataIndexSymmetryProtected2021,bourneClassificationSymmetryProtected2021,ogataGeneralLiebSchultzMattisType2021,ogataInvariantSymmetryProtected2022,ogata2DFermionicSPT2023}, \cites{bachmannQuantizationConductanceGapped2018,bachmannManybodyIndexQuantum2020,bachmannClassificationGchargeThouless2023,jappensSPTIndicesEmerging2024,carvalhoClassificationSymmetryProtected2024} and \cites{kapustinHigherdimensionalGeneralizationsBerry2020,kapustinHigherdimensionalGeneralizationsThouless2020,kapustinClassificationInvertiblePhases2021,sopenkoIndexTwodimensionalSPT2021,kapustinLocalNoetherTheorem2022,artymowiczQuantizationHigherBerry2023,kapustinAnomalousSymmetriesQuantum2024,artymowiczMathematicalTheoryTopological2024}.
\item The topology of Hamiltonians coming from matrix product operators has also studied actively in, e.g.,  \cites{ohyamaGeneralizedThoulessPumps2022,ohyamaHigherStructuresMatrix2024,shiozakiHigherBerryCurvature2023,ohyamaHigherBerryPhase2024,ohyamaHigherBerryConnection2024} and \cites{qiChartingSpaceGround2023,wenFlowHigherBerry2023,sommerHigherBerryCurvature2024,sommerHigherBerryCurvature2024a,beaudryClassifyingSpacePhases2025}. 
\end{enumerate}
We also refer the readers to \cites{beaudryHomotopicalFoundationsParametrized2023,spiegelWeakContractibilitySpace2024}, which deals with the topology of quantum spin systems from the viewpoint of operator algebra but with an approach different from this paper.

\subsection{Functional-analytic realization of Kitaev's conjecture}\label{subsection:intro1}
The first and primary goal in the present paper is to implement Kitaev's conjecture in a rigorous mathematical formulation of lattice quantum systems, \emph{quantum spin systems}, described in terms of functional analysis and operator algebra. Standard references of this subject are \cites{bratteliOperatorAlgebrasQuantum1987,bratteliOperatorAlgebrasQuantum1997,evansQuantumSymmetriesOperator1998}. This follows the line of (3) in the above list.

\begin{introthm}[{\cref{cor:spectrum,cor:fermionic.equivariant.spectrum}}]\label{thm:main1}
    There is an $\Omega$-spectrum $\{ \IP_d, \kappa_d\}_{d \in \bZ_{\geq 0}}$ such that the group $\rIP^d(\sM)\coloneqq [\sM, \IP_d]$ is isomorphic to the set of smooth homotopy classes of smooth families of invertible gapped uniformly almost local (IG UAL) Hamiltonians on a lattice of $\bR^d$ (in the sense of \cref{defn:lattice.set}). 
    Similarly, there also exist the $\Omega$-spectra of fermionic and on-site $G$-symmetric IG UAL Hamiltonians $\fIP_d$ and $\IP_d^G$ for any compact Lie group $G$. 
\end{introthm}
The statement will be more precise in \cref{subsection:sheaf.lattice}. 
As consequences of \cref{thm:main1}, we can freely apply elementary computational techniques in homotopy theory to determine the topology of the spaces $\IP_d$, $\fIP_d$, and $\IP_d^G$. 
This reproduces many known results without a direct and complicated analysis, as summarized in \cref{section:homotopy}. 

The proof of \cref{thm:main1} basically follows the line of Kitaev's original argument and its reinterpretations given in \cites{xiongMinimalistApproachClassification2018,gaiottoSymmetryProtectedTopological2019}. We construct the \emph{Kitaev pump} map $\kappa_d \colon \IP_d \to \Omega \IP_{d+1}$ in \cref{subsection:Kitaev}, and the \emph{adiabatic interpolation} map $\vartheta_d \colon \Omega \IP_{d+1} \to \IP_d$ in \cref{subsection:adiabatic.interpolation}. 
We then show that they are the weak homotopy inverse to each other.
However, some changes from the existing argument are necessary due to the subtleties of the analysis of quantum spin systems.

\medskip 

\paragraph{(1) Analysis of the gapped interpolation}
An essential step of Kitaev's argument is the construction of $\vartheta_d$, which assign to a smooth path of gapped Hamiltonians $\{ \sfH(t)\}_{t \in [0,1]}$ an invertible gapped defect interpolating $\sfH(0)$ and $\sfH(1)$. 
If one tries a simple construction, a gradual linear combination of varying Hamiltonians, one will have a problem in proving that the resulting interpolation has a spectral gap. 
Indeed, the analysis of quantum spin systems is much more delicate than the single-particle Hamiltonian in the free fermion theory, and the spectral gap is easily violated by small perturbations. 
It has become clear in recent years, e.g.\ by \cites{ogataValuedIndexSymmetryprotected2021,sopenkoIndexTwodimensionalSPT2021}, that this difficulty can be resolved by using Hastings' \emph{adiabatic theorem} or the \emph{automorphic equivalence}  \cites{hastingsLiebschultzmattisHigherDimensions2004,hastingsQuasiadiabaticContinuationQuantum2005,bachmannAutomorphicEquivalenceGapped2012,nachtergaeleQuasilocalityBoundsQuantum2019,moonAutomorphicEquivalenceGapped2020}. 
This theorem claims that, for a smooth path $\{ \sfH(t)\}_{t \in [0,1]}$, there is a smooth family of derivations $\sfG_{\sfH}(t)$ on the observable algebra whose (time-dependent) evolution carries the ground state of $\sfH(0)$ to that of $\sfH(1)$. 
Now, by carrying $\sfH(0)$ via the time evolution with respect to the truncation of this $\sfG_{\sfH}$ to the right half space, we get a new Hamiltonian that looks like $\sfH(0)$ at the left half space and $\sfH(1)$ at the right half space (\cref{thm:interpolation.loop}). 

In order to apply this technique, the class of Hamiltonians needs to be chosen carefully. 
It must be closed under a class of automorphisms, i.e., evolutions generated by derivations satisfying some locality assumption. Hence, it needs to contain infinite-range Hamiltonians. 
On the other hand, a Hamiltonian $\sfH$ needs to satisfy a strong locality assumption; $\sfH$ is represented by the infinite sum of local generators $\sfH_{\bm{x}}$ indexed by lattice points, each of which is required to decay faster than exponentials of fractional powers.
In \cref{subsection:quantum.spin.system}, we deal with this locality assumption in terms of the Fr\'{e}chet seminorms on the observable algebra following the line of  \cites{kapustinLocalNoetherTheorem2022,artymowiczQuantizationHigherBerry2023}. 

\medskip

\paragraph{(2) Difference of continuous and smooth homotopies} 
There is another requirement for applying the adiabatic theorem. 
Since the derivation $\sfG_{\sfH}$ is constructed only for smooth (differentiable) families of Hamiltonians, we must work in the smooth category dealing with smooth maps and their smooth homotopies. 
This requirement is incompatible with the framework of stable homotopy theory, which is built on the category of topological spaces. 
Indeed, it is not clear whether the space of gapped Hamiltonians is an infinite dimensional (Fr\'{e}chet) manifold, and even if it were known, it does not mean that any continuous homotopy can be approximated by a smooth one.
A framework of spaces that meets such a requirement is \emph{diffeological spaces}, or almost equivalently, \emph{smooth sets} or \emph{sheaves on the category $\Man$ of manifolds}. 

A diffeological space is a set $X$ equipped with a distinguished collection of smooth maps $C^\infty(\bD^n,X) \subset \Map(\bD^n,X)$, which satisfies the contravariant functoriality with respect to smooth maps and the gluing condition over open covers (a standard reference is \cite{iglesias-zemmourDiffeology2013}). 
These conditions are reworded to say that the assignment $\bD^n \mapsto C^\infty(\bD^n,M)$ is a concrete sheaf on the category of open disks.
Extending the domains of smooth maps to all manifolds and forgetting the underlying set $X$, one arrives at the definition of a sheaf on $\Man$. 
According to the theory of Madsen--Weiss \cite{madsenStableModuliSpace2007}, a sheaf $\sF$ on $\Man$ corresponds to a topological space $|\Sing \sF|$, the geometric realization of the smooth singular set, which remembers the smooth homotopy type of $\sF$. More precisely, for any manifold $\sM$, the continuous homotopy set $[\sM,|\Sing \sF|]$ is isomorphic to the smooth homotopy classes of smooth maps from $\sM$ to $\sF$. 
Since we also consider settings with on-site symmetry given by a compact Lie group, we need an equivariant version of this construction, which we discuss in \cref{section:Gsheaf}.
In \cref{section:sheaf.IP}, we define the space $\IP_d$ as the geometric realization of a sheaf $\sIP_d$ of smooth families of IG UAL Hamiltonians. 

\medskip

\paragraph{(3) Uniformity of the internal degrees of freedom}
While the homotopy $\vartheta_d \circ \kappa_d \simeq\id$ is given directly, the converse homotopy $\kappa_d \circ \vartheta_d \simeq\id$ is not as straightforward as it.
By comparing with the corresponding free fermion theory \cite{kubotaControlledTopologicalPhases2017}, we notice that this difficulty is essential and comes from the uniformity of the internal degrees of freedom. 
Indeed, the homotopy classes of $1$-dimensional complex free fermions, short-range bounded invertible self-adjoint operators on $\ell^2(\bZ,\bC^N)$ for some $N>0$, are classified by the K-theory of the uniform Roe algebra. 
This group is known to have uncountably many basis elements (\cite{spakulaKtheoryUniformRoe2008}*{Example 3.4}) coming from atomic Hamiltonians. 
By the Bott periodicity, Kitaev's conjecture breaks down in this formulation.  
On the other hand, as was observed by Higson--Roe--Yu \cite{higsonCoarseMayerVietorisPrinciple1993} in the context of coarse index theory, the conjecture is revived by relaxing the assumption of the Hilbert space of lattice wave functions so that the internal degrees of freedom $N$ depends on the lattice point $\bm{x}$ without any uniform bound (see e.g.\ \cite{higsonAnalyticHomology2000}*{Theorem 6.4.10}).
The heart of this version of Kitaev's conjecture is that the free fermion topological phase on the half Euclidean space $\bR^d_+=\bR^{d-1}\times \bR_{\geq 0}$ is trivial, and hence all atomic Hamiltonians are in the trivial phase. 
Its proof is given by a kind of Eilenberg swindle (i.e., an `$\infty + 1 =\infty$' type) argument.

When one tries to apply this idea to Kitaev's conjecture for quantum spin systems, one encounters an analytic problem: The local generator $\sfH_{\bm{x}}$ of the Hamiltonian appearing in the Eilenberg swindle argument becomes a sum of a large number of operators depending on the point, and hence their operator norms have no uniform upper bound. 
This causes a problem in the analysis of Hastings' adiabatic theorem.
In this paper, although it is a somewhat artificial assumption, we prepare a `virtual dimension' on the physical space and regard a uniformly discrete subset of $\bR^{d+l}$ for some $l\geq 0$ as a $d$-dimensional lattice if the projection $\pr_{\bR^d} \colon \Lambda \subset \mathbb{R}^{d+l} \to \mathbb{R}^d$ is linearly proper (\cref{defn:polynomially.proper}). 

\subsection{Computations of homotopy groups and SPT phases}
Once the $\Omega $-spectrum $\IP_d$ is defined, the next question is to determine its homotopy groups. 
The homotopy groups of $0$-th space $\IP_0$ is easily determined by the homotopy equivalence $\IP_0 \simeq\bC \bP^\infty $. 
The cases of fermionic and symmetry-protected versions are also determined similarly. 
As a consequence of \cref{thm:main1}, this determines the homotopy groups $\pi_{n+d}(\IP_d)$ of $d$-dimensional invertible spin systems, or equivalently, the stable homotopy groups $\pi_n^{\mathrm{st}}(\IP)$ for $n \geq 0$.
In particular, they vanish for $n \geq 3$, that is, the $\Omega$-spectrum $\IP_d$ is $2$-truncated.  
The remaining task is to determine the negative homotopy group $\pi_{-d}(\IP) \coloneqq \pi_0(\IP_d)$.

In \cref{subsection:homotopy.higher}, we prove the isomorphisms 
\[ 
    \pi_0(\IP_1) \cong 0, \quad \pi_0(\fIP_1) \cong \bZ/2.
\]
The isomorphisms of these kinds are already known in the prior works by Ogata \cite{ogataClassificationPureStates2021}, Kapustin--Sopenko--Yang \cite{kapustinClassificationInvertiblePhases2021}, Kapustin--Sopenko \cite{kapustinAnomalousSymmetriesQuantum2024}, and Carvalho--de Roeck--Jappens \cite{carvalhoClassificationSymmetryProtected2024} for the bosonic case, and Bourne--Ogata \cite{bourneClassificationSymmetryProtected2021} for the fermionic case.
Our proof, whose scope contains invertible phases on lattices that are $1$-dimensional in our relaxed sense, are essentially the same as the one of \cite{carvalhoClassificationSymmetryProtected2024}. 
In particular, it does not rely on Matsui's work \cite{matsuiBoundednessEntanglementEntropy2013} on the finiteness of the entanglement entropy of $1$-dimensional gapped Hamiltonians, which is bypassed by using the assumption of invertibility. 

On the other hand, our knowledge of the lower negative homotopy groups is restrictive. 
For example, $\pi_0(\IP_2) =\pi_{-2}(\IP)$ is conjectured to be isomorphic to $\bZ$, to be classified by the (as yet undefined) integer called the chiral central charge, and to be generated by Kitaev's $E_8$-state \cite{kitaevAnyonsExactlySolved2006}. A candidate of this $E_8$-phase is recently proposed by Sopenko \cite{sopenkoChiralTopologicallyOrdered2023}.

However, even with our limited knowledge, we can yield partial results for the problem of determining the groups of SPT phases. In our framework, it is possible to consider the following homomorphisms;
\[
    \mathrm{Ind}_G^d \coloneqq \pr_{EG}^* \colon \pi_0(\IP_d^G) =[\pt,\IP_d]^G \to [EG, \IP_d]^G \cong [BG, \IP_d]. 
\]
The group appearing on the right hand side has been proposed to give the classification of SPT phases. 
For example, this claim is referred to as the \emph{generalized cohomology hypothesis} in \cite{xiongMinimalistApproachClassification2018}.
This group is computable by methods of algebraic topology like the Atiyah--Hirzebruch spectral sequence. Moreover, it is related to the group cohomology in dimensions $d \leq 2$. 
Moreover, as is listed in \cref{exmp:list}, the map $\mathrm{Ind}_G^d$ recovers most topological invariants known so far. For this reason, we adopt the same terminology and refer to it as the \emph{SPT index} as well. 

According to a direct computation of $\pi_0(\IP^G_1)$ by the same argument as \cref{subsection:homotopy.higher}, the map $\mathrm{Ind}_G^1$ is an isomorphism. 
Even if $d >1$, certain information about the right hand side can be obtained. 
First, the two constructions of models, the lattice Dijkgraaf--Witten model and Araki's quasi-free quantization, show the split surjectivity of $\mathrm{Ind}_{G,\phi}^2$ in some cases. These models will be discussed in the subsequent papers \cites{kubotaStableHomotopyTheory2025a,kubotaStableHomotopyTheory2025b}. 
Second, if $G$ is a (connected) Lie group, then the existence of the Chern--Dold character isomorphism (\cref{subsection:Chern.Dold}) shows that the groups $[BG,\IP^d] \otimes \bQ$ contains a direct summand isomorphic to $\mathrm{H}^{d+2}(BG\,; \bQ)$. 
For example, when $G=\bT$, Kapustin--Sopenko \cite{kapustinLocalNoetherTheorem2022} studies the resulting invariant as an interacting model of the integer quantum Hall effect.

\subsection{Models of generalized homology theory}
The second theme of this paper is to advance the \emph{coarse homotopy theoretic} perspective of the invertible gapped topological phases found in the proof of \cref{thm:main1}, taking the coarse index theory as a guide.  
As a result, we provide a realization of another proposal by Kitaev \cites{kitaevClassificationShortrangeEntangled2013,kitaevHomotopytheoreticApproachSPT2015}, ``the invertible phases of quantum systems defined on $X$, viewed as a physical space, give a homology theory of $X$'', which has been accepted in the context of invertible field theory. We refer to \cite{freedInvertiblePhasesMatter2020}*{Ansatz 2.1}, \cite{debrayInvertiblePhasesMixed2021}*{Ansatz1.10}, and \cite{shiozakiGeneralizedHomologyAtiyah2023}.

Kitaev \cites{kitaevAnyonsExactlySolved2006,kitaevPeriodicTableTopological2009} initially noted that the K-theory classification of free fermion phases reflects the large-scale shape of the Euclidean space $\bR^d$. 
This idea was developed in the author's previous work \cite{kubotaControlledTopologicalPhases2017}. Here, the free fermion topological phases are formulated using the framework of coarse index theory, a branch of noncommutative geometry dealing with the large-scale nature of complete Riemannian manifolds and their Dirac type operators with application to geometry and differential topology. 
Standard references of this research field are  \cites{higsonAnalyticHomology2000,roeLecturesCoarseGeometry2003,willettHigherIndexTheory2020}. 
The coarse structure of a metric space $X$ is given by the family of subsets of $X \times X$ abstracting the uniform neighborhood $N_R(\Delta_X) $ of the diagonal set, which matches well with the concept of short-range property of Hamiltonians in the first place. 
This connection is a natural candidate for studying interacting invertible phases and has already appeared in various ways in the study of quantum spin systems. 
As explained in \cref{subsection:intro1} (3), the proof of the bulk-boundary correspondence given in \cite{kubotaControlledTopologicalPhases2017} by the coarse Mayer--Vietoris exact sequence is parallel to our proof of \cref{thm:main1}.
Moreover, recent progress on the higher order Berry phase by Kapustin--Spodyneiko \cites{kapustinHigherdimensionalGeneralizationsBerry2020,kapustinHigherdimensionalGeneralizationsThouless2020} and Artymowicz--Kapustin--Sopenko \cites{kapustinLocalNoetherTheorem2022,artymowiczQuantizationHigherBerry2023} also takes into account the idea of coarse geometry, using the pairing of the coarse ordinary homology and cohomology groups \cite{roeCoarseCohomologyIndex1993} in the process of extracting numerical invariants. Elokl--Jones \cite{eloklUniversalCoarseGeometry2024} also considers the coarse geometric viewpoint of the topological phases of lattice quantum systems.

As an application of this concept, we introduce a model of a generalized homology theory 
\begin{align*}
    \rIP_n(X) \coloneqq \pi_n^{\mathrm{st}}(X_+ \wedge \IP) =\colim_{k \to \infty} \pi_{n+k}(X_+ \wedge \IP_k ).
\end{align*}
associated with the IP spectrum, formulated as the homotopy groups of the set of `quantum spin systems defined on a metric space $X$.'
It is motivated by a well-known that coarse index theory provides an operator algebraic realization of the K-homology group of a topological space $X$ by the following general idea. 
Since coarse structure does not distinguish small and large finite distances, a loop in $X$ with finite diameter is ignored, and the topology of $X$ in the ordinal sense is violated. 
In particular, a compact metric space is identified with the point. 
In quantum spin systems, the sheaf $\sIP(X)$ of IG UAL Hamiltonians on $X$ can be defined even if $X$ is beyond the Euclidean space, as is considered in \cites{eloklUniversalCoarseGeometry2024,artymowiczMathematicalTheoryTopological2024}. However, two such sheaves, $\sIP(X)$ and $\sIP(Y)$, are indistinguishable if $X$ and $Y$ are equivalent as coarse spaces. 
For example, the invertible phases on $X=\bR^d$ and $Y=\bZ^d$ are precisely the same. 
However, an invariant in coarse geometry has an analogue that replaces the `finite' length scale with `infinitesimal', reflecting the space's local topology and eventually providing a generalized homology theory. 
In coarse index theory, there are two ways to realize the infinitesimal version of the Roe algebra: Higson--Roe's C*-algebra $D^*X/C^*X$ \cite{higsonAnalyticHomology2000} and Yu's localization C*-algebra $C^*_{\loc}X$ \cite{willettHigherIndexTheory2020}, both of which represent the K-homology group. 
Here, we follow the latter to define a model of the IP-homology group. 

The localization C*-algebra $C^*_{\loc}X$ consists of continuous paths of bounded operators $\{H_s\}_{s \in [1,\infty)}$ such that each $H_s$ has finite range $R_s$ that decays as $s \to \infty$ (more precisely, $C^*_{\loc}X$ is the closure of the set of such operators). 
This idea is immediately paraphrased to quantum spin systems, which leads us to the definition of \emph{invertible gapped (IG) localization flows} of UAL Hamiltonians in \cref{defn:localizing.path}. 
A similar idea also appears in \cite{kitaevTopologicalQuantumPhases2019}.
The concept of IG localization flow resembles to that of the renormalization group flow in quantum field theory, which is where the name localized flow comes from, but the existence of the scaling limit as $s \to \infty$ is not assumed in any sense. 
The smooth family of IG localization flows forms a sheaf $\sIP_{\loc}(X)$ on $\Man$. 
On the Euclidean space $\bR^d$, there is a smooth homotopy of an IG UAL Hamiltonian and its rescaled version, which shows that $\sIP_{\loc}(\bR^d)$ is weakly equivalent to $\sIP_d$ (\cref{thm:scaleable}). 

By taking the geometric realization of the singular set, we get the $\Omega$-spectrum $\IP_{\loc}(X)$.
By letting $\IP_{\loc,d}(X) \coloneqq \IP_{\loc}(\boldsymbol{\Sigma}^dX)$, where $\boldsymbol{\Sigma} X$ denotes the coarse geometric suspension $ X \times \bR$, the family of spaces $\{ \IP_{\loc,d}(X)\}_{d \in \bZ_{\geq 0}}$ forms an $\Omega$-spectrum as well as \cref{thm:main1}. 
Moreover, the following holds.
\begin{introthm}[{\cref{thm:IP.bivariant,thm:Atiyah.duality}}]\label{thm:main2}
    Let $X$ be a CW-complex embedded in a Euclidean space (\cref{defn:MCW}). 
    Then the assignment $\rIP_{\loc,n}(X) \coloneqq  \pi_{n}(\IP_{\loc }(X))$ forms a generalized homology functor, which is naturally isomorphic to the IP-homology theory $\rIP_n(X)$. 
\end{introthm}
As a consequence, the assignment
\[
    (X,\sY) \mapsto \rIP_{\loc,d}(X \midbar \sY) \coloneqq [\sY, \IP_{\loc,d}(X)] 
\]
forms a bivariant homology theory, a bifunctor that is homological with respect to the first variable and cohomological with respect to the second variable.

The proof of \cref{thm:main2} is given in the following way. 
First, we extend the definition of $\IP_{\loc,n}(X)$ to its relative version $\{ \IP_{\loc,n}(X,A)\}_n$ defined for a pair of metric spaces $(X,A)$. 
Then, we show that the assignment $(X,A ) \mapsto \pi_n(\IP_{\loc}(X,A))$ satisfies the axioms of generalized homology theory. 
Finally, the isomorphism $\rIP_{\loc,n}(X) \cong \rIP_n(X)$ is given by a standard Mayer--Vietoris argument by identifying the homology group $\rIP_n(X)$ with the cohomology group of the Spanier--Whitehead dual $\sD_X$ of $X$.

\subsection{Invertible phases with spatial symmetry and the assembly map}
The concept of IG localization flow on a general metric space has an application to the study of topological phases protected by a spatial symmetry given by a crystallographic group $\Gamma$, i.e., a discrete group of the rigid motion group $\bR^m \rtimes O(m)$ acting properly and cocompactly on $\bE_{\Gamma} \coloneqq \bR^m$. 
In \cref{section:assembly}, we give an appropriate formulation of the $\Gamma$-invariant topological phases of IG localization flows on a proper $\Gamma$-space $X$.
We then give a rigorous proof of another proposal that, stating that they are classified by the $\Gamma$-equivariant IP-homology group $\rIP_{\loc,0}^\Gamma(X)$, which is referred as the generalized cohomology hypothesis in \cite{xiongMinimalistApproachClassification2018} (see also \cite{freedInvertiblePhasesMatter2020}*{Ansatz 3.3} and \cite{debrayInvertiblePhasesMixed2021}*{Ansatz 0.1}).

For a technical reason that a $\Gamma$-manifold may not be a $\Gamma$-CW-complex for non-compact $\Gamma$, we do not formulate the crystallographic symmetry of invertible phases as $\Gamma$-equivariant geometric realizations of $\Gamma$-sheaves.
A crystallographic group $\Gamma$ has the finite index normal subgroup $N \coloneqq \Gamma \cap \bR^m$ consisting of translations. 
For a nice metric space $X$ with a proper $\Gamma$-action (\cref{defn:MCW.equivariant}), we consider the associated $\Gamma$-action on sheaves of $N$-invariant IG UAL Hamiltonians $\sIP_d^N(X)$ and IG localization flows $\sIP_{\loc,d}^N(X)$.
The proof of \cref{thm:main1} works to this equivariant setting, and we obtain the $\Gamma/N$-equivariant $\Omega$-spectra $\{ \IP_d^N(X)\}_d$ and $\{ \IP_{\loc,d}^N(X) \}_d$ (\cref{defn:spatial.equivariant.IP.cohomology}) whose associated equivariant cohomology functors
\[
    \rIP_{d}^\Gamma (X \midbar \sM) \coloneqq  [\sM,\IP_d^N(X)]^{\Gamma/N} , \quad \rIP_{\loc,d}^\Gamma (X \midbar \sM) \coloneqq [\sM,\IP_{\loc,d}^N(X)]^{\Gamma/N} 
\] 
remember the $\Gamma$-equivariant smooth homotopy sets of the $\Gamma$-sheaves $\sIP_d(X)$ and $\sIP_{\loc,d}(X)$ respectively.

Under this formulation, we compare two $\Gamma/N$-equivariant $\Omega$-spectra of IG localization flows from $X$; $\IP_{\loc}^N(X)$ and $\IP_{\loc}(X/N)$. 
In \cref{thm:covering}, we prove that they are equivariantly weakly equivalent, and hence, the associated equivariant bivariant homology theories are isomorphic as
\[
    \rIP_{\loc,n}^\Gamma(X \midbar \sY ) \cong \rIP_{\loc,n}^{\Gamma/N}(X/N \midbar \sY ).
\]
This isomorphism corresponds to a well-known fact in equivariant K-homology theory (for example, it is noted in \cite{baumClassifyingSpaceProper1994}*{(3.12)}). 
The latter groups are computable using the homological Mayer--Vietoris exact sequences or, more systematically, the Atiyah--Hirzebruch spectral sequence. 

The primary focus is the group $\rIP_0^\Gamma(\bE_{\Gamma} )\coloneqq \rIP_0^\Gamma(\bE_{\Gamma} \midbar \pt)$, the $\pi_0$-group of the set of $\Gamma$-invariant IG UAL Hamiltonians on the Euclidean space $\bE_\Gamma$. 
To achieve this, we compare it with another group  $\rIP_{\loc,0}^\Gamma(\bE_{\Gamma})$, which is better understood through the above isomorphism, using the map forgetting the information of localization flows at $s>1$. 
In the context of our guide, the free fermion theory, this map is known as the equivariant \emph{coarse assembly map} comparing the K-groups of the Roe and localization algebras.
When $\Gamma$ is trivial, it was conjectured that the coarse assembly map is an isomorphism for any uniformly contractible metric spaces, which is known as the coarse Baum-Connes conjecture \cites{higsonCoarseBaumConnesConjecture1995,yuCoarseBaumConnesConjecture1995}. Although there is a counterexample  (\cite{higsonCounterexamplesBaumConnesConjecture2002}), it is known to hold for a broad class of metric spaces that contains Euclidean spaces \cite{yuCoarseBaumConnesConjecture2000}. 
Among the existing proofs of this conjecture, the one of \cites{higsonCoarseBaumConnesConjecture1995,qiaoLocalizationAlgebraGuoliang2010} is directly paraphrased to IP cohomology theory, which is nothing but the weak equivalence $\IP_{\loc ,m} \simeq\IP_m$ that is already mentioned. 
When $\Gamma$ is non-trivial and $\Gamma \curvearrowright X$ is proper cocompact, the corresponding isomorphism conjecture is nothing but the Baum--Connes conjecture \cite{baumClassifyingSpaceProper1994}, which is known to hold for a broad class of groups that contains crystallographic groups \cite{higsonTheoryKKTheory2001}. 
A coarse geometric approach to this problem is the proof of the split injectivity that is addressed in \cite{roeIndexTheoryCoarse1996}*{Theorem 8.4}: The coarse Baum--Connes isomorphism implies the split injectivity of the Baum--Connes assembly map. 
This idea is immediately paraphrased to those of IP cohomology theory.

In conclusion, we obtain the following theorem.
\begin{introthm}[{\cref{thm:covering,thm:BCI}}]\label{thm:main3}
    Let $\Gamma$ be a crystallographic group. 
    Then the group $\rIP_n^\Gamma(\bE_{\Gamma})$ of homotopy classes of $\Gamma$-invariant IG UAL Hamiltonians on $\bE_{\Gamma}$ has a direct summand isomorphic to $\rIP_{\loc,n}^\Gamma(\bE_{\Gamma}) \cong \rIP_{\loc,n}^{\Gamma/N}(\bE_{\Gamma}/N)$. Moreover, if $\Gamma$ is torsion-free, then this direct summand is also isomorphic to the groups $\rIP_{\loc,n}(\bE_{\Gamma}/\Gamma) \cong \rIP_n(\bE_{\Gamma}/\Gamma)$.     
\end{introthm}
Since the quotient space $\bE_{\Gamma}/N$, an $m$-dimensional torus, is the classifying space $BN$, \cref{thm:main3} partially proves Xiong's generalized cohomology hypothesis.
Moreover, in the case that $\Gamma$ torsion-free, we have $\bE_{\Gamma}/\Gamma=B\Gamma$. 
That is, this means that $\rIP_{\loc,0}^\Gamma(\bE_{\Gamma})$ is isomorphic to the Borel equivariant IP-homology group of $\bE_{\Gamma}$, which partially proves the statement referred to as the crystalline equivalence principle in \cite{thorngrenGaugingSpatialSymmetries2018}.

More practically, \cref{thm:main3} not only reproves the result of Jappens \cite{jappensSPTIndicesEmerging2024} as a special case where $d\leq 2$ and $G= \bZ^2$, but also provides some new topological invariants for $2$-dimensional IG UAL Hamiltonians with crystallographic symmetries, e.g., \cref{cor:pg}. More generally, \cref{thm:main3} provides a mathematical basis of the computations of the crystalline-symmetric invertible phases via the Atiyah--Hirzebruch spectral sequence  \cites{shiozakiGeneralizedHomologyAtiyah2023,leeConnectionFreefermionInteracting2024,leeCrystallineequivalentTopologicalPhases2024} (\cref{cor:spectral.sequence}).

The idea of the coarse assembly map, or the `forget-control' map, goes back to Quinn \cite{quinnEndsMapsII1982}. In modern times, it shares a context with the Farrell--Jones conjecture in algebraic K-\ and L-theories \cite{farrellIsomorphismConjecture1993}.  
The concept of the assembly maps for general equivariant homology theories has been defined by Davis--L\"{u}ck in \cite{davisSpacesCategoryAssembly1998} so that it covers both the Baum--Connes and the Farrell--Jones assembly maps. 
Our equivariant coarse assembly map for IP cohomology theory is also identified with the Davis--L\"{u}ck assembly map under a canonical isomorphism of the domain groups. 
This comparison leads us to an optimistic conjecture that our assembly map for the equivariant IP-cohomology theory is also an isomorphism, as well as the successes of other assembly maps.

\subsection{Construction of the paper}
The paper is organized as follows.
\cref{section:analysis.spin} summarizes the analysis of the quantum spin system used in this paper.
\cref{section:sheaf.IP} introduces the mathematical formulation of the space of invertible gapped uniformly almost local Hamiltonians in terms of sheaves on $\Man$.
Under this terminology, \cref{section:Kitaev} proves of Kitaev's conjecture.
This main theorem and the results of the subsequent papers \cites{kubotaStableHomotopyTheory2025a,kubotaStableHomotopyTheory2025b} present the consequences for the computation of homotopy groups in \cref{section:homotopy}.
In \cref{section:localizing.path}, we give a coarse geometric model with localization flows to the generalized homology theory associated with the $\Omega$-spectrum constructed in \cref{section:Kitaev}.
In \cref{section:assembly}, we formulate a coarse assembly map for the space of IG UAL Hamiltonians with spatial symmetry and prove its split injectivity. 

\addtocontents{toc}{\protect\setcounter{tocdepth}{1}}%
\setcounter{tocdepth}{1}%

\subsection*{Notation}
    Throughout the paper, we use the following general notation. 
    \begin{itemize}
        \item For a (finite) set $X$, its cardinality denotes $\# X$ rather than $|X|$. 
        \item For a topological space $X$, we write $X_+$ for the disjoint union $X \sqcup \{ \pt \}$ regarded as a pointed topological space with the basepoint $\pt$. 
        \item We use Euler Script font, e.g., $\sM$, $\sF$, $\sX$, ..., for sheaves on $\Man$, or topological spaces that play the role of cohomological variables, i.e., parameter spaces of families of quantum systems. 
        We use italic font $X,Y,\dots$ for metric spaces on which a quantum system is placed, whose point is labeled by a bold small character, e.g., $\bm{x}, \bm{y} \in X$.
        \item For a metric space $X$, $\bm{x}\in X$ and $r >0$, the open ball denotes $B_r(\bm{x}) \coloneqq \{ \bm{y} \in X \mid \rmd(\bm{x},\bm{y}) <r\}$. For subspaces $Y,Z \subset X$, their distance $\mathop{\mathrm{dist}}(Y,Z)$ is the infimum of the distances of points $\bm{y} \in Y$ and $\bm{z} \in Z$. We write $N_r(Y)\coloneqq \{ \bm{y} \in X \mid \mathop{\mathrm{dist}}(\bm{y},Y) <r\} $ denotes its $r$-neighborhood.
        \item For any direct product $X_1 \times \cdots \times X_k $ of sets, topological spaces, or manifolds, we write $\pr_{X_i}$ for the projection onto the $X_i$ component, i.e., $\pr_{X_i}(x_1,\cdots,x_k)=x_i$ for $(x_1,\cdots,x_k) \in X_1 \times \cdots \times X_k$.  
        \item For a direct product $X \times X$ or a tensor product $\cA \otimes \cA$ of two copies of an object, we write $\flip$ for the flip automorphism acting on them. 
        \item We use the standard multi-index notation for differential: For $I = (i_1,\cdots,i_n) \in \bZ_{\geq 0}^n$, we write $|I| \coloneqq \sum_j i_j$ and $\partial^I = \partial_{x_1}^{i_1} \circ \cdots \circ \partial_{x_n}^{i_n}$. 
        \item When a subscript is added to a set, e.g., $[0,1]_s$ or $\bR_t$, it shall mean that an element of the set $[0,1]$ or $\bR$ is indicated by the parameter $s$ or $t$. 
        \item If $X$ is a Hilbert space, Banach space, or a Fr\'{e}chet space, the algebra of bounded operators on it denotes $\cB(X)$. 
        \item We write $1_A$, or just $1$, for the unit element of a C*-algebra $A$. 
    \end{itemize}

\subsection*{Acknowledgment}
The author is grateful to Ken Shiozaki for informing him of the problem and numerous variable discussions. 
He would also like to express his gratitude to Yoshiko Ogata for her comments and advice on the draft of this paper.
He also thanks Yosuke Morita and Shuhei Ohyama for helpful discussions.
The author is supported by RIKEN iTHEMS and JSPS KAKENHI Grant Numbers 22K13910, JPMJCR19T2. 
\addtocontents{toc}{\protect\setcounter{tocdepth}{2}}%
\setcounter{tocdepth}{2}%

\section{Analysis of quantum spin systems}\label{section:analysis.spin}
We begin by outlining the functional-analytic setup of quantum spin systems and the analytical tools used in this paper. 
We follow the literature on operator-algebraic studies of quantum spin systems \cites{bratteliOperatorAlgebrasQuantum1997,moonAutomorphicEquivalenceGapped2020,nachtergaeleQuasilocalityBoundsQuantum2019}, and also the framework recently developed by Kapustin, Sopenko, and Artymowicz \cites{kapustinLocalNoetherTheorem2022,artymowiczQuantizationHigherBerry2023}. 
Thereafter, in \cref{subsection:adiabatic.interpolation}, we give a construction of a new gapped Hamiltonian $\vartheta \sfH$ that spatially interpolates a given Hamiltonian $\sfH$ and the trivial Hamiltonian $\sfh$, obtained from a smooth homotopy $\sfH \simeq\sfh$. 
This construction serves as the analytic ingredient of our proof of Kitaev's conjecture in \cref{section:Kitaev}.

\subsection{Coarse geometry of a lattice of the Euclidean space}\label{subsection:coarse.lattice}
We begin by formulating the class of lattices on which quantum spin systems are placed. To realize Kitaev's conjecture, we must deal with a slightly broader class of discrete metric spaces than $\bZ^d$ or Delone sets in $\bR^d$.

\begin{defn}
Let $r>0$ and $N \in \bN$. We say that a map $\iota \colon \Lambda \to  \bR^{l_{\Lambda}}$ is \emph{$(r,N)$-weakly uniformly discrete} if $ \# \iota^{-1}(B_r(\bm{x})) \leq N$ for any $\bm{x} \in \bR^{l_{\Lambda}}$, and is weakly uniformly discrete if it is $(r,N)$-weakly uniformly discrete for some $(r,N)$. 
\end{defn}
We often abuse $\iota$ and say that $\Lambda $ is weakly uniformly discrete in $\bR^{l_{\Lambda}}$, and consider it as a `subset with multiplicity' of $\bR^{l_{\Lambda}}$. 
If $\Lambda $ is $(r,1)$-weakly uniformly discrete, then $\iota$ is injective and the image $\iota(\Lambda)$ is what is called an $r$-uniformly discrete subset of $\bR^{l_{\Lambda}}$. If it is also $R$-relatively dense (i.e., $ N_R(\iota(\Lambda)) =\bR^{l_{\Lambda}}$) for some $R>0$, then $\iota(\Lambda)$ is what is called a Delone set. 
We will deal with weakly uniformly discrete maps to $\bR^{l_{\Lambda}}$ but may not be relatively dense, where the superscript $l_{\Lambda}$ will not mean the dimension of the physical space (as we will discuss in \cref{section:sheaf.IP}). 

A weakly uniformly discrete map $\iota \colon \Lambda \to \bR^{l_{\Lambda}}$ imposes to $\Lambda$ the pseudo-metric $\rmd(\bm{x},\bm{y} ) \coloneqq \rmd_{\ell^2}(\iota(\bm{x}),\iota(\bm{y}))$, which is proper in the sense that any open ball has finite number of points. 
The analysis of quantum spin systems is controlled by this $\rmd$, or more essentially, its quasi-isometry class. 
In other words, the class of pseudo-metric spaces we deal with is quasi-isometrically embedded into $\bR^{l_{\Lambda}}$. 
We remind some terminology in metric geometry: A map $F$ between proper countable pseudo-metric spaces $(\Lambda,\rmd)$, $(\Lambda',\rmd')$ is a large-scale Lipschitz map if there are constants $A >0$ and $B>0$ such that $\rmd'(F(\bm{x}),F(\bm{y})) \leq A \cdot \rmd(\bm{x},\bm{y}) +B$ for any $\bm{x}, \bm{y}\in \Lambda$, and is a large-scale bi-Lipschitz map (or a quasi-isometric embedding) if it further satisfies $A^{-1}\rmd(\bm{x},\bm{y}) - B \leq \rmd'(F(\bm{x}), F(\bm{y}))$. 
A large-scale bi-Lipschitz map is a quasi-isometry if it has a relatively dense image.

The pseudo-metric space coming from a weakly uniformly discrete map $\iota \colon \Lambda \to \bR^{l_{\Lambda}}$ has the following two growth properties. 
\subsubsection*{(1) Polynomial growth}
First, such $\Lambda$ has \emph{polynomial growth}, i.e., there are positive constants $\kappa_{\Lambda}>0$ such that the growth rate function $Q_{\Lambda}(r) = \sup_{\bm{x} \in \Lambda} \#\iota^{-1}(B_{r}(\bm{x}))$ of $\Lambda$ is bounded above by $\kappa_{\Lambda} \cdot (1+r)^{l_{\Lambda}}$.

We remind the following elementary estimate concerned with polynomial growth property, which is also noted in \cite{nachtergaeleQuasilocalityBoundsQuantum2019}*{Subsections 8.1,8.2}, since it is frequently used throughout the paper. 
\begin{rmk}\label{lem:sum.exponential}
    Let $S$ be a countable set, let $l \geq 0$ and $r_0 \geq 1 $. 
    Let $\varphi \colon S \to \bR_{\geq 1}$ satisfies that  $\# \varphi^{-1}([1,r)) \leq \kappa \cdot r^l$ for any $r \geq r_0$. 
    Then, for any $\nu > l+1$, we have
    \begin{align*}
         \sum_{\bm{x} \in S , \  r_0 \leq \varphi (\bm{x})} \varphi(\bm{x})^{-\nu} \leq 
        2^{l+1}\kappa  \cdot r_0^{l+1-\nu}.
    \end{align*}
Indeed, since $\xi(t) = t^{l-\nu}$ is monotonously decreasing on $t \in [1,\infty) $, we have
\begin{align*}
     \sum_{r_0\leq f(\bm{x})} \varphi(\bm{x})^{-\nu} \leq {}&{} \sum_{n = 0}^\infty (r_0+n)^{-\nu}  \# \varphi^{-1}( [r_0+n ,  r_0+n +1 )) \\
     \leq {}&{}  \sum_{n =0}^\infty \kappa\frac{(r_0+n+1)^l}{(r_0+n)^l} \cdot (r_0+n)^{l-\nu}    
    \leq 2^l \kappa \cdot \bigg(r_0^{l-\nu} + \int_{r_0}^{\infty} t^{l-\nu} dt\bigg)\\
    \leq {}&{}  2^{l}\kappa (r_0^{-1} + (\nu-l)^{-1}) \cdot r_0^{l+1-\nu } \leq 2^{l+1}\kappa \cdot r_0^{l+1-\nu }. 
\end{align*}
\end{rmk}

\begin{lem}\label{prp:Ffunction}
    Let $\Lambda$ be a proper discrete metric space with the growth rate $Q_{\Lambda}(r) \leq \kappa_{\Lambda} \cdot (1+r)^{l_\Lambda}$. Then, for any $0\leq \mu<1$ and $\nu > l_{\Lambda}+1 $, the function 
    \begin{align*}
    f_{\nu , \mu}(r) \coloneqq (1+r)^{-\nu}\exp (-r^\mu )
    \end{align*}
    is an $F$-function on $\Lambda$ in the sense of \cite{nachtergaeleQuasilocalityBoundsQuantum2019}*{Section 3.1.1}. More precisely, there is $C_{\nu} > 0$ such that 
    \begin{align*} 
    \sum_{\bm{z} \in \Lambda} f_{\nu,\mu}(\rmd(\bm{x},\bm{z})) \leq 2^{l_{\Lambda} +1}\kappa_{\Lambda}, \quad 
    \sum_{\bm{z}\in \Lambda} f_{\nu,\mu}(\rmd(\bm{x},\bm{z})) \cdot f_{\nu,\mu}(\rmd(\bm{z},\bm{y})) \leq C_{\nu} \cdot f_{\nu,\mu }(\rmd(\bm{x},\bm{y})),
    \end{align*}
    for any $\bm{x}, \bm{y} \in \Lambda$.
\end{lem}
\begin{proof}
    Note that the polynomial growth assumption is rewritten to $\# \varphi^{-1}([1,r+1)) \leq \kappa_{\Lambda}\cdot (1+r)^{l_{\Lambda}}$ via $\varphi (\bm{y}) = 1+\rmd(\bm{x},\bm{y})$. 
    This, for any $\bm{x} \in \Lambda$, by \cref{lem:sum.exponential} we have 
    \begin{align*} 
    \sum_{\bm{y} \in \Lambda} f_{\nu,\mu}(\rmd(\bm{x},\bm{y}))  \leq \sum_{\bm{y} \in \Lambda} (1+\rmd(\bm{x},\bm{y}))^{-\nu} \leq 2^{l_\Lambda +1}\kappa_{\Lambda} <\infty.
    \end{align*}
    For $\bm{x}, \bm{y} \in \Lambda$, let $r \coloneqq \rmd(\bm{x},\bm{y})/2$. 
    For any $\bm{z} \in \Lambda \setminus N_{r}(\{ \bm{x} , \bm{y}\})$, we have $2^{-1}\rmd(\bm{z},\bm{y})  \leq \rmd(\bm{x}, \bm{z} ) \leq 2\rmd(\bm{z},\bm{y})$, and hence $6(1+\rmd(\bm{x},\bm{z}))(1+\rmd(\bm{z},\bm{y})) \geq  (1+\rmd(\bm{x},\bm{z}) + \rmd(\bm{z},\bm{y}))^2 $. 
    Moreover, by the concavity of $g(x) = x^\mu$ for $0<\mu<1$, we have $g(t) + g(1-t) \geq g(0)+g(1)=1$, and hence $\rmd(\bm{x},\bm{z})^\mu + \rmd(\bm{z},\bm{y})^\mu \geq (\rmd(\bm{x},\bm{z}) + \rmd(\bm{z},\bm{y}))^\mu \geq \rmd(\bm{x},\bm{y})^\mu$ for any $\bm{z} \in \Lambda$. 
    Thus 
    \begin{align*}
    {}&{}\sum_{\bm{z}\in \Lambda} f_{\nu,\mu}(\rmd(\bm{x},\bm{z})) \cdot f_{\nu,\mu}(\rmd(\bm{z},\bm{y}))\\
    = {}&{} 
    \bigg( \sum_{\bm{z} \in N_{r}(\{ \bm{x}, \bm{y}\} )} + \sum _{\bm{z} \in \Lambda \setminus N_{r}(\{ \bm{x},\bm{y}\})} \bigg) (1+\rmd(\bm{x},\bm{z}))^{-\nu} \cdot (1+\rmd(\bm{z},\bm{y}))^{-\nu} \cdot e^{- (\rmd(\bm{x},\bm{z})^\mu + \rmd(\bm{z},\bm{y})^\mu)}\\
    \leq {}&{} 2 \cdot \kappa_{\Lambda}2^{l_{\Lambda} +1} \cdot (1+r)^{-\nu}e^{-\rmd(\bm{x},\bm{y})^\mu} + 
    \kappa_{\Lambda}2^{l_{\Lambda} +1} \cdot 6^\nu  \cdot (1+2r)^{l_{\Lambda}+1-2\nu}e^{-\rmd(\bm{x},\bm{y})^\mu} \\
    \leq {}&{} \kappa_{\Lambda}2^{l_{\Lambda} +1} (2^{\nu +1} + 6^\nu) \cdot  f_{\nu,\mu}(\rmd(\bm{x},\bm{y})).
    \end{align*}
    For the first sum over $B_r(\bm{x})\subset N_r(\{\bm{x},\bm{y}\})$, we apply the inequalities $(1+\rmd(\bm{z},\bm{y}))^{-\nu} \leq (1+r)^{-\nu}$ and $ e^{- (\rmd(\bm{x},\bm{z})^\mu + \rmd(\bm{z},\bm{y})^\mu)} \leq e^{-\rmd(\bm{x},\bm{y})^\mu}$, and then use \cref{lem:sum.exponential} to handle the sum of $(1+\rmd(\bm{x},\bm{z}))^{-\nu}$.  For the second sum, apply \cref{lem:sum.exponential} for the function $\varphi(\bm{z})=1+\rmd(\bm{x},\bm{z}) + \rmd(\bm{z},\bm{y})$. 
\end{proof}

\subsubsection*{(2) Bricks}\label{subsubsection:brick}
Another important property of uniformly discrete $\Lambda$ in $\bR^{l_\Lambda}$ is that it is equipped with a poset of finite subsets called \emph{bricks} in \cite{kapustinLocalNoetherTheorem2022}*{Appendix C}. 
\begin{defn}\label{defn:brick}
Let $\bB=\bB_l$ denote the poset of subsets of $\bR^l$ the form $\rho \coloneqq \prod_{i=1}^l [n_i,m_i)$, where $n_i,m_i \in \bZ$. 
The \emph{brick} of $\Lambda$ is the ordered map $B \colon \bB_{l_{\Lambda}} \to 2^\Lambda$ given by $B_\rho \coloneqq \iota^{-1}(\rho)$. 
\end{defn}
Since the poset $\bB$ is locally finite, it has the M\"obius function $\mu_{\bB} \colon \bB \times \bB \to \bZ$, which is the unique map characterized by M\"obius inversion formula $f(\rho) = \sum_{\sigma \in \bB}\mu_{\bB} (\sigma, \rho)g(\sigma)$, where $g(\sigma) \coloneqq \sum_{\tau \leq \sigma}f(\tau ) $, for any map $f$ from $\bB$ to an arbitrary abelian group (cf.\ \cite{jacobsonBasicAlgebra1985}*{Section 8.6}). As is used in \cite{kapustinLocalNoetherTheorem2022}*{Proposition C.1}, the M\"{o}bius function of $\bB_{l_{\Lambda}}$ satisfies 
\begin{align}
    \sup_\rho \sum_{\sigma \leq \rho}|\mu_{\bB}(\sigma , \rho)| \leq 4^{l_{\Lambda}}.\label{eqn:brick.Mobius}
\end{align}

Moreover, the poset $\bB$ has some polynomial growth properties. 
For $\rho \in \bB$ and $\bm{x} \in \Lambda$, let
\begin{align}
    R(\rho,\bm{x}) \coloneqq \sup \{ r \geq 0 \mid \text{ there is $\sigma \lneq \rho \in \bB$ such that $B_r(\bm{x}) \cap B_{\rho} \subset B_\sigma$} \} \in [0,\infty]. \label{defn:R(rho.x)}
\end{align}
Note that $R(\rho,\bm{x}) = \infty$ if and only if there is $\sigma \lneq  \rho$ such that $B_\sigma = B_\rho$. This observation implies that, if $\rho=\prod_i [n_i,m_i)$ satisfies $R(\rho,\bm{x}) <\infty$, then its thickened faces $\partial_{i,+} \rho \coloneqq \{ \bm{x} \in \rho \mid x_i \leq n_i+1 \}$ and $\partial_{i,-}\rho  \coloneqq \{ \bm{x} \in \rho \mid x_i \geq m_i-1 \}$ correspond to non-empty subsets $B_{\partial_{i,\pm }\rho} \subset \Lambda$. By this, we have 
\begin{align}
    R(\rho,\bm{x}) = \sup_{\bm{y} \in \rho} \rmd_{\ell^1}(\iota(\bm{x}),\bm{y}) -1 \geq \frac{1}{\sqrt{l_\Lambda}} \sup_{\bm{y} \in \rho} \rmd_{\ell^2}(\iota(\bm{x}),\bm{y}) -1. 
    \label{eqn:R(rho.x)}
\end{align}
In particular, the infimum $\inf _{\bm{y} \in \Lambda } R(\rho,\bm{y})$ is bounded below by $\diam (\rho)/(2 \sqrt{l_{\Lambda}}) -1$. 
This bound immediately implies the following growth bounds. 

\begin{lem}\label{lem:R(rho.x).growth}
    There is $\kappa_{\bB} >0$ such that, for any $\bm{x} \in \Lambda$ and $\rho \in \bB$, we have
\begin{enumerate}
    \item $\# \{ \rho \in \bB \mid R(\rho,\bm{x}) \leq r \} \leq  \kappa_{\bB}  (1+r)^{2l_{\Lambda}}$, 
    \item $\# \{ \bm{x} \in \Lambda \mid R(\rho ,\bm{x}) \leq r \} \leq \kappa_{\bB}  (1+r)^{l_{\Lambda}}$,
    \item $\# \{ \rho \in \bB \mid \text{$B_\rho \ni \bm{x}$, $\inf_{\bm{y} \in \Lambda} R(\rho,\bm{y}) \leq r$} \}  \leq  \kappa_{\bB} (1+r)^{2l_{\Lambda }}$.
\end{enumerate}
\end{lem}
We remark that the above inequalities are still valid up to the choice of the constant $\kappa_{\bB}$ even if the metric of $\Lambda$ is replaced with another that is quasi-isometric to $\rmd$. 

\subsection{Quantum spin system}\label{subsection:quantum.spin.system}
The formulation of the quantum spin system in this paper basically follows the one of \cites{kapustinLocalNoetherTheorem2022,artymowiczQuantizationHigherBerry2023}. 
A local generator of a lattice Hamiltonian is indexed by lattice points $\bm{x} \in \Lambda$. 
However, there are some minor and technical differences. 
First, the Fr\'{e}chet algebra of almost local operators defined by using different seminorms, which is more compatible with the analysis of \cites{nachtergaeleQuasilocalityBoundsQuantum2019,moonAutomorphicEquivalenceGapped2020}.
Indeed, as is observed in \cref{prp:brick}, our formulation is equivalent to the common one in the literature  \cites{bratteliOperatorAlgebrasQuantum1997,nachtergaeleQuasilocalityBoundsQuantum2019,moonAutomorphicEquivalenceGapped2020}, which is described in terms of interactions indexed by finite subsets of $\Lambda$. 
Second, quantum spin systems are considered on a general weakly uniformly discrete metric space in $\bR^d$ that may not be Delone. 

This subsection is dedicated to verifying that the results of \cite{kapustinLocalNoetherTheorem2022} and \cites{nachtergaeleQuasilocalityBoundsQuantum2019,moonAutomorphicEquivalenceGapped2020} apply to our modified setting. It also serves as a survey to make the paper self-contained.

Let $\Lambda$ be weakly uniformly discrete in the Euclidean space $\bR^{l_{\Lambda}}$. 
Let us place a matrix algebra $\cA_{\bm{x}} =M_{n_{\bm{x}}}(\bC) =\cB(\bC^{n_{\bm{x}}})$ on each $\bm{x} \in {\Lambda}$. Here, we have made no assumptions on the uniform boundedness of the internal degrees of freedom $n_{\bm{x}} \in \bZ_{\geq 2}$. 
The observable algebras are defined by
\begin{align*} 
\cA^{\alg}_{{\Lambda}} \coloneqq \bigotimes_{\bm{x} \in {\Lambda}} \cA_{\bm{x}}, \quad \cA_{{\Lambda}}\coloneqq \overline{\cA^{\alg}_{{\Lambda}}},
\end{align*}
where the symbol $\bigotimes $ stands for the algebraic infinite tensor product (that is, the union of finite tensor products), and the completion is taken with respect to the operator $C^*$-norm. 
For a subset $Z \subset \Lambda$, let $\cA^{\alg}_{Z}$ and $\cA_{Z}$ denote the subalgebra $\bigotimes _{\bm{x} \in Z} \cA_{\bm{x}}$ of $\cA_{{\Lambda}}^{\alg}$ and its closure respectively. 
Let 
\begin{align*}
    \Pi_Z \coloneqq \id_{\cA_{Z}} \otimes \mathrm{tr}_{\cA_{Z^c}} \colon \cA_{{\Lambda}} \to \cA_Z
\end{align*}
denote the conditional expectation given by tracing out the tensor complement $\cA_{Z^c}$ of $\cA_Z$ in $\cA_{\Lambda}$ via the unique trace $\mathrm{tr}_{\cA_{Z^c}}$ normalized so that $\mathrm{tr}_{\cA_{Z^c}}(1)=1$. 
For $\bm{x} \in \Lambda$ and $r>0$, we use the abbreviated notation $\Pi_{\bm{x},r} \coloneqq \Pi_{B_r(\bm{x})}$. 

\begin{rmk}\label{rmk:CE.average}
    For a subset $Z \subset \Lambda$, let $\cU(\cA_{Z})$ denote the unitary group of the matrix algebra $\cA_Z$. If $\# Z<\infty$, then $\cU(\cA_{Z})$ is a compact group and hence has the unique normalized Haar measure.  
    For $a \in \cA_{\Lambda}^{\al}$, the conditional expectation $\Pi_Y$ is given by the integral
     \[
     \Pi_Y(a) = \lim_{Z \to Y^c}\int_{u \in \cU(\cA_Z)} uau^* \ du.
     \]
     This description of $\Pi_Y$ concludes that 
     \begin{align*}
         \| a - \Pi_{Y}(a) \| \leq 
         \lim_{Z \to Y^c}\int_{u \in \cU(\cA_Z)} \| uau^*  -a \| \ du \leq \sup_{u \in \cU(\cA_Y^c)}  \| uau^*-a\|. 
     \end{align*}
\end{rmk}

\begin{defn}\label{defn:Ffunction}
We define the class of functions on $\bR_{\geq 0}$ as
\begin{align*}
    \cF_0 \coloneqq {}&{}\{ f_{\mu} (r) \coloneqq \exp(-r^\mu ) \mid 0 \leq  \mu < 1 \},\\
    \cF \coloneqq {}&{} \{ f \colon \bR_{\geq 0} \to \bR_{\geq 0} \mid \text{$f$ is monotonously decreasing, $f(r) \geq C \cdot f_{\mu } (r)$ for some $0 \leq \mu <1$}\}.
\end{align*}
We call $a \in \cA_{{\Lambda}}$ is \emph{almost local} at $\bm{x} \in \Lambda$ if
    \begin{align*}
        \| a \|_{\bm{x}, f } \coloneqq  \max \Big\{ \| a\|  , \sup_{r \in \bR_{>0} } \big( f(r)^{-1} \cdot \| a - \Pi_{\bm{x},r}(a) \| \big) \Big\} < \infty
    \end{align*}
for any $ f \in \cF$.  
The subspace of almost local operators $\cA_{\Lambda}^\al \subset \cA_{{\Lambda}}$ is imposed the Fr\'echet topology by the (semi-)norms 
 $\| \blank \|_{\bm{x},f}$ for any $f \in \cF$. We use the abbreviated notation $\| \blank \|_{\bm{x},\mu } \coloneqq \| \blank \|_{\bm{x},f_{\mu}}$. 
\end{defn}

\begin{rmk}\label{rmk:seminorms}
    The subset 
\begin{align*} 
    \cF_1 = \{ f_{\nu,\mu} (r) = (1+r)^{-\nu }\exp (-r^\mu) \mid \nu >l_{\Lambda} +1,0 \leq \mu <1 \} \subset \cF,
\end{align*}
    determines the same topology as $\cF$ and $\cF_0$, i.e., an operator $a \in \cA_{\Lambda}$ is almost local if and only if $\|a\|_{\bm{x},f_{\nu,\mu}} <\infty$ for any $\nu,\mu$.
    In the same way as \cref{defn:Ffunction}, we use the abbreviated notation $ \| \blank \|_{\bm{x},\nu,\mu } \coloneqq \| \blank \|_{\bm{x},f_{\nu,\mu}}$. 
    By \cref{prp:Ffunction}, $\cF_1$ consists of $F$-functions.
\end{rmk}

\begin{rmk}\label{rmk:polynomial.coarse.invariance}
    For $0 \leq \mu < 1$, we introduce the notation $\mu_0 =\mu$ and $\mu_k \coloneqq (1+\mu_{k-1})/2$ inductively, to ensure that $\{\mu_k\}_{k \in \bN}$ is an increasing sequence in $[0,1)$. 
    For $A >1$, $\nu >l_{\Lambda}+1$ and $0\leq \mu<1$, we have
    \begin{align}
        c_{\nu,\mu,A} \coloneqq \sup_{r>0} f_{\nu,\mu_1}(r) \cdot f_{\nu,\mu}(Ar )^{-1} < \infty, \label{eqn:const.small.c}
    \end{align}
    that is, $f(A^{-1}r) \in \cF$ for any $f \in \cF$. By $(1+Ar+B)(1+r)^{-1} \leq \max \{ A, 1+B\}$ and the submultiplicativity of $f_{\mu}$ observed in the proof of \cref{prp:Ffunction}, we have
    \begin{align}
    \begin{split}
        {}&{}\sup_{r>B} f_{\nu,\mu_1}(A^{-1}r-A^{-1}B) \cdot f_{\nu,\mu}(r)^{-1} = \sup_{r>0} f_{\nu,\mu_1}(r) \cdot f_{\nu,\mu}(Ar+B )^{-1}  \\
        \leq {}&{} \sup_{r>0} f_{\mu_1}(r) \cdot f_{\mu}(Ar)^{-1}f_{\mu}(B)^{-1} \cdot (1+Ar+B)^{\nu}(1+r)^{-\nu}\\
        \leq {}&{} c_{0,\mu,A} \cdot f_{\mu}(B)^{-1} \cdot \max \{ A^\nu,(1+B)^\nu\}  < \infty, 
    \end{split}\label{eqn:large.scale.Lipschitz.F}
    \end{align}

    By \eqref{eqn:large.scale.Lipschitz.F}, the almost local subalgebra satisfies the following functoriality with respect to a large-scale Lipschitz map: 
    For an injective large-scale Lipschitz map $F \colon \Lambda \to \Lambda '$ and a family of unital $\ast$-homomorphisms $\phi_{\bm{x}} \colon \cA_{\bm{x}} \to \cA_{F(\bm{x})}$, the $\ast$-homomorphism $\bigotimes_{\bm{x}} \phi_{\bm{x}} \colon \cA_{\Lambda } \to \cA_{\Lambda '}$ sends $\cA_{\Lambda}^{\rm al}$ to $\cA_{\Lambda'}^{\al}$.
    In particular, for a weakly uniformly discrete $\Lambda $ in $\bR^{l_{\Lambda}}$, the Frech\'{e}t topology is unchanged if the metric on $\bR^l$ is replaced with the rescaled metric $A \cdot\rmd(\bm{x},\bm{y})$ for some $A \in \bR_{>0}$, or the $\ell^p$-metric for any $p \in [1,\infty]$. 
\end{rmk}

The space $\cA_{\Lambda}^{\al}$ forms a Fr\'echet $\ast$-algebra. Indeed, we have the inequality
\begin{align}
\begin{split}
    &f(r)^{-1} \| ab - \Pi_{\bm{x},r}(ab)\| \\
    \leq {} &f(r)^{-1}   (\|(\id - \Pi_{\bm{x},r})(a \cdot  (b-\Pi_{\bm{x},r}(b)))\| + \| (a-\Pi_{\bm{x},r}(a) ) \cdot \Pi_{\bm{x},r}(b)  \|) \\
    \leq{} & 2\| a \| \cdot \| b\|_{\bm{x},f}+\|a\|_{\bm{x},f} \cdot \| b\| \leq 3\|a\|_{\bm{x},f} \cdot \| b\|_{\bm{x},f}.
\end{split}\label{eqn:Frechet.algebra}
\end{align}
Here, we use the properties of the conditional expectation; $\Pi_{\bm{x},r}(a\Pi_{\bm{x},r}(b)) = \Pi_{\bm{x},r}(a) \Pi_{\bm{x},r}(b)$ and $\| \Pi_{\bm{x},r}(b)\| \leq \| b\| $. 
We remark that, for $a \in \cA_{\Lambda}$, the boundedness of $\| a \|_{\bm{x}, f}$ is independent of the choice of $\bm{x} \in \Lambda $. 

\begin{rmk}\label{rmk:almost.local.unitary}
    A unitary $u \in \cA_{\Lambda}$ is almost local if and only if there is a sequence of unitaries $u_r \in \cA_{B_r(\bm{x})}$ for $r \in \bZ_{\geq 1}$ such that $\sup_{r} f(r)^{-1}\| u_r -1 \| <\infty$ for any $f \in \cF$ and $u = \lim_{r \to \infty}u_{[1,r]}$, where $u_{[a,b]}\coloneqq u_b u_{b-1} \cdots u_{a+1}u_a$, in the norm topology. 
    The ``if part'' is verified as
    \[
    \| u - \Pi_{\bm{x},r}(u) \| \leq 2 \| u-u_{[1,r]} \| \leq \sum_{n=r+1}^\infty 2\| u_n -1\| \leq \sum_{n=r+1}^\infty 2C \cdot f_{\nu+l+1,\mu}(n) \leq  2C\kappa_{\Lambda} 2^{l_{\Lambda}+1} \cdot f_{\nu,\mu}(r+1)
    \]
    for some  $C >0$. To see the ``only if part'', observe that an operator $a \in \cA_{\Lambda}$ with $\| a -1 \| \leq 1/6$ satisfies $\| a \| <2 $, and hence $\| a^*a-1 \| \leq \|a^* \| \cdot \| a-1\| + \|a^*-1\| \leq 3 \| a-1\| \leq 1/2$.
    This estimate implies $\mathrm{spec}(a^*a) \subset [1/2,3/2]$, and hence $\| (a^*a)^{-1/2} -1 \| \leq \| a^*a-1 \| \leq 3\|a-1\|$ (note that $|x^{-1/2}-1| \leq |x-1|$ on $[1/2,3/2]$). In conclusion, we obtain that 
    \begin{align}
    \begin{split}
    \| a (a^*a)^{-1/2} -1 \| \leq {}&{} \| a\| \cdot \| (a^*a)^{-1/2} -1 \| + \|a -1\|
    \leq 7 \| a -1\|.
    \end{split}
    \label{eqn:polar.decomposition.norm}
    \end{align}
    Let $a_r \coloneqq \Pi_{\bm{x},r}(u)$ and apply the above inequality to $u^*a_r$. Then we obtain that $v_r \coloneqq a_r (a_r^*a_r)^{-1/2}$ satisfies $\| v_r -u \| = \| u^*v_r -1\| \leq 7 \| u^*a_r -1 \| \leq 7f(r) \cdot \| u \|_{\bm{x},f}$ for any $f \in \cF$. Now, let $v_0 \coloneqq 1$ and $u_r \coloneqq v_r v_{r-1}^* \in \cA_{B_r(\bm{x})}$. Then $\{ u_r\}_{r \in \bZ_{\geq 1}}$ is the desired sequence of unitaries.  

    Let $u \in \cA_{\Lambda}^{\al}$ be a unitary such that $-1$ is not contained in its spectrum. 
    Then $g \coloneqq \log (u)$ is the limit of $g_r \coloneqq \log (u_{[1,r]})$. 
    This limit is indeed the almost local convergence, which shows $ g \in \cA_{\Lambda}^{\al}$, since $\| g - g_r \|$ is bounded by a constant multiple of $\|u- u_{[1,r]}\|$. 
    Conversely, the same argument shows that the exponential of a skew-adjoint operator $g \in \cA_{\Lambda}^{\al}$ is again almost local. 
    More strongly, for a smooth function $g \colon [0,1] \to \cA_{\Lambda}^{\al}$ taking value in skew-adjoint operators, the solution $u(t)$ of the ordinary differential equation $u(t)^* \frac{d}{dt}u(t) = g(t)$, which is a priori a unitary in $\cA_{\Lambda}$, is almost local. Indeed, by \cref{rmk:CE.average}, we have
    \begin{align*}
    \|u(t) - \Pi_{\bm{x},r}(u(t)) \| \leq {}&{}
    \sup_{v \in \cU(\cA_{B_r(\bm{x})^c})} \| vu(t)v^* -u(t) \| 
    = \sup_{v \in \cU(\cA_{B_r(\bm{x})^c})} \| u(t)^*vu(t) -v \| \\
    \leq {}&{} \sup_{v \in \cU(\cA_{B_r(\bm{x})^c}) }\int_0^t \| [v,g(s)] \| ds \leq f(r) \cdot \| g\|_{\bm{x},f} \cdot t 
    \end{align*}
    for any $a \in \cA_{B_r(\bm{x})^c}$ and $f \in \cF$. 
\end{rmk}

\begin{defn}\label{defn:almost.local.Hamiltonian}
A \emph{uniformly almost local (UAL)} derivation is a linear map
\begin{align*} 
    \sfad (\sfG) = [\sfG,\blank ] \coloneqq  \sum_{\bm{x}\in {\Lambda}} [\sfG_{\bm{x}}, \blank ] \colon \cA_{\Lambda}^{\al} \to \cA_{\Lambda}
\end{align*}
    given by a collection $\sfG\coloneqq (\sfG_{\bm{x}})_{\bm{x} \in \Lambda }$ of skew-adjoint operators in $\cA_{\Lambda}^{\al}$, a local generator of the derivation,  such that
\begin{align*}
    \vvert \sfG \vvert _f \coloneqq  \sup _{\bm{x} \in \Lambda } \big\| \sfG_{\bm{x}} \big\|_{\bm{x} , f} < \infty.
\end{align*}
We write $\fDer ^{\al}_{\Lambda}$ for the set of local generators of UAL derivations. 
In this paper, we do not distinguish between a UAL derivation and its local generator in terminology and refer to the latter as a UAL derivation. 
\end{defn}

The above notation $\fDer^{\al}_{\Lambda}$ is inconsistent with that in \cites{kapustinLocalNoetherTheorem2022,artymowiczQuantizationHigherBerry2023}, where $\fD_{\al}$ refers to the set of UAL derivations, and the set we denote by $\fDer_{\Lambda}^{\al}$ is represented by $C_0(\fd_{\al})$.

The infinite sum in the definition of $\sfad(\sfG)$ exists since
\begin{align}
\begin{split}
    \| [a , b] \| \leq {}&{} \big\| [a, b-\Pi_{\bm{y},\rmd(\bm{x},\bm{y})/2}(b)] \big\| + \big\| [a - \Pi_{\bm{x},\rmd(\bm{x},\bm{y})/2} (a) , \Pi_{\bm{y},\rmd(\bm{x},\bm{y})/2}(b)] \big\|\\
    \leq {}&{} 2 f(\rmd(\bm{x},\bm{y})/2) \cdot  \|a\|_{\bm{x},f} \cdot \| b \|_{\bm{y},f}.
\end{split}\label{eqn:derivation.welldef}
\end{align}
Indeed, the estimate \eqref{eqn:derivation.welldef} for $a = \sfG_{\bm{x}}$ shows that the operator norms $\| [\sfG_{\bm{x}},a ]\|$ decay quickly as $\rmd(\bm{x} ,\bm{y} ) \to \infty$ so that the infinite sum $\sum_{\bm{x} \in \Lambda }\| [\sfG_{\bm{x}},a]\|$ converges.

\begin{lem}\label{lem:almost.local.continuity}
For $\sfG \in \fDer ^{\al}_{\Lambda}$, the associated UAL derivation $\sfad (\sfG) $ acts on $\cA_{\Lambda}^{\al}$ continuously with respect to its Fr\'{e}chet topology. 
\end{lem}
\begin{proof}
We show $\sfad(\sfG)(\cA^{\rm al}_{\Lambda}) \subset \cA_{\Lambda}^{\rm al} $. 
For $r>0$, $\bm{x},\bm{y} \in \Lambda$ such that $\rmd(\bm{x},\bm{y}) \leq r/2$, and $a,b \in \cA_{\Lambda}^{\al}$,  \eqref{eqn:const.small.c} and \eqref{eqn:Frechet.algebra} show that
\begin{align*}
\begin{split}
    \| [a, b] - \Pi_{\bm{x},r}([a,b])  \| 
    \leq 
    {}&{} \| ab - \Pi_{\bm{x},r}(ab)\| + \| ba - \Pi_{\bm{x},r}(ba)\|  \\
    \leq {}&{}  2 \cdot \| a\| \cdot \| (b - \Pi_{\bm{x},r}(b))\| + 2 \cdot \| (a - \Pi_{\bm{x},r}(a))\| \cdot \| \Pi_{\bm{x},r}(b)\|  \\ 
    \leq {}&{} 4 \cdot \| a\| \cdot \| (b - \Pi_{\bm{y},r/2}(b))\| + 4 \cdot \| (a - \Pi_{\bm{x},r/2}(a))\| \cdot \| \Pi_{\bm{x},r}(b)\|\\
    \leq {} &{} 8 f_{\nu,\mu_1}(r/2) \cdot  \| b\|_{\bm{y},\nu,\mu_1} \cdot \| a \|_{\bm{x},\nu,\mu_1 }\\
    \leq{}&{}  8 c_{\nu,\mu,2} \cdot f_{\nu,\mu}(r) \cdot \| b\|_{\bm{y},\nu,\mu_1} \cdot \| a \|_{\bm{x},\nu,\mu_1 }.
\end{split}
\end{align*}
Here, the inequality $\| a - \Pi_X(a)\| \leq \| a - \Pi_Y(a) \| + \| \Pi_X(a - \Pi_Y(a)) \| \leq 2 \| a - \Pi_Y(a)\|$ for $Y \subset X \subset \Lambda $ is also used. 
On the other hand, by letting $\nu_1 \coloneqq \nu+l_{\Lambda}+1$, \eqref{eqn:derivation.welldef} and \cref{lem:sum.exponential} show that 
\begin{align*}
    \sum_{\bm{y} \in \Lambda \setminus B_{r/2}(\bm{x})} \| [\sfG_{\bm{y}}, a] \| 
    \leq {}& \sum_{\bm{y} \in \Lambda \setminus B_{r/2}(\bm{x})} 2 \cdot  f_{\nu_1,\mu_2}(\rmd(\bm{x},\bm{y})/2)\cdot 
    \vvert \sfG \vvert _{\nu_1,\mu_2 } \cdot \|a\|_{\bm{x}, \nu_1,\mu_2} \\
    \leq {}& \sum_{\bm{y} \in \Lambda \setminus B_{r/2}(\bm{x})} 2c_{\nu_1,\mu_2,2} \cdot  f_{\nu_1,\mu_1}(\rmd(\bm{x},\bm{y}))\cdot 
    \vvert \sfG \vvert _{\nu_1,\mu_2 } \cdot \|a\|_{\bm{x}, \nu_1,\mu_2} \\    \leq {}& 2^{l_{\Lambda}+2} \kappa_{\Lambda} \cdot c_{\nu_1,\mu_2,2} \cdot f_{\nu,\mu_1}(r/2) \cdot \vvert \sfG \vvert _{\nu_1,\mu_2} \cdot \|a\|_{\bm{x}, \nu_1,\mu_2}\\
    \leq {}& 2^{l_{\Lambda}+2} c_{\nu_1,\mu_1,2}c_{\nu,\mu,2}\kappa_{\Lambda} \cdot f_{\nu,\mu}(r) \cdot \vvert \sfG \vvert _{\nu_1,\mu_2}  \cdot \|a\|_{\bm{x}, \nu_1,\mu_2}.
\end{align*}
By these inequalities, for $a \in \cA_\Lambda^{\al }$ and $\bm{x} \in { \Lambda}$, we obtain that 
\begin{align*}
    &f_{\nu,\mu}(r)^{-1} \cdot \| [\sfG,a] - \Pi_{\bm{x},r}([\sfG,a])\|\\
     \leq {}& f_{\nu,\mu}(r)^{-1} \cdot \bigg(\sum_{\bm{y} \in B_{r/2}(\bm{x})} \| [\sfG_{\bm{y}}, a] - \Pi_{\bm{x},r}([\sfG_{\bm{y}}, a]) \|  + \sum_{\bm{y} \in \Lambda \setminus B_{r/2}(\bm{x})} 2\| [\sfG_{\bm{y}}, a] \| \bigg)   \\
     \leq {}& 8 c_{\nu_1,\mu , 2}  \vvert \sfG \vvert_{\nu_1,\mu_1} \cdot \| a \|_{\bm{x},\nu_1,\mu_1}\cdot f_{\nu,\mu}(r)^{-1}f_{\nu_1,\mu}(r)  \cdot \# B_{r/2}(\bm{x}) \\
     {}&{} \qquad + 2^{l_{\Lambda}+3} c_{\nu_1,\mu_1,2}c_{\nu,\mu,2} \kappa_{\Lambda} \cdot \vvert \sfG \vvert_{\nu_1,\mu_2} \cdot \| a \|_{\bm{x},\nu_1,\mu_2}\\
     \leq {}&{}  d_{\nu,\mu} \cdot \vvert \sfG \vvert_{\nu_1,\mu_2} \cdot \| a \|_{\bm{x},\nu_1,\mu_2} <\infty,
\end{align*}
for any $R>0$, where $d_{\nu,\mu} \coloneqq  8\kappa_{\Lambda} c_{\nu_1,\mu,2}    + 2^{l_{\Lambda}+3} c_{\nu_1,\mu_1,2}c_{\nu,\mu,2} \kappa_{\Lambda}  $. In the third inequality, we have used 
\[
    f_{\nu,\mu}(r)^{-1}f_{\nu_1,\mu}(r) \cdot \# B_{r/2}(\bm{x}) \leq (1+r)^{-l_{\Lambda}-1}  \cdot \kappa_{\Lambda} (1+r/2)^{l_{\Lambda}} \leq 1. \qedhere 
\]
\end{proof}

For $\sfG \in \fDer ^{\al}_{\Lambda}$, the unbounded derivation $\sfad(\sfG)$ on the C*-algebra $\cA_{\Lambda}$ with the domain $\cA_{\Lambda}^{\mathrm{al}}$ is closable by \cite{bratteliOperatorAlgebrasQuantum1997}*{Lemma 3.1.14 and Proposition 3.2.22}. 
To justify this, we note that the square root of a positive almost local operator remains almost local. 
This follows from the inequality $\|a^{1/2} - b^{1/2}\| \leq \|a - b\|^{1/2}$ for any $a, b \geq 0$, which can be found in \cite{bhatiaMatrixAnalysis1997}*{Theorem~X.1.1}. Applying this inequality to $a \in \cA_\Lambda^{\al}$ and $b = \Pi_{\bm{x}, r}(a)$ yields the desired conclusion.
We use the same letter $\sfad(\sfG)$ for its closure. 
It generates a $1$-parameter group of $\ast$-automorphisms on $\cA_{\Lambda}$ which is explicitly given by $\alpha_{t}(a) = \lim_{Z \to \Lambda} e^{t \sfG_Z}a e^{-t\sfG_{Z}}$. 
\begin{rmk}
We add an explanation for the fact that $\sfad(\sfG)$ generates a $1$-parameter family of automorphism group.
First, it is proved in \cite{nachtergaeleQuasilocalityBoundsQuantum2019}*{Subsection 3.1} that the limit $\alpha_t(a)$ exists (we discuss this point later in \cref{rmk:parallel.transport.existence.uniqueness}). 
Its infinitesimal generator, say $\delta$, is also a closed unbounded derivation on $\cA_{\Lambda}$. 
The domain of $\delta$ contains $\cA_{\Lambda}^{\al}$ and $\delta(a)=\sfad (\sfG)(a)$ for any $\cA_{\Lambda}^{\al}$. Moreover, as will be observed in \cref{prp:Lieb.Robinson} (2), $\alpha_{t}$ preserves the subalgebra $\cA_{\Lambda}^{\al}$. 

These facts conclude that $\sfad(\sfG)$ also generates a $1$-parameter group of $\ast$-automorphisms, and hence $\sfad(\sfG) = \delta$ (in other words, $\cA_{\Lambda}^{\al}$ is a core of $\delta$). 
First, by \cite{bratteliOperatorAlgebrasQuantum1997}*{Theorem 3.2.50}, we have $\| a + [\sfG,a] \| \geq \liminf_{Z \to \Lambda} \| a + [\sfG_Z,a]\| \geq \|a\|$ for any $a \in \cA_{\Lambda}^{\al}$. Thus, again by \cite{bratteliOperatorAlgebrasQuantum1997}*{Theorem 3.2.50}, it suffices to show that $(\id +c \cdot \sfad(\sfG))(\cA_{\Lambda}^{\al})$ is dense in $\cA_{\Lambda}$ for any $c \in \bR \setminus \{0\}$. 
Assume the contrary. Then, there is a linear functional $\eta \colon \cA_{\Lambda} \to \bC$ such that $\eta(a +c [\sfG,a]) =0$ for any $a \in \cA_{\Lambda}^{\al}$. This means that $\frac{d}{dt}\eta(\alpha_{t}(a)) = \eta([\sfG,\alpha_{t}(a)]) = -c \eta(\alpha_{t}(a))$, and hence $\eta(\alpha_{t}(a)) = \eta(a) \cdot e^{-c t}$. This contradicts to the fact that any state on a C*-algebra is norm contractive (cf.\ \cite{murphyAlgebrasOperatorTheory1990}*{Theorem 3.3.6}).       
\end{rmk}

A collection of self-adjoint operators $\sfH =(\sfH_{\bm{x}})_{\bm{x} \in \Lambda}$ such that $\sfH_{\bm{x}} \in i \fDer ^{\al}_{\Lambda}$ serves as a lattice Hamiltonian generating the Heisenberg dynamics 
\[
    \tau_{\sfH,t}(a) \coloneqq \lim_{Z \to \Lambda} e^{-it \sfH_Z} a e^{it \sfH_Z}
\]
on $\cA_{{\Lambda}}$, which we call an \emph{UAL Hamiltonian}. 
A state, i.e.\ a bounded linear map $\omega \colon \cA_{\Lambda} \to \bC$ satisfying the positivity $\omega (a^*a) \geq 0$ and unitality $\omega (1)=1$, is said to be a ground state of $\sfH $ if $\omega (a^* [\sfH ,a]) \geq 0$ holds for any operator $a \in \cA_{\Lambda}^{\al}$. 
In terms of the GNS representation $(\sH _\omega, \pi _\omega , \Omega_\omega)$ of $(\cA_{\Lambda},\omega)$, this condition corresponds to the positivity of the generator $H_\omega $ of the $1$-parameter unitary group $U_\omega \colon \bR \to \cU(\sH_{\omega})$ given by $U_\omega (\pi_\omega(a)\Omega_\omega) = \pi_{\omega}(\tau_{\sfH,t}(a)) \Omega_\omega$. Such $U_{\omega}$ and $H_{\omega}$ are well-defined since a $\tau_{\sfH}$-ground state is $\tau_{\sfH}$-invariant (\cite{bratteliOperatorAlgebrasQuantum1997}*{Lemma 5.3.16}). 
Note that $\pi_\omega(\cA_{\Lambda}^{\al})\Omega_\omega \subset \mathop{\mathrm{Dom}}(H_\omega)$ is a core of $H_\omega$ and 
\begin{align*}
    H_\omega \pi_\omega(a)\Omega_{\omega} = 
    \pi_\omega \big( [\sfH, a]\big) \Omega_\omega 
    = 
    \sum_{\bm{x} \in \Lambda}\pi_{\omega}([\sfH_{\bm{x}},a])\Omega_{\omega }
\end{align*}
holds for any $a \in \cA_{\Lambda}^{\al}$. 
For $\gap >0$, the pair $(\sfH , \omega)$ of a UAL Hamiltonian $\sfH$ and its ground state $\omega$ is said to have a \emph{spectral gap} $\delta >0$ if the GNS Hamiltonian $H_\omega$ satisfies $\mathrm{spec}(H_\omega) \cap (0,\gap) = \emptyset$. 
Moreover, it is said to be \emph{non-degenerate} if $\Ker H_\omega = \bC \cdot \Omega_{\omega}$. 
A pair $(\sfH,\omega)$ is gapped and non-degenerate if and only if $\omega (a^* [\sfH,a]) \geq \gap \cdot \omega (a^*a)$ holds for any $a \in \cA_{\Lambda}^{\al}$ with $\omega (a)=0$. 

\begin{rmk}\label{rmk:Cstar}
A $\ast$-homomorphism $\phi \colon \cA_{\Lambda}^{\al} \to \cA_{\Lambda}^{\al}$ and a quantum state $\omega \colon \cA_{\Lambda}^{\al} \to \bC$ that are continuous with respect to almost local topology immediately extend to the C*-completion $\cA_{\Lambda}$ continuously. 
Indeed, for each bounded region $Z \subset \Lambda$, the restriction $\phi|_{\cA_{Z}}$ (resp.\ $\omega|_{\cA_Z}$) is a $\ast$-homomorphism (resp.\ state) on the matrix C*-algebra, and hence its norm measured by the C*-norm on $\cA_Z$ satisfies $\| \phi|_{\cA_Z}\| \leq 1$ (resp.\ $\| \omega|_{\cA_Z}\| \leq 1$). 
\end{rmk}

\begin{lem}
    A non-degenerate ground state $\omega_\sfH$ of a UAL Hamiltonian $\sfH$ is a pure state. 
\end{lem}
\begin{proof}
    In the proof, we abbreviate $\omega_{\sfH}$ to $\omega$. 
    We use the fact that an extremal point of the space of $\tau_{\sfH}$-ground states is a pure state (\cite{bratteliOperatorAlgebrasQuantum1997}*{Theorem 5.3.37}). 
    Let $\omega' $ be another $\tau_{\sfH}$-ground state of $\sfH$ such that $\omega' \leq c \omega$ for some $c \geq 1$. 
    By a standard fact in operator algebra (that is found in e.g.\ \cite{murphyAlgebrasOperatorTheory1990}*{Theorem 5.1.2}), there is a unique positive operator $V$ in the commutant von Neumann algebra $\pi_{\omega}(\cA_{\Lambda})' \subset \cB(\sH_{\omega})$ such that $\omega'(a)=\langle \pi_{\omega}(a)V\Omega_{\omega} ,\Omega_{\omega}\rangle =\langle \pi_{\omega}(a)V^{1/2}\Omega_{\omega} ,V^{1/2}\Omega_{\omega}\rangle$. 
    By the uniqueness of $V$ and the $\tau_{\sfH}$-invariance of $\omega'$, the equality $e^{itH_{\omega}}V^{1/2}e^{-itH_{\omega}} = V^{1/2}$, and hence $e^{itH_{\omega}}V^{1/2}\Omega_{\omega} = V^{1/2}\Omega_{\omega}$, holds. 
    By the non-degeneracy of $\omega$, it means that $V^{1/2}\Omega_{\omega}$ is contained in $\Ker H_{\omega} = \bC \cdot \Omega_{\omega}$. 
    This shows $\omega'=\omega$, which concludes that $\omega$ is an extremal point in the space of $\tau_{\sfH}$-invariant states. 
\end{proof}
\begin{defn}
    We say that a pair $(\sfH, \omega_{\sfH})$ a \emph{gapped UAL Hamiltonian} with gap $\delta$ if $\sfH$ is a UAL Hamiltonian and $\omega_{\sfH}$ is a non-degenerate ground state of $\sfH$ such that $(\sfH , \omega_{\sfH})$ has a spectral gap $\delta$. 
    The set of gapped UAL Hamiltonians $(\sfH,\omega_{\sfH})$ with gap $1$ denotes $\fH_\Lambda^{\al}$. 
    We abuse the notation and write $\sfH$ to represent a pair $(\sfH,\omega_{\sfH})$.
\end{defn}

An aspect of this paper is that we do not assume the uniqueness of the ground state. 
The non-degeneracy of a gapped ground state $\omega$ is a weaker notion of uniqueness. 
Indeed, the non-degeneracy implies the local uniqueness (rigidity) of such a ground state; there is no `smooth' perturbation of $\omega$ as gapped ground states. This point will be explained in the proof of \cref{thm:automorphic.equivalence}. 

Let $\sfH_1$ and $\sfH_2$ be gapped UAL Hamiltonians with gap $\delta$ defined on the lattices $\Lambda_1$ and $\Lambda_2$ respectively. 
Their composite system $\sfH = \sfH_1 \boxtimes \sfH_2$ is defined by 
\begin{align}
\begin{split}
    \sfH_{\bm{x}} = (\sfH_1 \boxtimes \sfH_2)_{\bm{x}} \coloneqq  {}&{}
    \begin{cases}
        \sfH_{1,\bm{x}} \otimes 1 & \text{ if $\bm{x} \in \Lambda_1$, }\\
        1 \otimes \sfH_{2,\bm{x}} & \text{ if $\bm{x} \in \Lambda_2$, }
    \end{cases}\\
    \in {}&{} \cA_{\Lambda_1}^{\al}  \otimes \cA_{\Lambda_2}^{\al} = \cA_{\Lambda_1 \sqcup \Lambda_2}^\al, \\
    \omega_{\sfH} \coloneqq {}&{} \omega_{\sfH_1} \otimes \omega_{\sfH_2} \in \fS(\cA_{\Lambda_1 \sqcup \Lambda_2}).
\end{split}
\label{eqn:Hamiltonian.product}
\end{align}
Then, the associated time evolution is $\tau_{\sfH , t} = \tau_{\sfH_1 , t} \otimes \tau_{\sfH_2, t} $, and the GNS representation of the C*-dynamical system $(\cA_{\Lambda_1 \sqcup \Lambda_2} , \omega_{\sfH}, \tau_{\sfH})$ decomposes to the tensor product as 
\[
    (\sH_{\omega_{\sfH}},\pi_{\omega_{\sfH}}, \Omega_{\omega_{\sfH}},H_{\omega_{\sfH}} )  =
    (\sH_{\omega_{\sfH_1}},\pi_{\omega_{\sfH_1}}, \Omega_{\omega_{\sfH_1}},H_{\omega_{\sfH_1}} )
    \otimes 
    (\sH_{\omega_{\sfH_2}},\pi_{\omega_{\sfH_2}}, \Omega_{\omega_{\sfH_2}},H_{\omega_{\sfH_2}} ).
\]
Thus, $(\sfH , \omega_{\sfH})$ is a gapped UAL Hamiltonian with gap $\delta$.
In \cref{subsection:sheaf.lattice}, we will fix a unit vector $\Omega_{\bm{x}} \in \bC^{n_{\bm{x}}}$ for each $\bm{x} \in \Lambda$, in order to define the trivial Hamiltonian
\begin{align}
    \sfh = (\sfh_{\bm{x}} )_{\bm{x} \in \Lambda} \in \fH_{\Lambda}^{\al}, \quad \sfh_{\bm{x}} \coloneqq 1 - \Omega_{\bm{x}} \otimes \Omega_{\bm{x}}, \quad \omega_0\coloneqq \bigotimes_{\bm{x}} \langle \blank \Omega_{\bm{x}},\Omega_{\bm{x}} \rangle, \label{eqn:trivial.Hamiltonian}
\end{align}
which will play the role of the unit of the above product.

Among the local generators realizing a given UAL derivation, there is a standard choice called the brick decomposition. 
It is also compatible with the notion of interaction, a family of operators indexed by finite subsets of $\Lambda$, used as a local generator of the Hamiltonian in the literature \cites{bratteliOperatorAlgebrasQuantum1997,nachtergaeleQuasilocalityBoundsQuantum2019,moonAutomorphicEquivalenceGapped2020}.
\begin{prp}[{\cite{kapustinLocalNoetherTheorem2022}*{Propositions C.2, D.2}}]\label{prp:brick}
    Let $\iota \colon \Lambda \to \bR^{l_{\Lambda}}$ be weakly uniformly discrete. For $0<\mu<1$ and $\nu\geq 0$, let $\nu_k\coloneqq \nu + k(l_{\Lambda}+1)$ and $\mu_k$ be as in \cref{rmk:polynomial.coarse.invariance}.
    \begin{enumerate}
        \item Let $a \in \cA_{\Lambda}^{\al}$. Its brick decomposition is $a=\sum_{\rho \in \bB} a^\rho $ is given inductively by 
        \begin{align*} a^\rho = \Pi_{B_\rho}(a) - \sum_{\rho \in \bB, \sigma < \rho} a^{\sigma}.  \end{align*}
        Then the infinite sum $a=\sum_{\rho \in \bB }a^\rho $ converges in the almost local topology. 
        \item Let $\iota(\bB)$ denote the subposet of $2^\Lambda$ consisting of $B_\rho$'s. 
        For a UAL derivation $\sfG$, the interaction generating the dynamics of $\sfG$ is a collection of skew-adjoint operators $\{ \Phi_{\sfG}(B_\rho) \}_{B_\rho \in \iota(\bB)}$ defined by 
        \begin{align*} 
        \Phi_\sfG( B_\rho ) \coloneqq  \sum_{\bm{x} \in \Lambda } \sfG_{\bm{x}}^{\rho} \in \cA_{B_\rho }. 
        \end{align*}
        This infinite sum converges in the operator norm topology. Moreover, for any $\nu >0$ and $0<\mu<1$, there is $L_{\nu,\mu}>0$ such that
        \begin{align} 
        \begin{split}
        \vvert \Phi_{\sfG} \vvert_{\nu,\mu} \coloneqq {}&{} \sup_{\bm{x},\bm{y} \in \Lambda} f_{\nu,\mu}(d(\bm{x},\bm{y}))^{-1} \cdot  \sum_{B_\rho \in \iota(\bB), \  B_\rho \ni \bm{x},\bm{y}}
        \| \Phi_\sfG(B_\rho )\| \\
        \leq {}&{} L_{\nu,\mu} \cdot \vvert \sfG \vvert_{\nu_3,\mu_1} <\infty.
        \end{split}\label{eqn:NSY}
        \end{align}
        \item Conversely, if a family of skew-adjoint operators $\Phi (B_\rho) \in \cA_{B_\rho}$ satisfies that the seminorms $\vvert \Phi \vvert_{\nu,\mu} $ in \eqref{eqn:NSY} is finite for any $0<\mu<1$ and $\nu >0$, then the collection of skew-adjoint operators $(\sfG_{\Phi,\bm{x}})_{\bm{x} \in \Lambda}$ given by
        \begin{align*}
        \sfG_{\Phi,\bm{x}} \coloneqq \sum_{B_\rho \ni \bm{x}} |B_\rho|^{-1}\Phi(B_\rho )
        \end{align*}
        forms a UAL derivation $\sfG_{\Phi} \in \fDer ^{\al}_{\Lambda}$ such that $\Phi_{\sfG_{\Phi}} = \Phi$ and $\sfad (\sfG_{\Phi_{\sfG}})=\sfad(\sfG)$.  
    \end{enumerate}
\end{prp}

For a finite subset $Z \subset \Lambda$, we can identify the finite dimensional algebra $\cA_Z$ with its GNS Hilbert space $L^2(\cA_{Z},\mathrm{tr})$ via the unique normalized trace. 
Under this identification, the conditional expectation $\Pi_Y$ is identified with the orthogonal projection onto $L^2(\cA_Y,\mathrm{tr})$, and the above $a^\rho$ is the projection onto the orthogonal complement of $\mathrm{span}_{\sigma <\rho }\cA_{B_\sigma}$. 
Hence the brick decomposition $a= \sum_{\rho}a^\rho$ is consistent with the one given in \cite{kapustinLocalNoetherTheorem2022}*{Subsection 3.2}.
In particular, $a^\rho =0$ if there is $\sigma \in \bB$ such that $\sigma <\rho$ and $B_\sigma=B_\rho$, and hence the brick decomposition is actually indexed by $\iota(\bB)$ rather than $\bB$. 

\begin{proof}
First, by $\Pi_{B_\rho}(a) = \sum_{\sigma \leq \rho} a^{\sigma}$, we have
\begin{align*}
    a= \lim_{\sigma \in \bB } \Pi_{B_\sigma}(a) = \lim_{\sigma \in \bB} \sum_{\rho \leq \sigma}a^{\rho} = \sum_{\rho \in \bB} a^\rho
\end{align*}
if the infinite sum in the right hand side converges. 
To see the convergence of the right hand side, let $\mu_{\bB}$ and $R(\rho,\bm{x}) $ be as in page \pageref{subsubsection:brick}. 
Then we have $a^{\rho}=\sum_{\sigma \leq \rho} \mu(\sigma,\rho)\Pi_{B_\sigma}(a)$ and $(\Pi_{\bm{x},R(\rho,\bm{x})}(a))^{\rho} =0$. Hence, by \eqref{eqn:brick.Mobius}, 
\begin{align}
\begin{split}
    \| a^{\rho}\| = {}& \| (a-\Pi_{\bm{x},R(\rho,\bm{x})}(a))^\rho \|
    \leq{} \sum_{\sigma <\rho} \mu_{\bB}(\sigma,\rho) \cdot \| \Pi_{B_\sigma}(a - \Pi_{\bm{x},R(\rho,\bm{x})}(a)) \| \\
    \leq {}& 4^{l_{\Lambda}} \cdot f(R(\rho,\bm{x})) \cdot \| a \|_{\bm{x},f}.
\end{split}
\label{eqn:brick.estimate}
\end{align}
Since $\Pi_{\bm{x},r}(a^{\rho})=0$ if $r\leq R(\rho,\bm{x})$ and $\Pi_{\bm{x},r}(a^{\rho}) = a^{\rho}$ if $r> R(\rho,\bm{x})$, we have 
\begin{align}
    \| a^\rho\|_{\bm{x},\nu,\mu}=f_{\nu,\mu}(R(\rho,\bm{x}))^{-1} \cdot \| a^\rho \| \leq 4^{l_{\Lambda}} \cdot (1+R(\rho,\bm{x}))^{-2l_{\Lambda} -2} \cdot \| a \|_{\bm{x},\nu_2,\mu}.\label{eqn:brick.estimate.F}
\end{align}
By \cref{lem:sum.exponential} and \cref{lem:R(rho.x).growth} (1), the sum of the right hand side is finite, which shows (1). 

The convergence of the right hand side of the definition of $\Phi_{\sfG}(B_\rho)$ follows from \cref{lem:R(rho.x).growth} (2) as 
\begin{align}
    \begin{split}
        \|\Phi_{\sfG}(B_{\rho})\| \leq \sum_{\bm{x} \in \Lambda }\| \sfG_{\bm{x}}^\rho\| \leq {}&{} \sum_{\bm{x} \in \Lambda} 4^{l_{\Lambda}} \cdot f_{\nu_3,\mu_1}(R(\rho,\bm{x}))\cdot \| \sfG_{\bm{x}}\|_{\bm{x},\nu_3,\mu_1} \\
        \leq {}&{} 2^{3l_{\Lambda} +1} \kappa_{\Lambda}  \cdot f_{\nu_2,\mu_1}  \big( \inf_{\bm{z} \in \Lambda} R(\rho,\bm{z}) \big)\cdot \vvert \sfG \vvert_{\nu_3,\mu_1 }
    \end{split}\label{eqn:brick.decomposition}
\end{align}
for any $\mu$, $\nu$. By the remark below \eqref{eqn:R(rho.x)}, for $\rho \in \bB$ such that $B_\rho \ni \bm{x}, \bm{y}$, we have
\begin{align*}
    \inf_{\bm{z} \in \Lambda}R(\rho,\bm{z}) \geq \frac{\diam(\rho)}{2\sqrt{l_{\Lambda}}} -1 \geq \frac{\rmd(\bm{x},\bm{y})}{2\sqrt{l_{\Lambda}}}-1.
\end{align*}
Thus, by \cref{lem:sum.exponential} and \cref{lem:R(rho.x).growth} (3), we have
\begin{align*}
    \sum_{B_\rho \in \iota(\bB), \ B_\rho \ni \bm{x},\bm{y}}  f_{\nu_2,\mu_1}\big(\inf_{\bm{z}} R(\rho,\bm{z})\big) \leq  2^{2l_{\Lambda} +1}  \kappa_{\bB} \cdot f_{\nu,\mu_1}\Big( \frac{\rmd(\bm{x},\bm{y})}{2\sqrt{l_{\Lambda}}} -1\Big).
\end{align*}
This, together with \eqref{eqn:large.scale.Lipschitz.F}, shows \eqref{eqn:NSY} by choosing the constant $L_{\nu,\mu}$ as 
\begin{align*}
    L_{\nu,\mu} \coloneqq 2^{5l_{\Lambda}+2} \kappa_{\Lambda}\kappa_{\bB} \cdot c_{0,\mu,2\sqrt{l_{\Lambda}}} \cdot f_{\mu}(2\sqrt{l_{\Lambda}}) \cdot (1+2\sqrt{l_{\Lambda}})^\nu <\infty. 
\end{align*}

Finally, we show (3). The well-definedness of $\sfG_{\Phi}$ is verified  as
\begin{align}
\begin{split}
    \| \sfG_{\Phi,\bm{x}} \|_{\bm{x},\nu,\mu} \leq {}&{}  \sup_{r>0}\bigg( f_{\nu,\mu}(r)^{-1} \cdot \sum_{\substack{B_\rho \in \iota(\bB), B_\rho\ni \bm{x} \\ B_\rho \not \subset B_r(\bm{x})}} \| \Phi (B_\rho)\| \bigg)\\
    \leq {}&{} \sup_{r>0}\bigg( f_{\nu,\mu}(r)^{-1} \cdot \sum_{\substack{\bm{y} \in \Lambda \\ \rmd (\bm{x},\bm{y})\geq r}} \sum_{B_\rho \ni \bm{x}, \bm{y}} \| \Phi(B_\rho)\| \bigg)\\
    \leq {}&{}  \sup_{r>0} \bigg( f_{\nu,\mu}(r)^{-1} \cdot\sum_{\substack{\bm{y} \in \Lambda \\ \rmd (\bm{x},\bm{y}) \geq r}} \vvert \Phi \vvert_{\nu_1,\mu} \cdot f_{\nu_1,\mu}(\rmd(\bm{x},\bm{y}))\bigg) \\
    \leq {}&{} 2^{l_{\Lambda}+1}\kappa_{\Lambda} \cdot \vvert \Phi \vvert_{\nu_1,\mu}  <\infty. 
\end{split}\label{eqn:interaction.to.local}
\end{align}
Here, \eqref{eqn:large.scale.Lipschitz.F} is used in the last inequality.
Moreover, by \eqref{eqn:derivation.welldef} and \eqref{eqn:brick.estimate.F}, we obtain that
\begin{align*}
    {}&{} \sum_{B_\rho \in \iota(\bB)} \sum_{\bm{y} \in \Lambda} \| [\sfG_{\bm{y}}^\rho,a] \| \leq  \sum_{B_\rho \in \iota(\bB)} \sum_{\bm{y} \in \Lambda} 2f_{\nu,\mu}(\rmd(\bm{x},\bm{y})/2) \cdot \| \sfG_{\bm{y}}^\rho \|_{\bm{y},\nu,\mu} \cdot \| a \|_{\bm{x},\nu,\mu} 
    \\
    \leq {}&{} 
    \sum_{B_\rho \in \iota(\bB)} \sum_{\bm{y} \in \Lambda} 2f_{\nu,\mu}(\rmd(\bm{x},\bm{y})/2) \cdot 4^{l_{\Lambda}}(1 + R(\rho, \bm{y}))^{-2l_{\Lambda} -2} \cdot \vvert \sfG \vvert_{\nu_2,\mu} \cdot \| a \|_{\bm{x},\nu,\mu} < \infty .
\end{align*}
It implies that the two local generators $\sfG$ and $\sfG_{\Phi_{\sfG}}$ defines the same UAL derivation on the common core $\cA_{\Lambda}^{\al}$.
\end{proof}

\subsection{Smooth family of Hamiltonians and the adiabatic theorem}\label{subsection:automorphic.equivalence}

\begin{defn}\label{defn:smooth}
Let $\sM$ be an $n$-dimensional manifold and let $\sfG \colon \sM \to \fDer ^{\al}_{\Lambda}$ be a map. 
We say that $\sfG$ is a \emph{smooth} family if the following conditions hold. 
\begin{enumerate}
    \item For each $\bm{x} \in {\Lambda}$, the assignment $p \mapsto \sfG_{\bm{x}}(p) \in \cA_{\Lambda}^{\al}$ is a $C^\infty$-map with respect to the almost local topology on $\cA_{\Lambda}^\al$. 
    \item For $k \in \bN$ and a relatively compact local chart $(\sU \, ; x_1,\cdots,x_n)$ of $\sM$, we have
    \begin{align*}
         \vvert \sfG \vvert _{\sU,C^k,f } &{}\coloneqq \sup_{\bm{x} \in \Lambda} \| \sfG_{\bm{x}}\|_{\sU,C^k,\bm{x},f} =  \sup _{p \in \sU } \sup_{\bm{x} \in {\Lambda}} \sup_{|I| \leq k} \| \partial^I \sfG_{\bm{x}} (p) \|_{{\bm{x}},f } <\infty.
    \end{align*} 
\end{enumerate}
We say that a family of gapped UAL Hamiltonians $\sfH \colon \sM \to \fH_{\Lambda}^{\al}$ is \emph{smooth} if $i\sfH$ satisfies the above (1), (2), and also the following condition (3).
\begin{enumerate}
    \item[(3)] For any $a \in \cA_{\Lambda}^{\al}$, the map $\omega_{\sfH}(a) \colon \sM \to \bC$ given by $p \mapsto \omega _{\sfH(p)} (a)$ is smooth. Moreover, for any relatively compact local chart $\sU $ of $\sM$, and $k \in \bN$, there is $f \in \cF$ such that
    \begin{align}
    \| \omega\|_{\sU,C^k,f} \coloneqq \sup_{\bm{x} \in \Lambda}\sup_{a \in \cA_{\Lambda}^{\al} \setminus \{ 0\} }\sup_{p \in \sU} \frac{\| \omega_{\sfH} (a) \|_{\sU,C^k}}{ \|a \|_{\bm{x},f}} <\infty. \label{eqn:differentiable}
    \end{align}
\end{enumerate}
\end{defn}
In contrast to \cref{rmk:Cstar}, \eqref{eqn:differentiable} does not imply that $\omega_{\sfH}$ is not smooth as a function taking value in the space of C*-states $\fS(\cA_{\Lambda})$.

By definition, smoothness is inherited by the precomposition with a smooth map $f \colon \sN \to \sM$. 
Indeed, the condition (2) is satisfied if $\vvert \sfG \vvert _{\sU_i,C^k,f}<\infty$ is verified only for each $\sU_j$ of an open cover $\{ \sU_j\}$ by relatively compact local charts.

\begin{rmk}\label{rmk:MO}
Let $c \colon (-\varepsilon, \varepsilon) \to \sM$ be an arbitrary smooth path. By the Banach space valued version of the Taylor theorem (e.g.~\cite{zeidlerAppliedFunctionalAnalysis1995}*{Theorem 4.C}), the condition (2) of \cref{defn:smooth} and \eqref{eqn:brick.decomposition} implies that
    \begin{align}
    \sup_{B_{\rho} \in \iota(\bB)} \sup_{-\varepsilon < t < \varepsilon} \sup_{\bm{x} \in {\Lambda}} \bigg\| \frac{\Phi_{\sfG(c(t))}(B_{\rho}) - \Phi_{\sfG(c(0)) }(B_\rho)}{t} - \frac{d}{dt}\bigg|_{t=0} \Phi_{\sfG(c(t))}(B_\rho )  \bigg\| \to 0 \quad \text{as $\varepsilon \to 0$. } \label{eqn:Banach.Taylor}
    \end{align}
Moreover, by \eqref{eqn:brick.decomposition} we have
    \begin{align*} 
    \| \partial_t \Phi_{\sfG}(B_\rho)(c(t))\| \leq \sum_{\bm{x} \in \Lambda} \|\partial_{t} \sfG_{\bm{x}}^\rho (c(t))\| \leq 2^{3l_{\Lambda} +1} \kappa_{\Lambda}  \cdot f_{\nu_2,\mu_1}  \Big(\frac{\mathrm{diam}(\rho)}{2\sqrt{l_{\Lambda}}}-1 \Big) \cdot \vvert \sfG \vvert_{\sU,C^1,\nu_3,\mu_1}.
    \end{align*}
They cover the assumptions (iii) and (iv) of \cite{moonAutomorphicEquivalenceGapped2020}*{Assumption 1.2}.  The other assumptions, except for (ii), are satisfied by definition. 
\end{rmk}

Let $\Aut(\cA_{\Lambda}^{\al}) \subset \cB(\cA_{\Lambda}^{\al})$ denote the set of continuous $\ast$-automorphisms on $\cA_{\Lambda}^{\al}$. 
\begin{defn}[{\cites{kapustinLocalNoetherTheorem2022,artymowiczQuantizationHigherBerry2023}}]\label{defn:parallel.transport}
Let $\sM$ be an $n$-dimensional smooth manifold. 
\begin{enumerate}
    \item A smooth \emph{connection $1$-form} of the trivial bundle $\sM \times \fDer ^{\al}_{\Lambda}$ is a section $\sfG$ of the vector bundle $T^*\sM \otimes \fDer ^{\al}_{\Lambda}$ such that, for any local chart $(\sU \, ;x_1,\cdots,x_n)$, its local presentation $\sfG=\sum_{i=1}^n \sfG_i \  dx_i$ satisfies that  each $\sfG_i$ is smooth. 
    \item Let $\alpha \colon \sM \to \Aut(\cA_{\Lambda}^{\al})$ be a map such that $\alpha_{\bullet }(a)$ is Fr\'{e}chet differentiable with respect to the almost local topology for any $a \in \cA_{\Lambda}^{\al}$. Its Maurer--Cartan form $\theta(\alpha)$ is a smooth $1$-form taking value in the space of continuous derivations on $\cA_{\Lambda}^{\al}$;
    \begin{align*} 
    \theta(\alpha) \coloneqq \alpha^{-1}d\alpha = \bigg( a \mapsto \sum_{i=1}^d \alpha_{\bullet }^{-1} \Big(\frac{\partial }{\partial x_i}(\alpha_{\bullet}(a)) \Big)dx_i \bigg) \in \Omega^1(\sM, \mathrm{Der}(\cA_{\Lambda}^{\al })).
    \end{align*}
    \item For a smooth connection $1$-form $\sfG \in \Omega^1(\sM,\fDer ^{\al}_{\Lambda})$ and a smooth path $c \colon [s_0,t_0] \to \sM$, a map 
    \begin{align*} 
     \alpha_c(\sfG \, ; \blank, \blank) \colon  \{ (s,t) \in \bR^2 \mid s_0 \leq s, t \leq t_0 \} \to \Aut(\cA_{\Lambda}^{\mathrm{al}})
    \end{align*}
    is said to be the parallel transport along $c$ if it satisfies $\alpha_c(\sfG \, ; s,s)=\id$ for any $s_0 \leq s \leq t_0$ and 
    \[
    \theta(\alpha_c (\sfG \,; \blank ,t) ) = -\sfad (c^*\sfG) \in \Omega^1([s_0,t],\mathrm{Der} (\cA ^{\al}_{\Lambda})).
    \]
    Equivalently, for any $a \in \cA_{\Lambda}^{\al}$, we have
    \begin{align} 
    \frac{d}{ds} \alpha_c(\sfG \, ; s,t) (a) + \alpha_c(\sfG \, ; s,t) \Big( \Big[\iota_{\frac{d}{ds}}\sfG (c(s)) , a\Big]\Big) =0. 
    \label{eqn:parallel.transport.ODE}
    \end{align}  
    In the case that the underlying manifold is $[s_0,t_0]$ and $c(t)=t$ (resp.\ $\sM \times [s_0,t_0]$ and $c(t) = (p,t)$), we write the parallel transport as $\alpha( \sfG  \,; s, t)$ (resp.\ $\alpha_p(\sfG \,;s, t)$). When $s=0$, we abbreviate $\alpha_c(\sfG\,;0, t)$ as $\alpha_c(\sfG\,; t)$.
    \item A $\ast$-automorphism $\alpha \in \Aut(\cA_{\Lambda})$ is said to be \emph{locally generated} (LG) if there is a smooth connection $\sfG \in \Omega^1 ([0,1],\fDer ^{\al}_{\Lambda})$ such that $\alpha = \alpha (\sfG \, ;1)$. 
    We write $\cG_{\Lambda}$ for the subgroup of LG automorphisms on $\cA_{\Lambda}$. 
    A map $\alpha \colon \sM \to \cG_{\Lambda}$ is said to be \emph{smooth} if, for any $p \in \sM$ there is an open neighborhood $p \in \sU$ and $\sfG \in \Omega^1 (\sU \times [0,1], \fDer ^{\al}_{\Lambda})$ such that $\alpha_{p'} = \alpha_{p'}(\sfG  \, ; 1)$ for any $m' \in \sU$.
    \item For a state $\omega \colon \cA_{\Lambda} \to \bC$, let $\cG_{\Lambda}^\omega$ denote the subgroup of $\cG_{\Lambda}$ consisting of LG automorphisms with $\omega \circ \alpha = \omega $. The smoothness of a map $\sM \to \cG_{\Lambda}^\omega $ is defined similarly. 
\end{enumerate}
\end{defn}
For example, the time evolution is written by the above notation as $\tau_{\sfH,t}\coloneqq \alpha(-i\sfH dt\,; t)$ for $t>0$. 

The estimate \eqref{eqn:NSY} means that the brick decomposition $\{ \Phi_{\sfG}(B_\rho)\}_{B_\rho \in \iota(\bB)}$ satisfies the quasi-locality assumption of \cite{nachtergaeleQuasilocalityBoundsQuantum2019}*{Subsection 3.1}. 
Therefore, analytic results concerning the Lieb--Robinson bound given in  \cite{nachtergaeleQuasilocalityBoundsQuantum2019} apply to our UAL derivations. 
\begin{rmk}\label{rmk:parallel.transport.existence.uniqueness}
We list known facts and immediate observations concerning the existence and uniqueness of parallel transport.
Let $\sfG \in C^\infty([0,1],\fD_{\Lambda}^{\al})$. 
\begin{enumerate}
    \item A solution of the ODE \eqref{eqn:parallel.transport.ODE} on $\cA_{\Lambda}^{\al}$ satisfy $\big( \frac{d}{ds} \alpha(\sfG dt \,; s,t)^{-1}\big) \circ \alpha(\sfG dt\,; s,t)(a) = [\sfG(s) ,a]$, in other words, 
    \[
    \frac{d}{ds}\alpha(\sfG dt \,; s,t)^{-1} = \sfad \big(\sfG(s) \big) \circ \alpha(\sfG dt \,; s,t)^{-1}.
    \]
    This is verified by differentiating $\alpha(\sfG dt \,; t,s)^{-1} \circ \alpha(\sfG dt \,; t,s)(a) = a$ by the variable $s$. 
    \item By (1), the ODE \eqref{eqn:parallel.transport.ODE} on $\cA_{\Lambda}^{\al}$ uniquely determines the evolution $\alpha(\sfG dt\,; \blank,\blank)$ on $\cA_{\Lambda}^{\al}$ if any. Indeed, if two evolutions $\alpha^{(0)}(\sfG dt \,; \blank,\blank)$ and $\alpha^{(1)}(\sfG dt \,; \blank,\blank)$ satisfy \eqref{eqn:parallel.transport.ODE}, then 
    \begin{align*}
    {}&{}\frac{d}{ds} \alpha^{(1)}(\sfG dt\,; s,t) \circ \alpha^{(0)}(\sfG dt \,; s,t)^{-1} \\
    ={}&{} \alpha^{(1)}(\sfG dt\,; s,t) \circ (\sfad (\sfG(s)) - \sfad(\sfG(s))) \circ  \alpha^{(0)}(\sfG dt \,; s,t)^{-1}=  0. 
    \end{align*}
    In particular, we obtain $\alpha(\sfG dt\,; t,u) \circ \alpha(\sfG dt\,; s,t) = \alpha(\sfG dt\,; s,u)$ for any $s \leq t \leq u$ or $s \geq t \geq u$. 
   \item By (1) and (2), the solution of \eqref{eqn:parallel.transport.ODE} satisfies 
    \[
    \frac{d}{dt}\alpha(\sfG dt \,; s,t)(a) = \big[ \sfG(t), \alpha(\sfG \,; s,t)(a) \big],
    \]
    which means that $\alpha(\sfG \,; t,s) = \alpha(\sfG\,; s,t)^{-1}$. 
    This is verified by differentiating $\alpha(\sfG dt \,; t,u) \circ \alpha(\sfG dt \,; s,t) = \alpha(\sfG dt \,; s,u)$ by the variable $t$.  
    \item As is remarked above, Subsection 3.1 (especially in Theorems 3.5, 3.9) of \cite{nachtergaeleQuasilocalityBoundsQuantum2019} is applicable to our setting. In conclusion, there is a $1$-parameter family of $\ast$-automorphisms of the C*-algebra $\cA_{\Lambda}$ that satisfies \eqref{eqn:parallel.transport.ODE} for any $a \in \cA_{\Lambda}^{\al}$. Note that \eqref{eqn:parallel.transport.ODE} makes sense even if $\alpha(\sfG \,; s,t) $ may not preserve the subalgebra $\cA_{\Lambda}^{\al}$.
\end{enumerate}
\end{rmk}

Moreover, the solution of \eqref{eqn:parallel.transport.ODE} constructed in \cite{nachtergaeleQuasilocalityBoundsQuantum2019} satisfies the following Lieb--Robinson bound (\cite{nachtergaeleQuasilocalityBoundsQuantum2019}*{Corollary 3.6 (i), (iii)}): For $\nu > l_{\Lambda}+1$, $0<\mu<1$, $X, Y \subset \Lambda$ such that $\# X <\infty$ and $X \cap Y =\emptyset$, $a \in \cA_{X}$, and $b \in \cA_Y$, we have
        \begin{align} 
            \| [\alpha(\sfG \, ;t)(a),b]\| 
            \leq {}&{} 2C_\nu^{-1} \cdot  D_{\nu,\mu}(X,Y) \cdot  (\exp (2t I_{\sfG,\nu,\mu}) -1)  \cdot \|a\| \cdot \| b\|,
            \label{eqn:Lieb.Robinson}
            \\
    \big\| \alpha( \sfG \,; t)(a)-\alpha(\Pi_{Y^c} (\sfG)\,; t) (a)  \big\| 
    \leq {}&{}  2C_\nu^{-1} \cdot tI_{\sfG,\nu,\mu} \exp ( 2t I_{\sfG,\nu,\mu} )  \cdot D_{\nu,\mu}(X,Y) \cdot \| a\|, \label{eqn:Lieb.Robinson.truncate}
        \end{align}
        where
        \begin{align*} 
        I_{\sfG,\nu,\mu} \coloneqq{}&{} C_\nu \vvert \Phi_{\sfG}\vvert_{\nu,\mu} \leq C_{\nu}L_{\nu,\mu} \vvert \sfG \vvert_{\nu_3,\mu_1}, \\
        D_{\nu,\mu} (X,Y) \coloneqq {}&{} \sum_{\bm{x} \in X} \sum_{\bm{y} \in Y}f_{\nu,\mu}(\rmd(\bm{x},\bm{y}))\\
        \leq {}&{}  2^{l_{\Lambda}+1}\kappa_{\Lambda}  \cdot \min \{ \# X ,\# Y\} \cdot f_{\nu-l_{\Lambda}-1,\mu}(\mathrm{dist}(X,Y)) <\infty.
        \end{align*}
    Hereafter, we abbreviate $C_{\nu}^{-1}$ since we may assume that the constant $C_{\nu} $ in \cref{prp:Ffunction} is not less than $1$ without loss of generality. 
By using them, in the following \cref{prp:Lieb.Robinson} (2), we prove that the solution by \cite{nachtergaeleQuasilocalityBoundsQuantum2019} indeed preserves the subalgebra $\cA_{\Lambda}^{\al}$, and hence is the unique parallel transport in the sense of \cref{defn:parallel.transport} (3).

\begin{prp}\label{prp:Lieb.Robinson}
    Let $0<\mu <1$ and $\nu >l_{\Lambda}+1$. For $k \in \bN$, let $\mu_k$ be as in \cref{rmk:seminorms} and $\nu_k \coloneqq \nu + k(l_{\Lambda}+1)$. 
    For $f,g \colon \bR_{\geq 0} \to \bR_{>0}$, we write $f \prec g$ if $\max \{ 1,f/g\}$ is integrable.  
    The following hold. 
\begin{enumerate}
    \item Let $\sfG \in \Omega^1([0,1],\fDer ^{\al}_{\Lambda})$. For $X \subset Y \subset \Lambda$ such that $ \# X <\infty$ and $a \in \cA_X^{\al}$, we have
    \begin{align*}
    \| \alpha (\sfG\,; t)(a) - \Pi_Y(\alpha(\sfG\,; t)(a)) \| \leq  2 D_{\nu,\mu}(X,Y^c) \cdot (\exp(2 t C_\nu L_{\nu,\mu} \vvert \sfG \vvert_{\nu_3,\mu_1} ) -1) \cdot \| a\|. 
    \end{align*}
    \item Let $\sfG \in \Omega^1([0,1],\fDer ^{\al}_{\Lambda})$. There is a positive $\bR$-valued function $\Upsilon_{\sfG,\nu,\mu}(t)$ on $\bR_{\geq 0}$ such that  $\Upsilon_{\sfG,\nu,\mu}(t) \prec \exp (t^{{\mu_\star}})$ for any $\mu/\mu_1 <{\mu_\star} <1$ and
            \begin{align*} 
            \| \alpha(\sfG \, ;t)(a) \|_{\bm{x},\nu,\mu } \leq \Upsilon_{\sfG , \nu,\mu}(t) \cdot  \| a \|_{\bm{x},\nu, \mu_1}
            \end{align*} 
    for any $a \in \cA_{{\Lambda}}^{\al}$.
    \item Let $\sfG, \sfG' \in \Omega^1([0,1],\fDer ^{\al}_{\Lambda})$. There is an $\bR_{> 0}$-valued function $\Upsilon^{\mathrm{rel}}_{\sfG,\sfG',\nu,\mu} (t)$ on $\bR_{\geq 0}$ such that $\Upsilon^{\mathrm{rel}}_{\sfG,\sfG',\nu,\mu}(t) \prec \exp(-t^{\mu_\star})$ for any $\mu_3/\mu_4 <{\mu_\star} < 1$ and
    \begin{align*}
    \| \alpha(\sfG  \, ;t)(a) - \alpha(\sfG ' \, ;t)(a)\|_{\bm{x},\nu,\mu} \leq \Upsilon^{\mathrm{rel}} _{\sfG,\sfG',\nu,\mu}(t) \cdot \vvert \sfG - \sfG'\vvert_{\nu_1,\mu_3}  \cdot \| a\|_{\bm{x},\nu_1,\mu_4} 
    \end{align*}
    for any $a \in \cA_{\Lambda}^{\al}$.
\end{enumerate}
\end{prp}

\begin{proof}
    By \cref{rmk:CE.average} and the Lieb--Robinson bound \eqref{eqn:Lieb.Robinson}, (1) is verified as
    \begin{align*}
    {}&{} \big\| \alpha (\sfG \,; t)(a)- \Pi_{Y}\big( \alpha (\sfG \,; t)(a)\big) \big\|\\
    \leq {}&{}  \sup_{u \in \cU(\cA_
    {Y^c})} \| [u, \alpha(\sfG\,;t)(a)]\| \leq   2 D_{\nu ,\mu}(X,Y^c) \cdot (\exp(2 C_\nu L_{\nu,\mu}\vvert \sfG \vvert_{\nu_3,\mu_1} t) -1 )  \cdot \| a\|. 
    \end{align*}

We then show (2) following the line of \cite{moonAutomorphicEquivalenceGapped2020}*{Lemmas 4.6, 4.11}. Since $\alpha(\sfG \,; t)$ and $ \Pi_{\bm{x},r}$ are contractive with respect to the operator norm, for any $r \in \bR_{>0}$ we have
\begin{align*}
    {}&{} f_{\nu,\mu}(r)^{-1} \cdot \| \alpha(\sfG \,; t)(a - \Pi_{\bm{x},r/2}(a)) - \Pi_{\bm{x},r}\big(\alpha(\sfG \,; t) (a - \Pi_{\bm{x},r/2}(a)) \big)\| \\
    \leq {}&{} 2f_{\nu,\mu}(r)^{-1} \cdot \| a - \Pi_{\bm{x},r/2}(a)\|
    \leq 2 c_{\nu,\mu,2} \cdot \| a\| _{\bm{x},\nu,\mu_1},
\end{align*}
where $c_{\nu,\mu,2}$ is the constant in \eqref{eqn:const.small.c}. Let
\begin{align*} 
    R_{\sfG,\nu,\mu}(t) \coloneqq 2 \cdot (2tC_{\nu_1} L_{\nu_1,\mu_1} \vvert \sfG \vvert_{\nu_4,\mu_2} )^{1/\mu_1}, \quad \text{ that is, } \, (R_{\sfG,\nu,\mu}(t)/2)^{\mu_1} = 2tC_{\nu_1}L_{\nu_1,\mu_1} \vvert \sfG \vvert_{\nu_4,\mu_2}.
\end{align*}

For any $t>0$, we have 
\begin{align*}
    {}&{} f_{\nu,\mu}(r)^{-1} \cdot \big\|  \alpha(\sfG  \, ;t)\big( \Pi_{\bm{x},r/2}(a) \big) - \Pi_{\bm{x},r}\big( \alpha(\sfG  \, ;t) \big( \Pi_{\bm{x},r/2}(a) \big) \big) \big\| \\
    \leq {}&{}
    f_{\nu,\mu}(r)^{-1} \cdot 2 \| a \| \leq 2f_{\nu,\mu}(R_{\sfG,\nu,\mu}(t))^{-1} \cdot \|a\|
\intertext{if $r \leq R_{\sfG,\nu,\mu}(t)$, and} 
    {}&{} f_{\nu,\mu}(r)^{-1} \cdot \big\|  \alpha(\sfG  \, ;t)\big( \Pi_{\bm{x},r/2}(a) \big) - \Pi_{\bm{x},r}\big( \alpha(\sfG  \, ;t) \big( \Pi_{\bm{x},r/2}(a) \big) \big) \big\| \\
    \leq {}&{}
    f_{\nu,\mu}(r)^{-1} \cdot 2 \cdot 2^{l_{\Lambda}+1}\kappa_{\Lambda}  \cdot \kappa_{\Lambda} (1+r/2)^{l_{\Lambda}}\cdot f_{\nu,\mu_1}(r/2) \cdot \exp (2tC_{\nu_1}L_{\nu_1,\mu_1}  \vvert \sfG \vvert_{\nu_4,\mu_2} )  \| a \| \\
    \leq {}&{} 2^{l_{\Lambda}+2}\kappa_{\Lambda}^2  \cdot  f_{\nu - l_{\Lambda},\mu_1}(r/2) \cdot f_{\nu,\mu}(r)^{-1} \cdot f_{\mu_1}(R_{\sfG,\nu,\mu}(t))^{-1} \cdot \| a \|
\end{align*}
if $r \geq R_{\sfG,\nu,\mu}(t)$. 
Here, we use (1) for the first inequality of the latter case. 
Moreover, by fixing $r_0 >0$ so that $f_{\nu-l_{\Lambda},\mu_1}(r/2)/f_{\nu,\mu}(r)$ is monotonously decreasing on $\bR_{\geq r_0}$, we obtain that  
\begin{align*}
    {}&{} \sup_{r \geq R_{\sfG,\nu,\mu}(t)} f_{\nu - l_{\Lambda},\mu_1}(r/2) \cdot f_{\nu,\mu}(r)^{-1} \cdot  f_{\mu_1}(R_{\sfG,\nu,\mu}(t))^{-1} \\
    \leq {}&{} 
       \begin{cases}
           \big( \sup_{r >0} f_{\nu-l_{\Lambda},\mu_1}(r/2) / f_{\nu,\mu}(r) \big) \cdot f_{\mu_1}(r_0/2)^{-1} & \text{ if $R_{\sfG,\nu,\mu}(t) \leq r_0/2$,} \\
           f_{\nu,\mu} \big( 2 R_{\sfG,\nu,\mu}(t) \big)^{-1} & \text{ if $R_{\sfG,\nu,\mu}(t) \geq r_0/2$.}
       \end{cases}
\end{align*}
Moreover, since $ f_{\nu,\mu}(2R_{\sfG,\nu,\mu}(t))^{-1}$ is a multiple of a polynomial on $t$ with \begin{align*}
    f_\mu(2R_{\sfG,\nu,\mu}(t))^{-1} = \exp \Big( 2^{2\mu} (2C_{\nu_1}L_{\nu_1,\mu_1} \vvert \sfG \vvert_{\nu_4,\mu_2})^{\mu/\mu_1} \cdot t^{\mu/\mu_1} \Big) ,
\end{align*}
we have $f_{\nu,\mu}(2R_{\sfG,\nu,\mu}(t))^{-1} \prec \exp (t^{\mu_{\star}})$ for any $\mu/\mu_1 < \mu_\ast < 1$. 
Therefore, the constant 
\begin{align}
\begin{split}
    \Upsilon_{\sfG,\nu,\mu}(t) \coloneqq \max \{ 2 , 2^{l_{\Lambda}+2}\kappa_{\Lambda}^2\} \cdot (\ast){}&{}, \\
    (\ast) = \max \Big\{  \big( \sup_{r >0} f_{\nu-l_{\Lambda},\mu_1}(r/2) / f_{\nu,\mu}(r) \big) \cdot f_{\mu_1}(r_0/2)^{-1} {}&{},  f_{\nu,\mu} \big( 2R_{\sfG,\nu,\mu}(t) \big)^{-1} \Big\} 
    \end{split}\label{eqn:Upsilon}
\end{align}
fulfills the desired conditions. Note that this $\Upsilon_{\sfG,\nu,\mu}(t)$ increases monotonously as a function on $t$. 

In the proof of (3), we use 
\begin{align}
\begin{split}
    \alpha(\sfG\,; t)(a) - \alpha(\sfG'\,; t)(a) ={}&{} \int_0^t \frac{d}{ds} \alpha(\sfG \,; s,t) \circ \alpha(\sfG' \,; 0,s) (a) ds \\
    ={}&{} \int_{[0,t]} \alpha(\sfG \,; s,t) \big( \big[ \sfG'(s) - \sfG(s) , \alpha(\sfG' \,;0,s)(a) \big] \big)ds.
    \end{split} \label{eqn:relative.evolution}
\end{align}
Here, the right hand side is the integration of a $\cA_{\Lambda}^{\al}$-valued $1$-form. This and \cref{lem:almost.local.continuity} proves (3) as
\begin{align}
\begin{split}
    {}&{} \| \alpha(\sfG \,; t)(a) - \alpha(\sfG' \,; t)(a)\|_{\bm{x},\nu,\mu} \\
    \leq {}&{} \int_0^t \big\| \alpha(\sfG \,; s,t) \circ \sfad \big( \sfG'(s) - \sfG(s)\big) \circ  \alpha(\sfG' \,; 0, s)(a) \big\|_{\bm{x},\nu,\mu} \\
    \leq {}&{} \int_0^t \Upsilon_{\sfG,\nu,\mu}(t-s) \cdot \big\| \sfad \big( \sfG'(s) - \sfG(s)\big) \circ \alpha(\sfG'\,;0, s)(a) \big\|_{\bm{x},\nu,\mu_1} ds \\
    \leq {}&{} d_{\nu,\mu_1}\int_0^t \Upsilon_{\sfG,\nu,\mu}(t-s) \cdot \vvert \sfG-\sfG' \vvert_{\nu_1,\mu_3} \cdot \| \alpha(\sfG'\,;0, s)(a) \|_{\bm{x},\nu_1,\mu_3}ds  \\
    \leq {}&{} d_{\nu,\mu_1}\int_0^t \Upsilon_{\sfG,\nu,\mu}(t-s) \cdot \vvert \sfG-\sfG' \vvert_{\nu_1,\mu_3} \cdot \Upsilon_{\sfG',\nu_1,\mu_3}(s) \cdot  \| a \|_{\bm{x},\nu_1,\mu_4} ds  \\
    \leq {}&{} \Upsilon^{\mathrm{rel}}_{\sfG,\sfG',\nu,\mu}(t) \cdot \vvert \sfG-\sfG'\vvert_{\nu_1,\mu_3} \cdot \|a\|_{\bm{x},\nu_1,\mu_4},
\end{split} \label{eqn:relative.Lieb.Robinson}
\end{align}
where $d_{\nu,\mu_1}$ is the constant in \cref{lem:almost.local.continuity} and
\begin{align*}
    \Upsilon^{\mathrm{rel}}_{\sfG,\sfG',\nu,\mu}(t) \coloneqq t \cdot d_{\nu,\mu_1}  \Upsilon_{\sfG,\nu,\mu}(t)\Upsilon_{\sfG',\nu_1,\mu_3}(t).
\end{align*}
Since $\mu/\mu_1 = 2\mu/(1+\mu)$ is monotonously increasing as $\mu \nearrow 1$, the function $\Upsilon^{\mathrm{rel}}_{\sfG,\sfG',\nu,\mu}(t) \cdot e^{-t^{\mu_{\star}}}$ is integrable for any ${\mu_\star} > \mu_3/\mu_4$, which concludes that $\Upsilon^{\mathrm{rel}}_{\sfG,\sfG',\nu,\mu}(t) \prec \exp(t^{\mu_\star})$. 
\end{proof}
By \cref{prp:Lieb.Robinson} (2), an LG automorphism $\alpha(\sfG \, ;t) $ acts on $\cA_{\Lambda}^{\al}$ continuously with respect to the almost local topology. 
Moreover, since the estimate in \cref{prp:Lieb.Robinson} (2) is uniform on $\bm{x} \in \Lambda$, the automorphism $\alpha (\sfG\,; t)$ also acts on the space of UAL derivations $\fDer ^{\al}_{\Lambda}$ by 
\begin{align}
    \alpha(\sfG \, ;t)( \sfG') \in \fDer ^{\al}_{\Lambda}, \quad \big( \alpha(\sfG \, ;t)( \sfG')\big)_{\bm{x}} \coloneqq \alpha(\sfG \, ;t)(\sfG'_{\bm{x}}).
    \label{eqn:UAL.acting.derivation}
\end{align}

\begin{lem}\label{lem:smooth}
Let the space $\cB(\cA_{\Lambda}^{\al})$ of bounded linear maps on the Fr\'{e}chet space $\cA_{\Lambda}^{\al} $ be equipped with the weak topology generated by seminorms $\sfT \mapsto \|\sfT(a)\|_{\bm{x},f}$ for any $f \in \cF$ and $a \in \cA_{\Lambda}^{\al}$. 
\begin{enumerate}
    \item A smooth function $\sfG \colon \sM \to \fDer ^{\al}_{\Lambda}$ in the sense of \cref{defn:parallel.transport} (1) induces a Fr\'{e}chet smooth map from $\sM$ to the set $\mathrm{Der}(\cA_{\Lambda}^{\al}) \subset \cB(\cA_{\Lambda}^{\al})$ of continuous derivations on $\cA_{\Lambda}^{\al}$.
    \item A smooth map $\alpha \colon \sM \to \cG_{\Lambda}$ in the sense of \cref{defn:parallel.transport} (4) is Fr\'{e}chet smooth map from $\sM $ to the group $\Aut(\cA_{\Lambda}^{\al }) \subset \cB(\cA_{\Lambda}^\al) $ of continuous $\ast$-automorphisms on $\cA_{\Lambda}^{\al}$.
    \item For smooth maps $\alpha \colon \sM \to \cG_{\Lambda}$ and $\sfG \colon \sM \to \fDer ^{\al}_{\Lambda}$, the following map is also smooth;
    \begin{align*} 
    \sM \ni p \mapsto \alpha_{p}(\sfG(p))  \in \fDer ^{\al}_{\Lambda}.
    \end{align*}
    \item For smooth maps $\alpha \colon \sM \to \cG_{\Lambda}$ and $\sfH \colon \sM \to \fH_{\Lambda}^{\al}$, the following map is also smooth;
    \begin{align*} 
    \sM \ni p \mapsto \alpha(\sfH ) (p) \coloneqq (\alpha_{p}(\sfH(p)), \omega_{\sfH(p)}\circ \alpha_p^{-1})  \in \fH_{\Lambda}^{\al}.
    \end{align*}
\end{enumerate}
\end{lem}

\begin{proof}
The claim (1) follows from \eqref{lem:almost.local.continuity} as
    \begin{align*}
    \| \partial^I([\sfG,a]) \|_{\bm{x},\nu,\mu} ={}&{}  \| [\partial^I\sfG,a] \|_{\bm{x},\nu,\mu}
    \leq d_{\nu,\mu} \cdot \vvert \partial^I\sfG \vvert_{\nu_1,\mu_2} \cdot \| a \|_{\bm{x},\nu_1,\mu_2}\\
    \leq {}&{}d_{\nu,\mu} \cdot \vvert \sfG_{\bm{y}} \vvert_{\sU,C^k,\nu_1,\mu_2} \cdot \| a \|_{\bm{x},\nu_1,\mu_2}.
    \end{align*}

The proof of (2) follows \cite{kapustinLocalNoetherTheorem2022}*{Proposition E.2}. First, by \eqref{eqn:relative.evolution}, $p \mapsto \alpha_p(\sfG\,; 1)(a)$ is continuous. Moreover, again by \eqref{eqn:relative.evolution}, we have 
\begin{align*}
    {}&{} \frac{\alpha_{p_s}(\sfG \,; 1)(a) - \alpha_{p_0}(\sfG\,; 1)(a)}{s} - \int_{u \in [0,1]}\alpha_{p_s}(\sfG\,; u,t) \circ \sfad (\partial_{s}\sfG (p_0,u)) \circ \alpha_{p_0}(\sfG \,; u) \\
    ={}&{} \int_{u \in [0,1]} \alpha_{p_s}(\sfG \,; u,t) \circ \sfad \bigg(\frac{\sfG (p_s,u) - \sfG(p_0,u)}{s}  - \partial_s \sfG (p_0,u) \bigg)  \circ \alpha_{p_0}(\sfG \,; u)(a) 
\end{align*}
for any smooth curve $s \mapsto p_s$ in $\sM$. 
By the same bound as \eqref{eqn:relative.Lieb.Robinson}, we obtain that the right hand side converges to $0 $ as $s \to 0$, and hence
\[
    \frac{d}{ds}\bigg|_{s=0} \alpha_{p_s}(\sfG\,; 1)(a) = \int_{u \in [0,1]} \alpha_{p_0}(\sfG\,; u,t) \circ \sfad (\partial_{s}\sfG) \circ \alpha_{p_0}(\sfG \,; u)(a),
\]
which is again continuous. Consequently, the function $\alpha(\sfG\,; 1)(a)$ is of class $C^1$. Iterating the same argument shows that $\alpha(\sfG \,; 1)(a)$ is a smooth function.

The above constants $C_1, C_2$ are independent of $\bm{x} \in \Lambda$. This also shows (3). 
Finally, (4) follows from (3), and the smoothness of 
\begin{align*}
    \sM \ni p \mapsto \omega_{\alpha_p(\sfH(p))} = \omega_{\sfH(p)} \circ \alpha_p^{-1}  \in \fS(\cA_{\Lambda}^{\al}),
\end{align*}
which is concluded from (2). 
\end{proof}

\begin{prp}\label{prp:Lieb.Robinson.Ck}
    Let $\sU$ be an open subset of $\bR^n$. The following hold. 
    \begin{enumerate}
        \item Let $\sfG  \in \Omega^1(\sU \times [0,1], \fDer ^{\al}_{\Lambda}) $.  Then there is a positive $\bR$-valued function $\Upsilon_{\sfG,\nu,\mu}^{(k)}(t)$ on $\bR_{\geq 0}$ such that $\Upsilon_{\sfG,\nu,\mu}^{(k)}(t) \prec \exp (t^{{\mu_\star}})$ for some ${\mu_\star} <1$ and 
        \begin{align*}
            \| \alpha_p(\sfG \,; t)(a) \|_{\sU,C^k,\bm{x},\nu,\mu} \leq \Upsilon^{ (k)}_{\sfG,\nu,\mu} (t) \cdot \| a \|_{\bm{x},\nu_k,\mu_{3k+1}}.
        \end{align*}
        \item Let $\sfG_1, \sfG_2 \in \Omega^1(\sU \times [0,1], \fDer ^{\al}_{\Lambda}) $.  Then there is a positive $\bR$-valued function $\Upsilon_{\sfG_1,\sfG_2,\nu,\mu}^{\mathrm{rel},(k)}(t)$ on $\bR_{\geq 0}$ such that $\Upsilon_{\sfG_1,\sfG_2,\nu,\mu}^{\mathrm{rel},(k)}(t) \prec \exp (t^{{\mu_\star}})$ for some ${\mu_\star} <1$ and 
        \begin{align*}
            \| \alpha_p(\sfG_1 \,; t)(a) - \alpha_p (\sfG_2 \,; t)(a) \|_{C^k,\bm{x},\nu,\mu} \leq \vvert \sfG_1 - \sfG_2 \vvert_{\sU,C^k,\bm{x},\nu_{k},\mu_{3k+1}} \cdot \Upsilon^{\mathrm{rel}, (k)}_{\sfG_1,\sfG_2,\nu,\mu} (t) \cdot \| a \|_{\bm{x},\nu_k,\mu_{3k+1}}.
        \end{align*}
    \end{enumerate}
\end{prp}
\begin{proof}
    We first show (1). Set $\sfG_t \coloneqq \iota_{\frac{d}{dt}} \sfG$. By \cref{lem:smooth} (2) and its proof, the differential $\partial^I \alpha (\sfG\,;t)(a)$ is given by a finite sum of the terms of the form 
    \begin{align*}
        {}&{}\int_{\Delta_j(t)} \alpha_p(\sfG \,; t_j,t) \circ \sfad (\partial^{I_j} \sfG_t(p,t_j) ) \circ 
        \alpha_p(\sfG \,; t_{j-1},t_j) \circ \sfad( \partial^{I_{j-1}} \sfG_t(p,t_{j-1})) \circ \\
        {}&{} \hspace{15ex} \cdots \circ \alpha_p(\sfG \,; t_j,1) \circ \sfad(\partial^{I_1}\sfG_t(p,t_1) ) \circ \alpha_p(\sfG \,; 0,t_1) \ dt_1 \cdots dt_j,
    \end{align*}
    where $\Delta_j(t) \coloneqq \{ (t_1,\cdots,t_j) \mid 0 \leq t_1 \leq \cdots \leq t_j \leq t\}$ and $I_l$'s are multi-indices such that $|I_l| \geq 1$ and $I_1 + \cdots + I_j=I$. It is bounded by a constant multiple of 
    \begin{align*}
    \bigg( \prod_{l=0}^{j-1} \Upsilon_{\sfG,\nu_{l},\mu_{3l}}(t_{j-l+1}-t_{j-l}) \cdot d_{\nu_l,\mu_{3l+1}} \vvert \partial^{I_{j-l}}\sfG \vvert_{\nu_{l+1},\mu_{3l+3}}\bigg)  \cdot \Upsilon_{\sfG,\nu_j,\mu_{3j}}(t_1) \cdot \| a \|_{\nu_{j},\mu_{3j+1}}
    \end{align*}
    where $t_{j+1} \coloneqq t$. Since $ \vvert \partial^{I_{j-l}}\sfG \vvert_{\nu_{l+1},\mu_{3l+3}} \leq \vvert \sfG \vvert_{\sU,C^k,\nu_j,\mu_{3j}}$, there are constants $C_{I_1,\cdots,I_j}>0$ that does not depend on $\sfG$ such that 
    \[
    \Upsilon^{(k)}_{\sfG,\nu,\mu}(t) \coloneqq \sum_{\substack{|I_m| \geq 1 \\ |I_1 + \cdots + I_j| \leq k }} C_{I_1,\cdots,I_j} \cdot \Big(\prod_{l=0}^{j} \Upsilon_{\sfG,\nu_{l},\mu_{3l}}(t)\Big) \cdot  \vvert \sfG \vvert_{\sU,C^k,\nu_{j},\mu_{3j}}^j 
    \]
    is the desired function on $t$. 

    The claim (2) is proved similarly. Set $\sfG_{i,t} \coloneqq \iota_{\frac{d}{dt}} \sfG_i$. The differential $\partial^I (\alpha (\sfG_1\,;t)(a) - \alpha (\sfG_2\,;t)(a))$ is given by a finite sum of the terms of the form 
    \begin{align*}
        {}&{}\int_{\Delta_j(t)} \alpha_p(\sfG_1 \,; t_j,t) \circ \sfad (\partial^{I_j}\sfG_{1,t}(p,t_j) ) \circ 
        \alpha_p(\sfG_1 \,; t_{j-1},t_j) \circ \sfad (\partial^{I_{j-1}}\sfG_{1,t}(p,t_{j-1})) \circ \\
        {}&{} \hspace{15ex} \cdots \circ \alpha_p(\sfG_1 \,; t_m, t_{m+1}) \circ \sfad ( \partial^{I_m} (\sfG_{1,t}(p,t_m) - \sfG_{2,t}(p,t_m))) \circ \alpha_p(\sfG_2\,; t_{m-1},t_m) \circ  \\
        {}&{} \hspace{15ex} \cdots \circ \alpha_p(\sfG \,; t_j,1) \circ \sfad (\partial^{I_1}\sfG_{2,t}(p,t_1)) \circ \alpha_p(\sfG_2 \,; 0,t_1) \ dt_1 \cdots dt_j ,
    \end{align*}
    where $I_l$'s are multi-indices such that $|I_l| \geq 1$ unless $l=m$ and $I_1 + \cdots + I_j=I$. Thus, there are constants $C_{I_1,\cdots,I_j,m}'$ that does not depend on $\sfG_1, \sfG_2$ such that 
    \begin{align*}
    \Upsilon^{\mathrm{rel},(k)}_{\sfG_1,\sfG_2,\nu,\mu}(t) \coloneqq {}&{} 
    \sum_{\substack{|I_l| \geq 1 \text{ if $l \neq m$} \\ |I_1 + \cdots + I_j| \leq k}} C_{I_1,\cdots,I_j,m}' \cdot  \vvert \sfG_1 \vvert_{\sU,C^k,\nu_{j},\mu_{3j}}^{m'}    \cdot  \vvert \sfG_2 \vvert_{\sU,C^k,\nu_{j},\mu_{3j}}^{j-m'-1}\\
    {}&{} \qquad  \qquad \qquad  \cdot 
    \bigg( \prod_{l=0}^{m'}\Upsilon_{\sfG_1,\nu_{l},\mu_{3l}}(t)\bigg)   \cdot 
    \bigg( \prod_{l=m'+1}^{j} \Upsilon_{\sfG_2,\nu_{l},\mu_{3l}}(t) \bigg),
    \end{align*}
    where $m'\coloneqq j-m$, is the desired function on $t$. 
\end{proof}

 \begin{lem}\label{prp:relative.evolution.is.evolution}
    Let $\sfG, \sfG' \in \Omega^1([0,1], \fDer ^{\al}_{\Lambda})$. Then 
    \begin{align*}
    \alpha (\sfG' \,; t)^{-1}  \circ \alpha(\sfG \,; t) = \alpha\big( \alpha(\sfG' \,; t)^{-1} \big(\sfG - \sfG' \big) \,; t \big).
    \end{align*}
\end{lem}
\begin{proof}
    The well-definedness of the right hand side follows from \cref{lem:smooth} (3). Now, since
    \begin{align*}
        \frac{d}{dt} \Big( \alpha(\sfG' \,; t)^{-1} \circ \alpha(\sfG\,; t) \Big)= {}&{} \alpha(\sfG' \,; t )^{-1} \circ \sfad ( \sfG(t) - \sfG'(t) ) \circ  \alpha (\sfG\,; t) \\
        = {}&{} \sfad \big( \alpha(\sfG' \,; t)^{-1} \big(\sfG(t) - \sfG'(t)\big) \big)  \circ \alpha(\sfG' \,; t )^{-1} \circ \alpha (\sfG\,; t),
    \end{align*} 
    the left and the right hand sides satisfy the same ordinary differential equation. 
\end{proof}
  
\begin{defn}[{\cites{hastingsLiebschultzmattisHigherDimensions2004,hastingsQuasiadiabaticContinuationQuantum2005,nachtergaeleQuasilocalityBoundsQuantum2019,moonAutomorphicEquivalenceGapped2020}}]\label{defn:automorphic.connection}
Let $\sfH \colon \sM  \to \fH_{\Lambda}^{\al}$ be a smooth family of gapped UAL Hamiltonians. We fix $0<v<\gapone$. 
The \emph{adiabatic connection $1$-form} associated with $\sfH$ is defined by
\begin{align*}
    \sfG_{\sfH,\bm{x}}(p)\coloneqq  i\int_{-\infty} ^{\infty} \bigg( \int_0^t \tau_{\sfH(p) ,u} \big( d \sfH_{\bm{x}}(p) \big) du \bigg)  w_v (t) dt. 
\end{align*}  
\end{defn}
Here, the exterior derivative $d\sfH$ is a smooth $\fDer ^{\al}_{\Lambda}$-valued $1$-form defined as
\begin{align*} 
    d\sfH = \sum_{j=1}^n  \frac{\partial \sfH }{\partial x_j} dx_j \in \Omega^1(\sM, \fDer ^{\al}_{\Lambda})
\end{align*}
and $w_v (t)$ is a positive even smooth function such that $\int_{-\infty}^{\infty} w_v(t) dt = 1 $, the Fourier transform $\hat{\omega}_v$ is supported on $[-v, v]$, and 
\begin{align}
    w_v(t) \leq 2(ev)^{2}t \cdot \exp \Big( -\frac{2}{7} \cdot \frac{vt}{\log(vt)^2} \Big) \prec (1+t)^{-\nu} \exp(-t^\mu)\label{eqn:w.function.bound}
\end{align}
for any $\nu \geq 0$ and $0< \mu <1$. Such a function is constructed by \cite{bachmannAutomorphicEquivalenceGapped2012}*{Lemma 2.3} as 
\begin{align*} w_v(t)=v w_1(v t), \quad w_1(t)\coloneqq c \cdot \prod _{n=1}^\infty \bigg( \frac{\sin(a_nt)}{a_nt} \bigg)^2, \quad a_n\coloneqq \frac{a_1}{n \log(n)^2}, \end{align*}
where the normalization constants $a_1$ and $c$ are chosen so that $\sum_{n=1}^\infty a_n = 1/2$ and $\int_{-\infty}^\infty w_1(t)dt=1$. 
\begin{lem}\label{lem:adiabatic.well-defined}
The above $\sfG_{\sfH}$ defines a smooth $1$-form on $\sM$.     
\end{lem}
\begin{proof}
For any relatively compact local chart $(\sU \, ;x_1,\cdots ,x_n)$ of $\sM $, the $C^k$-norm of $\sfG_{\sfH} =\sum_j \sfG_{\sfH,j} dx_j $ is bounded above by
    \begin{align*}
    \vvert \sfG_{\sfH}\vvert_{\sU,C^k,\nu,\mu} \coloneqq {}&{} \max_{j\in \{ 1,\cdots, n\} }\vvert \sfG_{\sfH ,j} \vvert_{\sU,C^k,\nu,\mu} \\
    \leq {}&{} \max_{j\in \{ 1,\cdots, n\} } 2\int_{0} ^{\infty} \bigg( \int_0^t \Upsilon_{\sfH ,\nu,\mu}^{(k)}(u) \cdot \vvert \partial_{x_j} \sfH(p)\vvert_{\sU,C^{k},\nu_k,\mu_{3k+1}} du \bigg)  w_v (t) dt \\
    \leq  {}&{} \vvert \sfH \vvert _{\sU,C^{k+1},\nu_k,\mu_{3k+1}} \cdot 2\int_{0}^\infty t\Upsilon_{\sfH,\nu,\mu}^{(k)}(t) \cdot w_v(t)dt,
    \end{align*}
    and the last integral is finite by \cref{prp:Lieb.Robinson.Ck} and \eqref{eqn:w.function.bound}. 
\end{proof}

For a smooth path $c \colon [0,1] \to \sM$, the parallel transport $\alpha_{t}\coloneqq \alpha_c(\sfG_\sfH  \,; t)^{-1} \colon [0,1] \to \cG_{\Lambda}$ satisfies the differential equation
\begin{align*}
    \frac{d}{dt}\alpha_{t}^{-1}(a) = {}&{} \iota_{\dot{c}} \big( \big[ \sfG_{\sfH}, \alpha_t^{-1}(a) \big]\big) \\
    ={}&{} i\int_{-\infty}^\infty \bigg( \int_0^s \big[ \tau_{\sfH (c(t)),u}\big( \partial_t ( \sfH _{c(\bullet)})\big) , \alpha_t^{-1}(a) \big] du\bigg) w_v(s) ds\\
    ={}&{} \int_{-\infty}^\infty \bigg( \int_0^s  \tau_{\sfH (c(t)),u} \circ \sfad(i\partial_t\sfH_{c(\bullet)}) \circ \tau_{\sfH (c(t)),u}^{-1} ( \alpha_t^{-1}(a)) du\bigg) w_v(s) ds
\end{align*}
with the initial condition $\alpha_0=\id$ (compare the right hand side with \cite{moonAutomorphicEquivalenceGapped2020}*{(2.11)}).

A UAL derivation $\sfG \in \Omega ^1(\sM,\fDer ^{\al}_{\Lambda})$ acts on the space of continuous linear functionals $(\cA_{\Lambda}^{\al})^*$ by $(\sfG \omega) (a) \coloneqq \omega ([\sfG , a])$.  
For a connection $1$-form $\sfG \in \Omega^1(\sM,\fDer ^{\al}_{\Lambda})$, the associated covariant derivative also acts on the space of smooth functions taking value in $(\cA_{\Lambda}^{\al})^*$ as
\begin{align*}
    d  + \sfG \colon \Omega^0(\sM , (\cA_{\Lambda}^{\al})^* ) \to \Omega^1(\sM, (\cA_{\Lambda}^{\al})^* ).
\end{align*}
Let $\sM =[0,1]$ and  $\omega \colon [0,1]\to (\cA_{\Lambda}^{\al})^*$. Then we have
\begin{align}
\begin{split}
    \frac{d}{dt} \big( \omega \circ \alpha(\sfG \,; t) \big) = {}&{} \frac{d \omega }{dt}  \circ \alpha(\sfG \,; t)  + \omega \circ \sfad (\sfG) \circ \alpha(\sfG\,; t) \\
    ={}&{}
    \Big( \iota_{\frac{\partial}{\partial t}}(d + \sfG ) \omega \Big)   \circ \alpha(\sfG\,; t) .
\end{split}\label{eqn:parallel.transport.connection}
\end{align}
\begin{thm}[{\cite{moonAutomorphicEquivalenceGapped2020}}]\label{thm:automorphic.equivalence}
For a smooth map $\sfH \colon \sM \to \fH_{\Lambda}^{\al}$, we have $(d + \sfG_{\sfH}) \omega_\sfH =0$. 
\end{thm}
\begin{proof}
By \eqref{eqn:parallel.transport.connection}, the statement is equivalent to
\begin{align} 
    \frac{d}{dt} (\omega_{\sfH (c(t))} \circ \alpha_c(\sfG_{\sfH} \,; t) ) =0 \label{eqn:GS.automorphic}
\end{align}
for any smooth map $c \colon [0,1] \to \sM$, which is exactly the same as \cite{moonAutomorphicEquivalenceGapped2020}*{Theorem 1.3}. Notably, \eqref{eqn:GS.automorphic} gives an equivalent form of \cref{thm:automorphic.equivalence}: For any smooth curve $c \colon [0,1] \to \sM$, the transported Hamiltonian $\alpha_c(\sfG_{\sfH} \,; t) (\sfH(c(t)))$ has constant ground state. 

There are two remarks on the application of \cite{moonAutomorphicEquivalenceGapped2020}*{Theorem 1.3} to our setting. 
\begin{enumerate}
    \item The uniqueness assumption of the ground state, although it is assumed in \cite{moonAutomorphicEquivalenceGapped2020}, is not used anywhere in the proof by Moon--Ogata. 
    What is actually proved by them is the following statement: If one has a smooth $1$-parameter family of UAL Hamiltonians $\sfH(t)$ and a smooth $1$-parameter family of states $\omega_t$ such that each $\omega_t$ is a non-degenerate gapped ground state of $\sfH(t)$, then $\omega _t = \omega_0 \circ \alpha (\sfG _{\sfH} \,; t)^{-1}$ holds. This is exactly what we want. 
    Indeed, applying their proof to the constant family of UAL Hamiltonians $\sfH(t)=\sfH$, we obtain the following rigidity of $\omega$; there is no non-constant smooth path $\omega_t$ of non-degenerate ground state that makes $(\sfH,\omega_t) $ a gapped UAL Hamiltonian with gap $\delta$ uniformly on $t \in [0,1]$. 
    Therefore, once a choice of gapped ground states $\omega_0$ of $\sfH(0)$ is fixed, then $\omega_0 \circ \alpha(\sfG_{\sfH} \,; t)$ is the unique smooth family of non-degenerate gapped ground states starting from $\omega_0$, if any.  
    \item As is noted in \cref{rmk:MO}, the brick decomposition $\{ \Phi_\sfH(B_\rho) \}$ of $\sfH$ satisfies \cite{moonAutomorphicEquivalenceGapped2020}*{Assumption 1.2} except for (ii), the uniform finiteness of the range of the Hamiltonian. 
This assumption is used to show the condition 2 of \cite{moonAutomorphicEquivalenceGapped2020}*{Lemma 2.1} (see also \cite{moonAutomorphicEquivalenceGapped2020}*{Lemmas 4.6, 4.11}). The corresponding estimate in our setting is proved in \cref{prp:Lieb.Robinson} (2) and \eqref{eqn:w.function.bound}.  \qedhere 
\end{enumerate}
\end{proof}

\subsection{Adiabatic interpolation}\label{subsection:adiabatic.interpolation}
We use \cref{thm:automorphic.equivalence} as follows. 
Suppose that we have a smooth path of gapped UAL Hamiltonians $\sfH \colon [0,1] \to \fH_{\Lambda}^{\al}$. 
We truncate the adiabatic connection 1-form $\sfG_{\sfH}$ to the half space and apply the associated parallel transport to $\sfH (0)$. 
This construction yields a new Hamiltonian $\vartheta \sfH$. 
Its ground state is transported only in one half of the space, where it is asymptotically equal to the ground state of $\sfH_1$. 
On the other half, it remains unchanged and is asymptotically equal to the ground state of $\sfH_0$.
With a little more effort, this construction can be refined to a spatial interpolation of Hamiltonians itself if $\sfH_1 $ is the trivial Hamiltonian \eqref{eqn:trivial.Hamiltonian}.

For $\sfG =(\sfG_{\bm{x}})_{\bm{x} \in \Lambda } \in \fDer ^{\al}_{\Lambda}$ and $Z \subset \Lambda$, we write its truncation as
\begin{align*} 
    \Pi_Z(\sfG) = (\Pi_Z(\sfG_{\bm{x}}) )_{\bm{x} \in \Lambda} \in \fDer ^{\al}_{\Lambda}.
\end{align*}
Indeed, we have $\| \Pi_Z(\sfG_{\bm{x}}) - \Pi_{\bm{x},r} (\Pi_Z(\sfG_{\bm{x}})) \| \leq \| \sfG_{\bm{x}} - \Pi_{\bm{x},r}(\sfG_{\bm{x}}) \| $ by $\Pi_Z \Pi_{\bm{x},r} = \Pi_{\bm{x},r}\Pi_Z$.

\begin{lem}\label{lem:LGA.cone.decomposition}
    Let $\Lambda$ be weakly uniformly discrete in $\bR^{l_{\Lambda}}$, and let $Y\subset \Lambda$. 
    For a smooth $1$-form of UAL derivations $ \sfG \in \Omega^1 (\sM \times [0,1]_t,  \fDer^{\al}_{\Lambda} )$, let 
    \begin{align*}
    \Pi_{Y}^{\star}(\sfG) (p,t) \coloneqq \alpha_p(\Pi_{Y^c}(\sfG)  \,; t)^{-1}\big( \sfG - \Pi_{Y^c}(\sfG) \big). 
    \end{align*}
    Then we have
    \begin{align*}
        \alpha_p(\sfG\,; t) = \alpha_p(\Pi_{Y^c}(\sfG)\,; t) \circ \alpha (\Pi_{Y}^{\star}(\sfG)\,; t). 
    \end{align*}
    Moreover, there is a function $\Upsilon^{\star, (k)}_{\sfG,\nu,\mu}(t)$ such that $\Upsilon_{\sfG,\nu,\mu}^{\star, (k)}(t) \prec e^{-t^{\mu_{\star}}}$ for sufficiently large $\mu_{\star}>0$ and
    \begin{align*}
         \sup_{a \in \cA_{\Lambda}^{\al}} \| \alpha_p(\Pi_{Y}^\star(\sfG)\,; t)(a) -a \|_{\sU,C^k,\bm{x},\nu,\mu} \leq f_{\nu,\mu}(\mathrm{dist}(\bm{x},Y)/2) \cdot \| a\|_{\bm{x},\nu_{k},\mu_{3k+2}} \cdot \Upsilon^{\star, (k)}_{\sfG,\nu,\mu} (t) 
    \end{align*}
    for any $\bm{x} \in \Lambda$ and $t \in \bR_{\geq 0}$. 
\end{lem}
\begin{proof}
    The former equation follows from \cref{prp:relative.evolution.is.evolution}. To see the latter inequality, we use the following bound of almost local $C^k$-norms for $f_{\nu,\mu} ,f_{\nu,\mu_1} \in \cF_1$;
    \begin{align}
    \| a - \Pi_{\bm{x},r}(a) \|_{\bm{x},\nu,\mu} \leq \sup_{r' \geq r} f_{\nu,\mu}(r')^{-1} \cdot \|a - \Pi_{\bm{x},r'}(a) \| \leq g_{\nu,\mu} \cdot f_{\nu,\mu}(r) \cdot  \| a \|_{\bm{x},\nu,\mu_1},\label{eqn:norm.comparison.truncated}
    \end{align}
    where
    \begin{align*} 
    g_{\nu,\mu} \coloneqq \sup_{r \geq 0}f_{\nu,\mu}(r)^{-2}f_{\nu,\mu_1}(r) <\infty. 
    \end{align*}
    Let $\sfG_r \coloneqq \Pi_{N_{r/2}(Y)}\Pi_Y^\star(\sfG)$. Then, the almost local norm $\vvert \Pi_{Y}^{\star} (\sfG) - \sfG_r \vvert _{\sU,C^k,\nu,\mu}$ is bounded as
\begin{align*}
        {}&{} \| \Pi_{Y}^{\star} (\sfG)_{\bm{x}} - \sfG_{r,\bm{x}} \|_{\sU,C^k,\bm{x},\nu,\mu}\\
        \leq {}&{} \| \alpha_p(\Pi_{Y^c}(\sfG) \,; t)^{-1} (\sfG_{\bm{x}} - \Pi_{Y^c}(\sfG)_{\bm{x}} ) - \Pi_{N_{r/2}(Y)}\alpha_p(\Pi_{Y^c}(\sfG) \,; t)^{-1} (\sfG_{\bm{x}} - \Pi_{Y^c}(\sfG)_{\bm{x}} ) \|_{\sU,C^k,\bm{x},\nu,\mu}\\
        \leq {}&{} 
        \begin{cases}
            2 g_{\nu,\mu} \cdot f_{\nu,\mu }(r/4) \cdot \Upsilon_{\sfG ,\nu,\mu_1}^{(k)}(t) \cdot 2\| \sfG_{\bm{x}}\|_{\sU,C^k,\bm{x},\nu_k,\mu_{3k+2}} & \text{ if $\bm{x} \in N_{r/4}(Y)$}, \\
            2 \cdot \Upsilon_{\sfG,\nu,\mu}^{(k)}(t) \cdot g_{\nu_k,\mu_{3k+1}} \cdot f_{\nu_k,\mu_{3k+1}}(r/4) \cdot 2 \| \sfG_{\bm{x}}\|_{\sU,C^k,\bm{x}, \nu_k,\mu_{3k+2}} & \text{ if $\bm{x} \in N_{r/4}(Y)^c$}.
        \end{cases}
    \end{align*}
    This, together with \cref{prp:Lieb.Robinson.Ck} (2), we get
    \begin{align*}
        {}&{}\| \alpha_p(\Pi_{Y}^\star(\sfG)\,; t)(a) -a \|_{\sU,C^k,\bm{x},\nu,\mu} \\
        \leq {}&{}  \| \alpha_p(\Pi_{Y}^\star(\sfG)\,; t)( \Pi_{\bm{x},r/2}(a)) - \alpha_p (\Pi_{N_r(Y)} \Pi_Y^\star (\sfG)\,; t)( \Pi_{\bm{x},r/2}(a)) \|_{\sU,C^k,\bm{x},\nu,\mu} \\
        {}&{} + \Upsilon_{\Pi_Y^{\star}(\sfG), \nu,\mu}^{(k)}(t)  \cdot \| a -  \Pi_{\bm{x},r/2}(a) \|_{\bm{x},\nu_k,\mu_{3k+1}} + \| a -  \Pi_{\bm{x},r/2}(a) \|_{\bm{x},\nu,\mu} \\
        \leq {}&{} \vvert \Pi_{Y}^{\star} (\sfG) - \sfG_r \vvert _{\sU,C^{k},\nu_k,\mu_{3k+1}} \cdot \Upsilon_{\Pi_Y^\star \sfG, \sfG_r,\nu,\mu}^{\mathrm{rel},(k)}(t) \cdot \|  \Pi_{\bm{x},r/2}(a) \|_{\bm{x}, \nu_k,\mu_{3k+1}} \\
        {}&{} + \Upsilon_{\Pi_Y^{\star}(\sfG), \nu,\mu}^{(k)}(t)  \cdot g_{\nu_{k},\mu_{3k+1}} \cdot f_{\nu_{k},\mu_{3k+1}}(r/2) \cdot \| a\|_{\bm{x},\nu_k,\mu_{3k+2}} + g_{\nu,\mu} \cdot f_{\nu,\mu}(r/2) \cdot \| a  \|_{\bm{x},\nu,\mu_1} \\
        \leq {}&{}   \Big( 4 \max \{ g_{\nu_k,\mu_{3k+1}}, g_{\nu_{2k},\mu_{6k+2}}\} \cdot  \vvert \sfG \vvert_{\sU,C^k,\nu_{2k},\mu_{6k+3}} \cdot  \Upsilon_{\sfG,\nu_{k},\mu_{3k+2}}^{(k)}(t) \cdot \Upsilon_{\Pi_Y^\star \sfG, \sfG_r,\nu,\mu}^{\mathrm{rel},(k)} (t) \cdot c_{\nu,\mu,2} \\
        {}&{} \hspace{10ex}+  g_{\nu_k,\mu_{3k+1} } \cdot \Upsilon_{\Pi_Y^{\star}(\sfG), \nu,\mu}^{(k)}(t)   +g_{\nu,\mu} \Big)  \cdot f_{\nu,\mu}(r/2) \cdot\|a \|_{\bm{x}, \nu_k,\mu_{3k+2}}.
    \end{align*}
    This shows the desired inequality. 
\end{proof}

\begin{lem}\label{cor:Lieb.Robinson.approx.Ck}
Let $Y \subset \Lambda$.
\begin{enumerate}
    \item Let $\sfG \in \Omega^1(\sM \times [0,1],\fDer  (\cA^{\al}_{\Lambda}))$. Then 
    \begin{align*}
        {}&{} \sup_{a \in \cA_{\Lambda}^{\al}} \| \alpha_p(\sfG \,; t) (a) - \alpha_p(\Pi_{Y}(\sfG)\,; t) (a)\|_{\sU,C^k,\bm{x},\nu,\mu} \\
        \leq {}&{} f_{\nu_{k},\mu_{3k+1}}(\operatorname{dist}(\bm{x},Y^c)/2) \cdot \|a \|_{\bm{x},\nu_{2k},\mu_{6k+3}} \cdot  \Upsilon^{(k)}_{\Pi_Y(\sfG),\nu_k,\mu_{3k+2}}(t) \cdot \Upsilon^{\star, (k)}_{\sfG,\nu,\mu}(t)
    \end{align*}
    for any $t \in \bR_{\geq 0}$ and $\bm{x} \in \Lambda$. 
    \item Let $\sfH \colon \sM \to \fH_{\Lambda}^{\al}$ be a smooth map and let $\sfH_r \coloneqq \Pi_{N_rY}\sfH$. Then, for any $f_{\nu,\mu} \in \cF_1$, we have 
    \begin{align*}
        \sup_{r>0} f_{\nu,\mu}(r/4)^{-1} \cdot \vvert \Pi_{Y}(\sfG_{\sfH}) - \Pi_Y(\sfG_{\sfH_r}) \vvert_{\sU,C^k,\nu,\mu}  <\infty. 
    \end{align*}
\end{enumerate}
\end{lem}
\begin{proof}
    The claim (1) follows from \cref{prp:Lieb.Robinson.Ck,lem:LGA.cone.decomposition}. We show (2). First, 
    \begin{align*}
        {}&{} \| \sfG_{\sfH,\bm{x}} -\sfG_{\sfH_r,\bm{x}} \|_{\sU,C^k,\nu,\mu}   \\
        \leq {}&{}
        2\int_{0}^\infty \bigg( \int_0^t 
        \| \tau_{\sfH , u} (d\sfH_{\bm{x}} - d\sfH_{r/2,\bm{x}} )\|_{\sU,C^k,\bm{x},\nu,\mu} + \| \tau_{\sfH_r , u} (d\sfH_{\bm{x}} - d\sfH_{r/2,\bm{x}} ) \|_{\sU,C^k,\bm{x},\nu,\mu}  \\
        {}&{}\hspace{15ex} + \| \tau_{\sfH , u} ( d\sfH_{r/2,\bm{x}} ) - \tau_{\sfH_r ,u}(d\sfH_{r/2,\bm{x}} )\|_{\sU,C^k,\bm{x},\nu,\mu} \ du \bigg) w_v(t) dt .
    \end{align*}
    For $\bm{x} \in N_{r/4}(Y)$, the first and the second integrands are bounded above by 
    \begin{align*}
        \| \tau_{\sfH , u} (d\sfH_{\bm{x}} - d\sfH_{r/2,\bm{x}} ) \|_{\sU,C^k,\bm{x},\nu,\mu} \leq {}&{}  \Upsilon_{\sfH , \nu,\mu}^{(k)}(t) \cdot 2g_{\nu_k,\mu_{3k+1}} f_{\nu_k,\mu_{3k+1}}(r/4) \cdot  \vvert \sfH \vvert_{\sU,C^{k+1},\nu_k,\mu_{3k+2}}, \\
        \| \tau_{\sfH_r , u} (d\sfH_{\bm{x}} - d\sfH_{r/2,\bm{x}} ) \|_{\sU,C^k,\bm{x},\nu,\mu} \leq {}&{} \Upsilon_{\sfH_r , \nu,\mu}^{(k)}(t) \cdot 2g_{\nu_k,\mu_{3k+1}} f_{\nu_k,\mu_{3k+1}}(r/4) \cdot  \vvert \sfH \vvert_{\sU,C^{k+1},\nu_k,\mu_{3k+2}}, 
    \end{align*}
    and the third integrand is bounded above by (1) as    
    \begin{align*} 
        {}&{} \| \tau_{\sfH , u} ( d\sfH_{r/2,\bm{x}} ) - \tau_{\sfH_r ,u}(d\sfH_{r/2,\bm{x}} )\|_{\sU,C^k,\bm{x},\nu,\mu} \\
        \leq {}&{} f_{\nu,\mu}(r/4) \cdot \vvert \sfH_{r/2} \vvert_{\sU,C^{k+1},\bm{x},\nu_{2k},\mu_{6k+3}} \cdot \Upsilon^{(k)}_{\sfH_r,\nu_k,\mu_{3k+2}}(t) \cdot \Upsilon^{\star, (k)}_{\sfH,\nu,\mu}(t).
    \end{align*}
    Since $\Upsilon^{(k)}_{\sfH_r,\nu_k,\mu_{3k+2}}(t) \cdot \Upsilon^{\star,(k)}_{\sfH,\nu,\mu}(t) \prec  w_v(t)^{-1}$, we obtain that $\| \sfG_{\sfH,\bm{x}} - \sfG_{\sfH_r,\bm{x}} \|_{\sU,C^k,\bm{x},\nu,\mu}$ is bounded by a constant multiple of $f_{\nu,\mu}(r/4) $. In conclusion, there is $C>0$ such that
    \begin{align*}
        \sup_{\bm{x} \in N_{r/4}(Y)}\| \Pi_Y(\sfG_{\sfH,\bm{x}}) - \Pi_Y(\sfG_{\sfH_r,\bm{x}}) \|_{\sU,C^k,\bm{x},\nu,\mu} 
        \leq 
        \sup_{\bm{x} \in N_{r/4}(Y)}\| \sfG_{\sfH,\bm{x}} - \sfG_{\sfH_r,\bm{x}} \|_{\sU,C^k,\bm{x},\nu,\mu} 
            \leq  C \cdot f_{\nu,\mu}(r/4) .
    \end{align*}
    On the other hand, if $\bm{x} \in N_{r/4}(Y)^c$, by \cref{lem:adiabatic.well-defined} we have
    \begin{align*}
            {}&{} \| \Pi_Y(\sfG_{\sfH, \bm{x}}) - \Pi_Y(\sfG_{\sfH_r , \bm{x}})\|_{\sU,C^k,\bm{x},\nu,\mu} \\
            \leq  {}&{} \| \Pi_Y(\sfG_{\sfH,\bm{x}})  - \operatorname{tr}(\sfG_{\sfH,\bm{x}})\|_{\sU,C^k,\bm{x},\nu,\mu} +  \| \Pi_Y(\sfG_{\sfH_r,\bm{x}})- \operatorname{tr}(\sfG_{\sfH,\bm{x}})\|_{\sU,C^k,\bm{x},\nu,\mu} \\\leq {}&{}2 g_{\nu,\mu}f_{\nu,\mu}(r/4) \cdot C' \cdot  (\vvert \sfH \vvert_{\sU,C^{k+1},\nu_k,\mu_{3k+2}} + \vvert \sfH_r \vvert_{\sU,C^{k+1},\nu_k,\mu_{3k+2}} )
    \end{align*}
    for some $C'>0$. They conclude (2).
\end{proof}

\begin{prp}\label{prp:cut.trivial.Hamiltonian}
    Let $\Lambda$ be weakly uniformly discrete in  $\bR^{l_{\Lambda}}$ and let $Y \subset \Lambda $.
    Let
    $\sfH \colon \sM \to \fH_{\Lambda}^{\al}$ be a smooth family of gapped UAL Hamiltonians and let $\omega \colon \sM \to \fS(\cA_{Y^c})$ be a smooth family of pure states. 
    For $\bm{x} \in Y$, let
    \begin{align*} 
    (\Theta_{Y, \omega}\sfH)_{\bm{x}} (p) \coloneqq 
    (\id_{\cA_{Y}} \otimes \omega)( \sfH_{\bm{x}}(p)) + \sum_{\bm{y} \in Y^c}\sum_{\substack{B_\rho \in \iota(\bB), B_\rho \ni \bm{x}\\ Y \cap B_\rho \neq \emptyset }} |B_\rho \cap Y|^{-1} \cdot (\id_{\cA_{Y}} \otimes \omega)(\sfH_{\bm{y}}^\rho(p)).
    \end{align*}
    Assume that the ground state $\omega_{\sfH}$ of $\sfH$ restricts to $\omega$ on $\cA_{Y^c}$. 
    Then the above $\Theta_{Y,\omega}\sfH$ forms a smooth family of gapped UAL Hamiltonians on $\cA_Y$ with the distinguished gapped ground state $\omega_Y \coloneqq \omega_{\sfH}|_{\cA_{Y}}$. 
    Moreover, if $\omega$ is the distinguished ground state of a smooth family of gapped UAL Hamiltonians $\sfH_{Y^c} \colon \sM \to \fH_{\Lambda}^{\al}$, then $\sfH$ and $\Theta_{Y,\omega}\sfH \boxtimes \sfH_{Y^c}$ on $\cA_{\Lambda}$ has the same ground state, and hence are smoothly homotopic. 
\end{prp}
The second summand in the above definition of $\Theta_{Y,\omega}\sfH$ is the redistribution of the remaining part of the truncated local generators, originally placed at $\bm{y} \in Y^c$, to points in $Y$. 
\begin{proof}
    First, we check that the right hand side is well-defined. The almost local norm of $\Theta_{Y,\omega} \sfH$ is bounded above as 
    \begin{align*}
    {}&{} \vvert (\Theta_{Y,\omega} \sfH)_{\bm{x}}\vvert_{\sU,C^k,\bm{x},\nu,\mu} \\
    \leq {}&{} 
    \vvert (\id_{\cA_{Y}} \otimes \omega)( \sfH_{\bm{x}}) \vvert_{\sU,C^k,\bm{x},\nu,\mu}
    +
    \bigg\| \sum_{\bm{y} \in Y^c}\sum_{B_\rho \in \iota(\bB), B_\rho \ni \bm{x} }   (\id_{\cA_{Y}} \otimes \omega)(\sfH_{\bm{y}}^\rho)\bigg\|_{\sU,C^k,\bm{x},\nu,\mu}.
    \end{align*}
    The first term is bounded above by
    \begin{align}
        \begin{split}
    {}&{}\sup_{r>0} f_{\nu,\mu}(r)^{-1} \cdot \| (\id \otimes \omega)(\sfH_{\bm{x}}) - \Pi_{\bm{x},r}\big( \id \otimes  \omega)(\sfH_{\bm{x}}) \big)\|_{\sU,C^k} \\
    \leq {}&{} \sup_{r>0} f_{\nu,\mu}(r)^{-1} \cdot \| (\id \otimes \omega)\big( \sfH_{\bm{x}} - \Pi_{B_r(\bm{x}) \cup Y^c}(\sfH_{\bm{x}})  \big)\|_{\sU,C^k}\\
    \leq {}&{} \sup_{r>0} f_{\nu,\mu}(r)^{-1}\cdot \| \omega\|_{\sU,C^k,\nu,\mu} \cdot  \| \sfH_{\bm{x}} - \Pi_{B_r(\bm{x}) \cup Y^c}(\sfH_{\bm{x}})  \|_{\sU,C^k,\nu,\mu}\\
    \leq {}&{} \| \omega \|_{\sU,C^k,\nu,\mu} \cdot g_{\nu,\mu} \cdot \vvert \sfH \vvert_{\sU,C^k,\nu,\mu_1},
        \end{split}\label{eqn:redistribution.1}
    \end{align}
    where $g_{\nu,\mu}$ is the constant in \eqref{eqn:norm.comparison.truncated}. 
    By \eqref{eqn:NSY} and \eqref{eqn:interaction.to.local}, the second term is bounded above by
\begin{align}
    \begin{split}
     {}&{}\sup_{r>0}f_{\nu,\mu}(r)^{-1} \cdot \sum_{\substack{B_\rho \in \iota(\bB), B_\rho \ni \bm{x} \\ B_\rho \cap Y \not \subset B_r(\bm{x})} }  \Big\| \sum_{\bm{y} \in Y^c}(\id_{\cA_{Y}} \otimes \omega)(\sfH_{\bm{y}}^\rho)\Big\|_{\sU,C^k} \\
    \leq {}&{}\sup_{r>0} f_{\nu,\mu}(r)^{-1} \cdot \sum_{\substack{B_\rho \in \iota(\bB), B_\rho \ni \bm{x}\\ B_\rho \not \subset B_r(\bm{x})} } \Big\| \sum_{\bm{y} \in \Lambda} \sfH_{\bm{y}}^\rho \Big\| _{\sU,C^k} \\
    \leq  {}&{} 2^{l_{\Lambda} +1} \kappa_{\Lambda} \cdot \vvert \Phi_{\sfH} \vvert _{\sU,C^k,\nu_1,\mu} \leq  2^{l_{\Lambda} +1} \kappa_{\Lambda} \cdot L_{\nu_1,\mu} \cdot \vvert \sfH \vvert_{\sU,C^k,\nu_4,\mu_1}.
        \end{split}\label{eqn:redestribution.2}
\end{align}

    We show that $\Theta_{Y,\omega} \sfH$ is gapped. 
    First, the ground state $\omega_{\sfH}$ decomposes as $\omega_\sfH=\omega \otimes \omega_{Y}$ by a general fact in C*-algebra theory that a pure state $\omega$ on the tensor product C*-algebra $A \otimes B$ decomposes as $\omega =\omega|_A \otimes \omega|_B$ if $\omega|_A $ is a pure state. See for example  \cite{brownAlgebrasFinitedimensionalApproximations2008}*{Corollary 3.4.3}.
    Let $a \in \cA_{Y}^{\al}$ with $\|a\|=1$ and $\omega_{\sfH}(a) =0$. 
    Then, by the spectral gap assumption of $\sfH$, we have 
    \begin{align*}
        \omega_{Y}(a^*[\Theta_{Y,\omega}\sfH,a]) ={}&{} \sum_{\bm{x} \in \Lambda}  \omega_{Y} (a^* [(\id \otimes \omega_{{Y^c}})(\sfH_{\bm{x}}), a])
        =  \sum_{\bm{x} \in \Lambda}  \omega_{Y} \circ (\id \otimes \omega_{{Y^c}}) ( a^* [\sfH_{\bm{x}}, a])\\
        = {}&{} \sum_{\bm{x} \in \Lambda}  \omega_{\sfH}( a^* [\sfH_{\bm{x}}, a]) \geq \gap .        \qedhere
    \end{align*}
\end{proof}

In order to state the main result of this section, \cref{thm:interpolation.loop}, we introduce the following constant. 
For $Y \subset \Lambda$, a smooth family $\sfG \colon \sM \to \fDer ^{\al}_{\Lambda}$, $f_{\nu,\mu} \in \cF$, $k \in \bN$ and a relatively compact open chart $\sU$ of $\sM$, let
\begin{align} 
    K_{Y,\sfG,\nu,\mu}^{(k)} \coloneqq \sup_{\bm{x} \in Y^c} \sup_{p \in \sU}  f_{\nu,\mu}(\mathrm{dist} (\bm{x},Y))^{-1} \cdot \|  \sfG _{\bm{x}} \|_{\sU,C^k}. \label{eqn:constant.K}
\end{align}
If $Y $ is a single point, we abbreviate it and simply write $K_{\sfG,\nu,\mu}^{(k)}$. 

\begin{rmk}\label{rmk:cutoff.approximate}
    The truncated Hamiltonian $\Theta_{Y,\omega}\sfH$ in \cref{prp:cut.trivial.Hamiltonian} satisfies $K_{Y^c,\sfH - \Theta_{Y,\omega}\sfH,\nu,\mu}^{(k)} < \infty$ for any $f_{\nu,\mu} \in \cF$, $k \in \bN$ and a relatively compact open chart $\sU$ of $\sM$. 
    This can be verified by using the inequalities \eqref{eqn:norm.comparison.truncated}, \eqref{eqn:redistribution.1}, and \eqref{eqn:redestribution.2}, applied to $\sfH - \Pi_r\sfH$, where $\Pi_r\sfH$ is the family of operators $\Pi_r\sfH$ given by $(\Pi_r\sfH )_{\bm{x}} \coloneqq \Pi_{\bm{x},r}(\sfH_{\bm{x}})$. 
    Indeed, by letting $r \coloneqq \mathrm{dist}(\bm{x},Y^c)$, we obtain 
    \begin{align*}
        {}&{} \| \sfH_{\bm{x}} - (\Theta_{Y,\omega}\sfH) _{\bm{x}} \|_{\sU,C^k} \\
        \leq {}&{}\| (\sfH -\Pi_{r}\sfH_{\bm{x}}) - \Theta_{Y,\omega}(\sfH -\Pi_{r}\sfH )_{\bm{x}} \|_{\sU,C^k} \\
        \leq {}&{} f_{\nu,\mu}(r) \cdot \vvert \sfH \vvert_{\sU,C^k,\nu,\mu_1} + \| \omega\|_{\sU,C^k,\nu,\mu} \cdot g_{\nu,\mu}g_{\nu,\mu_1} \cdot f_{\nu,\mu_1}(r) \cdot \vvert \sfH \vvert_{\sU,C^k,\nu,\mu_2}\\
        {}&{} \qquad \qquad + 2^{l_{\Lambda} +1}\kappa_{\Lambda }L_{\nu_1,\mu} g_{\nu_4,\mu_1}f_{\nu_4,\mu_1}(r) \cdot \vvert \sfH \vvert_{\sU,C^k,\nu_4,\mu_2}.   
    \end{align*}
\end{rmk}

\begin{thm}\label{thm:interpolation.loop}
    Let $\Lambda$ be weakly uniformly discrete in $\bR^{l_{\Lambda}}$ and let $Y \subset \Lambda$. Let $\sfH \colon \sM \times [0,1] \to \fH_{\Lambda}^{\al}$ be a smooth map such that $\sfH(p,1) = \sfh \boxtimes \sfH_{Y^c}$ for some smooth family $\sfH_{Y^c} \colon \sM \times [0,1] \to \fH_{Y^c}^{\al}$. Then, there is another smooth map $\tilde{\vartheta} \sfH \colon \sM \times [0,1] \to \fH_{Y^c}^{\al}$ such that 
\begin{enumerate}
    \item $\tilde{\vartheta}\sfH (p,1) = \sfH (p,1)$, and
    \item the constant $K_{Y,\tilde{\vartheta} \sfH \boxtimes \sfh - \sfH,\nu,\mu}^{(k)}$ in \eqref{eqn:constant.K} is finite for any $f \in \cF$, $k \in \bN$ and  $\sU$.
\end{enumerate}
\end{thm}
\begin{proof}
    The smooth family $\tilde{\vartheta} \sfH $ is defined by
    \[
        \tilde{\vartheta} \sfH(p,t) \coloneqq \Theta_{Y^c , \omega_0} \circ \alpha \big( \Pi^{\star}_{Y}(\sfG_{\sfH}) \,; t,1 \big)^{-1} (\sfH(p,t)). 
    \]
    Indeed, by \cref{lem:LGA.cone.decomposition}, we have
    \begin{align*}
    \omega_{\sfH (p,t)} \circ \alpha_t \big( \Pi^{\star}_{Y}(\sfG_{\sfH}) \,; t,1 \big) (a) = {}&{} \omega_{\sfH (p,t)} \circ \alpha_t \big(\sfG_{\sfH} \,; t,1 \big) \circ \alpha_t \big( \Pi_{Y^c} (\sfG_{\sfH}) \,; t,1 \big) (a) \\
    = {}&{} \omega_{\sfH (p,t)} \circ \alpha_p \big(\sfG_{\sfH} \,; t,1 \big) (a) = \omega_0(a)
    \end{align*}
    for any $a \in \cA_{Y}$, and hence we can apply \cref{prp:cut.trivial.Hamiltonian} to obtain that the above $\tilde{\vartheta} \sfH (p,t)$ has the same spectral gap as that of $\sfH$. 
    It satisfies the conditions (1) by definition. Finally, the condition (2) follows from \cref{lem:LGA.cone.decomposition} since the operator norm  $\| \blank \|$ is by definition bounded by $ \| \blank \|_{\bm{x},\nu,\mu}$. 
\end{proof}

The above condition (2) is available to truncate a gapped UAL Hamiltonian.
The following definition is inspired from \cite{ludewigQuantizationConductanceCoarse2023}.

\begin{defn}\label{defn:linearly.coarsely.transverse}
    A pair of subspaces $Y,Z \subset \Lambda $ is \emph{linearly coarsely transverse} if, for some (any) reference point $\bm{x}_0 \in \Lambda$, there is $k>0$ such that $N_{r}(Y) \cap  N_r(Z) \subset B_{k(1+r)}(\bm{x}_0)$ for any $r>0$. 
\end{defn}
Typically, if $\Lambda$ is a Delone set of $\bR^{d}$, the following conical regions $Z_L$, $Z_R$ gives a linearly coarsely transverse pair (see \cref{fig:transverse.cone}):
\begin{align}
\begin{split}
    Z_L^{\theta_L} \coloneqq {}&{} \{ \bm{x} \in \Lambda \mid x_d \leq -\| \bm{x}\| \cdot \sin(\theta_L) \}, \\
    Z_R^{\theta_R} \coloneqq {}&{} \{ \bm{x} \in \Lambda \mid x_d \geq \| \bm{x} \| \cdot \sin(\theta_R) \}.
\end{split}\label{eqn:conical.region}
\end{align}
Here, $x_d$ denotes the $d$-th coordinate function on $\bR^{d}$ and $\theta_L, \theta_R \in [0,\pi/2)$ are fixed angles such that either of them is non-zero (we deal with the cases $(\theta_L, \theta_R) = (\pi /4, 0)$ or $(0,\pi/4)$). 
\begin{figure}
    \centering
    \begin{tikzpicture}[scale=0.5]
        \draw[->] (-6,0) -- (6,0);
        \draw[->] (0,-3.5) -- (0,3.5);
        \path[pattern = north east lines] (-6.0,-3.5) -- (-4.0,-3.5) -- (0,0) -- (-4.0,3.5) -- (-6.0,3.5) -- cycle;
        \path[pattern = north east lines] (6.0,-3.5) -- (2.0,-3.5) -- (0,0) -- (2.0,3.5) -- (6.0,3.5) -- cycle;
        \draw (0.0,0.0) -- (0.0, 1.0) arc[x radius = 1.0, y radius = 1.0, start angle = 90, end angle = 138.81] -- cycle;
        \draw (0.0,0.0) -- (0.0,-1.0) arc[x radius = 1.0, y radius = 1.0, start angle = -90, end angle = -60.26] -- cycle;
        \draw (-4.0,-3.5) -- (0,0) -- (-4.0,3.5);
        \draw (2.0,-3.5) -- (0,0) -- (2.0,3.5);
        \node[fill=white] at (-4.5,-1.5) {$Z_L$};
        \node[fill=white] at (4.5,-1.5) {$Z_R$};
        \node at (-0.8,1.8) {$\theta_L$};
        \node at (0.5,-1.8) {$\theta_R$};
    \end{tikzpicture}
    \caption{Example of linearly coarsely transverse pair $Z_L$, $Z_R$. }
    \label{fig:transverse.cone}
\end{figure}
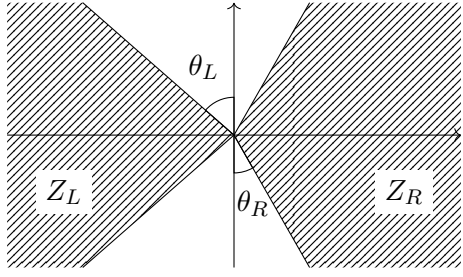

\begin{lem}\label{lem:inner.derivation}
    Let $\sfG \colon \sM \to \fDer ^{\al}_{\Lambda}$ be a smooth map such that he constant $K_{\sfG,\nu,\mu}^{(k)}$ in \eqref{eqn:constant.K} is finite 
    for any $f_{\nu,\mu} \in \cF_1$, any $k \in \bN$, and any relatively compact open chart $\sU$. 
    Then the infinite sum $g \coloneqq \sum_{\bm{x} \in \Lambda} \sfG_{\bm{x}}$ converges to a smooth map $g \colon \sM \to \cA_{\Lambda }$ in the almost local topology. 
\end{lem}
\begin{proof}
We have
    \begin{align}
    \begin{split}
        {}&{} f_{\nu,\mu}(r)^{-1} \cdot  \| g - \Pi_{\bm{x},r}(g)\|_{\sU,C^k} \\
        \leq {}&{}  \sum_{\bm{x} \in B_{r/2}(\bm{x}_0)} f_{\nu,\mu}(r)^{-1} \cdot \| \sfG_{\bm{x}} - \Pi_{\bm{x}_0,r/2}(\sfG_{\bm{x}}) \|_{\sU,C^k} + \sum_{\bm{x} \not \in B_{r/2}(\bm{x}_0)} 2f_{\nu,\mu}(r)^{-1} \cdot \| \sfG_{\bm{x}} \|_{\sU,C^k} \\
        \leq {}&{} \kappa_{\Lambda}  \cdot c_{\nu,\mu,2} \cdot \vvert \sfG \vvert_{\sU,C^k,\nu_1,\mu_1} 
        + 2K_{\sfG,\nu,\mu_1}^{(k)} \cdot 2^{l_\Lambda +1}\kappa_{\Lambda} \cdot c_{\nu,\mu,2} <\infty,
    \end{split}\label{eqn:inner.norm}
    \end{align}
    where $c_{\nu,\mu,2}$ be the constant as in \eqref{eqn:const.small.c}.
\end{proof}

\begin{lem}\label{lem:asymptotically.equal.derivation}
    Let $\sfG, \sfG' \in \Omega^1(\sM \times [0,1], \fDer ^{\al}_{\Lambda})$ be smooth $1$-form of UAL derivations such that the constant $K_{\sfG-\sfG',\nu,\mu}^{(k)}$ in \eqref{eqn:constant.K} is finite for any $k \in \bZ_{\geq 0}$, a relatively compact open chart $\sU \subset \sM \times[0,1]$, and $f_{\nu,\mu} \in \cF_1$. Then, there is a Fr\'{e}chet smooth map $u \colon \sM \times [0,1] \to \cA_{\Lambda}^{\al}$ such that each $u(p,t)$ is unitary and $\alpha_p(\sfG \,; t)  = \alpha_p(\sfG' \,; t) \circ \Ad(u(p,t))$. 
\end{lem}
\begin{proof}
    By \cref{lem:inner.derivation}, the infinite sum 
    \[
    g'(p,t) \coloneqq \sum_{\bm{x} \in \Lambda} (\sfG_{\bm{x}}(p,t) - \sfG_{\bm{x}}'(p,t))
    \]
    converges to a smooth $1$-form of almost local operators. Let $g (p,t) \coloneqq \alpha_p(\sfG' \,; t)^{-1}(g'(p,t))$, which is again a smooth function taking value in $\cA_{\Lambda}^\al$ by \cref{lem:smooth}.
    Now, the solution $u(p,t)$ of the ordinary differential equation 
    \[
    u(p,t)^{*}\frac{d}{dt}u(p,t) = g(p,t),
    \]
    which takes value in $\cA_{\Lambda}^{\al}$ by \cref{rmk:almost.local.unitary}, satisfies the desired equation $\Ad(u(p,t)) =  \alpha_p(\sfG' \,; t)^{-1} \circ \alpha_p(\sfG \,; t)$ by \cref{prp:relative.evolution.is.evolution}. 
\end{proof}

\begin{lem}\label{lem:cut.diffused}
    Let $\sfH, \sfH' \colon \sM \to \fH_{\Lambda}^{\al}$ be a smooth family such that $\sum_{\bm{x} \in \Lambda} \| \sfH_{\bm{x}} - \sfH_{\bm{x}}'\|  < 1/4$, $\| \omega_{\sfH} - \omega_{\sfH'} \| < 2$, and the constant $K_{\sfH-\sfH',\nu,\mu}^{(k)}$ in \eqref{eqn:constant.K} is finite 
    for any $k \in \bZ_{\geq 0}$, a relatively compact open chart $\sU \subset \sM \times[0,1]$, and $f_{\nu,\mu} \in \cF_1$. Then, for any $(p,t) \in \sM \times [0,1]$, the UAL Hamiltonian 
    \[
    \widetilde{\sfH}(p,t) \coloneqq 2t\sfH'(p) + 2(1-t)\sfH
    \]
    has a spectral gap $1$ with respect to an appropriate choice of smooth families of distinguished ground states. Moreover, there is a smooth family $\omega \colon \sM \times [0,1] \to \fS(\cA_{\Lambda}^{\al})$ such that $\omega_{p,0} = \omega_{\sfH (p)}$, $\omega_{p,1}= \omega_{\sfH'(p)}$, and  each $\omega_{p,t}$ is a non-degenerate gapped ground state for $\widetilde{\sfH}(p,t)$. 
\end{lem}
The composition of this $\widetilde{\sfH}$ with the paths of constant multiples gives a smooth homotopy connecting $\sfH$ and $\sfH'$ in $\fH_{\Lambda}^{\al}$.
\begin{proof}
    By \cref{lem:inner.derivation}, the infinite sum
    \[
        v(p) \coloneqq \sum_{\bm{x} \in \Lambda} \sfH_{\bm{x}} (p) - \sfH_{\bm{x}}'(p)
    \]
    converges to a smooth map $v \colon \sM \to \cA_{\Lambda}^{\al}$ in the smooth Fr\'{e}chet topology. Moreover, since $\| v \| < 1/4$, we have that the unbounded self-adjoint operator 
    \begin{align}
    H_{\omega_{\sfH(p)}} + t\pi_{\omega_{\sfH(p)}} (v(p)) \colon \mathrm{Dom}(H_{\omega_{\sfH(p)}}) \to \sH_{\omega_{\sfH(p)}} \label{eqn:perturbed.unbounded.operator}
    \end{align}
    still has a spectral gap in the interval $(t/4,1-t/4)$. 
    By the $C^\infty$-boundedness of the perturbation, the family of pure states
    \[
        \omega_{p,t} (a) \coloneqq \mathrm{Tr} \big( \pi_{\omega_{\sfH(p)}}(a) \cdot P_{\leq 1/2} (H_{\omega_{\sfH(p)}} + t\pi_{\omega_{\sfH(p)}}(v(p)) \big),
    \]
    where $P_{\leq 1/2}(A)$ denotes the spectral projection of a self-adjoint operator $A$ associated with the spectrum $(-\infty,1/2]$, is smooth as a function on $\sM \times [0,1]_t$. Therefore, by \cref{thm:automorphic.equivalence}, we have
    \[
    \omega_{p,t} = \omega_{\sfH(p)} \circ \alpha_p \big( \sfG_{\widetilde{\sfH}} \,; t \big)^{-1} ,
    \]
    which is smooth as a function on $\sM \times [0,1]$. 

    The remaining task is to show that $\omega_{p,1} = \omega_{\sfH'(p)}$ for any $p \in \sM$. 
    Since $\| \omega_{\sfH(p)} - \omega_{\sfH'(p)}\|<2$, the pure states $\omega_{\sfH(p)}$ and $\omega_{\sfH'(p)}$ induce equivalent GNS representations $(\sH_{\omega_{\sfH(p)}}, \pi_{\omega_{\sfH(p)}}) \cong (\sH_{\omega_{\sfH'(p)}},\pi_{\omega_{\sfH'(p)}})$  (cf.\ \cite{kadisonFundamentalsTheoryOperator1997}*{Corollary 10.3.8}). 
    This implies that $\omega_{\sfH'(p)}$ is implemented by a unit vector $\Omega' \in \sH_{\omega_{\sfH(p)}}$.
    By the ground state property of $\omega_{\sfH'(p)}$ with respect to the derivation $\sfad (\sfH')$, the vector $\Omega'$ must be a constant multiple of the lowest eigenvector of the unbounded operator \eqref{eqn:perturbed.unbounded.operator}. This shows $\omega_{p,1} = \omega_{\sfH'(p)}$. 
\end{proof}

Let $\sfH \colon \sM \times [0,1] \to \fH_{\Lambda}^{\al}$ be a smooth map such that $\sfH(p,0) = \sfH(p,1) = \sfh$. 
Let $Z \subset \Lambda$ be a subspace linearly coarsely transverse to $Y$. For $r >0$, let $Z_r \coloneqq Z \setminus N_r(Y)$. 
Then, by \cref{prp:cut.trivial.Hamiltonian,thm:interpolation.loop}, the smooth families $\Theta_{Z_r^c, \omega_0} (\tilde{\vartheta}\sfH)$ and $\vartheta \sfH$ satisfy the assumptions of \cref{lem:cut.diffused} for sufficiently large $r>0$. In conclusion, they are smoothly homotopic via the convex combination.

\subsection{The ground state of invertible gapped Hamiltonians}\label{subsection:GS}
As will be considered in \cref{subsection:sheaf.lattice}, we say that a gapped UAL Hamiltonian $\sfH \in \fH_{\Lambda}^{\al}$ is \emph{invertible} if there is another gapped UAL Hamiltonian $\check{\sfH} \in \fH_{\Lambda}^{\al}$ and a smooth path of gapped UAL Hamiltonians $\overline{\sfH}$ connecting $\sfH \boxtimes \check{\sfH}$ and the trivial Hamiltonian $\sfh$ defined in \eqref{eqn:trivial.Hamiltonian}. In particular, the path $\omega_{\overline{\sfH}}$ of distinguished ground states gives a smooth homotopy of $\omega_{\sfH} \otimes \omega_{\check{\sfH}}$ and $\omega_0$. Here, we remind some properties of the ground state of an invertible gapped UAL Hamiltonian, which will be used in \cref{subsection:1d} and the subsequent paper \cite{kubotaStableHomotopyTheory2025a}.   
\begin{defn}
    Let $\Lambda$ be weakly uniformly discrete in $\bR^l$. 
    \begin{enumerate}
        \item A pure state $\omega$ on $\cA_{\Lambda}$ has the $\cF$-clustering property if, for any $f \in \cF$, we have 
        \[
        \sup_{X,Y \subset \Lambda }\sup_{a \in \cA_{X}}\sup_{b \in \cA_Y} f(\mathrm{dist}(X,Y))^{-1} \cdot \frac{|\omega(ab) - \omega(a)\omega(b) |}{\|a\| \cdot \|b\|} < \infty. 
        \]
        \item Two states $\omega_1$, $\omega_2$ on $\cA_{\Lambda}$ are $\cF$-asymptotically equivalent if 
        \[
        \sup_{X \subset \Lambda }\sup_{a \in \cA_{X}}f(\mathrm{dist}(\bm{x}_0,X))^{-1} \cdot \frac{|\omega_1(a) - \omega_2(a) |}{\| a \| } < \infty. 
        \]
        \item A pure state $\omega $ on $\cA_{\Lambda}$ has the split property with respect to a decomposition $\Lambda = Y \sqcup Z$ if the double commutant $\pi_{\omega }(\cA_Y)''$ in the GNS representation of $\omega$ is a type I factor (\cites{doplicherStandardSplitInclusions1984,matsuiSplitPropertySymmetry2001}). 
    \end{enumerate}
\end{defn}
\begin{rmk}
    For an arbitrary subspace $Y \subset \Lambda $, the double commutant von Neumann algebra $\pi_{\omega}(\cA_{Y})''$ is a factor, i.e., its center $\cZ(\pi(\cA_{Y})'') $ is trivial. This is verified as
    \[
    \cZ(\pi_{\omega}(\cA_Y)'') = \pi_{\omega}(\cA_Y)' \cap \pi_{\omega}(\cA_Y)'' \subset \pi_{\omega}(\cA_{Y})' \cap \pi_{\omega}(\cA_{Y^c})' = \pi_{\omega}(\cA_{\Lambda})' =\bC.
    \]
    The last equality follows from the pureness of $\omega$. 
    By a standard argument, this implies that the weak completion of $\cA_Y$ via the GNS representation $(\sH_{\omega_Y},\pi_{\omega_Y},\Omega_{\omega_Y})$ of  $\omega_Y \coloneqq \omega|_{\cA_{Y}}$ is isomorphic to that via the subrepresentation $(P_Y\sH_\omega, P_Y\pi_\omega(\blank), \Omega_\omega)$, where $P_Y$ denotes the projection onto $\sH_Y \coloneqq \overline{\pi_{\omega}(\cA_{Y})\Omega_{\omega}} \subset \sH_{\omega}$ (note that they are not a priori isomorphic). In particular, $\omega_Y$ is a factor state, i.e., generates a factor $\pi_{\omega_Y}(\cA_{Y})''$. 
\end{rmk}

\begin{lem}[{\cite{artymowiczQuantizationHigherBerry2023}*{Proposition A.1}}]\label{lem:asymp.equal.state.trivial}
    Let $\omega $ be a pure state on $\cA_{\Lambda}$ that is asymptotically equivalent to $\omega_0$. 
    Then there is a unitary $u \in \cA_{\Lambda}^{\al}$ such that $\omega  = \omega_0 \circ \Ad(u)$. 
\end{lem}
\begin{proof}
    We write $(\sH,\pi,\Omega)$ for the GNS representation of $(\cA_{\Lambda}, \omega_0)$. 
    By \cite{bratteliOperatorAlgebrasQuantum1987}*{Corollary 2.6.11} (note that $\cA_{\Lambda}$ is a quasi-local algebra in the sense of \cite{bratteliOperatorAlgebrasQuantum1987}*{Definition 2.6.3} indexed by the set of finite subsets of $\Lambda$) and \cite{kadisonFundamentalsTheoryOperator1997}*{Proposition 10.3.7}, the GNS representation for $\omega$ is quasi-equivalent, and hence is equivalent, to $(\sH,\pi)$. 
    Therefore, there is $\xi \in \sH $ such that $\omega(a)=\langle \pi(a)\xi,\xi\rangle$ for any $a \in \cA_{\Lambda}$. 
    For $r \in \bR_{>0}$, we have $\sH = \sH_r \otimes \sH_r^c$, where $(\sH_r, \pi_r,\Omega_r)$ and $(\sH_r^c, \pi_r^c,\Omega_r^c)$ are the GNS representation of $(\cA_{B_r(\bm{x})}, \omega_0)$ and $(\cA_{B_r(\bm{x})^c},\omega_0)$ respectively. 

    Let $P_r$ denote the projection onto $\sH_r \otimes \bC \cdot \Omega_r^c$ and let $\xi_r' \coloneqq P_r \xi$. In particular, $P_0 \coloneqq \Omega \otimes \Omega^*$. 
    We show that $f(r)^{-1}\| \xi - \xi_r' \| $ is a bounded function on $r$ for any $f \in \cF$. 
    Let $\xi= \sum_i \lambda_i \cdot \eta_i \otimes \zeta_i$ be the Schmidt decomposition with respect to $\sH \cong \sH_r \otimes \sH_r^c$ such that $\| \eta_i \|=\| \zeta_i\|=1$ and $\sum_i \lambda_i^2 =1$. Let $\rho_r \coloneqq \sum_{i} \lambda_i^2 \zeta_i \otimes \zeta_i^* \in \cB(\sH_r^c)$ and $p_r \coloneqq \Omega_r^c \otimes (\Omega_r^c)^* \in \cB(\sH_r^c)$. 
   Then, by the asymptotic equivalence $\omega \sim_{\mathrm{asymp}} \omega_0$, for any $f \in \cF$ there is $C_f>0$ such that 
   \begin{align*} 
   \bigg|\sum \lambda_i^2 \cdot |\langle \zeta_i,\Omega_r^c \rangle|^2 -1\bigg| = |\langle ( \rho_r -p_r) \Omega_r^c,\Omega_r^c\rangle| \leq \| \rho_r - p_r \|_{1} =\| (\omega - \omega_0 )|_{\cA_{ B_{r}(\bm{x})^c}} \| \leq C_f \cdot f(r)^2,
   \end{align*}
   where $\| \blank \|_1$ denotes the Schatten-$1$ norm. Here, we used the inequality $f_\mu(r)^2 = f_\mu (2^{1/\mu}r) \geq c_{0,\mu,2^{1/\mu}}^{-1} \cdot f_{\mu_1}(r)$ in order to bound the left hand side by $f(r)^2$ rather than $f(r)$. 
   Hence 
    \begin{align*}
        \| \xi - \xi_r'\|^2 = 2-2\langle \xi,\xi_r\rangle = 2-2 \sum_i \lambda_i^2 \cdot  \langle \zeta_i, \langle \zeta_i,\Omega_r^c \rangle \Omega_r^c \rangle = 2\Big( 1- \sum_i \lambda_i^2 \cdot  |\langle \zeta_i,  \Omega_r^c \rangle|^2 \Big) \leq 2C_f \cdot f(r)^2. 
    \end{align*}
    
    Let $\xi_r$ be the normalization $\| \xi_r'\|^{-1} \xi_r'$. 
    The above estimate shows $|\| \xi_r'\| -1| \leq 2C_f^{1/2} f(r)$, and hence $\| \xi -\xi_r\| \leq 4C_f^{1/2} f(r)$.
    Let $P_\xi$ and $P_{\xi_r}$ be the rank $1$ projections onto the subspaces generated by the subscript vectors. Since $\| \xi_{r+1} -\xi_r \| \leq 8C_f^{1/2} f(r)$, we have $\|P_{\xi_{r+1}} - P_{\xi_r}\| \leq 16C_f^{1/2}  f(r)$, and hence
    \begin{align*} 
    a_r \coloneqq P_{\xi_{r+1}} P_{\xi_r} + (1-P_{\xi_{r+1}})(1-P_{\xi_r})\in \cA_{B_{r+1}(\bm{x})}  
    \end{align*}
    satisfies $a_r \xi_r= \xi_{r+1}$ and $\| a_r -1 \| \leq 2\| P_{\xi_{r+1}}-P_{\xi_r}\| \leq  32C_f^{1/2}  f(r)$. 
    By \eqref{eqn:polar.decomposition.norm}, the unitary $u_r \coloneqq a_r(a_r^*a_r)^{-1/2} \in \cA_{B_r(\bm{x})}$ satisfies $\| u_r -1 \| \leq 216C_f^{1/2}  f(r)$. Thus, for sufficiently large $r_0>0$, \cref{rmk:almost.local.unitary} implies that $u = \lim_{r \to \infty} u_r \cdots u_{r_0}$ converges to an almost local unitary such that $u \xi_{r_0} =\xi$. Since there is a unitary $v \in \cA_{B_{r(\bm{x})}}$ such that $v\xi_{r_0}=\Omega$, this proves the lemma. 
\end{proof}

\begin{lem}[{\cite{artymowiczQuantizationHigherBerry2023}*{Proposition A.1}}]\label{lem:asymp.equal.state.invertible}
        Let $\sfH \in \fH_{\Lambda}^{\al}$ be an invertible gapped UAL Hamiltonian with the ground state $\omega$. Let $\omega' $ be a pure state on $\cA_{\Lambda}$ such that  $\omega$ is asymptotically equivalent to $\omega'$. 
        Then, the pure states $\omega \otimes \omega_0 \otimes \omega_0$ and $\omega'\otimes \omega_0 \otimes \omega_0$ on $\cA_{\Lambda}^{\otimes 3}$ are unitarily equivalent by an almost local unitary. 
\end{lem}
\begin{proof}
    The case of $\omega =\omega_0$ is already proved in \cref{lem:asymp.equal.state.trivial}. 
    Consider the ground state $\omega$ of a general IG UAL Hamiltonian $\sfH$. Let  $\check{\sfH}$ be a homotopy inverse of $\sfH$ with the null-homotopy $\overline{\sfH}$ of $\sfH \boxtimes \check{\sfH}$, let $\check{\omega}=\omega_{\check{\sfH}}$ be the distinguished ground state of $\check{\sfH}$, and let $\alpha \coloneqq \alpha(\sfG_{\overline{\sfH}} \,; 1)^{-1} $. Then, by \cref{thm:automorphic.equivalence}, we have $(\omega \otimes \check{\omega})\circ \alpha = \omega_0 \otimes \omega_0$. Moreover, $(\omega' \otimes \check{\omega}) \circ \alpha$ and $\omega_0=(\omega \otimes \check{\omega}) \circ \alpha$ satisfy the assumption of \cref{lem:asymp.equal.state.trivial} as 
    \begin{align*}
     {}&{} |(\omega' \otimes \check{\omega}) \circ \alpha (a) - (\omega \otimes \check{\omega}) \circ \alpha (a) | \\
     \leq{}&{} |(\omega' \otimes \check{\omega}) \circ \Pi_{B_r(\bm{x}_0)^c}(\alpha(a)) - (\omega \otimes \check{\omega}) \circ \Pi_{B_r(\bm{x}_0)^c}(\alpha( a)) | \\
     {}&{} \quad +  2 D_{\nu_2,\mu}(B_r(\bm{x}_0),B_{2r}(\bm{x}_0)^c) \cdot  (\exp(2C_{\nu_2} L_{\nu_2,\mu} \vvert \sfG_{\overline{\sfH}} \vvert_{\nu_5,\mu_1})-1) \cdot \| a \|
    \end{align*}
   for any $a \in \cA_{B_{2r}(\bm{x}_0)}$, since and $D_{\nu_2,\mu}(B_r(\bm{x}_0),B_{2r}(\bm{x}_0)^c) \leq 2^{l_{\Lambda}+1}\kappa_{\Lambda}^2f_{\nu,\mu}(r)$. 
   Hence we obtain a unitary $u \in \cA_{\Lambda}^{\al}$ such that $(\omega' \otimes \check{\omega}) \circ \alpha = (\omega_0 \otimes \omega_0) \circ \Ad(u)$. Now, 
    \begin{align*}
    (\omega' \otimes \omega_0 \otimes \omega_0) ={}&{}
    \big( \big( (\omega' \otimes \check{\omega}) \circ \alpha \big) \otimes \omega \big) \circ \alpha_{12}^{-1} \circ \alpha_{32}\\
    ={}&{} \big( ((\omega_0 \otimes \omega_0) \circ \Ad(u)) \otimes \omega \big) \circ \alpha_{12}^{-1} \circ \alpha_{32}\\
    ={}&{} (\omega \otimes \omega_0 \otimes \omega_0) \circ \Ad \big( \alpha_{12} \circ \alpha_{32}^{-1} (u \otimes 1) \big).
    \end{align*}
    Here, we used the leg notation so that $\alpha_{ij}$ is the automorphism $\alpha $ acting on $i$-th and $j$-th tensor components of $\cA_{\Lambda}^{\otimes 3}$.
\end{proof}

\begin{lem}\label{lem:quantitative.BR2.6.10}
    The distinguished ground state $\omega_{\sfH}$ of an invertible gapped UAL Hamiltonian $\sfH$ has the $\cF$-clustering property.
\end{lem}
\begin{proof}
    By replacing $\omega_{\sfH}$ with $\omega_{\sfH} \otimes \omega_{\check{\sfH}}$ and $a ,b$ with $a \otimes 1$, $b \otimes 1$ if necessary, we may assume that $\sfH$ is homotopic to $\sfh$ by a path $\overline{\sfH}$. Let $\alpha \coloneqq \alpha (\sfG_{\overline{\sfH}} \,; 1)$ and let $r \coloneqq \mathrm{dist}(X,Y)/2$. By \eqref{eqn:Lieb.Robinson}, we have 
    \begin{align*}
        {}&{}|\omega(ab) -\omega(a)\omega(b)| \\
        \leq {}&{}|\omega_0((\alpha(a)-\Pi_{N_r(X)}(\alpha(a)))\alpha(b))| + |\omega_0((\alpha(a)-\Pi_{N_r(X)}(\alpha(a)))\omega_0(\alpha(b))| \\
        &{} \quad + |\omega_0 (\Pi_{N_r(X)}(\alpha(a))(\alpha(b)-\Pi_{N_r(Y)}(\alpha(b))))|\\
        &{} \quad + |\omega_0(\Pi_{N_r(X)}(\alpha(a)))\omega_0(\alpha(b)-\Pi_{N_r(Y)}(\alpha(b)))|\\
        \leq {}&{} 2 (\|\alpha(a) - \Pi_{N_r(X)}(\alpha(a))\| \cdot \|b\| + \| a \| \cdot \| \alpha(b) - \Pi_{N_r(Y)}(\alpha (b))\|)\\
        \leq {}&{}  2(2^{l_{\Lambda} +1} \kappa_{\Lambda}^2 \cdot \kappa_{\Lambda}(1+r)^{l_{\Lambda}} \cdot f_{\nu,\mu}(r)) \cdot \big(2 \exp(2 C_\nu L_{\nu,\mu}\vvert \sfG_{\overline{\sfH}} \vvert_{\nu_3,\mu_1}) -1 \big)  \cdot \|a \| \cdot \|b\|.
    \end{align*}
    By \cref{rmk:polynomial.coarse.invariance}, this concludes the lemma. 
\end{proof}

Hereafter, we assume that $\Lambda$ is decomposed to a disjoint union of linearly coarsely transverse pair $(Y,Z)$ as $\Lambda = Y \sqcup Z$. For example, if there is a proper large-scale Lipschitz map $p \colon \Lambda \to \bR$ such that $\mathrm{diam} (p^{-1}([-r,r]))$ has a linear growth (cf.\ \cref{defn:polynomially.proper}), then $Y=p^{-1}(\bR_{\geq 0})$ and $Z = p^{-1}(\bR_{<0})$ satisfies the assumption. 
A typical example of such $\Lambda$ is the $r$-neighborhood of the union of conical regions $N_r(Z_L^{\theta_L} \cup Z_R^{\theta_R})$ with the notation in \eqref{eqn:conical.region}. 

\begin{lem}\label{lem:split.1d}
    Let $\Lambda = Y \sqcup Z$ be a linearly coarsely transverse decomposition. Then, the distinguished ground state $\omega_{\sfH}$ of an IG UAL Hamiltonian has the split property with respect to this decomposition.  
\end{lem}
\begin{proof}
    First, we show the lemma when $\sfH$ is smoothly homotopic to the trivial Hamiltonian. We take a smooth null-homotopy $\widetilde{\sfH}$ of $\sfH$. Then, by \cref{lem:LGA.cone.decomposition} and \cref{lem:asymptotically.equal.derivation} applied to $\Pi_{Y}^{\star}\sfG$ and $\Pi_{Z}\Pi_{Z}^\star \sfG$, we have LG automorphisms 
    \[ 
        \alpha_Y \coloneqq \alpha \big( \Pi_{Y}\big( \sfG_{\widetilde{\sfH}} \big) \,; 1 \big)^{-1}, 
        \quad 
        \alpha_Z \coloneqq \alpha \big(\Pi_{Z}\Pi_{Z}^{\star}\big( \sfG_{\widetilde{\sfH}}\big) \,; 1\big)^{-1}
    \]
    and a unitary $u \in \cA_{\Lambda}^{\al}$ such that $\omega = \omega_0 \circ (\alpha_{Y} \otimes  \alpha_Z) \circ \Ad(u)$. 
    This means that $\omega$ and $\omega_0 \circ (\alpha_{Y} \otimes  \alpha_Z)$, and hence $\omega_{Y} \coloneqq \omega|_{\cA_{Y}}$ and $\omega_0 \circ \alpha_Y$, are $\cF$-asymptotically equivalent. 
    Now, by \cite{bratteliOperatorAlgebrasQuantum1987}*{Corollary 2.6.11} applied to the states $\omega_0 \circ \alpha_Y$ and $\omega_{Y} $, we obtain that $\omega_0 \circ \alpha_{Y}$ and $\omega_{Y}$ are quasi-equivalent, and hence $\pi_{\omega_{L}}(\cA_{Y})''$ is a type I factor. 

    For a general IG UAL Hamiltonian $\sfH$ with a homotopy inverse $\check{\sfH}$, let $\omega$ and $\check{\omega}$ denote the distinguished ground state of $\sfH$ and $\check{\sfH}$. Then, by the above argument, $(\omega \otimes \check{\omega})|_{\cA_{Y} \otimes \cA_Y} = \omega|_{\cA_{Y}} \otimes \check{\omega}|_{\cA_{Y}}$ is a type I factor state. This means that 
    \begin{align*} 
    \pi_{\omega_{Y} \otimes \check{\omega}_{Y}} (\cA_{Y} \otimes \cA_{Y}) '' \cong \pi_{\omega_{Y}}(\cA_{Y} ) '' \mathbin{\bar{\otimes}} \pi_{\check{\omega}_{Y}}(\cA_{Y})''
    \end{align*}
    is a type I factor, which concludes by \cite{takesakiTheoryOperatorAlgebras2002}*{Theorem V.2.30} that $\pi_{\omega_{Y}}(\cA_{Y} ) ''$ is also a type I factor.
\end{proof}

\begin{lem}[{\cite{carvalhoClassificationSymmetryProtected2024}*{Corollary 5.3}}]\label{prp:left.equivalent.state.unitary}
    Let $\Lambda = Y \sqcup Z$ be a linearly coarsely transverse decomposition. Let $\sfH$ be an invertible gapped UAL Hamiltonian with the ground state $\omega$. 
    Then, by taking the tensor product with the trivial Hamiltonian $\sfh$ if necessary, there exists a unitary $u \in \cA_{\Lambda}^{\al}$ and pure states $\omega_{Y}$, $\omega_{Z}$ such that $ \omega \circ \Ad(u) = \omega_{Y} \otimes \omega_{Z}$. 
\end{lem}
\begin{proof}
    By \cref{lem:split.1d}, the GNS representation $(\sH_{\omega},\pi_{\omega })$ decomposes into the tensor product as $(\sH_{Y} \otimes \sH_{Z} , \pi_{Y} \otimes \pi_{Z})$.
    Let $\Omega_{{\omega}} = \sum_{i=1}^\infty \lambda_i \cdot \xi_i \otimes \eta_i$ be the Schmidt decomposition of the vacuum state with respect to this decomposition. 
    It suffices to show that the vector state $\omega' \coloneqq \langle \blank (\xi_1 \otimes \eta_1), \xi_1 \otimes \eta_1\rangle $ satisfies the assumption of \cref{lem:asymp.equal.state.invertible}. 

    By \cref{lem:quantitative.BR2.6.10}, $\omega_{\sfH}$ satisfies the $\cF$-clustering property.
    By the von Neumann double commutant theorem and the Kaplansky density theorem (\cite{takesakiTheoryOperatorAlgebras2002}*{Theorems II.3.9, II.4.8}), there is a net $\{ b_i \} $ of norm $\leq 1$ elements in $\cA_{Z}$ that $\sigma$-strongly converges to $\eta_1 \otimes \eta_1^* \in \cB(\sH_Z)$. 
    Therefore, for any $a \in \cA_{N_r(Z)^c}$, we have
    \begin{align*}
        |\omega'(a) -\omega(a)| = {}&{} |\lambda_1^{-1} \langle a \otimes (\eta_1 \otimes \eta_1^*)\Omega , \Omega \rangle - \langle a\Omega,\Omega \rangle | \\
        \leq  {}&{}\lambda_1^{-1} \cdot\limsup_{i}  |\omega(ab_i) - \omega(a)\omega(b_i)| \\     
        \leq {}&{} C \cdot \lambda_1^{-1} \cdot f(r) \cdot \| a\|
    \end{align*}
    for some $C>0$. 
    The same inequality also follows for $a \in \cA_{N_r(Z)^c}$. This shows that $\omega$ and $\omega'$ are $\cF$-asymptotically equivalent. 
\end{proof}

\subsection{Fermionic and \texorpdfstring{$G$}{G}-invariant systems}
The argument in the previous subsections immediately runs parallel in the fermionic quantum spin system.

Let us remind of the terminology of $\bZ/2$-graded C*-algebras. We refer the readers to \cite{blackadarTheoryOperatorAlgebras1998}*{Section 14}. A $\bZ/2$-grading of a C*-algebra $A$ is an involutive $\ast$-isomorphism $\gamma$. If $A$ is a matrix algebra, then such $\gamma$ is of the form $\Ad(\Gamma)$ for some self-adjoint unitary $\Gamma \in A$. 
We write $A^{(0)}$ and $A^{(1)}$ for the subspace of $A$ consisting of elements $\gamma(a) = a$ and $\gamma(a)=-a$ respectively. 
We use the bracket notation $[\blank,\blank]$ for the graded commutator 
\begin{align}
    [a_1,a_2] = a_1a_2 - (-1)^{|a_1|\cdot |a_2|}a_2a_1\label{eqn:graded.commutator}
\end{align}
for $a_i \in A^{(|a_i|)}$. 
A $\bZ/2$-graded $\ast$-representation of $A$ is a $\ast$-homomorphism $\pi \colon A \to \cB(\sH)$ and a self-adjoint unitary $\Gamma \in \cB(\sH)$ such that $\pi(\gamma(a)) = \Gamma \pi(a)\Gamma$. 
If a state $\omega $ on $A$ is even, i.e., $\omega \circ \gamma = \omega$, then its GNS representation $(\sH_{\omega},\pi_\omega,\Omega_{\omega})$ gives an example of a graded $\ast$-representation via the $\bZ/2$-grading $\Gamma \pi_\omega(a)\Omega = \pi_{\omega}(\gamma(a))\Omega$ chosen so that the vacuum vector $\Omega_\omega$ is even. 

The (reduced) graded tensor product $A \hotimes B$ of two $\bZ/2$-graded C*-algebras $A$, $B$ is a C*-completion of the universal algebra generated by $A$ and $B$ so that $a \in A$ and $b \in B$ are graded commutative.
It is generated by the elements of the form $a \hotimes b$ for $a \in A$ and $b \in B$, with the relation $(a_1 \otimes b_1) (a_2 \hotimes b_2) = (-1)^{|a_2| \cdot |b_1|} a_1a_2 \hotimes b_1b_2$ for $a_i \in A^{(|a_i|)}$ and $b_i \in B^{(|b_i|)}$. 
The norm on $A \hotimes B$, making it a C*-algebra, is defined in the following way. 
Pick a faithful graded $\ast$-representation $\pi_A \colon A \to \cB(\sH_A)$, $\pi_B \colon B \to \cB(\sH_B)$. 
Then $A \hotimes B$ is identified with a subalgebra of $\cB(\sH_A\otimes \sH_B)$ via a $\ast$-representation $\pi$ defined by $\pi(a \hotimes b) \coloneqq \pi_A(a) \otimes \pi_B(b_0) + \pi_A(a )\Gamma_A \otimes \pi_B(b_1)$ if $b = b_0+b_1$ is the decomposition of $b$ with $b_i \in B^{(i)}$. 
The involution $\gamma\coloneqq \Ad(\Gamma_1 \otimes \Gamma_2)$ makes $A \hotimes B$ a $\bZ/2$-graded C*-algebra. 
By definition, $A \hotimes B$ is $\ast$-isomorphic to $A \otimes B$ if the $\bZ/2$-grading $\Gamma $ on $A$ is inner, e.g., if $A$ is the matrix algebra. 
\begin{rmk}\label{rmk:graded.conditional.exp}
The graded tensor product $A \hotimes B$ is canonically regarded as a C*-subalgebra of $(A \rtimes \bZ/2) \otimes B$ generated by operators of the form $a \otimes b_0 + a v \otimes b_1$, where $v \in \bZ/2$ is the generator. 
For this reason, many properties of $\bZ/2$-equivariant C*-algebras inherit to $\bZ/2$-graded C*-algebras.
For example, if $A$, $B$, $D$ are $\bZ/2$-graded and $\varphi \colon B \to D$ is a completely positive map preserving the $\bZ/2$-grading, then
\begin{align*}
    1 \hotimes \varphi \colon A \hotimes B \to A \hotimes D, \quad (1 \hotimes \varphi)(a \hotimes b) = a \hotimes \varphi(b) ,
\end{align*}
is well-defined and is also completely positive. 
In particular, if $\varphi$ is an even state on $B$, then $1 \hotimes \varphi$ is a conditional expectation preserving the $\bZ/2$-grading. 
\end{rmk}

The notion of a fermionic quantum spin system is formulated in the same way as the above subsections, except that the internal degree of freedom $\widehat{\cA}_{\bm{x}} \cong M_{2n_{\bm{x}}}(\bC)$ is equipped with a $\bZ/2$-grading $\Ad(\mathsf{P}_{\bm{x}})$. 
Here, the self-adjoint unitary $\mathsf{P}_{\bm{x}}$, the fermion parity operator, is the diagonal matrix $\diag(1_{n_{\bm{x}}},-1_{n_{\bm{x}}})$. 
The fermionic observable algebras $\widehat{\cA}_{\Lambda}^{\alg}$, $\widehat{\cA}_{\Lambda}$, $\widehat{\cA}_{\Lambda}^{\al}$ are defined by the graded tensor product as
\begin{align*} 
    \widehat{\cA}_{\Lambda}^{\alg} \coloneqq  \mathop{\widehat{\bigotimes}}_{\bm{x} \in {\Lambda}} \cA_{\bm{x}}, \quad \widehat{\cA}_{\Lambda}\coloneqq \overline{\widehat{\cA}_{\Lambda}^{\alg}}, \quad \widehat{\cA}_{\Lambda }^{\al} \coloneqq  \{ a \in \widehat{\cA}_{\Lambda} \mid \| a\|_{\bm{x},f } <\infty \text{ for any $f \in \cF$}\},
    \end{align*}
where the fermionic almost local norm 
$\| \cdot\|_{\bm{x},f}$ is defined by using the conditional expectation
\begin{align*}
    \Pi_Z \coloneqq  1 \hotimes \mathrm{tr}_{Z^c} \colon \widehat{\cA}_\Lambda \to \widehat{\cA}_{Z}.
\end{align*}
Here, $\mathrm{tr}_{Z^c}$ is the normalized trance of $\widehat{\cA}_{\Lambda}$, which is unique (and hence $\Ad(\sfP)$-invariant) since $\widehat{\cA}_{\Lambda}$ is a colimit of matrix algebras. We write $\Ad (\sfP)$ for the $\bZ/2$-grading imposed on these algebras. 
As is remarked in \cref{rmk:graded.conditional.exp}, this $\Pi_Z$ is a graded conditional expectation. Note that, by fixing an odd unitary $v \in \cA_{\bm{x}}^{\al}$ such that $\bm{x} \not \in Z$, $\Pi_Z$ has an integral form
\begin{align*}  
    \Pi_Z(a) = \lim_{Y \subset Z^c} \int_{u \in \cU(\cA_{Y}^{(0)})} \frac{u(a + vav^*)u^*}{2} du,
\end{align*}
and hence we obtain the fermionic version of \cref{rmk:CE.average};
\[
    \| a - \Pi_Z(a) \| \leq \sup_{u \in \cA_Z^{(0)}} \max \{ \| uau^* -a\|, \|uvav^*u^*-a \| \}.
\]
By substituting this inequality, the proof of \cref{prp:Lieb.Robinson} and the subsequent propositions in \cref{subsection:quantum.spin.system,subsection:automorphic.equivalence,subsection:adiabatic.interpolation} are immediately generalized to fermionic versions under the following definitions.

\begin{defn}
A fermionic UAL Hamiltonian is a collection of operators $\sfH = ( \sfH_{\bm{x}} )_{\bm{x} \in {\Lambda}}$ such that $\vvert H\vvert_{f} \coloneqq  \sup_{\bm{x} \in \Lambda} \| \sfH_{\bm{x} }\| _{\bm{x},f} < \infty$ for any $f\in \cF$ and each $\sfH_{\bm{x}}$ is contained in $\cA_{\Lambda}^{(0)}$.
We write $\ffH_{\Lambda} ^{\al}$ for the set of fermionic gapped UAL Hamiltonians $(\sfH , \omega_{\sfH})$ with gap $1$ that is invariant under the $\bZ/2$-grading $\Ad(\mathsf{P})$. The notion of smoothness of a family of fermionic UAL Hamiltonians is defined in the same way as \cref{defn:smooth}. 
\end{defn}

Indeed, as is mentioned by Bourne--Ogata in \cite{bourneClassificationSymmetryProtected2021}*{Proposition 3.4 and Lemma B.4}, the same proof as \cites{nachtergaeleQuasilocalityBoundsQuantum2019,moonAutomorphicEquivalenceGapped2020} shows the fermionic versions of the Lieb--Robinson bound \eqref{eqn:Lieb.Robinson} and the automorphic equivalence (\cref{thm:automorphic.equivalence}) as well, just by replacing the commutator with the graded commutator \eqref{eqn:graded.commutator}.

Instead of repeating the full argument, we note the following two remarks. 
Here, the trivial Hamiltonian $\sfh_{\bm{x}}= 1- \Omega_{\bm{x}} \otimes \Omega_{\bm{x}}^* $ is fixed so that its ground state is an even unit vector, i.e., $\sfP_{\bm{x}} \Omega_{\bm{x}}=\Omega_{\bm{x}}$. This is a stronger assumption than just requiring that $\sfh_{\bm{x}} \in \widehat{\cA}_{\bm{x}}^{(0)}$.
\begin{lem}\label{lem:fermion.onsite.invariant.analysis}
The following hold.
\begin{enumerate}
    \item Let $\fFDer^{\al}_{\Lambda}$ be the set of local generators of even UAL derivations on $\widehat{\cA}_{\Lambda}$, i.e., families $(\sfG_{\bm{x}})_{\bm{x} \in \Lambda}$ of even skew-adjoint operators such that $\vvert \sfG \vvert_f =\sup_{\bm{x}\in\Lambda}\|\sfG_{\bm{x}}\|_{\bm{x},f} < \infty$. 
    Then, for a smooth map $\sfG \colon \sM \to \fFDer^{\al}_{\Lambda}$ and a path $c \colon [0,1] \to \sM$, the parallel transport $\alpha_{c}(\sfG\,; t)$ preserves the $\mathbb{Z}/2$-grading, i.e., $\alpha_{c}( \sfG  \,; t) \circ \Ad(\sfP) = \Ad(\sfP) \circ \alpha_{c}(\sfG \,; t)$ holds. 
    \item Let $(\sfH , \omega_{\sfH})$ be a fermionic gapped UAL Hamiltonian. Assume that $\sfH$ is invertible, that is, there is another fermionic gapped UAL Hamiltonian $(\check{\sfH} , \omega_{\check{\sfH}})$ such that $(\sfH \boxtimes \check{\sfH} , \omega_{\sfH} \otimes \omega_{\check{\sfH}})$ is connected to the trivial Hamiltonian by a smooth path of fermionic gapped UAL Hamiltonians.
    Then $\omega_{\sfH}$ preserves the $\bZ/2$-grading, i.e., $\sfH \in \ffH_{\Lambda}^{\al}$.
\end{enumerate}
\end{lem}
\begin{proof}
    The claim (1) is easily verified since $\Ad(\mathsf{P}) \circ \alpha_{c}(\sfG\,; t)\circ \Ad(\sfP) $ and $\alpha_{c}(\Ad(\mathsf{P})(\sfG)  \,; t) = \alpha_{c}(\sfG \,; t)$ satisfies the same ODE. 
    As for (2), assume that $\omega_{\sfH}\circ \Ad(\mathsf{P}) \neq \omega_{\sfH}$. Let $\overline{\sfH}$ be a gapped smooth homotopy of $\sfH \boxtimes \check{\sfH}$ and $\sfh$. Then, since $\Ad(\sfP)( \sfG_{\overline{\sfH}}) = \sfG_{\overline{\sfH}}$ and $\omega_0 \circ \Ad(\mathsf{P}) = \omega_0$, by \cref{thm:automorphic.equivalence} we get
    \begin{align*}
    \omega_{\sfH} ={}&{} \omega_0 \circ \alpha\big(\sfG_{\overline{\sfH}} \,; s\big)^{-1} |_{\widehat{\cA}_{\Lambda}}
    = \omega_0 \circ \Ad(\mathsf{P}) \circ \alpha\big( \sfG_{\overline{\sfH}} \,; s\big)^{-1}|_{\widehat{\cA}_{\Lambda}}\\
    = {}&{}\omega_0 \circ \alpha \big(\sfG_{\overline{\sfH}} \,; s \big)^{-1} \circ \Ad(\mathsf{P})|_{\widehat{\cA}_{\Lambda}} = \omega_{\sfH } \circ \Ad(\mathsf{P}). \qedhere
    \end{align*}
\end{proof}
Similarly, if there is an action of a compact Lie group $G$ on each matrix algebra $\cA_{\bm{x}}$ by a linear representation, then the same holds as \cref{lem:fermion.onsite.invariant.analysis}. 
That is, a $G$-invariant UAL derivation generates a $G$-invariant automorphism and, if $\sfH$ a $G$-invariant Hamiltonian that has a $G$-invariant homotopy inverse and a $G$-invariant null-homotopy, then the distinguished ground state of $\sfH$ is $G$-invariant.

Finally, we remark on a fermionic version of the split property discussed in \cref{subsection:GS}.
Note that, for a $\bZ/2$-graded C*-algebra $A$ and $\pi \colon A \to \cB(\sH)$ is a graded $\ast$-representation, the double commutant $\pi(A)''$, which coincides with the weak closure, is a $\bZ/2$-graded von Neumann algebra. 
A $\bZ/2$-graded von Neumann algebra is called a graded factor if it has no non-trivial even central projection. 
We say that a state $\omega \in \fS(\widehat{\cA}_{\Lambda})$ has the (graded) split property with respect to $\Lambda = Y \sqcup Z$ if $\pi_{\omega}(\widehat{\cA}_{Y})''$ is a graded type I factor (\cites{matsuiSplitPropertyFermionic2020,bourneClassificationSymmetryProtected2021}).
In contrast to the ungraded case, there are two graded type I factors; $\cB(\sH)$ and $\cB(\sH) \hotimes \bC \ell_1$. Here, the Clifford algebra $\bC\ell_1$ is isomorphic to $\bC \oplus \bC$ with the $\bZ/2$-grading $(a,b) \mapsto (b,a)$. Note that $\bC \ell_1 \hotimes \bC \ell_1 \cong \bC \ell_2 \cong M_2(\bC)$.

For the fermionic version of \cref{lem:split.1d}, we remark that the graded tensor product version of \cite{takesakiTheoryOperatorAlgebras2002}*{Theorem V.2.30} follows as well; if the graded tensor product $M_1 \hotimes  M_2$ of two balanced graded factors $M_1, M_2$ becomes a type I factor, then the von Neumann subalgebra $M_1^{(0)} \mathrel{\bar{\otimes}} M_2$ is of type I as well (\cite{bourneClassificationSymmetryProtected2021}*{Lemma A.1}), and hence both $M_1$ and $M_2$ are of type I. 
    Here, a graded von Neumann algebra is said to be balanced if it has an odd self-adjoint unitary. The fermionic version of \cref{prp:left.equivalent.state.unitary} also holds under a stronger assumption.

\begin{lem}\label{prp:left.equivalent.state.unitary.fermion}
    Let $\Lambda = Y \sqcup Z$ be a coarsely transverse decomposition. Let $\sfH$ be a fermionic invertible gapped UAL Hamiltonian with respect to the ground state $\omega$. Suppose that the GNS representation of $\omega$ satisfies that $\pi_{\omega}(\cA_{Y})'' \cong \cB(\sH)$. 
    Then, by taking the tensor product with the trivial Hamiltonian $\sfh$ if necessary, there exists an even unitary $u \in \widehat{\cA}_{\Lambda}^{\al}$ and even pure states $\omega_{Y}$, $\omega_{Z}$ such that $ \omega \circ \Ad(u) = \omega_{Y} \otimes \omega_{Z}$. 
\end{lem}
\begin{proof}
    The proof is given in the same way as \cref{prp:left.equivalent.state.unitary}. In the proof, the unit vectors $\xi_1, \eta_1$ can be taken to be of pure degree. Moreover, the unitary $u \in \widehat{\cA}_{\Lambda}^{\al}$ in \cref{lem:asymp.equal.state.invertible} can be chosen to be even. 
\end{proof}

\section{The sheaf of invertible gapped quantum spin systems}\label{section:sheaf.IP}
In this section, we give a definition of the ``space'' of invertible gapped quantum spin systems. 
As a basis of the analysis of smooth homotopies developed in \cref{section:analysis.spin}, it is relevant to formulate the space of IG UAL Hamiltonians as a \emph{smooth set}, that is, as a sheaf over the category $\Man$ of manifolds.
The topological space of our interest is obtained by taking the geometric realization of its smooth singular set. 
Many elementary homotopy theories, such as path and loop spaces, fibrations, and stable homotopy theory, can be defined in the framework of sheaves on $\Man$, which corresponds to their counterpart in topological spaces through geometric realizations.

\subsection{Sheaves on the category \texorpdfstring{$\mathsf{Man}$}{Man}}\label{subsection:sheaf}
We start with a brief review of the definition and a little elementary homotopy theory of sheaves on the category $\Man$ following the line of Madsen--Weiss \cite{madsenStableModuliSpace2007}*{Sections 2 and A}. 
For a general theory and language of sheaves on categories, we refer the readers to \cite{maclaneSheavesGeometryLogic1994}*{Chapters II and III}. 

\subsubsection{Definition}\label{subsubsection:definition.sheaf}
Let $\Man$ denote the category of smooth manifolds and smooth maps, and let $\Set_*$ denote the category of pointed sets and pointed maps. 
A $\Set _*$-valued presheaf on $\Man$, i.e., a contravariant functor $\sF \colon \Man \to \Set_*$, is said to be a \emph{sheaf} if it satisfies the gluing condition with respect to the Gr\"{o}thendieck topology on $\Man$ given by coverings via open embeddings (\cite{maclaneSheavesGeometryLogic1994}*{Sections III.1, III.4}), which is explicitly described as following; for any manifold $\sM $ and its open cover $\{ \sU_i\}_{i \in I}$ (note that each $\sU_i$ is also a manifold), any collection of sections $(\sfs_i) \in \prod_i \sF(\sU_i)$ satisfying $\sfs_i|_{\sU_{i} \cap \sU_j} = \sfs_j|_{\sU_{i}\cap \sU_{j}}$ for any $i,j \in I$ is uniquely glued to $\sfs \in \sF(\sM)$ such that $\sfs_i=\sfs|_{\sU_i}$. 
Here, for a submanifold $\sN \subset \sM$, we abbreviate pull-back by the inclusion to $\sfs|_{\sN}$.

\begin{rmk}\label{rmk:smooth.section.boundary}
For a closed subset $ \sA \subset \sM$, we write $\sF(\sA) \coloneqq \colim_{\sU}\sF(\sU)$, where $\sU$ runs over all open neighborhoods of $\sA$. For example, one can input a manifold with boundary into $\sF$ by employing such a notation.
\end{rmk}

Two sections $\sfs_0, \sfs_1 \in \sF(\sM )$ are said to be \emph{smoothly homotopic} (or \emph{concordant}) if there is $\tilde{\sfs} \in \sF(\sM \times [0,1])$ (that makes sense by \cref{rmk:smooth.section.boundary}) such that $\tilde{\sfs}|_{\sU_0} = \mathrm{pr}_\sM^*(\sfs_0)$ and $\tilde{\sfs}|_{\sU_1} = \mathrm{pr}_\sM^*(\sfs_1)$ for some open neighborhood of $\sM \times \{ 0\}$ and $\sM \times \{1\}$, where $\mathrm{pr}_\sM$ denote the projection to $\sM$. 
We write $\sF[\sM]$ for the set of smooth homotopy classes of sections on $\sM$. Similarly, for $\sfu \in \sF(\sA)$, we write $\sF(\sM,\sA  \, ; \sfu)$ for the subset of $\sF(\sM)$ consisting of $\sfs \in \sF(\sM)$ such that $\sfs|_{\sA} =\sfu$, which means that $\sfs|_{\sV} = \sfu|_{\sV}$ on some open neighborhood $\sA \subset \sV$. 
We say that $\sfs_0,\sfs_1 \in \sF(\sM,\sA \, ; \sfu)$ are smoothly homotopic relative to $\sfu$ if the smooth homotopy $\tilde{s}$ connecting them is taken from $\sF(\sM \times [0,1], \sA \times [0,1]  \, ; \pr_{\sU}^* \sfu ) $. 
Let $\sF[\sM,\sA  \, ; \sfu]$ denote the corresponding set of relative smooth homotopy classes of sections. For the special case that the reference section $\sfu$ is the base point $\mathbf{1}_\sF$ of $\sF$, we abbreviate it and write $\sF(\sM, \sA)$ and $\sF[\sM,\sA]$.

We write $\Sh (\Man)$ for the category of $\Set_*$-valued sheaves on $\Man$. 
We also write $\sSet_*$, $\kTop_*$, and $\CW_*$ for the category of pointed simplicial sets, pointed compactly generated Hausdorff spaces, and pointed $\CW$-complexes, respectively. 
Let $j \colon \Delta \to \Man$ be the covariant functor from the simplex category $\Delta$ to the category $\Man$, which sends the object $[n] = \{ 0,\cdots, n\}$ to the extended standard simplex $\Delta_e^n\coloneqq \{ (x_0,\cdots, x_n) \in \bR^{n+1} \mid \sum_i x_i=1 \}$. 
For a $\Set_*$-valued sheaf $\sF$ on $\Man$, the composition $\Sing (\sF) : = \sF \circ j$, called the \emph{smooth singular set} of $\sF$, is a pointed simplicial set. 
In combination to the geometric realization functor $| \blank | \colon \mathsf{sSet}_* \to \CW_*$, we get a functor 
\begin{align*} 
    |\Sing (\blank )| \colon \Sh(\Man) \to \sSet_* \to \CW_* \subset \kTop_*.  
\end{align*}
In short, we refer to it as the \emph{realization} of a sheaf.
The topological space obtained here remembers the smooth homotopy theory of $\sF$ in the following way.

\begin{thm}[{\cite{madsenStableModuliSpace2007}*{Proposition A.1}}]\label{thm:MW}
For any manifold $\sM \in \Man$ and a closed subset $\sA \subset \sM$, there are isomorphisms 
\begin{align*} 
    \sF[\sM] \cong {}& [\sM, |\Sing (\sF)|], \\
    \sF[\sM,\sA] \cong {}& [(\sM, \sA) , (|\Sing (\sF)| , \pt )].
\end{align*}
\end{thm}

\subsubsection{Elementary homotopy theory}\label{subsubsection:homotopy.sheaf}
The category $\Sh (\Man)$ inherits some elementary homotopy theory of $\CW_*$ through \cref{thm:MW}.  
With the language of abstract homotopy theory, this is understood as the Quillen equivalence of model categories (in this direction, we refer the readers to \cites{christensenHomotopyTheoryDiffeological2014,kiharaQuillenEquivalencesModel2017,kiharaModelCategoryDiffeological2019,pavlovProjectiveModelStructures2022}) and also \cref{subsection:sheaf.model}).
Here, we quickly rephrase some elementary notions of homotopy theory in terms of sheaves on $\Man$. 

\paragraph{(1) Homotopy groups}
The homotopy group of $\sF \in \Sh(\Man)$ with the basepoint $\sfs_0 \in \sF(\pt)$ is defined by 
\begin{align*}
    \pi_n(\sF \,; \sfs_0) \coloneqq  \sF[S^n,\pt  \, ; \sfs_0] \cong \pi_n \big( |\Sing (\sF)| \,; \sfs_0 \big).
\end{align*} 
We abbreviate the basepoint $\sfs_0$ when it is chosen to be $\sfs_0 = \mathbf{1}$. 
Similarly, for a subsheaf $\sG \leq \sF$, the relative homotopy group with respect to the basepoint $\sfs_0 \in \sG(\pt)$ is defined by
\begin{align*} 
    \pi_n (\sF,\sG  \,; \sfs_0)  \coloneqq \{ \sfs \in \sF(\bD^n) \mid s|_{\sU} \in \sG(\sU) \text{ for some open neighborhood $\sU$ of $\partial \bD^n$}, \sfs(\pt )=\sfs_0  \}/\sim,  
\end{align*}
where $\pt $ is a fixed basepoint in $S^{n-1}$ and $\sim$ is the suitable smooth homotopy.  
By \cref{thm:MW}, we have 
\[
    \pi_n(\sF,\sG \,; \sfs_0) \cong \pi_n(|\Sing \sF|, |\Sing \sG| \,; \sfs_0).
\]
For pairs $(\sF_i, \sG_i)$ of sheaves on $\Man$ ($i=1,2$), a morphism $\phi \colon \sF_1 \to \sF_2$ with $\phi(\sG_1) \subset \sG_2$ is said to be a \emph{weak equivalence} if $\phi$ induces the isomorphism of smooth homotopy groups. 
By \cref{thm:MW} and the Whitehead theorem, this holds if and only if $\phi $ induces isomorphisms of smooth homotopy sets for any $\sM \in \Man$. 

\paragraph{(2) Limit and colimit}
The general notion of limit and colimit of sheaves (see e.g.~\cite{maclaneSheavesGeometryLogic1994}*{Section III.6}) covers many fundamental constructions.
What is used in this paper includes the direct product given by $(\sF \times \sG)(\sM) = \sF(\sM) \times \sG(\sM)$, the fiber product along the morphisms $\phi_i \colon \sF_i \to \sG$ is given by 
    \begin{align*} 
    (\sF_1 \times_\sG \sF_2)(\sM) = (\sF_1 \times_{(\phi_1,\phi_2)} \sF_2)(\sM)  \coloneqq \{ (\sfs_1,\sfs_2) \in \sF_1 \times \sF_2(\sM) \mid \phi_1(\sfs_1) = \phi_2(\sfs_2) \},
    \end{align*} 
and the colimit $\colim_{i \in I}\sF_i$ over a directed system $I$ given by the sheafification of the presheaf $\sM \mapsto \colim_{i \in I} \sF_i(\sM)$. 
If one has an increasing sequence $\{\sU_j\}_{j \in \bN}$ of relatively compact open submanifolds that covers $\sM$, an element of $\colim_{i \in I} \sF_i(\sM)$ is given by a family $(\sfs_{j})_{j}$ where $\sfs_j \in \sF_{i(j)}(\sU_j)$ for some $i(j) \in I$ such that $i(j) \leq i(j+1)$ and $\sfs_{j+1}|_{\sU_j} = \sfs_j$ for any $j \in \bN$.

\paragraph{(3) Mapping sheaves and lifting properties}\label{paragraph:fibration}
For $\sF \in \Sh(\Man)$ and $\sM \in \Man$, the mapping sheaf (internal Hom sheaf) is defined by $\sHom(\sM , \sF) \coloneqq \sF(\sM \times \blank)$. 
More generally, for $\sF, \sG \in \Sh(\Man)$, the set $\Hom (\sF,\sG)$ of morphisms of sheaves (i.e., natural transformations) is endowed with the sheaf structure to form the internal Hom sheaf $\sHom (\sF,\sG)$ given by
\begin{align}
    \sHom (\sF,\sG)(\sM) \coloneqq \Hom (\sF, \sHom (\sM,\sG)).\label{eqn:internal.hom}
\end{align}
Via the Yoneda embedding identifying $\sM \in \Man$ with $C^\infty(\blank, \sM) \in \Sh(\Man)$, the two definitions of $\sHom(\sM, \sF)$ above are actually identical.  

The path sheaf $\cP \sF$ is the subsheaf of $\sHom([0,1],\sF)$ defined by 
\begin{align} 
    \cP \sF (\sM) \coloneqq \bigg\{ \sfs \in \sF(\sM \times [0,1]) \mathrel{\bigg|} 
    \begin{gathered}
    \text{$\sfs|_{\sU_i} = \pr_\sM^*(\ev_i\sfs)$ for some open neighborhoods} \\ 
    \text{$\sU_i$ of $\sM \times \{i\}$ for $i=0,1$.}
    \end{gathered}
     \bigg\},
    \label{eqn:path.sheaf}
\end{align}
where $\ev_t \colon \sF(\sM \times [0,1]) \to \sF(\sM)$ denotes the restriction to $\sM \times \{ t \}$. 
The evaluations $\ev_i \colon \cP \sF \to \sF$ are morphisms of sheaves. By the gluing condition of sheaves on $\Man$, one can define the cocatenation paths $\sfs_1 \circ \sfs_2$ if $\ev_1 \sfs_2 = \ev_0 \sfs_1$. 
We also define the subsheaves $\cP_{\sfs_0} \sF \coloneqq  (\ev_0)^{-1}(\sfs_0)$, $\cP_{\sfs_0,\sfs_1}\sF \coloneqq (\ev_0,\ev_1)^* (\sfs_0,\sfs_1)$ for $\sfs_0,\sfs_1 \in \sF(\pt)$. In particular, the based loop sheaf $\Omega \sF$ is defined by $\Omega \sF \coloneqq \cP_{\mathbf{1}, \mathbf{1}} \sF$.
We say that two morphisms $\phi_0 , \phi_1 \colon \sF \to \sG$ are homotopic if there exists $\tilde{\phi} \colon \sF \to \cP \sG $ such that $\phi_i = \ev_i \circ \tilde{\phi}$ for $i=0,1$. 

A morphism $\phi \colon \sF \to \sG$ is a (Serre-) \emph{fibration} if it has the smooth homotopy lifting property for sections on manifolds, i.e., the map 
\begin{align*} 
    \cP\sF(\sM) \to \mathrm{Cocyl}(\phi)(\sM) \coloneqq \cP\sG(\sM) \times _{(\ev_0 ,\phi)} \sF(\sM)
\end{align*}
is surjective for any $\sM$. 
For such $\phi$, the canonical inclusion
\[
\phi^{-1}(\sfu) \to \operatorname{hofib}_{\sfu}(\phi)\coloneqq \big(\ev_1 \colon  \mathop{\mathrm{Cocyl}}(\phi) \to \sG\big)^{-1}(\sfu),
\]
given by $\sfs \mapsto (\sfu,\sfs)$, is a weak equivalence. Note that the sheaves $\operatorname{hofib}_{\sfu}(\phi)$ are weakly equivalent for all choices of $\sfu$, as long as it is contained in the same connected component in $\sG$.

It is straightforward to check that the above definitions are consistent with those of $|\Sing (\sF)|$ in the following way. 
\begin{lem}\label{lem:survey.homotopy}
The following hold. 
\begin{enumerate}
    \item There are homotopy equivalences $|\Sing (\cP \sF)| \simeq\cP |\Sing(\sF)|$, $|\Sing(\cP_{\sfs} \sF)| \simeq\cP_{\sfs}|\Sing(\sF)|$, and $|\Sing (\Omega \sF)| \simeq\Omega |\Sing (\sF)|$. 
    \item If $\phi \colon \sF \to \sG$ is a fibration, then the fiber sheaves $\sE_{\sfu} \coloneqq \phi^{-1}(\sfu)$ are all weakly equivalent independent of the choice of $\sfu \in \sG(\pt)$. Moreover, the morphisms $\Omega \sG \to \operatorname{hofib}_{\sfu}(\phi) \leftarrow \sE_{\sfu}$ 
    induce the connecting map of the homotopy exact sequence
    \begin{align*} 
    \cdots \to \pi_n(\sE_{\sfu}) \to \pi_n(\sF)  \to \pi_n(\sG)  \xrightarrow{\delta} \pi_{n-1}(\sE_{\sfu}) \to \pi_{n-1}(\sF) \to  \cdots .
    \end{align*}
\end{enumerate}
\end{lem}

\begin{defn}
An $\Omega$-spectrum object in $\Sh(\Man)$ is a family $\{ \sF_n\}_{n \in \bZ_{\geq 0 }}$ of $\Set$-valued sheaves on $\Man$ with weak equivalences $\epsilon_n \colon \sF_n \to \Omega \sF_{n+1}$. 
\end{defn}

\begin{prp}
Let $\{ \sF_n\}_{n \in \bN}$ be an $\Omega$-spectrum object in $\Sh(\Man)$. Then the associated family of realizations $\mathit{F}_n \coloneqq |\Sing \sF_n|$ form a CW-$\Omega$-spectrum. The associated cohomology theory $\mathrm{F}^n(\sM) \coloneqq  [\sM, \mathit{F}_n]$ is isomorphic to the smooth homotopy set $\sF_n[\sM]$ for any $\sM \in \Man$.
\end{prp}

\subsubsection{Equivariant sheaves}
We are also interested in the equivariant homotopy theory of sheaves on $\Man$ for a compact Lie group $G$.

\begin{defn}\label{defn:Gsheaf}
A \emph{$G$-sheaf} on $\Man$ is a pair $(\sF,\nu)$, where $\sF$ is a sheaf on the category $\Man$ and $\nu \colon \sF \to \sHom (G,\sF)$ is a morphism of sheaves, such that $\sHom (\id_G , \nu) \circ \nu = \sHom(m^* ,\id_{\sF} ) \circ \nu$ and $\sHom (\iota^*,\id_{\sF}) \circ \nu =\id_{\sF}$ holds, where $m \colon G \times G \to G$ denotes the multiplication and $\iota \colon \pt \to G$ denotes the embedding to the origin. 
\end{defn}
Pointwisely on $G$, by denoting $\nu_g \coloneqq \iota_g \circ \nu \colon \sF \to \sF$, the above relations are rephrased to $\nu_g \circ \nu_h = \nu_{gh}$ and $\nu_e=\id$, i.e., $g \mapsto \nu_g$ gives a left action of $G$ on $\sF$. 
The assumption that $\nu$ is realized as a morphism of sheaves corresponds to the smoothness of this action. In particular, a $G$-action on the sheaf $C^\infty(\blank,\sM)$ is the same thing as a smooth $G$-action on $\sM$.

Let $(\sM, \gamma)$ be a smooth $G$-manifold, i.e., $\gamma \colon G \times \sM  \to \sM$ be a smooth map satisfying the axiom of group actions from the left. 
The $G$-action onto $\sF(\sM)$ is defined by 
\begin{align*} 
    \tilde{\nu}\colon 
    \sF(\sM) \xrightarrow{\gamma^*} 
    \sF(G \times \sM) \xrightarrow{\nu} 
    \sF(G \times (G \times \sM)) \xrightarrow{(\check{\Delta}_G \times \id_{\sM})^* } 
    \sF(G \times \sM ),
\end{align*} 
where $\check{\Delta}_G \colon G \to G \times G$ is given by $\check{\Delta}_G(g)=(g,g^{-1})$. 
Note that $\nu$ and $\gamma^*$ commutes, i.e., $\tilde{\nu} = (\check{\Delta}_G \times \id_{\sM})^* \circ (\id_G \times \gamma)^* \circ \nu$ also holds. 
We write $\sF(\sM)^G$ for the $G$-invariant subset of $\sF(\sM)$ with respect to this action, i.e., 
\begin{align*}
    \sF^G(\sM) \coloneqq \{ \sfs \in \sF(\sM) \mid \tilde{\nu}(\sfs) = \pr_{\sM}^* (\sfs) \}.
\end{align*}
By restricting it to manifolds with trivial $G$-actions, the assignment $\sM \to \sF^G(\sM)$ gives a sheaf on $\Man$, which we write $\sF $ and call the $G$-fixed point sheaf. 

Two sections $\sfs_1, \sfs_2 \in \sF^G(\sM)$ are \emph{$G$-equivariantly smoothly homotopic} if there is a $G$-invariant smooth homotopy $\tilde{\sfs} \in \sF^G(\sM \times \bR)$ connecting $\sfs_1$ and $\sfs_2$. 
Let $\sF^G[\sM]$ denote the set of smooth $G$-homotopy classes. 
The relative version $\sF^G(\sM,\sA)$ and $\sF^G[\sM,\sA]$ are also defined similarly.  

In the same way as \cref{thm:MW}, there is a $G$-space whose $G$-homotopy type represents the functor $\sM \mapsto \sF^G[\sM]$ in the following way.  We write $\Sh_G(\Man)$ and $\kTop_G$ for the categories of $G$-sheaves and $G$-spaces respectively.
\begin{thm}\label{thm:equivariant.MW}
There is a functor $|\Sing_G (\blank)|_G \colon \Sh_G(\Man) \to \kTop_G$, which we call the \emph{$G$-realization}, that satisfies the following property. For any $G$-sheaf $\sF$ on $\Man$, any $G$-manifold $\sM$ and any closed subset $\sA \subset \sM$, there are isomorphisms
\begin{align*} 
    \sF^G [\sM] \cong{}& \big[ \sM, |\Sing_G \sF|_G\big]^G,\\
    \sF ^G[\sM,\sA] \cong{}& \big[ (\sM,\sA) , (|\Sing_G\sF|_G , \mathbf{1}) \big]^G.
\end{align*}
\end{thm}
\cref{section:Gsheaf} provides a more comprehensive treatment of these objects from the viewpoint of equivariant algebraic topology, including the proof of \cref{thm:equivariant.MW}.

\begin{rmk}
    The $G$-space $|\Sing _G \sF|_G$ is a $G$-CW-complex. If a $G$-equivariant morphism $\phi \colon \sF \to \sG$ of $G$-sheaves induces weak equivalences of fixed point subsheaves $\phi \colon \sF^H \to \sG^H$ for any closed subgroup $H \leq G$, then, by the Elmendorf theorem \cite{elmendorfSystemsFixedPoint1983}, the induced map $\phi \colon |\Sing _G \sF|_G \to |\Sing _G \sG|_G$ is indeed a $G$-homotopy equivalence, i.e., has a $G$-homotopy inverse $\phi^{-1} \colon |\Sing _G \sG|_G \to |\Sing _G \sF|_G$. 
\end{rmk}

A concrete sheaf, or a diffeological space, is a sheaf $\sF$ such that $\sF(\sM) \to \prod_{p \in \sM} \sF(p)$ is injective (\cite{baezConvenientCategoriesSmooth2011}). 
\begin{lem}\label{lem:inducction.sheaves}
    Let $\sF$ be a $G$-sheaf that is also a diffeological space. Then, for a closed subgroup $K$ of $G$ and a $K$-manifold $\sM$, there is a bijection $\mathrm{ind}_K^G \colon \sF^K(\sM) \to \sF^G(G \times_K \sM)$, which we call the induction. 
\end{lem}
\begin{proof}
    By the associativity of $\nu$, the composition 
    \[ 
     \iota_{gK}^* \circ \nu \colon \sF(\sM) \to \sF(G \times \sM) \to \sF(gK \times \sM)
    \]
    sends $\sfs \in \sF^K(\sM)$ to $\Im \pr_{\sM}^*(\sfs)$.  
    Since $\sF$ is assumed to be a diffeological space, this means that $\nu \sfs$ is invariant under the pull-back by the right $K$-action onto $G$. Now, take a local trivialization $(\{\sU_i , \varphi_i\})$ of the principal $K$-bundle $G \times \sM \to G \times_K \sM$. Then, $\{\sU_i\}$ is an open cover of $G \times_K\sM$ and the family of sections $\varphi_i^*\sfs \in \sF(\sU_i)$ satisfy the gluing condition. 
    We define $\mathrm{ind}_K^G (\sfs)$ by the section obtained by gluing them. 
    Again, by using the assumption of diffeological space, a pointwise comparison shows the bijectivity of $\mathrm{ind}_K^G$.
\end{proof}

\subsection{Local commutative H-monoids and H-groups}\label{subsubsection:Hmonoid}
A sheaf $\sF$ is said to be an H-monoid if it is equipped with a morphism $\mu \colon \sF \times \sF \to \sF$ and the smooth homotopies $\mu \circ (\id \times \iota) \simeq\id$ and $\mu \circ (\mu \times \id )  \simeq\mu \circ (\id \times \mu)$, where $\iota \colon \mathbf{1} \to \sF$ denotes the initial morphism. This $\mu$ induces the monoid structure on the set $\sF[\sM]$ for each $\sM \in \Man$. 
An H-monoid $(\sF,\mu)$ is called commutative if it is equipped with a homotopy $\mu \simeq\mu \circ \flip$. We introduce their relaxed versions when a sheaf $\sF$ is obtained as a colimit. 

For a directed set $I$, we write $NI$ for its nerve, the simplicial set given by $NI[n] = \Hom ([n],I) $ (where the right hand side means the set of ordered maps). 
For a projective system $\sG_{\bullet} \colon I^{\mathrm{op}} \to \Sh(\Man)$, we define the bar resolution $B(I, \sG_{\bullet})$ as the simplicial set whose set of $n$-cells are given by
\[
    B(I,\sG_{\bullet })[n] \coloneqq \bigsqcup_{i_0 \leq \cdots \leq i_n } \Sing \sG_{i_0} [n] ,
\]
with appropriate morphisms. It has a canonical projection $B(I, \sG) \to B(I,\mathbf{1}_{\bullet})= NI$ induced from the terminal map $\sG_{\bullet} \to \mathbf{1}$. 
If $\{ \sF_{i}\}_{i \in I}$ is an inductive system, then the corresponding system of internal Hom sheaves $\sHom(\sF_{\bullet}, \sG)$ forms a projective system indexed by $I^{\mathrm{op}}$.

\begin{defn}\label{defn:local.commutative.Hmonoid}
    Let $\{ \sF_i\}_{i \in I}$ be a net of sheaves on $\Man$ and let $\sF \coloneqq \colim_{i} \sF_i$. 
    We say that two morphisms $\phi_0,\phi_1 \colon \sF \to \sG$ are \emph{locally homotopic} if there is a simplicial section
    \[
     \tilde{\phi} \in \Gamma ( NI, B(I, \sHom (\sF_{\bullet}, \cP \sG))
    \]
    such that $\ev_0 \tilde{\phi}_{i_0,\cdots,i_n} = \phi_0$ and $\ev_1 \tilde{\phi}_{i_0,\cdots,i_n}= \mu (\iota \times \mu)$ independent of $\{ i_0, \cdots , i_n\} \in NI$. 
    
    We say that $\mu \colon \sF \times \sF \to \sF$ is \emph{locally homotopy associative} if there is a local homotopy connecting $\mu (\mu \times \iota)$ and $\mu(\iota \times \mu)$. 
    The local homotopy unit condition and the local homotopy commutativity are defined similarly. 
    We say that $(\sF,\mu)$ is a \emph{local commutative H-monoid} if it is locally homotopy associative, locally homotopy commutative, and satisfies the local unit condition.
\end{defn}

Explicitly, $\mu$ is locally homotopy associative if the following conditions hold: For any $i_0 \leq \cdots \leq  i_n \in I$, there is a smooth $n$-homotopy $\eta_{i_0,\cdots,i_n} \colon \sF_{i_0}^3 \to \sHom (\Delta_n^e, \cP \sF)$ such that $\ev_0 \eta_{i_0,\cdots,i_n}=\mu (\mu \times \id) $, $\ev_1 \eta_{i_0,\cdots,i_n} = \mu(\id \times \mu)$, and $f^*\eta_{i_0,\cdots,i_n} = \eta_{f^*(i_0,\cdots,i_n)}$ for any $f \colon [k] \to [n]$.     
Roughly, this means that the product $\mu|_{\sF_i} \colon \sF_i \times \sF_i \to \sF$ has a coherence homotopy for each $i \in I$, and such homotopies defined for different $i,j \in I$ are consistent up to homotopy. 
If $\mu$ is locally homotopy associative, then the homotopy colimit $\hocolim \sF_i $ is homotopy associative. 

\begin{rmk}\label{rmk:local.Hmonoid.structural.homotopy}
    Let $\sF =\colim _i \sF_i $ be a local commutative H-monoid. Then, for any $\sM \in \Man$ and $\sfs_1,\sfs_2,\sfs_3 \in \sF(\sM)$, there exists a smooth homotopy connecting  $\mu ( \mu(\sfs_1 , \sfs_2), \sfs_3) $ and $\mu ( \sfs_1, \mu(\sfs_2, \sfs_3)) $. If $\sM$ is compact, it immediately follows from the fact that $\sfs_1,\sfs_2,\sfs_3 \in \sF_i(\sM)$ for some $i \in I$. If not, the desired homotopy is obtained by using the following data; an open cover $\sM = \bigcup_{i \in J} \sU_j$, a partition of unity $\rho_j$ subordinate with it, $i(j) \in I$ for each $j \in \bN$ such that $i(j_1) \leq i(j_2)$ if $j_1 \leq j_2$ and $\sfs_k|_{\sU_j} \in \sF_{i(j)}(\sU_j)$. Indeed, the smooth section
    \[
    \tilde{\eta}(\sfs_1,\sfs_2,\sfs_3) (p) \coloneqq \eta_{i_0,\cdots, i_n}(\rho_{i_0}(p),  \cdots, \rho_{i_n}(p)) (\sfs_1,\sfs_2,\sfs_3)
    \]
    on the open subset $\{ p \in \sM \mid \rho_i(p) \neq 0 \text{ unless $i \in \{ i_0,\cdots,i_n \}$} \}$ are glued together, which gives the desired smooth homotopy. 
    Furthermore, any two such homotopies differ only by a choice of partitions of unity (after replacing $i(j)$'s with sufficiently large ones if necessary. Hence, they are homotopic to each other.
    Similarly, there are smooth homotopies $\mu(\sfs ,\mathbf{1}) \simeq\sfs$ and $\mu(\sfs_1,\sfs_2) \simeq\mu(\sfs_2,\sfs_1)$. 
\end{rmk}

Since we are dealing with $\Set_*$-valued sheaves, the basepoint $\mathbf{1}_{\sF} \in \sF(\pt)$ should play the role of the unit of a local commutative H-monoid $\sF$. It immediately satisfies that $\mu(\mathbf{1}_{\sF},\mathbf{1}_{\sF}) =\mathbf{1}_{\sF}$. Moreover, the structural homotopies $\mu (\mathbf{1}_{\sF} , \mathbf{1}_{\sF}) \simeq\mathbf{1}_{\sF}$ and $\mu (\mathbf{1}_{\sF}, \mathbf{1}_{\sF}) \simeq\mu \circ \flip (\mathbf{1}_{\sF},\mathbf{1}_{\sF})$ are the constant homotopies. 

Let us say that the unit $\mathbf{1}_{\sF}$ of a local commutative H-monoid $\sF = \colim \sF_i$ is \emph{weakly strict} if, for any $\sfs, \sfu \in \sF(\sM)$, the following loops of structural homotopies (obtained by \cref{rmk:local.Hmonoid.structural.homotopy} uniquely up to homotopy) are filled by a $2$-homotopy;
    \begin{itemize}
        \item $ \mu(\sfs, \sfu) \simeq \mu(\mathbf{1}_{\sF}, \mu(\sfs,\sfu)) \simeq \mu(\mu(\mathbf{1}_{\sF}, \sfs),\sfu) \simeq \mu(\sfs,\sfu)$, 
        \item $\mu(\sfs, \sfu) \simeq \mu(\sfs, \mu(\mathbf{1}_{\sF},\sfu)) \simeq \mu(\mu(\sfs, \mathbf{1}_{\sF}),\sfu) \simeq \mu(\sfs,\sfu)$, 
        \item $\mu(\sfs, \sfu)  \simeq \mu(\sfs, \mu(\sfu,\mathbf{1}_{\sF})) \simeq \mu(\mu(\sfs,\sfu),\mathbf{1}_{\sF}) \simeq \mu(\sfs,\sfu)$. 
    \end{itemize}

\begin{lem}\label{lem:Hmonoid.invertible}
    Let $(\sF , \mu )$ be a local commutative H-monoid object of $\Sh (\Man) $. The following hold.
    \begin{enumerate}
        \item For $\sfs \in \sF(\sM)$, $[\sfs]\in \sF[\sM]$ has an inverse if and only if $[\sfs(p)] \in \sF[\pt]$ has an inverse for any $p \in \sM$.  
        \item For $\sfs \in \sF(\sM,\sA)$, $[\sfs]\in \sF[\sM,\sA]$ has an inverse if and only if $[\sfs(p)] \in \sF[\pt]$ has an inverse for any $p \in \sM$. 
    \end{enumerate} 
\end{lem}
We say that a section $\sfs \in \sF(\sM)$ is invertible if so is its smooth homotopy class $[\sfs] \in \sF[\sM]$. 
By (1), the set of invertible sections $\sF^\times (\sM)$ form a sheaf on $\Man$. 
Moreover, for any $\sM \in \Man$ and a closed subset $\sA \subset \sM$, the sets $\sF^{\times}[\sM]$ and $\sF^{\times}[\sM,\sA]$ are imposed the structure of abelian groups.
\begin{proof}
This lemma is proved in an analogous way to the following fact known for H-spaces: If a connected CW-complex is an H-monoid, then it is indeed an H-group (see e.g.~\cite{srinivasAlgebraicKtheory2008}*{Corollary A.47}).

Let $\sM $ be a connected manifold and let $\sfs \in \sF(\sM)$ such that $[\sfs(p)] \in \sF[\pt]$ is invertible for any $p \in \sM$. By definition, there is $\sfu \in \sF(\pt)$ such that $[\mu(\sfs(p_0),\sfu )]= [\mathbf{1}_{\sF}] \in \sF[\pt]$. 
By replacing $\sfs$ to $\mu(\sfs,\sfu)$ if necessary, we may assume that $[\sfs(p)] = [\mathbf{1}_{\sF}]$ for any $p \in \sM$. 
Let $\prescript{\circ}{}{\sF}$ denote the subsheaf of $\sF$ consisting of sections $\sfs \in \sF^{\times}(\sM)$ with $[\sfs(p)]=[\mathbf{1}_{\sF}]$ for any $p \in \sM$. By definition, we have $\pi_0(\prescript{\circ}{}{\sF}) =0$. 
Consider the morphism
    \begin{align*}
    \Phi \colon (\prescript{\circ}{}{\sF} \times \prescript{\circ}{}{\sF})(\sM) \to (\prescript{\circ}{}{\sF} \times \prescript{\circ}{}{\sF})(\sM), \quad \Phi(\sfs,\sfs') \coloneqq (\sfs,\mu(\sfs,\sfs')).
    \end{align*}
Note that there is a smooth homotopy
\[
    \mu(\sfs,\sfs') \simeq \mu(\sfs,\sfs') \circ \mathbf{1}_{\sF} = \mu (\sfs \circ \mathbf{1}_{\sF}, \sfs' \circ \mathbf{1}_{\sF}) \simeq \mu (\sfs \circ \mathbf{1}_{\sF}, \mathbf{1}_{\sF} \circ \sfs' ) = \mu(\sfs ,\mathbf{1}_{\sF}) \circ \mu(\mathbf{1}_{\sF},\sfs') \simeq \sfs \circ \sfs' 
\]
by \cref{rmk:local.Hmonoid.structural.homotopy}. Therefore, the induced morphism in homotopy groups $\pi_n(\sF \times \sF)$ for $n > 0$ satisfy $\Phi_*([\sfs],[\sfs']) = ([\sfs],[\sfs]+[\sfs'])$ on the homotopy groups, and hence is an isomorphism. 
Thus, by \cref{thm:MW}, the morphism $\Phi_* $ is bijective for any $\sM$. 
The smooth homotopy class $\Phi_*^{-1}([\sfs],\mathbf{1}_{\sF})$ is represented by some section $\check{s} \in \sF(\sM)$ as $([\sfs],[\check{\sfs}])$. This concludes that $\mu(\sfs,\check{\sfs})$ is smoothly homotopic to $\mathbf{1}_{\sF}$.  
\end{proof}
\begin{lem}\label{lem:bold.invertible.sheaf}
    Let $\sF$ be a local commutative H-monoid in $\Sh(\Man)$ with a weakly strict unit. Let $\boldsymbol{\sF}^\times $ be the sheaf on $\Man$ given by
  \begin{align*}
  \boldsymbol{\sF}^\times (\sM) \coloneqq  \{ (\sfs,\check{\sfs},\bar{\sfs}) \in \sF^\times(\sM) \times \sF^\times (\sM) \times \cP\sF^{\times}(\sM ) \mid \text{$\bar{\sfs} \in \cP_{\mu (\sfs,\check{\sfs} ), \mathbf{1}_{\sF}} \sF(\sM)$}  \}. 
  \end{align*}
  Then $\boldsymbol{\sF}^{\times }$ is weakly equivalent to $\sF^{\times}$ by the forgetful morphism $\mathsf{fg} \colon \boldsymbol{\sF}^\times \to \sF^\times $ given by $\mathsf{fg}(\sfs,\check{\sfs},\bar{\sfs}) = \sfs$. 
\end{lem}
We call such $\boldsymbol{\sF}^\times$ the \emph{refinement} of a commutative H-group $\sF^\times$ with a weakly strict limit. 
\begin{proof}
  First, the morphism $\mathsf{fg}$ is a fibration. Indeed, for $\sfs \in \sF^\times (\sM \times [0,1]_u)$, a lift $(\sfs|_{\sM}, \check{\sfs}|_{\sM}, \bar{\sfs}|_{\sM})  \in \boldsymbol{\sF}^\times (\sM) $ of $\sfs|_{\sM}$ extends to $\sM \times [0,1]_u$ by $\check{\sfs} (p,u)\coloneqq \check{\sfs}|_{\sM}(p)$ and  
  \begin{align} 
  \bar{\sfs} \coloneqq  \text{(composition of smooth homotopies $\mu (\sfs|_{\sM \times \{u\}},\check{\sfs}|_{\sM}) \simeq \mu(\sfs|_{\sM}$, $\check{\sfs}|_{\sM}) \simeq \mathbf{1}_{\sF}$)}. \label{eqn:fibration.forget}
  \end{align}
  The fiber $\mathsf{fg}^{-1}(\sfs_0)$ consists of pairs $(\check{\sfs},\bar{\sfs})$ such that $\bar{\sfs}$ is a smooth homotopy of $\mu (\sfs_0, \check{\sfs})$ and $\mathbf{1}_{\sF}$. 
  Therefore, it is enough to show that any section $(\check{\sfs},\bar{\sfs}) \in \mathsf{fg}^{-1}(\sfs_0)(S^n)$ is smoothly homotopic to the constant section $(\check{\sfs}_0,\bar{\sfs}_0)$. 

  In the rest of the proof, we present a smooth homotopy by an arrow. 
  First, let $\bar{\bar{\sfs}}_0$ be the composition of the smooth homotopies
  \begin{align*}
  \bar{\bar{\sfs}}{}^{(1)}_0 \colon {}&{} \mu(\sfs_0,\check{\sfs}_0) \xrightarrow{\mathrm{unit}} \mu (\mathbf{1}_{\sF}, \mu(\sfs_0,\check{\sfs}_0)) \xrightarrow{\mu(\bar{\sfs}, \mu(\sfs_0,\check{\sfs}_0))^{\mathrm{rev}}}  \mu (\mu(\sfs_0 ,\check{\sfs}), \mu(\sfs_0,\check{\sfs}_0)) \xrightarrow{\mathrm{assoc}} \mu(\sfs_0,\mu(\check{\sfs},\mu(\sfs_0,\check{\sfs}_0))), \\
  \bar{\bar{\sfs}}{}^{(2)}_0 \colon {}&{}\mu(\sfs_0,\mu(\check{\sfs},\mu(\sfs_0,\check{\sfs}_0)))\xrightarrow{\mathrm{assoc}} \mu (\sfs_0, \mu (\mu(\check{\sfs},\sfs_0) , \check{\sfs}_0) ) \xrightarrow{\mathrm{comm}} 
    \mu ( \sfs_0, \mu (\mu(\sfs_0,\check{\sfs}) , \check{\sfs}_0) ) \\
     {}&{}  \hspace{40ex} \xrightarrow{\mu(\sfs_0,\mu(\bar{\sfs},\check{\sfs}_0))}   
    \mu (\sfs_0,\mu(\mathbf{1}_{\sF},\check{\sfs}_0)) \xrightarrow{\mathrm{unit}} \mu(\sfs_0,\check{\sfs}_0) \xrightarrow{\bar{\sfs}_0} \mathbf{1}_{\sF}.
  \end{align*}
  Here, $\mathrm{unit}$, $\mathrm{assoc}$, and $\mathrm{comm}$ denotes the structural homotopies.
    Then, the squares in the diagram of $1$-homotopies
   \begin{align}
   \begin{split}
   \xymatrix@C=3em{
   \mu(\sfs_0,\mu(\check{\sfs},\mu(\sfs_0,\check{\sfs}_0))) \ar[rr]^{\mu(\sfs_0,\mu(\check{\sfs}, \bar{\bar{\sfs}}_0))} \ar[d]^{\mathrm{assoc}}
   && \mu(\sfs_0,\mu(\check{\sfs} , \mathbf{1}_{\sF})) \ar[d]^{\mathrm{assoc}} \ar[r]^{\mathrm{unit}} & \mu(\sfs_0,\check{\sfs}) \ar@{=}[d] \ar[r]^{\bar{\sfs}} & \mathbf{1}_{\sF} \ar@{=}[d]  \\
   \mu(\mu(\sfs_0,\check{\sfs}), \mu(\sfs_0,\check{\sfs}_0)) \ar[rr]^{\mu(\mu(\sfs_0,\check{\sfs}), \bar{\bar{\sfs}}_0))} && \mu(\mu(\sfs_0,\check{\sfs}),\mathbf{1}_{\sF}) \ar[r]^{\mathrm{unit}} & \mu(\sfs_0,\check{\sfs})\ar[r]^{\bar{\sfs}} & \mathbf{1}_{\sF} 
   }
   \end{split} \label{eqn:diagram.homotopy.equivalence.invertible}
   \end{align}
   are filled by $2$-homotopies, where the weak strictness of the unit is used. 
   Hence the pair $(\check{\sfs}, \bar{\sfs}) \in \mathsf{fg}^{-1}(\sfs_0) (S^n)$ can be replaced under a smooth homotopy as
   \[
   (\check{\sfs}, \bar{\sfs}) \simeq (\check{\sfs}',\text{(the clockwise composition of \eqref{eqn:diagram.homotopy.equivalence.invertible})}) \simeq  (\check{\sfs}',\text{(the anticlockwise composition of  \eqref{eqn:diagram.homotopy.equivalence.invertible})}).
   \]
   where $\check{\sfs}' \coloneqq \mu(\check{\sfs},\mu(\sfs_0,\check{\sfs}_0))$. We write $\bar{\sfs}'$ for the anticlockwise composition of \eqref{eqn:diagram.homotopy.equivalence.invertible}. 
   Next, the squares in the following diagram are also filled by $2$-homotopies:
   \begin{align*}
   \xymatrix@C=2em{
   \mu(\sfs_0,\mu(\check{\sfs}, \mu(\sfs_0, \check{\sfs}_0))) \ar[r]^{\mathrm{assoc}} & \mu(\mu(\sfs_0,\check{\sfs}), \mu(\sfs_0,\check{\sfs}_0)) 
   \ar[rr]^{\mu(\bar{\sfs} , \mu(\sfs_0,\check{\sfs}))} \ar[d]^{\mu(\mu(\sfs_0,\check{\sfs}), \bar{\bar{\sfs}}{}_0^{(1)})} && 
   \mu(\mathbf{1}_{\sF}, \mu(\sfs_0,\check{\sfs}_0)) 
   \ar[r]^{\ \ \ \ \mathrm{unit}} \ar[d]^{\mu(\mathbf{1}_{\sF}, \bar{\bar{\sfs}}{}_0^{(1)})} & 
   \mu(\sfs_0,\check{\sfs}_0) \ar[d]^{\bar{\bar{\sfs}}{}_0^{(1)}} \\
   &\mu(\mu(\sfs_0,\check{\sfs}), \mu(\sfs_0,\check{\sfs}')) 
   \ar[rr]^{\mu(\bar{\sfs},\mu(\sfs_0,\check{\sfs}'))} && 
   \mu(\mathbf{1}_{\sF}, \mu(\sfs_0,\check{\sfs}') ) 
   \ar[r]^{ \ \ \ \ \mathrm{unit}} & 
   \mu(\sfs_0,\check{\sfs}').
   }
   \end{align*}
    After the postcomposition with $ \bar{\bar{\sfs}}_0^{(2)}$, the anticlockwise composition of this diagram is homotopic to the anticlockwise composition of \eqref{eqn:diagram.homotopy.equivalence.invertible}. 
    On the other hand, the clockwise composition is the composition of $\bar{\bar{\sfs}}{}_0^{(1)}$ and its reversion, and hence is homotopic to the constant path.  
    This shows that $(\check{\sfs}', \bar{\sfs}') \simeq (\check{\sfs}', \bar{\bar{\sfs}}{}^{(2)}_0)$ as section of $\mathsf{fg}^{-1}(\sfs_0)$. 
    Finally, notice that the homotopy $\bar{\bar{\sfs}}{}_0^{(2)}$ is the composition of $\bar{\sfs}_0$ and a homotopy of the form $\mu(\sfs_0, \sfu)$, where $\sfu$ is a certain homotopy connecting $\check{\sfs}'$ and $\check{\sfs}_0$. A reparametrization of $\sfu$ gives a smooth homotopy $(\check{\sfs}',\bar{\bar{\sfs}}{}_0^{(2)}) \simeq (\check{\sfs}_0,\bar{\sfs}_0)$.  
    This shows the lemma. 
\end{proof}

\subsection{The sheaf of invertible gapped lattice Hamiltonians}\label{subsection:sheaf.lattice}
We give a definition of the space of IG UAL Hamiltonians as a sheaf $\sIP_d$ on $\Man$. We then apply \cref{thm:MW} to obtain the topological space $\IP_d$, which inherits the smooth homotopy type of $\sIP_d$.  

\begin{defn}\label{notn:internal.degree}
    We call $\fR$ a set of \emph{bosonic internal degrees of freedom} if it is a countable set of pairs $\lambda=(\sH_\lambda,\Omega_\lambda)$, where $\sH_\lambda $ is a Hilbert space with $2 \leq \dim \sH_\lambda < \infty $ and $\Omega_\lambda$ is a fixed unit vector. 
    For $\lambda \in \fR$, let $\cA_{\lambda} \coloneqq \cB(\sH_{\lambda})$ and $\sfh_{\lambda}\coloneqq  2(1-\Omega_{\lambda} \otimes \Omega_{\lambda}^*)$. 
\end{defn}

\begin{rmk}
    The most natural choice for $\fR$ is to contain one representative from each isomorphism class of internal degrees of freedom, which forms a countable set.
    Nevertheless, we retain the information of the choice of $\fR$ here because we are aware of previous research on quantum cellular automata (QCA).
    The connection between quantum spin systems and QCAs has been considered (we refer to \cites{ranardConverseLiebRobinson2022,haahNontrivialQuantumCellular2023,haahInvertibleSubalgebras2023,kapustinAnomalousSymmetriesQuantum2024}). 
    The GNVW index \cite{grossIndexTheoryOne2009}, which is a topological invariant that distinguishes the connected components of the space of $1$-dimensional QCAs, takes values in the additive group $\bZ[\{ \log p_i\}]$ generated by the logarithms of prime numbers $ p_i$ appearing in the prime factor of the dimensions of the internal degrees of freedom.
\end{rmk}

\begin{defn}\label{defn:polynomially.proper}
    Let $X$ and $Y$ be proper metric spaces. We say that a map $F \colon X \to Y$ is \emph{linearly proper} if, for any (some) $\bm{x}_0 \in X$, there is $m >0$ such that the diameter of the inverse image $F^{-1}(B_r(\bm{x}))$ is not more than $m \cdot (1+r)$.    
\end{defn}
Typically, the projection $\pr_{\bR^d} \colon X \subset \bR^{d+l} \to \bR^d$ is linearly proper if $X$ is contained in the conical region $\{ (\bm{x},\bm{y}) \in \bR^{d+l} \mid \rmd(\bm{y},\bm{y}_0) \leq m(1+ \rmd(\bm{x},\bm{x}_0)) \}$ (the unshaded region in \cref{fig:transverse.cone} projected to the $y$-axis). 
We remark that the composition of two linearly proper maps is again linearly proper. 

\begin{defn}\label{defn:lattice.set}
For a countable set $\fR$, let $\fL_d=\fL_{\fR, d}$ denote the set of subsets $\Lambda \subset \bR^{d+\infty} \times \fR \times \bN$ satisfying the following conditions:
\begin{enumerate}
    \item[(i)] The set $\Lambda $ is contained in $\bR^{d+l_\Lambda} \times \fR \times [0,n_{\Lambda} ]$ for some $l_{\Lambda} \in \bZ_{\geq 0}$ and $n_{\Lambda} \in \bZ_{\geq 0}$.
    \item[(ii)] The projection $\pr_{\bR^{d+l}} \colon \Lambda \to \bR^{d+l}$ is weakly uniformly discrete in the sense of \cref{subsection:coarse.lattice}, which imposes the pseudo-metric $\rmd$ on $\Lambda$.
    \item[(iii)] The map $\pr_{\bR^d} \colon \Lambda \to \bR^d$ is linearly proper.  
\end{enumerate}
\end{defn}
By (ii), $\Lambda \in \fL_d$ has polynomial growth with the growth rate less than $\kappa_{\Lambda} \cdot (1+r)^{d+l_{\Lambda}}$. Moreover, (iii) implies that there is $\kappa_{\Lambda}' >0$ such that 
\begin{align*}
    |\pr_{\bR^d}^{-1}(B_r(\pr_{\bR^d}(\bm{x})))| \leq |B_{m (1+r)}(\bm{x})| \leq \kappa_{\Lambda} \big( 1+ m (1+r) \big)^{d+l_{\Lambda}}\leq \kappa_{\Lambda}' (1+r)^{d+l_{\Lambda}}.
\end{align*}
\begin{defn}\label{defn:local.observable}
Let $\Lambda \in \fL_{\fR, d}$. 
For $\bm{x} \in \Lambda$, we write as $\lambda(\bm{x})\coloneqq \pr_{\fR}(\bm{x})$.
On each $\bm{x} \in \Lambda $, we place the local observable algebra $\cA_{\bm{x}}\coloneqq \cA_{\lambda(\bm{x})}$ and the trivial Hamiltonian $\sfh_{\bm{x}} \coloneqq \sfh_{\lambda(\bm{x})}$ with the ground state $\Omega_{\bm{x}} \coloneqq \Omega_{\lambda(\bm{x})}$.  
In the same way as \eqref{eqn:trivial.Hamiltonian}, set
\begin{align*}
    \cA_{\Lambda} \coloneqq \bigotimes _{\bm{x} \in \Lambda} \cA_{\bm{x}}, \quad  \sfh = (\sfh_{\bm{x}} )_{\bm{x} \in \Lambda} \in \fH_{\Lambda}^{\al}, \quad \omega_0\coloneqq \bigotimes_{\bm{x}} \langle \blank \Omega_{\bm{x}},\Omega_{\bm{x}} \rangle. 
\end{align*}
\end{defn}
As is observed in \cref{subsection:coarse.lattice}, $ \Lambda  \in \fL_{\fR,d}$ is of polynomial growth and has a brick with respect to the pseudo-metric $\rmd $ induced from the Euclidean metric of $\bR^{d+l}$.

\begin{defn}\label{defn:gapped.sheaf}
We define the following sheaves on $\Man$. 
\begin{enumerate}
    \item Let $\Lambda \in \fL_{\fR,d}$. 
    The sheaf of gapped UAL Hamiltonians $\sGP(\Lambda \midbar \blank)$ on $\Lambda$ is defined by 
    \begin{align*} 
    \sGP ( \Lambda \midbar \sM ) \coloneqq \{  \sfH \colon \sM \to \fH_{\Lambda}^{\al} \mid \text{ $\sfH$ is smooth in the sense of \cref{defn:smooth}} \}. 
    \end{align*}
    \item The sheaf $\sGP_{\fR,d}(\blank)$ on $\Man$ is defined by the colimit as
\begin{align*} 
    \sGP_{\fR,d} \coloneqq {}&  \colim_{\Lambda \in \fL_{\fR,d}} \sGP(\Lambda \midbar \blank) .
\end{align*} 
    Here, the morphisms $\sGP(\Lambda_1 \midbar \blank) \to \sGP(\Lambda_2 \midbar \blank) $ for $\Lambda_1 \leq \Lambda_2$, which forms an inductive system, is given by 
\begin{align*} 
    \sGP(\Lambda_1 \midbar \sM)  \ni \sfH \mapsto \sfH \boxtimes \sfh_{\Lambda_2 \setminus \Lambda_1 } \in \sGP(\Lambda_2 \midbar \blank ).
\end{align*}
Hereafter, we often abbreviate $\fR$ and write $\sGP_d$ unless the argument depends on a particular choice of such $\fR$. 
    \item For $\sfH \in \sGP_{\fR,d}(\pt)$, we call the smallest $\Lambda \in \fL_{\fR, d}$ such that $\sfH \in \fH_{\Lambda}^{\al}$ the \emph{support} of $\sfH$. 
\end{enumerate}
\end{defn}

\begin{rmk}\label{line.mu}
Throughout the paper, we fix a bijection $\fm \colon \bN \sqcup \bN \to \bN$.     
\end{rmk}

The product on $\fL_d$ is defined by 
\begin{align}
    \Lambda_1 \boxtimes \Lambda_2 \coloneqq (\id_{\bR^\infty} \times \id_{\fR} \times \fm) (\Lambda_1 \sqcup \Lambda_2) \in \fL_d,
    \label{eqn:composition.lattice}
\end{align}
where the disjoint union $\Lambda_1 \sqcup \Lambda_2 $ is regarded as the subset of $\bR^{d+\infty} \times \fR \times (\bN \sqcup \bN)$. 
The composition of quantum systems given in \eqref{eqn:Hamiltonian.product} and the map $\fm \colon \Lambda_1 \sqcup \Lambda_2 \to \Lambda_1 \boxtimes \Lambda_2$ gives a morphism of sheaves  
\begin{align} 
    \blank \boxtimes \blank  \colon \sGP(\Lambda_1 \midbar \blank) \times \sGP(\Lambda_2 \midbar \blank ) \to \sGP(\Lambda_1 \boxtimes \Lambda_2 \midbar \blank ).
    \label{eqn:composition.system}
\end{align}
By definition, this product commutes with the inclusions $\sGP(\Lambda_i \midbar \blank) \to \sGP(\Lambda_i' \midbar \blank )$ for $i=1,2$, and hence extends to the morphism of sheaves $\blank \boxtimes \blank \colon \sGP_d \times \sGP_d \to \sGP_d$.

\begin{prp}\label{lem:H-monoid}
    The composite operation $\boxtimes$ makes the sheaf $\sGP_d$ a local commutative H-monoid with a weakly strict unit in the sense of \cref{defn:local.commutative.Hmonoid}. 
\end{prp}
The proof given below may appear somewhat indirect. However, this is intended with the later generalization to the fermionic and the $(G,\phi)$-symmetric versions in mind.

\begin{lem}\label{lem:flip.homotopy}
    There is a smooth homotopy $\widetilde{\flip} \colon [0,1] \to \Aut (\cA_{\Lambda}^{\otimes 4})$ of local automorphisms such that $\widetilde{\flip}{}^s(\sfh) = \sfh$ for any $s \in [0,1]$ and $\widetilde{\flip}{}^0 =\id$, $\widetilde{\flip}{}^1 =  \flip_{12} \otimes \flip_{34}$ hold. 
\end{lem}
\begin{proof}
    First, we remark that the transpositions $(1 \, 2)$ and $(3 \, 4)$ are homotopic in the unitary (orthogonal) group of the real $\ast$-algebra $\bR[\fS_4]$. 
    Indeed, since transpositions are mutually conjugate, they have the same determinant on each irreducible representation $\sV$ of $\fS_4$. This means that they are homotopic in the orthogonal group $O(\sV)$, and hence in $O(\bR[\fS_4])$ by the Peter--Weyl theorem. 
    
    Let $\sfh_{\lambda}' \coloneqq \sfh_{\lambda} \boxtimes \sfh_{\lambda} \boxtimes \sfh_{\lambda} \boxtimes \sfh_{\lambda} \in \cA_{\lambda}^{\boxtimes 4}$. 
    The permutation gives a unitary representation of the symmetric group $v \colon \fS_4 \to \sH_{\lambda}^{\otimes 4}$. In particular, the flip automorphisms are realized by this representation as $\flip_{12} \otimes \flip_{34} = \Ad (v_{(1 \, 2) (3 \, 4)})$. 
    Moreover, this $v$ gives rise to a $\ast$-representation $\phi \colon \bR [\fS_4] \to \cB(\sH_{\lambda}^{\otimes 4}) = \cA_{\lambda}^{\otimes 4}$ such that any element of $\phi(\bR[\fS_4])$ commutes with $\sfh_{\lambda}'$. By the argument in the above paragraph, this proves the lemma. 
\end{proof}

\begin{lem}\label{lem:Ainfty}
    Let us regard the group $\cG_{\Lambda}$, and its subgroup $\cG^{\sfh}_{\Lambda}$ consisting of automorphisms that fix the trivial Hamiltonian, as sheaves on $\Man $ via the smoothness defined in \cref{defn:parallel.transport} (4). 
    The morphism of sheaves 
    \[
    i_{\Lambda} \circ \blank \colon \cG_{\Lambda}^{\sfh} \to \sHom (\sGP(\Lambda \midbar \blank ), \sGP_d )
    \]
    given by $\alpha  \mapsto i_{\Lambda} \circ \alpha$ is smoothly homotopic to the constant morphism $\alpha \mapsto i_{\Lambda}$. 
\end{lem}
\begin{proof}
    Since $\pr_{\bN}(\Lambda) \subset [0,n_{\Lambda}]$, we can take $3$ copies $\Lambda_k \subset \pr_{\bN}^{-1}[kn_{\Lambda}, (k+ 1)n_{\Lambda}]$ (for $k=1,2,3$) of $\Lambda$ that are mutually disjoint. The homotopy of morphisms 
    \[
    \cG_\Lambda^{\sfh} \to \sHom (\sGP(\Lambda \midbar \blank ), \sGP (\Lambda \sqcup \Lambda_1 \sqcup \Lambda_2 \sqcup \Lambda_3 \midbar \blank ) )
    \]
    given by
    \[
    \alpha \mapsto  \big( \sfH \mapsto \widetilde{\flip}{}^s \circ (\alpha \otimes \id \otimes \id \otimes \id) \circ \widetilde{\flip}{}^s (\sfH \boxtimes \sfh \boxtimes \sfh \boxtimes \sfh )\big) 
    \]
    connects $i_{\Lambda} \circ \blank $ and the constant morphism  $\alpha \mapsto i_{\Lambda} $. It is a morphism of $\sSet_*$-valued sheaves since this homotopy fixes the trivial Hamiltonian $\sfh$. 
\end{proof}

\begin{proof}[Proof of \cref{lem:H-monoid}]
    Here, we only prove the local homotopy associativity. The local homotopy unit condition and the local homotopy commutativity are proved similarly. 
    We construct the coherent homotopies $\eta_{\Lambda_0,\cdots,\Lambda_n}$ by induction on the dimension of cells. 
    Suppose that, for $0 \leq k <n$ and any increasing sequence $\Lambda_0 \leq \cdots \leq \Lambda_k$, there is another lattice $U(\Lambda_0,\cdots,\Lambda_k) \in \fL_{\fR,d}$ that contains $(\Lambda_0 \boxtimes \Lambda_0) \boxtimes \Lambda_0 \cup \Lambda_0 \boxtimes (\Lambda_0 \boxtimes \Lambda_0)$ and a smooth map  
    \[
    \eta_{\Lambda_{0},\cdots, \Lambda_{k}} \colon \Delta_k^e \to \cP \cG^{\sfh}_{U(\Lambda_0,\cdots,\Lambda_k)}
    \]
    for any $0 \leq k \leq n-1$ such that $\eta_{\Lambda_{f(0)} , \cdots , \Lambda_{f(l)}} = f^* \eta_{\Lambda_0,\cdots,\Lambda_k}$ for any $f \in [l] \to [k]$ and, on $\sGP(\Lambda_{i_0} \midbar \blank )^3 $, they satisfy the relations 
    \[
    \ev_0 (\eta_{\Lambda_0, \cdots, \Lambda_k})  \circ \mu (\mu \times \iota)  = \mu (\mu \times \iota), \quad  \ev_1 (\eta_{\Lambda_{i_0},\cdots, \Lambda_{i_k}})  \circ \mu (\mu \times \iota)  = \mu (\iota \times \mu ),
    \]
    where the evaluations $\ev_i$ are the ones with respect to the path variable on $\cP \cG^{\sfh}_{U_{\Lambda_0,\cdots,\Lambda_k}}$. Then, by \cref{lem:Ainfty}, they extend to a consistent family of $n$-homotopies $\eta_{\Lambda_0,\cdots,\Lambda_n}$ by taking $U(\Lambda_0,\cdots,\Lambda_n)$ as the disjoint union of $\bigcup_{k=0}^n U(\Lambda_0,\cdots, \hat{\Lambda}_k,\cdots, \Lambda_n)$ and $3$ copies of it. 
    Now, $\eta_{\Lambda_0,\cdots , \Lambda_n} \circ \mu(\mu \times \iota)$ is the desired family.
\end{proof}

\begin{rmk}\label{rmk:Ainfty}
    The same argument as the proof of \cref{lem:H-monoid}, we would get arbitrarily higher coherence homotopies representing the $A_{\infty}$-relation. In simpler terms, we need not care for the subtlety of the order of the product of lattices, e.g.,\ $(\Lambda_1 \boxtimes \Lambda_2) \boxtimes (\Lambda_3 \boxtimes \Lambda_4)$ and $\Lambda_1 \boxtimes ((\Lambda_2 \boxtimes \Lambda_3) \boxtimes \Lambda_4)$.
\end{rmk}

\begin{defn}\label{defn:gapped.invertible}
    Let $\sfH \in \sGP_d(\sM)$. We say that $\check{\sfH} \in \sGP(\check{\Lambda} \midbar \sM )$ is a homotopy inverse system of $\sfH$ if there is and a smooth homotopy $\overline{\sfH} \in \cP\sGP_d(\sM )$ of $\sfH \boxtimes \check{\sfH}$ and $\sfh$. 
    We say that $\sfH$ is \emph{invertible} if, for any $p \in \sM$, there is an open neighborhood $p \in \sU$ such that $\sfH|_{\sU}$ has an inverse system. 
    We write $\sIP_d (\sM )$ for the subset of invertible gapped uniformly almost local (IG UAL) Hamiltonians. By definition, $\sIP_d$ forms a sheaf on $\Man$. 
\end{defn}

In terms of \cref{subsubsection:Hmonoid}, $\sIP_d$ is nothing but $\sGP_d^\times$. 
By \cref{lem:Hmonoid.invertible} and \cref{lem:H-monoid}, the operation $\boxtimes$ makes the sets $\sIP_d[\sM]$ and $\sIP_d[\sM,\sA]$ abelian groups. 
Moreover, the following sheaf is weakly equivalent to its refinement $\bsIP_d \coloneqq \boldsymbol{\mathscr{GP}}_d^\times$ in the sense of \cref{lem:bold.invertible.sheaf} for $\sF=\sGP_d$. Here, we repeat the definition. 
\begin{defn}\label{eqn:IP.bold}
We define the refinement sheaf $\bsIP_d$ of $\sIP_d$ as
\begin{align*}
    \bsIP{}_d(\sM)\coloneqq  \{ \bsfH=(\sfH,\check{\sfH}, \overline{\sfH} ) \in  \sIP_d(\sM )^2  \times \cP\sGP_d(\sM) \mid  \text{$\ev_0\overline{\sfH} =\sfH \boxtimes \check{\sfH}$, $\ev_1 \overline{\sfH} = \sfh$} \}.
\end{align*}    
For a relatively compact open subset $\sU \subset \sM$, there is $\Lambda \in \fL_d$ such that $\sfH, \check{\sfH} \in \sGP(\Lambda \midbar \sU) $ and $\overline{\sfH} \in \sGP(\Lambda \boxtimes \Lambda \midbar \sU)$. We call the smallest $\Lambda$ with this property the support of $\bsfH |_{\sU}$. 
\end{defn}
The following lemma follows from \cref{lem:Hmonoid.invertible}.
\begin{lem}\label{lem:forget.invertible}
The forgetful map $\mathsf{fg} \colon \bsIP{}_d (\sM) \to \sIP_d (\sM)$ given by $\mathsf{fg}(\sfH, \check{\sfH}, \overline{\sfH}) \coloneqq  \sfH$ is a weak equivalence of sheaves. 
\end{lem}

\begin{rmk}
    The sheaf $\bsIP{}_d$ is also endowed with the structure of a local commutative H-monoid via the product 
    \begin{align*} 
    (\sfH_1 , \check{\sfH}_1, \overline{\sfH}_1) \boxtimes (\sfH_2 , \check{\sfH}_2, \overline{\sfH}_2) \coloneqq (\sfH_1 \boxtimes \sfH_2, \check{\sfH}_1 \boxtimes \check{\sfH}_2, \flip_{23}(\overline{\sfH}_1 \boxtimes \overline{\sfH}_2) ) \in \bsIP_d (\sM), 
    \end{align*}
    where $\flip_{23}$ denotes the flip of the second and the third tensor components of $\cA_{\Lambda}^{\otimes 4}$. Strictly speaking, the above definition is insufficient in that it does not specify the order of the three products in $\sfH_1 \boxtimes \sfH_2 \boxtimes \check{\sfH}_1 \boxtimes \check{\sfH}_2$, which corresponds to different ways of identifying the disjoint union $\Lambda^{\sqcup 4}$ with an element of $\fL_d$ by the multiple use of $\fm$. However, by \cref{rmk:Ainfty}, we need not care for this subtlety. 
\end{rmk}

\begin{defn}
The topological space $ \IP_d $ is defined to be the realization $|\Sing (\bsIP_d) |$ of the sheaf $\bsIP_d$ (cf.~\cref{thm:MW}). 
\end{defn}

The inclusion $\sGP(\Lambda_1 \midbar \blank) \subset \sGP(\Lambda_2 \midbar \blank)$ for $\Lambda_1 \leq \Lambda_2$ is generalized to the covariant functoriality with respect to a map of lattices. 
Let $\Lambda_1, \Lambda_2 \in \fL_d$. We say that $F \colon \Lambda_1 \to \Lambda_2$ is a map over $\fR$ if $\lambda \circ F = \lambda$. 
Let $F \colon \Lambda_1 \to \Lambda_2$ be an injective linearly proper large-scale Lipschitz map over $\fR$. 
Then it induces the unital $\ast$-homomorphism of C*-algebras
    \begin{align*} 
    F_* \coloneqq \bigotimes_{\bm{x} \in \Lambda_1} (\cA_{\lambda(\bm{x})} \xrightarrow{\cong} \cA_{\lambda(F(\bm{x}))}) \colon \cA_{\Lambda_1} \hookrightarrow \cA_{\Lambda_2}. 
    \end{align*}

\begin{lem}\label{lem:induced.map.GP}
    The above $F_*$ induces the continuous homomorphism of almost local algebras $F_* \colon \cA_{\Lambda_1}^{\al} \to \cA_{\Lambda_2}^{\al}$. Moreover, it also induces the morphisms $F_* \colon \fDer ^{\al}_{\Lambda_1}  \to \fDer ^{\al}_{\Lambda_2}$ and $F_* \colon \fH_{\Lambda_1}^{\al} \to \fH_{\Lambda_2}^{\al}$ explicitly given by 
    \begin{align*}
        (F_*\sfG)_{\bm{x}} \coloneqq 
        \begin{cases}
            F_*(\sfG_{F^{-1}(\bm{x})}) & \text{ if $\bm{x} \in F(\Lambda_1)$},\\
            0 & \text{ otherwise, }
        \end{cases}
        \quad 
        (F_*\sfH)_{\bm{x}} \coloneqq 
        \begin{cases}
            F_*(\sfH_{F^{-1}(\bm{x})}) & \text{ if $\bm{x} \in F(\Lambda_1)$},\\
            \sfh_{\bm{x}} & \text{ otherwise, }
        \end{cases}
    \end{align*}
    and 
    \begin{align*}
        \omega_{F_*\sfH} \coloneqq (\omega_{\sfH} \circ F^{-1}_*)|_{\cA_{F(\Lambda_1)}}\otimes \omega_0|_{\cA_{F(\Lambda_1)^c}} 
    \end{align*}
    for $\sfG \in \fDer ^{\al}_{\Lambda_1}$ and $\sfH \in \fH_{\Lambda_1}^{\al}$. 
    These morphisms preserve the smoothness of families parametrized by manifolds. 
    In particular, the above definitions give a morphism of sheaves $F_* \colon \sGP(\Lambda_1 \midbar \blank ) \to \sGP(\Lambda_2 \midbar \blank ) $.
\end{lem}
\begin{proof}
    Let $\alpha , \beta >0$ be constants such that $\rmd(F(\bm{x}),F(\bm{y}) ) \leq \alpha \rmd(\bm{x},\bm{y}) +\beta$. 
    For any $a \in \cA_{\Lambda}^{\al}$, by \eqref{eqn:large.scale.Lipschitz.F} we have 
    \begin{align}
    \begin{split}
        \vvert F_*(a) \vvert _{F(\bm{x}),\mu} = {}&{}  \sup_{r>0} f_\mu(r)^{-1} \| F_*(a) - \Pi_{B_r(F(\bm{x}))} (F_*(a)) \| \\
        \leq {}&{}  \sup_{r>0} 2f_\mu(r)^{-1} \| F_*(a - \Pi_{B_{\alpha^{-1} r-\alpha^{-1}\beta }(\bm{x})} (a)) \| \\
        \leq {}&{} 2 \cdot \sup_{r>0} f_\mu(r)^{-1} \cdot f_{\mu_1} (\alpha^{-1} r - \alpha^{-1}\beta) \cdot  \|a \|_{\bm{x},\mu_1}\\
        \leq {}&{} 2c_{0,\mu,\alpha} \cdot f_{\mu_1}(\alpha^{-1}\beta)^{-1} \cdot \| a \|_{\bm{x},\mu_1}.
        \label{eqn:f-norm.push}
    \end{split}
    \end{align}
    
    This also proves that $F_*\sfG \in \fDer ^{\al}_{\Lambda_2}$. 
    Since the GNS representation of $(\cA_{\Lambda_2}, \omega_{F_*\sfH}, F_*\sfH )$ is isomorphic to the tensor product of those of $(\cA_{\Lambda_1}, \omega_{\sfH}, \sfH)$ and $(\cA_{F(\Lambda_1)^c}, \omega_0,\sfh)$, it is gapped. 

    For the same reason, for a smooth family $\sfG \colon \sM \to \fDer ^{\al}_{\Lambda}$, its almost local $C^k$-norms are all bounded above on each relatively compact open chert. 
    Moreover, for a smooth family $\sfH$ of gapped UAL Hamiltonians, there is $0<\mu<1$ such that 
    \[ 
    |\partial^I \omega_{F_*\sfH(p)}(a)| \leq C \cdot \| F_*(a) \|_{\bm{x},\mu} \leq  C \cdot 2c_{0,\mu,\alpha} \cdot f_{\mu_1}(\alpha^{-1}\beta)^{-1} \cdot \| a \|_{\bm{x}, \mu_1}.
    \]
    Therefore, $F_*\sfH$ is also smooth. 
    Finally, since $F_*\sfh = \sfh$, the morphism $F_*$ is a morphism of $\Set_*$-valued sheaves. 
\end{proof}
The following lemma is proved in the same way as \cref{lem:H-monoid}. 
\begin{lem}\label{lem:flip.homotopy.use}
    Let $\Lambda_1, \Lambda_2 \in \fL_d$,
    and let $F_1 , F_2 \colon \Lambda_1 \to \Lambda_2$ be injective large-scale Lipschitz map over $\fR$. 
    Assume that $F_1$ and $F_2$ are near, i.e., $\sup_{\bm{x} \in \Lambda_1}\rmd(F_1(\bm{x}) , F_2(\bm{x})) < \infty$. 
    Then the induced morphisms $F_{1,*}, F_{2,*} \colon \sGP(\Lambda_1  \midbar \blank) \to \sGP(\Lambda_2 \midbar \blank) \to \sGP_{d}$ are homotopic. 
\end{lem}

For a linearly proper large-scale Lipschitz map $F \colon \bR^d \to \bR^{d'}$, the map
\begin{align*}
    \widetilde{F} \coloneqq (F \times \id_{\bR^d}) \times \id_{\bR^\infty} \times \id_{\fR} \colon \bR^d \times \bR^{\infty} \times \fR \to (\bR^{d'} \times \bR^d) \times \bR^\infty \times \fR 
\end{align*}
is an example of an injective linearly proper large-scale Lipschitz map over $\fR$ that satisfies the following two conditions: (i) $\pr_{\bR^{d'}} \circ \tilde{F} = F \circ \pr_{\bR^d}$, and (ii) $\widetilde{F}$ sends $\Lambda \in \fL_{\fR,d}$ to $\widetilde{F}(\Lambda) \in \fL_{\fR,d'}$. 
By using this $\widetilde{F}$, we define the morphism of sheaves
\begin{align}
    F_* \coloneqq \colim_{\Lambda \in \fL_{\fR,d}} \Big( (\widetilde{F}|_{\Lambda})_* \colon \sGP(\Lambda \midbar \sM) \to \sGP(\widetilde{F}(\Lambda) \midbar \sM) \Big) \colon \sGP_{\fR,d}(\sM) \to  \sGP_{\fR,d'}(\sM).\label{eqn:push.quantum.system}
\end{align}

It will be proved in \cref{exmp:Lipschitz.homotopy} that, for $F \colon \bR^d \to \bR^d$, two choices of its lifts $\widetilde{F}_1$, $\widetilde{F}_2$ satisfying (i), (ii) induce the homotopic morphisms $\widetilde{F}_{1,*} \simeq \widetilde{F}_{2,*}$ on the sheaves of gapped UAL Hamiltonians. 
In particular, the assignment $F \mapsto F_*$ is functorial up to homotopy; there is a canonical homotopy of $F_{2,*} \circ F_{1,*}$ and $(F_2 \circ F_1)_*$ for any $F_1 \colon \bR^{d_1} \to \bR^{d_2}$ and $F_2 \colon \bR^{d_2} \to \bR^{d_3}$. This functoriality will be generalized to sheaves of IG UAL Hamiltonians on a more general class of metric spaces in \cref{subsection:localizing.path}.

\subsection{The sheaf of fermionic and on-site symmetric invertible phases}
The spaces of fermionic or symmetry-protected topological phases are defined similarly. 
Let $G$ be a compact Lie group, not necessarily connected. 
Let $\phi \colon G \to \bZ/2$ be a group homomorphism. 
We say that $(\sH,\pi)$ is a $\phi$-linear unitary representation of $G$ if $\pi$ is a continuous group homomorphism from $G$ to the group of linear or antilinear unitaries on $\sH$ such that $\pi_g$ is linear if $\phi(g)=0$ and antilinear if $\phi(g) =1$. 
\begin{defn}\label{notn.internal.degree.spt}
We introduce the following variations of \cref{notn:internal.degree}.
    \begin{itemize}
        \item We say that $\fR$ is a \emph{fermionic internal degrees of freedom} if it is a countable set of pairs $\lambda =(\widehat{\sH}_{\lambda},\Omega_{\lambda})$, where $\widehat{\sH}_{\lambda}$ is a $\bZ/2$-graded Hilbert space with $2 \leq \dim \widehat{\sH}_{\lambda} <\infty$ and $\Omega_{\lambda} \in \sH_{\lambda}$ is a fixed even unit vector.
        \item We say that $\fR$ is a \emph{bosonic (resp.\ fermionic) internal degrees of freedom with on-site $(G,\phi)$-symmetry} if it is set of triples $\lambda =(\sH_{\lambda},\pi_\lambda,\Omega_{\lambda})$, where $(\sH_{\lambda}, \pi_{\lambda})$ is a representation (resp.\ $\bZ/2$-graded representation) of $G$ that is $\phi$-linear and $2 \leq \dim \sH_{\lambda} <\infty$, and $\Omega_{\lambda} \in \sH_{\lambda}$ is a fixed $G$-invariant unit vector (resp.\ even unit vector).
    \end{itemize}
    For $\lambda \in \fR$ in each case, let $\cA_{\lambda} \coloneqq \cB(\sH_{\lambda})$ or $\cB(\widehat{\sH}_{\lambda})$ and $\sfh_{\lambda} \coloneqq 2(1-\Omega_{\lambda} \otimes \Omega_{\lambda}^*)$.  
\end{defn}

Using these internal degrees of freedom, we define the set of lattices $\fL_{\fR, d}$ in the same way as in \cref{defn:lattice.set}. 
In the case of the bosonic $(G,\phi)$-symmetric lattice, we place on each $\bm{x} \in \Lambda$ the local observable algebra $\cA_{\bm{x}} \coloneqq \cA_{\lambda(\bm{x})}$ and the trivial Hamiltonian $\sfh_{\bm{x}} \coloneqq \sfh_{\lambda(\bm{x})}$. 
For $\Lambda \in \fL_{\fR, d}$, the local observable algebra $\cA_{\Lambda}$ is defined in the same way as \cref{defn:local.observable}. 
In the case of fermionic lattice, possibly with on-site $(G,\phi)$-symmetry, what we place on each $\bm{x}$ is the same, but the local observable algebra is replaced with the graded tensor product as
\begin{align*}
\widehat{\cA}_{\Lambda} \coloneqq \mathop{\widehat{\bigotimes}} _{\bm{x} \in \Lambda} \cA_{\bm{x}}, \quad  \sfh = (\sfh_{\bm{x}} )_{\bm{x} \in \Lambda} \in \fH_{\Lambda}^{\al}, \quad \omega_0\coloneqq \bigotimes_{\bm{x}} \langle \blank \Omega_{\bm{x}},\Omega_{\bm{x}} \rangle. 
\end{align*}
For these cases, the sheaves of smooth families of gapped UAL Hamiltonians are defined in the same way as \cref{defn:gapped.sheaf}. 
Moreover, the composite product of lattices and Hamiltonians are defined in the same way as  \eqref{eqn:composition.lattice} and \eqref{eqn:composition.system} coming from the map $\fm$ fixed in \cref{line.mu}. 
This product makes $\sGP_{\fR,d}$ a local commutative H-monoid, which leads us to the definition of the sheaves of IG UAL Hamiltonians. 
We only remark that the proof of \cref{lem:H-monoid} works even in this setting since $\bR[\fS_4] \subset \cB(\sH_{\lambda}^{\otimes 4})$ commutes with the $\phi$-twisted representation of $G$, and is an even subalgebra in the fermionic case.

In the bosonic case, we use the same letters $\sGP(\Lambda \midbar \sM)$, $\sGP_{\fR,d}(\sM)$, $\sIP_{\fR,d}(\sM)$ for the corresponding sheaves if $\fR$ has an on-site $(G,\phi)$-symmetry. 
In the fermionic case, we introduce the following notation. 
\begin{defn}
    Let $\fR$ be a set of fermionic internal degrees of freedom, possibly with on-site $(G,\phi)$-symmetry. 
    For $\Lambda \in \fL_{\fR,d}$, let 
    \begin{align*}
    \sfGP(\Lambda \midbar \sM ) \coloneqq \{ \sfH \colon \sM \to \ffH_{\Lambda}^{\al} \mid \text{smooth in the sense of \cref{defn:smooth}} \}. 
    \end{align*} 
    The sheaves $\sfGP_{\fR,d}(\sM)$ and $\sfIP_{\fR,d}(\sM)$ are defined in the same way as \cref{defn:gapped.sheaf} and \cref{defn:gapped.invertible}, coming from the product on $\sfGP_{\fR,d}(\sM)$ via the map $\fm$. 
\end{defn}

For $\Lambda \in \fL_{\fR,d}$, the action of $G$ on the observable algebra is denoted by
\begin{align*} \beta_g \coloneqq \bigotimes_{\bm{x} \in \Lambda } \Ad (\pi_g) \in \overline{\Aut}(\widehat{\cA}_{\Lambda}^{\al}), \end{align*}
where $\overline{\Aut}(\widehat{\cA}_{\Lambda}^{\al})$ denotes the set of linear or antilinear $\ast$-automorphisms.
\begin{lem}
    Let $\fR$ be a set of internal degrees of freedom with on-site $(G,\phi)$-symmetry, either bosonic or fermionic. 
    With the $G$-action $\beta (\sfH)(p,g ) \coloneqq \beta_g(\sfH(p))$, the sheaves $\sGP_{\fR}(\Lambda \midbar \blank)$, $\sGP_{\fR, d}$, $\sIP_{\fR, d}$, and $\bsIP_{\fR,d}$ (resp.~ $\sfGP_{\fR}(\Lambda \midbar \blank)$, $\sfGP_{\fR, d}$, $\sfIP_{\fR, d}$, and $\bsfIP_{\fR,d}$ in the fermionic case) are $G$-sheaves in the sense of \cref{defn:Gsheaf}.
\end{lem}
\begin{proof}
    It suffices to show that $\beta (\sfH)$ is smooth in the sense of \cref{defn:smooth}. It follows from \cref{lem:smooth} (3).  
\end{proof}

\begin{defn}
    For bosonic and fermionic internal degrees of freedom $\fR$ with on-site $(G,\phi)$-symmetry, use the same letter $\IP_{\fR,d}$ and $\fIP_{\fR,d}$ for the $G$-spaces
        \begin{align*} 
        \IP_{\fR,d} \coloneqq |\Sing_G \bsIP_{\fR,d} |_G, \quad \fIP_{\fR,d} \coloneqq |\Sing_G \bsfIP_{\fR,d} |_G. 
        \end{align*}
\end{defn}

\begin{rmk}
    When taking the group action into account, a $G$-invariant section of an $\sIP_{\fR,d}$ is not necessarily $G$-invariantly invertible, i.e., the homotopy inverse $\check{\sfH}$ and the null-homotopy $\overline{\sfH}$ are taken as $G$-invariant smooth maps. 
    In other words, the $G$-equivariant version of \cref{lem:bold.invertible.sheaf} may fail: the fixed point sheaves $\sIP_{\fR,d}^G$ and $\bsIP_{\fR,d}^G$ may not be weakly equivalent.  
    In the definition of the $G$-realization $\IP_{\fR,d}$, it is more relevant to use the $G$-invariant sections of $\bsIP_{\fR,d}$.
\end{rmk}

\begin{rmk}
In this paper, for simplicity, we consider only the case where a fermionic on-site symmetry consists of a group $G$ that commutes with the fermionic parity $\mathsf{P}$. 
However, in papers such as \cites{wangConstructionClassificationSymmetryprotected2020,aasenCharacterizationClassificationFermionic2022,barkeshliClassification2+1DInvertible2022}, a fermionic symmetry is formulated in terms of a group $G_f$ that incorporates these two types of symmetries, which is a central extension by the group $\bZ/2$ of fermionic parity. 
Our framework can be immediately generalized to handle such symmetries.
\end{rmk}

\section{Kitaev's conjecture}\label{section:Kitaev}
In this section, the first central part of this paper, we mathematically realize Kitaev's conjecture that the sequence of spaces $\{\IP_d  \}_d$ forms an $\Omega$-spectrum.
We begin by defining a continuous map $\kappa_d \colon \IP_d \to \Omega \IP_{d+1}$, called the Kitaev pump.
We then show that $\IP_d $ and $\Omega \IP_{d+1}$ are weakly equivalent by using the idea of adiabatic interpolation from \cref{thm:interpolation.loop}, and then show that our weak equivalence is attained by $\kappa_d$. 
Our argument for the second step differs from that of \cite{kitaevClassificationShortrangeEntangled2013}. 
A key ingredient is a version of the bulk boundary correspondence, \cref{prp:Eilenberg.swindle}, proved by using an Eilenberg swindle argument. 
Our relaxation of the concept of a lattice in $\bR^d$ to allow a non-uniform finite distance in the direction of the internal degrees of freedom (\cref{defn:lattice.set}) is precisely for the purpose of applying this argument.

\subsection{Kitaev's pump}\label{subsection:Kitaev}
This section deals with the path and loop sheaves, defined in \cref{subsubsection:homotopy.sheaf}, of the sheaves $\sIP_d$ and $\bsIP_d$ defined in \cref{eqn:IP.bold}.
For $\sfH \in \cP \sIP_d(\sM)$, we write $\ev_t \sfH \in \sIP_d(\sM)$ for its restriction to $\sM \times \{t\} $. 

\begin{defn}\label{defn:Kitaev.pump}
The Kitaev pump is a morphism of sheaves
    \begin{align*} 
        \kappa_d \colon \boldsymbol{\sIP}_d \to \Omega \sIP_{d+1} 
    \end{align*}
    that sends $\bsfH \coloneqq (\sfH , \check{\sfH}, \overline{\sfH}) \in \bsIP_d(\sM)$ to the loop $\kappa_d (\bsfH)  \in \Omega \sIP_{d+1}(\sM)$ defined by
    \begin{align*}
        \kappa_d (\bsfH)|_t = \left\{
        \begin{array}{rcccccccccccccll}
            \cdots &\boxtimes & \sfh & \boxtimes & \sfh & \boxtimes & \sfh &\boxtimes & \sfh & \boxtimes & \sfh& \boxtimes & \cdots & t= 0, \\
            \cdots &\boxtimes & \multicolumn{3}{c}{\overline{\sfH}|_{1-2t}}  & \boxtimes &\multicolumn{3}{c}{\overline{\sfH}|_{1-2t}} & \boxtimes & \multicolumn{2}{r}{\overline{\sfH}|_{1-2t}} &  \cdots & 0<t<1/2, \\
            \cdots &\boxtimes & \sfH  & \boxtimes & \check{\sfH} & \boxtimes & \sfH & \boxtimes & \check{\sfH} & \boxtimes & \sfH  & \boxtimes & \cdots & t=1/2, \\            
            \cdots&\multicolumn{2}{c}{ \flip \overline{\sfH}|_{2t-1}} & \boxtimes & \multicolumn{3}{c}{\flip \overline{\sfH}|_{2t-1}} & \boxtimes &\multicolumn{3}{c}{\flip \overline{\sfH}|_{2t-1}} & \boxtimes &  \cdots &  1/2<t<1, \\            
            \cdots &\boxtimes & \sfh & \boxtimes & \sfh & \boxtimes & \sfh &\boxtimes & \sfh & \boxtimes & \sfh& \boxtimes &  \cdots & t = 1.
        \end{array}
        \right.
    \end{align*}
    Here, $\sfH|_t $ is an abbreviation for $\ev_t \sfH= \sfH(\blank, t)$. 
    We use the same letter $\kappa_d \colon \IP_d \to \Omega \IP_{d+1}$ for the associated map of their realizations.
\end{defn}

\begin{figure}[t]
    \centering
        \begin{tikzpicture}
        \draw[->] (-1.3,1.0) -- (-1.3,-2.0);
        \node at (-2,0.85) {$t =0$};
        \node at (-2,-0.35) {$t=1/2$};
        \node at (-2,-1.55) {$t=1$};
        \node[draw] at (-0.75,0.85) {$\sfh$};
        \node[draw] at (0.25,0.85) {$\sfh$};
        \node[draw] at (1.25,0.85) {$\sfh$};
        \node[draw] at (2.25,0.85) {$\sfh$};
        \node[draw] at (3.25,0.85) {$\sfh$};
        \node[draw] at (4.25,0.85) {$\sfh$};
        \node[draw] at (5.25,0.85) {$\sfh$};
        \node[draw] at (6.25,0.85) {$\sfh$};
        \node[draw] at (7.25,0.85) {$\sfh$};
        \node[draw] at (8.25,0.85) {$\sfh$};
        \draw (-1.2,0) -- (-0.5,0.0) -- (-0.5,0.5) -- (-1.2,0.5);
        \draw (0,0) rectangle (1.5,0.5);
        \node at (0.75,0.25) {$\overline{\sfH}|_{1-2t}$};
        \draw (2,0) rectangle (3.5,0.5);
        \node at (2.75,0.25) {$\overline{\sfH}|_{1-2t}$};
        \draw (4,0) rectangle (5.5,0.5);
        \node at (4.75,0.25) {$\overline{\sfH}|_{1-2t}$};
        \draw (6,0) rectangle (7.5,0.5);
        \node at (6.75,0.25) {$\overline{\sfH}|_{1-2t}$};
        \draw (8.7,0) -- (8,0.0) -- (8,0.5) -- (8.7,0.5);
        \node[draw] at (-0.75,-0.35) {$\check{\sfH}$};
        \node[draw] at (0.25,-0.35) {$\sfH$};
        \node[draw] at (1.25,-0.35) {$\check{\sfH}$};
        \node[draw] at (2.25,-0.35) {$\sfH$};
        \node[draw] at (3.25,-0.35) {$\check{\sfH}$};
        \node[draw] at (4.25,-0.35) {$\sfH$};
        \node[draw] at (5.25,-0.35) {$\check{\sfH}$};
        \node[draw] at (6.25,-0.35) {$\sfH$};
        \node[draw] at (7.25,-0.35) {$\check{\sfH}$};
        \node[draw] at (8.25,-0.35) {$\sfH$};
        \draw (-1,-0.7) rectangle (0.5,-1.2);
        \node at (-0.25,-0.95) {$\flip \overline{\sfH}_{2t-1}$};
        \draw (1,-0.7) rectangle (2.5,-1.2);
        \node at (1.75,-0.95) {$\flip \overline{\sfH}_{2t-1}$};
        \draw (3,-0.7) rectangle (4.5,-1.2);
        \node at (3.75,-0.95) {$\flip \overline{\sfH}_{2t-1}$};
        \draw (5,-0.7) rectangle (6.5,-1.2);
        \node at (5.75,-0.95) {$\flip \overline{\sfH}_{2t-1}$};
        \draw (7,-0.7) rectangle (8.5,-1.2);
        \node at (7.75,-0.95) {$\flip \overline{\sfH}_{2t-1}$};
        \node[draw] at (-0.75,-1.55) {$\sfh$};
        \node[draw] at (0.25,-1.55) {$\sfh$};
        \node[draw] at (1.25,-1.55) {$\sfh$};
        \node[draw] at (2.25,-1.55) {$\sfh$};
        \node[draw] at (3.25,-1.55) {$\sfh$};
        \node[draw] at (4.25,-1.55) {$\sfh$};
        \node[draw] at (5.25,-1.55) {$\sfh$};
        \node[draw] at (6.25,-1.55) {$\sfh$};
        \node[draw] at (7.25,-1.55) {$\sfh$};
        \node[draw] at (8.25,-1.55) {$\sfh$};
        \end{tikzpicture}
    \caption{The picture of $\kappa_d(\bsfH)$ (here $\sfH|_t$ is an abbreviation for $\ev_t\sfH$). }
    \label{fig:picture.kappa.bold}
\end{figure}
The precise definition of the infinite composite systems appearing in the above definition is as follows: 
If $(\sfH , \check{\sfH}, \overline{\sfH})$ is supported on $\Lambda \in \fL_d$ in the sense of \cref{eqn:IP.bold}, then the IG UAL Hamiltonian $\cdots \boxtimes \sfH \boxtimes \check{\sfH} \boxtimes \sfH \boxtimes \cdots $, which is also denoted by $(\sfH \boxtimes \check{\sfH})^{\boxtimes \bZ}$ hereafter, is supported on $\Lambda \times \bZ \in \fL_{d+1}$. 
Let $\iota_n \colon \Lambda \to \Lambda \times \{n\} \subset \Lambda \times \bZ$ denotes the inclusion. 
The Hamiltonian in question is given by placing an almost local operator $\iota_{2n,*}(\sfH_{\bm{x}}) \in \cA_{\Lambda \times \{2n\} }$ on the lattice point $(\bm{x},2n) \in \Lambda \times \bZ$, and $\iota_{2n+1,*}(\check{\sfH}_{\bm{x}} ) \in \cA_{\Lambda \times \{ 2n+1\} }$ on $(\bm{x},2n+1) \in \Lambda \times \bZ$. 
The resulting time evolution and the distinguished ground state are the tensor product of copies of $\tau_{\sfH}$ and $\omega_{\sfH}$. 
Therefore, it is again gapped. 
Moreover, it is invertible since $\Omega \sGP_d = \Omega \sIP_d$ by definition.

This $\kappa_d$ gives a well-defined morphism of sheaves on $\Man$ since the assignment $\bsfH \mapsto \kappa_d(\bsfH)$ clearly commutes with the inclusion of lattices $\Lambda_1 \leq \Lambda_2$.
By \cref{lem:forget.invertible}, a continuous map of realizations
\[ 
    \kappa_d \colon \IP_d = |\Sing \bsIP_d| \to |\Sing \Omega \sIP_{d+1}| \simeq |\Sing \Omega \bsIP_{d+1}| \simeq \Omega \IP_{d+1}
\]
is induced. More strongly, such $\kappa_d$ is explicitly realized by a morphism of sheaves in the following way. 
The morphism $\boldsymbol{\kappa}_d$ makes the sequence $\{ \bsIP_d, \boldsymbol{\kappa}_d \}$ a spectrum object in the category $\Sh(\Man)$. 
\begin{prp}\label{thm:inverse.pump}
There is a morphism of sheaves $\boldsymbol{\kappa }_d \colon \bsIP_d \to \Omega \bsIP_{d+1}$ such that $\mathsf{fg} \circ \boldsymbol{\kappa}_d = \kappa_d$. 
\end{prp}
\begin{proof}
Let $\check{\kappa}_d(\bsfH)$ be the inverse path, and let $\bar{\kappa}_d (\bsfH )$ be the composition of smooth homotopies
\begin{align}
\begin{split}
    \kappa_d(\bsfH) \boxtimes \check{\kappa}_d(\bsfH) \sim {}&{} 
    (\sfh \boxtimes \check{\kappa}_d(\bsfH) ) \circ 
    (\kappa_d(\bsfH) \boxtimes \sfh )\\
    \sim {}&{} 
    (\check{\kappa}_d(\bsfH) \boxtimes \sfh) \circ (\kappa_d(\bsfH) \boxtimes \sfh)\\
    \sim {}&{} \sfh,
\end{split}\label{eqn:loop.reverse.homotopy}
\end{align}
where the first homotopy is given by the reparametrization of $\kappa_d(\bsfH)$ and $\check{\kappa}_d(\bsfH)$, the second homotopy is given by \cref{lem:flip.homotopy}, and the last homotopy is given by the reparametrization 
\begin{align*}
    (p,t,s) \mapsto 
    \begin{cases}
        \kappa_d(\bsfH)(p,st) \boxtimes \sfh & \text{ if $t \leq 1/2$, }\\
        \kappa_d(\bsfH)(p,s(1-t)) \boxtimes \sfh & \text{ if $t \geq 1/2$. }
    \end{cases}
\end{align*}
Indeed, this is a special case of the general construction of the weak inverse $\Omega\sF^\times \to \Omega \boldsymbol{\sF}^\times$ of a local commutative H-monoid sheaf $\sF$. 
\end{proof}

\begin{rmk}\label{rmk:smoothing.path}
    In the above proof, and throughout this paper, we treat a section on $\sM \times [0,1]$ as a smooth path even if it is not constant at both endpoints $t \in \partial [0,1]$, although it is not a section of the path sheaf $\cP\sF$ in the sense of \eqref{eqn:path.sheaf}. 
    A typical example is the convex combination $t\sfH_1 + (1-t)\sfH_2$.   
    Although we do not explicitly mention this each time, one should assume that these paths are implicitly reparametrized by a bump function, that is, a monotonically increasing smooth function that takes the value $0$ on $(-\infty,\varepsilon]$ and $1$ on $[1-\varepsilon,\infty)$ for some $\varepsilon >0$.
\end{rmk}

This Kitaev pump has a version truncated to the half subspace, which plays a fundamental role in the theory.
\begin{lem}\label{lem:truncated.pump}
There is a morphism $\boldsymbol{\kappa}_d^R \colon \bsIP_d \to \cP_{\sfh} \bsIP_{d+1}$ such that the following hold:  
\begin{enumerate}
    \item For any $\bsfH \in \bsIP_d(\sM)$ supported on $\Lambda$, $ \boldsymbol{\kappa}_d^R (\bsfH)$ is supported on $\Lambda \times \bZ_{\geq 0}$. 
    \item For any $\bsfH \in \bsIP_d(\sM)$, we have $ \kappa_d^R (\bsfH)_{(\bm{x} ,n)} = \kappa_d (\bsfH)_{(\bm{x} ,n)}$, $ \check{\kappa}_d^R (\bsfH)_{(\bm{x} ,n)} = \check{\kappa}_d (\bsfH)_{(\bm{x} ,n)}$, and $ \bar{\kappa}_d^R (\bsfH)_{(\bm{x} ,n)} = \bar{\kappa}_d (\bsfH)_{(\bm{x} ,n)}$ for any $n \geq 2$. 
    \item For any $\bsfH \in \bsIP_d(\sM)$ supported on $\Lambda$, we have $\ev_1\boldsymbol{\kappa}_d^R (\bsfH) = \bsfH \boxtimes \sfh^{\boxtimes \bZ_{\geq 1}} \in \bsIP(\Lambda \times \bZ_{\geq 0}\midbar \sM)$.
\end{enumerate}
\end{lem}
\begin{proof}
    The morphism $\kappa_d^R \colon \bsIP_d \to \cP_{\sfh} \sIP_{d+1}$ is defined by the truncated version of the Kitaev's pump in \cref{defn:Kitaev.pump} as
\begin{align*}
    \ev_t\kappa_d^R (\bsfH) = \left\{
        \begin{array}{ccccccccccccll}
             \sfh & \boxtimes & \sfh & \boxtimes & \sfh &\boxtimes & \sfh & \boxtimes & \sfh& \boxtimes & \cdots & t\leq 0, \\
             \multicolumn{3}{c}{\overline{\sfH}|_{1-2t}}  & \boxtimes &\multicolumn{3}{c}{\overline{\sfH}|_{1-2t}} & \boxtimes & \multicolumn{2}{r}{\overline{\sfH}|_{1-2t}} &  \cdots & 0<t<1/2, \\
             \sfH  & \boxtimes & \check{\sfH} & \boxtimes & \sfH & \boxtimes & \check{\sfH} & \boxtimes & \sfH  & \boxtimes & \cdots & t=1/2, \\            
            \sfH&  \boxtimes & \multicolumn{3}{c}{\flip \overline{\sfH}|_{2t-1}} & \boxtimes &\multicolumn{3}{c}{ \flip \overline{\sfH}|_{2t-1}} & \boxtimes &  \cdots &  1/2<t<1, \\            
             \sfH & \boxtimes & \sfh & \boxtimes & \sfh &\boxtimes & \sfh & \boxtimes & \sfh& \boxtimes &  \cdots & t \geq 1.
        \end{array}
        \right.
\end{align*}

Notice that the reversed path $\check{\kappa}_d(\sfH,\check{\sfH},\overline{\sfH})$ is identical to $ T_* \kappa_d(\check{\sfH},\sfH, \flip \overline{\sfH})$, where $T \colon \Lambda \times \bZ \to \Lambda \times \bZ$ is the shift in the $\bZ$-direction by $1$. Therefore, 
\[
    \check{\kappa}_d^R (\sfH,\check{\sfH},\overline{\sfH}) \coloneqq T_* \kappa_d^R(\check{\sfH}, \sfH, \flip \overline{\sfH}) = \sfh \boxtimes \kappa_d^R(\check{\sfH}, \sfH, \flip \overline{\sfH})
\]
satisfies (2). Along the construction in \eqref{eqn:loop.reverse.homotopy}, we obtain the smooth homotopy $\bar{\kappa}_d^R(\bsfH)$ satisfying the condition (2) by the composition
\begin{align*}
    \kappa_d^R(\bsfH) \boxtimes \check{\kappa}_d^R(\bsfH) \simeq  {}&{} (  (\ev_1 \kappa_d^R(\bsfH)) \boxtimes \check{\kappa}_d^R(\bsfH)) \circ (\kappa_d^R(\bsfH) \boxtimes \sfh^{\boxtimes \bN}) \\
    \simeq    {}&{} \big( \big( 
    \sfH \boxtimes \kappa_d^R(\check{\sfH},\sfH,\overline{\sfH})\big) \boxtimes \sfh^{\boxtimes \bN} \big) \circ (\kappa_d^R(\bsfH) \boxtimes \sfh^{\boxtimes \bN}) \\
    \simeq {}&{} \overline{\sfH}{}^{\mathrm{rev}} \simeq \sfh.
\end{align*}
Here, we distinguish the trivial Hamiltonians on $\Lambda \times \bN$ and $\Lambda$ by denoting them as $\sfh^{\boxtimes \bN}$ and $\sfh$, respectively. 
The third smooth homotopy is given by canceling $\overline{\sfH}$ and its reversion, as is seen from \cref{fig:picture.kappa.bold.tr}. The last homotopy is given by 
\[
    (p,s,t) \mapsto \overline{\sfH}(p,\min \{ 1+s-t ,1 \}),
\]
and hence its evaluation at $t=1$ is nothing but $\overline{\sfH}$. This shows (3). 
\end{proof}
\begin{figure}[t]
    \centering
    \begin{minipage}[b]{0.45\hsize}
        \begin{tikzpicture}
        \draw[->] (1.6,1.0) -- (1.6,-2.0);
        \node at (1.3,-1.8) {$t$};
        \node[draw] at (2.25,0.85) {$\sfh$};
        \node[draw] at (3.25,0.85) {$\sfh$};
        \node[draw] at (4.25,0.85) {$\sfh$};
        \node[draw] at (5.25,0.85) {$\sfh$};
        \node[draw] at (6.25,0.85) {$\sfh$};
        \node[draw] at (7.25,0.85) {$\sfh$};
        \node[draw] at (8.25,0.85) {$\sfh$};
        \draw (2,0) rectangle (3.5,0.5);
        \node at (2.75,0.25) {$\overline{\sfH}|_{1-2t}$};
        \draw (4,0) rectangle (5.5,0.5);
        \node at (4.75,0.25) {$\overline{\sfH}|_{1-2t}$};
        \draw (6,0) rectangle (7.5,0.5);
        \node at (6.75,0.25) {$\overline{\sfH}|_{1-2t}$};
        \draw (8.7,0) -- (8,0.0) -- (8,0.5) -- (8.7,0.5);
        \node[draw] at (2.25,-0.35) {$\sfH$};
        \node[draw] at (3.25,-0.35) {$\check{\sfH}$};
        \node[draw] at (4.25,-0.35) {$\sfH$};
        \node[draw] at (5.25,-0.35) {$\check{\sfH}$};
        \node[draw] at (6.25,-0.35) {$\sfH$};
        \node[draw] at (7.25,-0.35) {$\check{\sfH}$};
        \node[draw] at (8.25,-0.35) {$\sfH$};
        \node[draw] at (2.25,-0.95) {$\sfH$};
        \draw (3,-0.7) rectangle (4.5,-1.2);
        \node at (3.75,-0.95) {$\flip \overline{\sfH}_{2t-1}$};
        \draw (5,-0.7) rectangle (6.5,-1.2);
        \node at (5.75,-0.95) {$\flip \overline{\sfH}_{2t-1}$};
        \draw (7,-0.7) rectangle (8.5,-1.2);
        \node at (7.75,-0.95) {$\flip \overline{\sfH}_{2t-1}$};
        \node[draw] at (2.25,-1.55) {$\sfH$};
        \node[draw] at (3.25,-1.55) {$\sfh$};
        \node[draw] at (4.25,-1.55) {$\sfh$};
        \node[draw] at (5.25,-1.55) {$\sfh$};
        \node[draw] at (6.25,-1.55) {$\sfh$};
        \node[draw] at (7.25,-1.55) {$\sfh$};
        \node[draw] at (8.25,-1.55) {$\sfh$};
        \end{tikzpicture}
    \end{minipage}
    \qquad 
    \begin{minipage}[b]{0.45\hsize}
        \begin{tikzpicture}
        \draw[->] (1.6,0.6) -- (1.6,-2.0);
        \node at (1.3,-1.8) {$t$};
        \draw (2,0) rectangle (3.5,0.5);
        \node at (2.75,0.25) {$\overline{\sfH}|_{1-4t}$};
        \draw (4,0) rectangle (5.5,0.5);
        \node at (4.75,0.25) {$\overline{\sfH}|_{1-4t}$};
        \draw (6,0) rectangle (7.5,0.5);
        \node at (6.75,0.25) {$\overline{\sfH}|_{1-4t}$};
        \draw (8.7,0) -- (8,0.0) -- (8,0.5) -- (8.7,0.5);
        \node[draw] at (2.25,-0.35) {$\sfH$};
        \draw (3,-0.1) rectangle (4.5,-0.6);
        \node at (3.75,-0.35) {$\flip \overline{\sfH}_{4t-1}$};
        \draw (5,-0.1) rectangle (6.5,-0.6);
        \node at (5.75,-0.35) {$\flip \overline{\sfH}_{4t-1}$};
        \draw (7,-0.1) rectangle (8.5,-0.6);
        \node at (7.75,-0.35) {$\flip \overline{\sfH}_{4t-1}$};
        \draw (3,-0.7) rectangle (4.5,-1.2);
        \node[draw] at (2.25,-0.95) {$\sfH$};
        \node at (3.75,-0.95) {$\flip \overline{\sfH}_{3-4t}$};
        \draw (5,-0.7) rectangle (6.5,-1.2);
        \node at (5.75,-0.95) {$\flip \overline{\sfH}_{3-4t}$};
        \draw (7,-0.7) rectangle (8.5,-1.2);
        \node at (7.75,-0.95) {$\flip \overline{\sfH}_{3-4t}$};
        \node[draw] at (2.25,-1.55) {$\sfH$};
        \node[draw] at (3.25,-1.55) {$\check{\sfH}$};
        \draw (4,-1.3) rectangle (5.5,-1.8);
        \node at (4.75,-1.55) {$\overline{\sfH}|_{4t-3}$};
        \draw (6,-1.3) rectangle (7.5,-1.8);
        \node at (6.75,-1.55) {$\overline{\sfH}|_{4t-3}$};
        \draw (8.7,-1.3) -- (8,-1.3) -- (8,-1.8) -- (8.7,-1.8);
        \end{tikzpicture}
    \end{minipage}
    \caption{The picture of $\kappa_d^R(\bsfH)$ and $\kappa_d^R(\bsfH) \circ  \check{\kappa}_d^R(\bsfH)$. }
    \label{fig:picture.kappa.bold.tr}
\end{figure}
\subsection{Large-scale Lipschitz homotopy invariance}
As an application of the construction of the truncated Kitaev pump, we prove the invariance under coarse geometric homotopy, in the sense of \cite{mitchenerCoarseHomotopyGroups2020}, of the homotopy groups of the invertible phases. 

For a large-scale Lipschitz map $\varphi \colon \Lambda \to \bR_{\geq 1}$, set 
    \begin{align}
    \begin{split}
    \mathbf{I}_{\varphi }\Lambda&{} \coloneqq \{ (\bm{x},n) \in \Lambda \times \bZ_{\geq 0} \mid 0 \leq n  \leq \varphi (\bm{x}) \}, \\
    \partial_{0}\mathbf{I}_{\varphi }\Lambda &{}\coloneqq \{ (\bm{x},n) \in \Lambda \times \bZ_{\geq 0} \mid n=\varphi(\bm{x})  \}, \\
    \partial_1 \mathbf{I}_{\varphi } \Lambda&{}\coloneqq \Lambda \times \{0\}.
    \end{split}\label{eqn:coarse.homotopy}
    \end{align}
    Both $\partial_0 \mathbf{I}_\varphi\Lambda$ and $\partial_1 \mathbf{I}_\varphi\Lambda$ are large-scale bi-Lipschitz equivalent to $\Lambda$ by the projection $\pr_\Lambda$. For linearly proper large-scale Lipschitz maps $F_0, F_1 \colon \Lambda \to X$ to a metric space $X$, a \emph{large-scale Lipschitz homotopy} connecting them is a linearly proper large-scale Lipschitz map $\widetilde{F} \colon \mathbf{I}_{\varphi}\Lambda \to X$ such that $\widetilde{F} \circ \iota_i = F_i$, where $\iota_i \colon \Lambda \to \partial_i \mathbf{I}_{\varphi}\Lambda$ is the canonical inclusions, for $i =0,1$.

\begin{prp}\label{prp:coarse.homotopy}
    Let $\Lambda \in \fL_{d}$ and let $\varphi \colon \Lambda \to \bR_{\geq 1}$ be a linearly proper large-scale Lipschitz map (e.g.\ $\varphi(\bm{x}) = 1 +\rmd(\bm{x},\bm{x}_0)$). 
    Then, for any $\bsfH=(\sfH, \check{\sfH},\overline{\sfH}) \in \bsIP_d(\sM)$ supported on $\Lambda$, there is $r>0$ and a smooth homotopy of IG UAL Hamiltonians supported on $I_{2 \varphi + r}\Lambda$ connecting $\bsfH$ and another smooth family $\bsfH' \in \bsIP_d(\sM)$ that is supported on $\mathbf{I}_{2\varphi +r}\Lambda \setminus \mathbf{I}_{\varphi}\Lambda$. 
\end{prp}
\begin{proof}
    Apply \cref{thm:interpolation.loop,lem:cut.diffused} to the truncated Kitaev pump $\kappa_d^R(\sfH, \check{\sfH},\overline{\sfH})$, the subspaces $Y \coloneqq (\Lambda \times \bZ_{\geq 0}) \setminus \mathbf{I}_{2\varphi +r} \Lambda$ and $Z = \mathbf{I}_{\varphi}\Lambda \setminus \mathbf{I}_{\varphi}\Lambda$. 
    Then, the resulting smooth family 
    \[
    \Theta_{\mathbf{I}_{2\varphi +r}\Lambda \setminus \mathbf{I}_{\varphi}\Lambda , \omega_0} \circ \vartheta (\boldsymbol{\kappa}_d^R(\bsfH)) \in \bsIP_d(\sM)
    \]
    is supported on $\mathbf{I}_{2\varphi +r} \Lambda \setminus \mathbf{I}_{\varphi}\Lambda$, and is smoothly homotopic to $\sfH$ as sections of $\bsIP_d$ supported on $\mathbf{I}_{2\varphi+r}\Lambda$ by the concatenation of the convex combination homotopy $\Theta_{\mathbf{I}_{2\varphi +r}\Lambda \setminus \mathbf{I}_{\varphi}\Lambda , \omega_0} \circ \vartheta (\boldsymbol{\kappa}_d^R(\bsfH))  \simeq \vartheta (\boldsymbol{\kappa}_d^R(\bsfH )) $ and 
    $\tilde{\vartheta}\kappa_d^R(\bsfH)$ in \cref{thm:interpolation.loop}.    
\end{proof}

\begin{cor}\label{cor:coarse.homotopy.Hamiltonian}
    Let $F_0,F_1 \colon \Lambda \to \bR^{d+l}$ be weakly uniformly discrete maps such that $ \pr_{\bR^d} \circ F_i$ is linearly proper. 
    Assume that $F_0$ and $F_1$ are large-scale Lipschitz homotopic. 
    Then, for any $\sfH \in \sIP_d(\Lambda \midbar \sM)$, the images $F_{0,*}(\sfH)$ and $F_{1,*}(\sfH)$ are smoothly homotopic in $\sIP_d(\sM)$. 
\end{cor}
\begin{proof}
    Let $\widetilde{F} \colon \mathbf{I}_{\varphi}\Lambda  \to \bR^{d+l}$ be a large-scale Lipschitz homotopy connecting $F_0$ and $F_1$. By definition, the composition $\pr_{\bR^d} \circ \widetilde{F}$ is linearly proper. 
    We may assume that $r>0$ is sufficiently large so that \cref{prp:coarse.homotopy} holds for $\sfH$. 
    We modify $\widetilde{F}$ to extend the domain as 
    \begin{gather*}
    \widetilde{F}' \colon \mathbf{I}_{2\varphi +r }\Lambda \to \bR^{d+l} \times \bR = \bR^{d+l+1}, \quad 
    \widetilde{F}'(\bm{x},n) = (\widetilde{F}(\bm{x},\min \{ n , \varphi(\bm{x}) \} ),n).
    \end{gather*}
    Similarly, we also define $F'_1 \colon \mathbf{I}_{2\varphi+r}\Lambda \to \bR^{d+l+1}$ by $F_1' (\bm{x},n) \coloneqq (F_1(\bm{x}),n)$. 
    Then, both $\widetilde{F}'$ and $F_1'$ are weakly uniformly discrete. Indeed, the restrictions of $\widetilde{F}'$ to $\pr_{\bZ_{\geq 0}}^{-1}(\{n\})$ is near to $F_0$ and hence is weakly uniformly discrete. 
    Moreover, we have $\widetilde{F}' = F_1'$ on $\mathbf{I}_{2\varphi +r}\Lambda \setminus \mathbf{I}_{\varphi}\Lambda$. 
    By \cref{prp:coarse.homotopy}, we get the smooth homotopy 
    \begin{align*}
    F_{0,*}(\sfH) =\widetilde{F}_*'(\sfH \boxtimes \sfh) \simeq \widetilde{F}_*'(\sfH') = F_{1,*}'(\sfH') \simeq F_{1,*}'(\sfH \boxtimes \sfh) = F_{1,*}(\sfH). 
    \end{align*}
    This proves the claim since both $\widetilde{F}'$ and $F'_1$ are weakly uniformly discrete and are linearly proper after composition with $\pr_{\bR^d}$. 
\end{proof}

\begin{exmp}\label{exmp:Lipschitz.homotopy}
    The following maps are examples of large-scale Lipschitz homotopies.
\begin{enumerate}
    \item By replacing $\Lambda$ and $\bZ_{\geq 0}$ in \eqref{eqn:coarse.homotopy} with $X$ and $\bR_{\geq 0}$, the space $\mathbf{I}_\varphi X$ and the notion of (linearly proper) large-scale Lipschitz homotopy is defined for maps whose domain is a non-discrete metric space $X$. 
    Let $F \colon \mathbf{I}_{\varphi}\bR^{d_1} \to \bR^{d_2}$ be a large-scale Lipschitz homotopy. Then, for any $\Lambda \in \fL_d$, the restriction 
    \[
    F|_{\Lambda} \coloneqq (F \times \id)|_{\Lambda} \colon \Lambda \to \bR^{d_2} \times \bR^\infty \times \fR \times \bN
    \]
    is linearly proper after composition with $\pr_{\bR^d}$. We introduce two examples of such maps. Here, the reference point is chosen as $\bm{x}_0 = \bm{0} \in \bR^{d+l}$.  
\begin{itemize}
    \item For $0<\lambda <1$ and $\varphi (\bm{x}) = 1+\| \bm{x} \|$, the map
    \begin{align*}
    \widetilde{F}_{\lambda } \colon \mathbf{I}_{\varphi } \bR^{d+l} \to \bR^{d+l}, \quad \widetilde{F}_{\lambda}(\bm{x},v) = \big(1 + (1- \lambda )(1+\| \bm{x}\|)^{-1}v \big) \cdot \bm{x}
    \end{align*} 
    is a continuous large-scale Lipschitz homotopy connecting the scaling map $\widetilde{F}_{\lambda}(\bm{x} , \varphi(\bm{x}))=\lambda \bm{x}$ and the identity map $\widetilde{F}_\lambda(\bm{x},0)=\bm{x}$. 
    \item Let $U = e^{T} \in SO(d+l)$, where $T$ is a real antisymmetric matrix. Then, the continuous large-scale Lipschitz map 
    \begin{align*}
    F_{T} \colon \mathbf{I}_{\varphi} \bR^{d+l} \to \bR^{d+l}, \quad F_{\lambda}(\bm{x},v) = e^{v (1+\| \bm{x}\|)^{-1} T} \cdot \bm{x}
    \end{align*} 
    connects the rotation by $U$ to the identity. 
\end{itemize}
    \item Let $F_0,F_1 \colon \Lambda \to \bR^{d + l}$ be large-scale Lipschitz maps such that $\pr_{\bR^d} \circ F_i$ are linearly proper and $\pr_{\bR^d} \circ F_0 $ and $ \pr_{\bR^d}  \circ F_1$ are near. 
    Then there is $m>0$ such that $\rmd(F_0(\bm{x}),F_1(\bm{x} )) \leq \varphi(\bm{x}) \coloneqq m(1+\rmd (\bm{x} , \bm{x}_0))$. Hence $F_0$ and $F_1$ are large-scale Lipschitz homotopic by $\widetilde{F} \colon \mathbf{I}_{\varphi}\Lambda  \to \bR^{d+l}$ given by
\begin{align*} 
    \widetilde{F}(\bm{x},n) = \frac{\varphi(\bm{x})-n}{\varphi(\bm{x})} \cdot F_0(\bm{x}) + \frac{n}{\varphi(\bm{x})} \cdot F_1(\bm{x}).
\end{align*} 
    This $\widetilde{F}$ satisfies that $(\pr_{\bR^{d}} \circ \widetilde{F})^{-1}(\bm{x},n) \subset (\pr_{\bR^{d}} \circ F_i)^{-1}(\bm{x}) \times [0,\varphi(\bm{x})]$, and hence $\pr_{\bR^{d}} \circ \widetilde{F}$ is linearly proper. 
\end{enumerate}
\end{exmp}

\subsection{A bulk-boundary correspondence}\label{subsection:bulk.boundary}
Next, we prove that the Kitaev pump is a weak equivalence. A key ingredient of the proof is the following version of bulk--boundary correspondence.
Here, we use the conical regions $Z_{L}^{\theta_L}$ and $Z_{R}^{\theta_R}$ introduced in \eqref{eqn:conical.region} and \cref{fig:transverse.cone}, especially the following subspaces; 
\begin{align*}
    Y_L \coloneqq {}&{} Z_R^{0} = \{ \bm{x} \in \bR^d \mid x_d \leq 0 \}, \\
    Y_R \coloneqq {}&{} Z_R^{0} = \{ \bm{x} \in \bR^d \mid x_d \geq 0 \}, \\
    Z_{L,r} \coloneqq {}&{} Z_L^{\pi/4} \setminus N_r(Y_R) =\{ \bm{x} \in \bR^d \mid x_d \leq \min \{ -r,  -\| \bm{x}\| /\sqrt{2} \} \}, \\
    Z_{R,r} \coloneqq {}&{}Z_R^{\pi/4} \setminus N_r(Y_L) = \{ \bm{x} \in \bR^d \mid x_d \geq \max \{r, \| \bm{x} \| /\sqrt{2} \} \}.
\end{align*}
For $X \subset \bR^d$, we use the same letter $X$ for the inverse images $\pr_{\bR^d}^{-1}(X)\subset \Lambda$.

\begin{thm}\label{prp:Eilenberg.swindle}
    The subsheaf $\bsIP_{d}(\blank )_{R}$ (resp.\ $\bsIP_{d}(\blank)_L$) of $\bsIP_{d}$ consisting of triples $(\sfH,\check{\sfH},\overline{\sfH})$ supported on $Y_R$ (resp.\ $Y_L$) is weakly contractible.   
\end{thm}
This theorem is analogous to a part of the bulk-boundary correspondence in free fermion theory (cf.\ \cite{kubotaControlledTopologicalPhases2017}), which claims that a $d$-dimensional bulk system belongs to the trivial topological phase if and only if, when a boundary is inserted into the system, it admits a boundary condition that preserves the spectral gap.
A consequence of \cref{prp:Eilenberg.swindle} is the following: If an IG UAL Hamiltonian $\sfH$ can be deformed into a boundary-gapped one, i.e., of the form $\sfH_L \boxtimes \sfH_R $, where $\sfH_L$ and $\sfH_R$ are IG UAL Hamiltonians supported on $Y_L$ and $Y_R$, then $\sfH$ belongs to the trivial phase.

The following proof is inspired by the work of Higson--Roe--Yu \cite{higsonCoarseMayerVietorisPrinciple1993} in the K-theory of coarse C*-algebras.

\begin{figure}[t]
    \centering
        \begin{minipage}[b]{0.3\hsize}
        \begin{tikzpicture}
        \draw (4,0.0) -- (0,0.0) -- (0.9375,1.25) -- (4.0,1.25);
        \node at (2.0,0.625) {$\overline{\sfH}_{1-2t}$};
        \draw (4,1.5) -- (1.125,1.5) -- (2.0625,2.75) -- (4.0,2.75);
        \node at (2.8,2.125) {$\overline{\sfH}_{1-2t}$};
        \draw (4,3) -- (2.25,3) -- (3.1875,4.25) -- (4,4.25);
        \node at (3.6,3.675) {$\overline{\sfH}_{1-2t}$};
        \end{tikzpicture}
        \subcaption{The case of $0 \leq t \leq 1/2$.}
    \end{minipage}
        \begin{minipage}[b]{0.3\hsize}
        \begin{tikzpicture}
        \draw (4,0) -- (0.0,0.0) -- (0.375,0.5) -- (4.0,0.5);
        \node at (2.0,0.25) {$\sfH$};
        \draw (4,0.75) -- (0.5625,0.75) -- (0.9375,1.25) -- (4.0,1.25);
        \node at (2.35,1) {$\check{\sfH}$};
        \draw (4,1.5) -- (1.125,1.5) -- (1.5,2) -- (4.0,2);
        \node at (2.7,1.75) {$\sfH$};
        \draw (4,2.25) -- (1.6875,2.25) -- (2.0625,2.75) -- (4.0,2.75);
        \node at (3.05,2.5) {$\check{\sfH}$};
        \draw (4,3) -- (2.25,3) -- (2.625,3.5) -- (4.0,3.5);
        \node at (3.4,3.25) {$\sfH$};
        \draw (4,3.75) -- (2.8125,3.75) -- (3.1875,4.25);
        \end{tikzpicture}
        \subcaption{The case of $t=1/2$.}
    \end{minipage}
        \begin{minipage}[b]{0.3\hsize}
        \begin{tikzpicture}
        \draw (4,0) -- (0.0,0.0) -- (0.375,0.5) -- (4.0,0.5);
        \node at (2.0,0.25) {$\sfH$};
        \draw (4,0.75) -- (0.5625,0.75) -- (1.5,2) -- (4.0,2);
        \node at (2.6,1.375) {$\flip \overline{\sfH}_{2t-1}$};
        \draw (4,2.25) -- (1.6875,2.25) -- (2.625,3.5) -- (4.0,3.5);
        \node at (3.2,2.875) {$\flip \overline{\sfH}_{2t-1}$};
        \draw (4,3.75) -- (2.8125,3.75) -- (3.1875,4.25);
        \end{tikzpicture}
        \subcaption{At $1/2 \leq  t \leq 1$.}
    \end{minipage}
    \caption{An Eilenberg swindle in the proof of \cref{prp:Eilenberg.swindle}.}
    \label{fig:Eilenberg.swindle}
\end{figure}
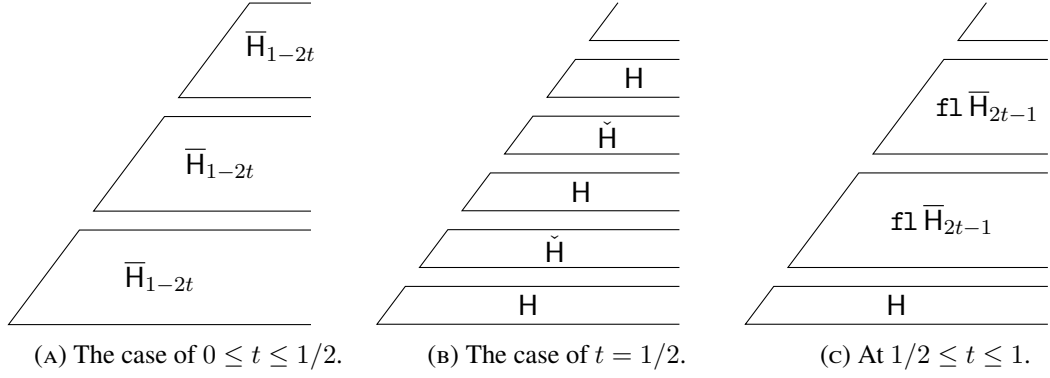

\begin{proof}[Proof of {\cref{prp:Eilenberg.swindle}}]
We show that the set of concordance classes $\bsIP_d[\sM]_R$ is contractible for any compact manifold $\sM \in \Man$. 
Let $\bsfH = (\sfH, \check{\sfH}, \overline{\sfH}) \in \bsIP_d(\sM)_R $. 
Since $\sM$ is compact, there is $\Lambda \in \fL_d$ on which $\bsfH$ is supported in the sense of \cref{eqn:IP.bold}, and $\pr_{\bR^d}(\Lambda)\subset Y_R$.

Let us define $T_d^y \colon \bR^{d+l} \to \bR^{d+l}$ by the shift by $y$ in the $x_d$-direction;
\begin{align*} 
    T_d^y(x_1,\cdots,x_{d+l}) \coloneqq (x_1,\cdots,x_{d-1},x_d+y,x_{d+1},\cdots,x_{d+l}),
\end{align*}
and let
\begin{align*}
    \widetilde{T}_d  \colon   \bR^{d+l} \times \bR_{\geq 0} \to \bR^{d+l} \times \bR_{\geq 0}, \quad \widetilde{T}_d (\bm{x},y) = T_d^y(\bm{x}).
\end{align*}
We define the lattice 
\begin{align*}
    \widetilde{\Lambda} \coloneqq \widetilde{T}_d \bigg( \bigcup_{n \in \bZ_{\geq 0}} (\id \times i_{2n})(\overline{\Lambda}) \cup(\id \times i_{2n})(\overline{\Lambda}) \bigg) \subset \big( \bR^{d} \times \bR^{l} \times \fR \times \bN \big) \times \bR \cong \bR^d \times \bR^{l+1} \times \fR \times \bN.
\end{align*}
By $\Lambda \subset Y_R$, the above $\widetilde{\Lambda} $ satisfies the condition (iii) of \cref{defn:lattice.set} (see also \cref{fig:Eilenberg.swindle}), and hence $\widetilde{\Lambda} \in \fL_d$. 
Therefore, as is illustrated in \cref{fig:Eilenberg.swindle}, the truncated Kitaev pump (\cref{lem:truncated.pump}) on $\widetilde{\Lambda}$ is regarded as a smooth family of IG UAL Hamiltonians
\begin{align*}
    \widetilde{T}_{d,*} \big( \boldsymbol{\kappa}_d^R (\sfH , \check{\sfH}, \overline{\sfH}) \big) \in \cP \bsIP_d( \sM )_R,
\end{align*}
which gives a smooth homotopy 
\begin{align*}
    \bsfH=(\bsfH \boxtimes \sfh \boxtimes \sfh \boxtimes \cdots) \simeq (\sfh \boxtimes \sfh \boxtimes \sfh \boxtimes \cdots) = \sfh.
\end{align*}
This shows that $[\bsfH] = [\sfh] \in \bsIP_d[\sM]_R$. The weak contractibility of $\bsIP_d(\blank)_L$ is proved similarly by reversing the left and the right.  
\end{proof}

\begin{defn}\label{defn:swithcing.IP.sheaf}
We define the following sheaves on $\Man$. 
\begin{enumerate}
    \item Let $\bsIP_{d+1}(\sM)_{R\star }$ denote the subset of $\bsIP_{d+1}(\sM)$ consisting of triples $(\sfH,\check{\sfH},\overline{\sfH})$ such that, for any $k \in \bN$, $f_{\nu,\mu} \in \cF$ and a relatively compact open chart $\sU$ of $\sM$, the following hold for the constants in \eqref{eqn:constant.K}; 
    \[
    K_{Y_R, \sfH-\sfh, \nu,\mu}^{(k)}< \infty, \quad K_{Y_R, \check{\sfH}-\sfh, \nu,\mu}^{(k)}< \infty, \quad  K_{Y_R, \overline{\sfH}-\sfh, \nu,\mu}^{(k)}< \infty.
    \]
    \item Let $\cS \bsIP_{d+1}$ denote the subsheaf of $\cP \bsIP_{d+1}$ whose section on $\sM \in \Man$ consists of smooth families of paths $\sfH \in \cP\sIP_{d+1}(\sM)$ such that $\ev_0 \sfH \in \sIP_{d+1}(\sM)_{R\star} $ and $\ev_1 \sfH \in \sIP_{d+1}(\sM)_L$. 
\end{enumerate}
\end{defn}

\begin{lem}\label{lem:switch.loop.equivalence}
    The inclusion $i \colon \Omega \bsIP_{d+1} \to \cS \bsIP_{d+1}$ is a weak equivalence. 
\end{lem}
\begin{proof}
    By definition, the morphism 
    \[
    (\ev_0,\ev_1) \colon \cS \bsIP_{d+1} \to \sIP_{d+1}(\blank)_{R\star} \times \bsIP_{d+1}(\blank)_L
    \]
    is a fibration in the sense of page \pageref{paragraph:fibration}.  By the long exact sequence in \cref{lem:survey.homotopy} (2), it suffices to show that the sheaves $ \bsIP_{d+1}(\blank)_{R\star}$ and $\bsIP_{d+1}(\blank)_L$ are weakly contractible.

    The weak contractibility for $\bsIP_{d+1}(\blank)_L$ is already proved in \cref{prp:Eilenberg.swindle}. 
    Moreover, the same argument as \cref{prp:Eilenberg.swindle} also shows that the subsheaf of $\bsIP_d$ consisting of smooth families $\bsfH$ that is supported on $Z_{L,r}^c$ is also weakly contractible as well.
    By \cref{lem:cut.diffused}, any section $\sfH \in \bsIP_{d+1}(\blank) _{R\star}$ is smoothly homotopic to another Hamiltonian that is supported on $Z_{L,r}^c$ for some $r>0$. They show the weak contractibility of $\bsIP_{d+1}(\blank)_{R\star}$. 
\end{proof}

The inclusion $j\colon \bR^d \to \bR^{d+1}$ induces the morphism of sheaves
\begin{gather*}
    j_* \colon \bsIP_{d} \to \bsIP_{d+1}(\blank)_{L,R\star} \coloneqq \bsIP_{d+1}(\blank)_L \cap \bsIP_{d+1}(\blank)_{R\star} , \\
    \iota \colon \bsIP_{d+1}(\blank)_{L,R\star} \hookrightarrow\cS \bsIP_{d+1}, 
\end{gather*}
where $\iota$ sends $\bsfH \in \sIP_d(\sM)_{L,R\star}$ to the constant path $\iota \bsfH (p,t) = \bsfH(p)$. 

\begin{lem}\label{lem:weak.equivalence.stick}
    The inclusion $j_* \colon \bsIP_{d} \to \bsIP_{d+1}(\blank)_{L, R\star}$ is a weak equivalence. 
\end{lem}
\begin{proof}
    Let $\sM$ be a compact manifold. We first show that $j_* \colon \bsIP_d[\sM] \to \bsIP_{d+1}[\sM]_{L,R\star}$ is surjective. 
    Let $(\sfH , \check{\sfH}, \overline{\sfH} ) \in \bsIP_{d+1}(\sM)_{L,R\star}$ be supported on $\Lambda$.
    Then, by \cref{lem:cut.diffused}, there is $r>0$ such that the convex combination gives a smooth homotopy of $\bsfH$ and 
    \[
    \Theta_{Z_{L,r}^c , \omega_0}(\bsfH) \coloneqq (\Theta_{Z_{L,r}^c , \omega_0}(\sfH) , \Theta_{Z_{L,r}^c , \omega_0}(\check{\sfH}), \Theta_{Z_{L,r}^c , \omega_0}(\overline{\sfH}) ) \in \bsIP_{d+1}(\sM)_{L,R\star}
    \]
    Now, $\Theta_{Z_{L,r}^c, \omega_0}(\bsfH)$ is supported on $Z_{L,r}^c \cap Y_L$. 
    Under the identification of Euclidean spaces $\bR^{d+1} \times \bR^{l} \cong \bR^d \times \bR^{l+1}$, the subspace $\pr_{\bR^d}^{-1}(Z_{L,r}^c \cap Y_L) \subset \Lambda$ is identified with an element of $\fL_{d}$. 
    Now, both $\Theta_{Z_{L,r}^c, \omega_0}(\bsfH)$ and $j_* \Theta_{Z_{L,r}^c, \omega_0}(\bsfH)$ are regarded as an element of $\bsIP_d(\sM)_{L,R\star}$, and are smoothly homotopic by \cref{exmp:Lipschitz.homotopy} (2). 

    The injectivity is proved similarly. If $\bsfH \in \bsIP_{d}(\sM)$ satisfies that there is a null-homotopy $\widetilde{\bsfH}$ of $j_*\bsfH$, then $\Theta_{Z_{L,r}^c , \omega_0}(\widetilde{\bsfH}) \in \bsIP_{d}(\sM)$ gives a null-homotopy of $j_*'\bsfH$, where $j' \colon \bR^d \times \bR^l \to \bR^d \times \bR^{l+1} $ is given by $j'(x_1,\cdots,x_{d+l}) = (x_1,\cdots,x_d,0,x_{d+1},\cdots,x_{d+l})$. Moreover, again by \cref{exmp:Lipschitz.homotopy} (2), $j_*' \bsfH$ and $\bsfH$ are homotopic in the sheaf $\bsIP_d$. They show that $[\bsfH] = [j_*'\bsfH] = [\sfh] \in \bsIP_d[\sM]$. 
\end{proof}

\begin{lem}\label{lem:switch.constant.equivalence}
    The inclusion $\iota \colon \bsIP_{d+1}(\blank)_{L,R\star} \to \cS \bsIP_{d+1}$ is a weak equivalence. 
\end{lem}
\begin{proof}
    We apply \cref{thm:interpolation.loop} to $Y= Y_R$. Then, the adiabatic interpolation gives a morphism of sheaves
    \begin{gather*}
    \vartheta_d \colon \cS \bsIP_{d+1} \to \bsIP_{d+1}(\blank)_{L,R\star} , \\
    \vartheta_d (\sfH,\check{\sfH},\overline{\sfH})(p) \coloneqq  (\tilde{\vartheta}\sfH (p,0), \tilde{\vartheta}\check{\sfH} (p,0), \tilde{\vartheta}\overline{\sfH} (p,0)). 
    \end{gather*}
    The only non-trivial part of this statement is $\tilde{\vartheta}\sfH \in \sIP_{d+1}(\sM)_R$, which we now verify. Indeed, we have
    \begin{align}
    \begin{split}
    {}&{} \| \tilde{\vartheta}\sfH_{\bm{x}}  - \sfh_{\bm{x}} \|_{\sU,C^k} \\
    = {}&{} \big\| \Theta_{Y_R^c , \omega_0} \circ \alpha \big( \Pi_{Y_R}^\star (\sfG_{\sfH}) \,; t, 1\big)^{-1} (\sfH_{\bm{x}}) - \sfh_{\bm{x}} \big\|_{\sU,C^k} \\
    \leq {}&{} \| \Theta_{Y_R^c , \omega_0}  (\alpha \big( \Pi_{Y_R}^\star (\sfG_{\sfH}) \,; t,1\big)^{-1} (\sfH))_{\bm{x}} -  \alpha \big( \Pi_{Y_R}^\star (\sfG_{\sfH}) \,; t,1\big)^{-1} (\sfH_{\bm{x}}) \|_{\sU,C^k} \\
    {}&{} \quad + \big\| \alpha \big( \Pi_{Y_R}^\star (\sfG_{\sfH}) \,; t,1\big)^{-1} (\sfH_{\bm{x}}) - \sfH_{\bm{x}}\big\|_{\sU,C^k} + \big\| \sfH_{\bm{x}} - \sfh_{\bm{x}} \big\|_{\sU,C^k}.
    \end{split}\label{eqn:theta.decay}
    \end{align}
    The third term decays faster than $f(\mathrm{dist}(\bm{x},Y_R))$ for any $f \in \cF$ by assumption. The same is true for the first and the second terms by \cref{rmk:cutoff.approximate,lem:LGA.cone.decomposition} respectively.     
    
    By definition, we have $\vartheta_d \circ \iota  =\id$. Therefore, it suffices to show that $\iota \circ \vartheta_d (\bsfH) $ is smoothly homotopic to $\bsfH $ for any $\bsfH  \in \bsIP_{d+1}(\sM)_{L,R\star}$. 
    First, by \eqref{eqn:theta.decay}, we have that a function
    \[
       (p,t,s) \mapsto ( \tilde{\vartheta}\sfH (p,st), \tilde{\vartheta}\check{\sfH} (p,st), \tilde{\vartheta}\overline{\sfH} (p,st)) 
    \]
    gives an element of $\cS\bsIP_{d+1}(\sM \times [0,1]_s)$. By regarding it as a smooth homotopy of $\cS\bsIP_{d+1}$, with respect to the parameter $s \in [0,1]$, we obtain that
    \[
    \iota \circ \vartheta_d(\bsfH) \simeq ( \tilde{\vartheta}\sfH , \tilde{\vartheta}\check{\sfH} , \tilde{\vartheta}\overline{\sfH} ) \in \cS \bsIP_{d+1}(\sM). 
    \]
    Moreover, since 
    \[
    ( \tilde{\vartheta}\sfH  , \tilde{\vartheta}\check{\sfH} , \tilde{\vartheta}\overline{\sfH} )(p,1) = (\sfH,\check{\sfH},\overline{\sfH})(p,1) \in \bsIP_d(\sM)_L
    \]
    is supported on $Y_L$, the convex combination gives the second smooth homotopy
    \begin{align*}
    {}&{}( \tilde{\vartheta}\sfH , \tilde{\vartheta}\check{\sfH} , \tilde{\vartheta}\overline{\sfH} )\\
    ={}&{} \Theta_{Y_R^c , \omega_0} \big( \alpha(\Pi_{Y_R}^{\star}\sfG_{\sfH}\,; t,1)^{-1}(\sfH(p,t)), \alpha(\Pi_{Y_R}^{\star}\sfG_{\check{\sfH}}\,; t,1)^{-1}(\check{\sfH}(p,t)),  \alpha(\Pi_{Y_R}^{\star}\sfG_{\overline{\sfH}}\,; t,1)^{-1}(\overline{\sfH}(p,t))\big)     \\
    \simeq {}&{}\big( \alpha(\Pi_{Y_R}^{\star}\sfG_{\sfH}\,; t,1)^{-1}(\sfH(p,t)), \alpha(\Pi_{Y_R}^{\star}\sfG_{\check{\sfH}}\,; t,1)^{-1}(\check{\sfH}(p,t)),  \alpha(\Pi_{Y_R}^{\star}\sfG_{\overline{\sfH}}\,; t,1)^{-1}(\overline{\sfH}(p,t))\big) .    
    \end{align*}
    Finally, by unwinding the LG automorphisms, the right hand side is smoothly homotopic to $(\sfH,\check{\sfH},\overline{\sfH})$. These three homotopies conclude $[\iota \circ \vartheta_d (\bsfH)]=[\bsfH ] \in \cS\bsIP_{d+1}[\sM]$ as desired. 
\end{proof}

\begin{lem}\label{lem:switch.diagram.commute}
    The following diagram commutes up to homotopy:
    \[
    \xymatrix{
    \bsIP_{d} \ar[rr]^{\boldsymbol{\kappa}_d} \ar[rd]_{\simeq}^{\iota \circ j_*} && \Omega \bsIP_{d+1}\ar[ld]^{\simeq}_{i} \\
    &\cS \bsIP_{d+1}.&
    }
    \]
\end{lem}
\begin{proof}
    The homotopy connecting $\boldsymbol{\kappa}_d(\bsfH)$ and the constant path taking value in $\bsfH$ in the sheaf $\cS \bsIP_{d+1}$ is explicitly given by
    \begin{align}
        \widetilde{\bsfH}_{(\bm{x},n)}(p,t,s)\coloneqq 
        \begin{cases}
            \boldsymbol{\kappa}_d(\bsfH)_{(\bm{x},n)} (p,t-s) & \text{ if $n >0$ or $ (n,t) \in \{ 0 \} \times [0,1/2-s] $, } \\
            \boldsymbol{\kappa}_d(\bsfH)_{(\bm{x},n)} (p,t+s) & \text{ if $n <0$ or $ (n,t) \in \{ 0 \} \times [1/2+s,1] $, }\\
            \bsfH_{\bm{x}} & \text{ if $n=0$ and $t \in [1/2-s , 1/2+s]$}
        \end{cases}\label{eqn:pump.constant.homotopy}
    \end{align}
    See also \cref{fig:pump.constant.homotopy}. 
\end{proof}

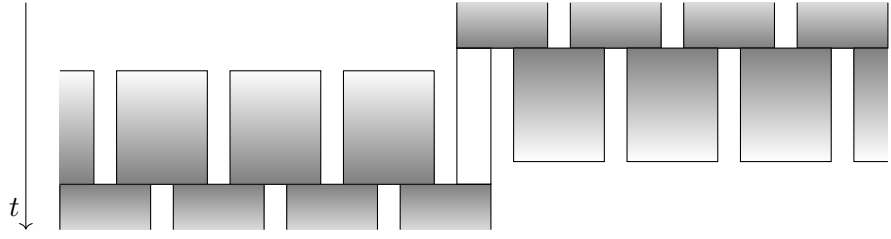
\begin{figure}
    \centering
    \begin{tikzpicture}[xscale=1.5, yscale=1.5]
        \draw[->] (-0.3,2.0) -- (-0.3,0.0);
        \node at (-0.4,0.2) {$t$};
        \draw[shade, top color = gray, bottom color = {rgb:black,1;white,6}] (0,0) -- (0,0.4) -- (0.8,0.4) -- (0.8,0);
        \draw[shade, top color = gray, bottom color = {rgb:black,1;white,6}] (1,0) -- (1,0.4) -- (1.8,0.4) -- (1.8,0);
        \draw[shade, top color = gray, bottom color = {rgb:black,1;white,6}] (2,0) -- (2,0.4) -- (2.8,0.4) -- (2.8,0);
        \draw[shade, top color = gray, bottom color = {rgb:black,1;white,6}] (3,0) -- (3,0.4) -- (3.8,0.4) -- (3.8,0);
        \draw[shade, top color = white, bottom color = gray] (0.0,0.4) -- (0.3,0.4) -- (0.3,1.4) -- (0.0,1.4);
        \draw[shade, top color = white, bottom color = gray] (0.5,0.4) -- (0.5,1.4) -- (1.3,1.4) -- (1.3,0.4) -- cycle;
        \draw[shade, top color = white, bottom color = gray] (1.5,0.4) -- (1.5,1.4) -- (2.3,1.4) -- (2.3,0.4) -- cycle;
        \draw[shade, top color = white, bottom color = gray] (2.5,0.4) -- (2.5,1.4) -- (3.3,1.4) -- (3.3,0.4) -- cycle;
        \draw[shade, bottom color = gray, top color = {rgb:black,1;white,6}] (3.5,2) -- (3.5,1.6) -- (4.3,1.6) -- (4.3,2);
        \draw[shade, bottom color = gray, top color = {rgb:black,1;white,6}] (4.5,2) -- (4.5,1.6) -- (5.3,1.6) -- (5.3,2);
        \draw[shade, bottom color = gray, top color = {rgb:black,1;white,6}] (5.5,2) -- (5.5,1.6) -- (6.3,1.6) -- (6.3,2);
        \draw[shade, bottom color = gray, top color = {rgb:black,1;white,6}] (6.5,2) -- (6.5,1.6) -- (7.3,1.6) -- (7.3,2);
        \draw[shade,top color = gray, bottom color = white] (4.0,1.6) -- (4.0,0.6) -- (4.8,0.6) -- (4.8,1.6) -- cycle;
        \draw[shade,top color = gray, bottom color = white] (5.0,1.6) -- (5.0,0.6) -- (5.8,0.6) -- (5.8,1.6) -- cycle;
        \draw[shade,top color = gray, bottom color = white] (6.0,1.6) -- (6.0,0.6) -- (6.8,0.6) -- (6.8,1.6) -- cycle;
        \draw[shade,top color = gray, bottom color = white] (7.3,1.6) -- (7.0,1.6) -- (7.0,0.6) -- (7.3,0.6);
        \draw (3.5,0.4) -- (3.5,1.6) -- (3.8,1.6) -- (3.8,0.4) -- cycle;
    \end{tikzpicture}
    \caption{The homotopy \eqref{eqn:pump.constant.homotopy}. The shaded squares represent the homotopy $\overline{\sfH}$. The middle white square stands for $\sfH$. }
    \label{fig:pump.constant.homotopy}
\end{figure}

\begin{thm}\label{cor:spectrum}
The invertible gapped phase spectrum $\{ \IP_d, \kappa_d\}$ is an $\Omega$-spectrum.  
\end{thm}
\begin{proof}
    It follows from \cref{lem:switch.constant.equivalence,lem:switch.loop.equivalence,lem:switch.diagram.commute,prp:Eilenberg.swindle}.
\end{proof}

\subsection{Fermionic and equivariant versions}\label{subsection:SPT}
The proofs of \cref{thm:inverse.pump,prp:Eilenberg.swindle,cor:spectrum} are immediately made fermionic and $G$-equivariant. Therefore, we get the following variations of $\Omega$-spectra.

\begin{cor}\label{cor:fermionic.equivariant.spectrum}
The following hold.
\begin{enumerate}
    \item Let $\fR$ be a set of fermionic internal degrees of freedom. Then the sequence of sheaves $\{ \sfIP_d\}$, and hence the sequence of spaces $\{ \fIP_{\fR,d}\}$, form an $\Omega$-spectrum object via the Kitaev pump $\kappa_d$ (\cref{defn:Kitaev.pump}) as connecting maps. 
    \item Let $G$ be a compact Lie group, $\phi \colon G \to \bZ/2$ be a homomorphism, and let $\fR$ be a set of bosonic (resp.\ fermionic) internal degrees of freedom with on-site $(G,\phi)$-symmetry. Then the sequence of $G$-sheaves $\{ \fIP_{\fR,d}\}$, and hence the sequence of $G$-spaces $\{ \IP_{\fR,d}\}$, form a naive $G$-spectrum via the Kitaev pump $\kappa_d$ as connecting maps.
\end{enumerate}          
\end{cor}
Here, the adjective `naive' means that the $G$-equivariant connecting map is not defined for $\IP_d \to \Omega_V \IP_{d+V}$ for non-trivial representations $V$ of $G$ (cf.\ \cite{mayEquivariantHomotopyCohomology1996}*{p.~141}).

\begin{lem}\label{lem:equivariant.spectrum}
Let $\fR$ be a set of internal degrees of freedom and let $\fI$ be a subset such that the prime factors of $\{ \dim \sH_\lambda \}_{\lambda \in \fR}$ and those of $\{ \dim \sH_\lambda \}_{\lambda \in \fI}$ generate the same multiplicative subgroup of $\bQ^\times$. 
Then, the inclusion $\sIP_{\fI,d} \to \sIP_{\fR,d}$ is a morphism of $G$-sheaves that is non-equivariantly weakly equivalent. 
\end{lem}
\begin{proof}
    It is clear from the definition that the inclusion is $G$-equivariant. It follows from \cref{lem:flip.homotopy.use} that the inclusion is a weak equivalence.
\end{proof}
The typical case in which we apply this lemma is when $\fI$ consists of trivial representations.
In this case, \cref{lem:equivariant.spectrum} means that $\IP_d$ is a split $G$-spectrum in the sense of \cite{mayEquivariantHomotopyCohomology1996}*{p.197, Definition XVI.2.1}. 
As a consequence, for any free $G$-space $\sX$, a Mayer--Vietoris argument shows that the induced map
\begin{align}
    \Hom (\sX/G,\IP_{\fI,d}) = \Hom (\sX,\IP_{\fI,d})^G \to \Hom (\sX,\IP_{\fR,d})^G \label{eqn:split.pre}
\end{align}
is a weak equivalence (cf.\ \cite{mayEquivariantHomotopyCohomology1996}*{Theorem XVI.2.4}). 
This is a simple yet effective application of \cref{cor:spectrum}. 

As will be discussed in \cref{section:homotopy}, we define the cohomology functors $\rIP^d \colon \kTop \to \mathsf{Ab}$ and $\rIP^d_G \colon \kTop_G \to \mathsf{Ab}$ corresponding to the IP-spectra as 
\begin{align*}
    \rIP^d(\blank ) \coloneqq [\blank, \IP_{\fR,d}], \quad \rIP^d_G(\blank ) \coloneqq [\blank, \IP_{\fR,d}]^G,
\end{align*}
in which the choice $\fR$ is abbreviated just for simplicity of notation. 
With this notation, \eqref{eqn:split.pre} is rewritten as
\begin{align} 
    \rIP_{\fJ,d}(\sX/G) = [\sX/G  , \IP_{\fJ,d}] \cong [\sX , \IP_{\fR,d}]^G =\rIP_{\fR,d}^G(\sX) .\label{eqn:split}
\end{align}
In particular, for any free $G$-manifold $\sM$, the homotopy sets of smooth sections are compared as
\begin{align*} 
    \sIP_{\fJ,d}[\sM/G] \cong \rIP^d_G(\sM) \cong \rIP^d(\sM/G) \cong \sIP_{\fI,d}[\sM /G]. 
\end{align*}

When $\phi$ is non-trivial, the right hand side of the above isomorphism should be replaced with the IP-cohomology group with twisted coefficient. Here, the twisted $\IP$-cohomology group with the twisting data $\phi \in \mathrm{H}^1(\sM\,; \bZ/2) \cong \Hom(\pi_1(\sM), \bZ/2)$ is defined by
\[
    \prescript{\phi}{}{\rIP}^{d}(\sX) \coloneqq [\sX_{\phi}, \IP_{\fR,d}]^{\bZ/2},
\]
where $\sX_\phi$ denotes the double covering associated with $\phi$ and $\fR$ is a set of internal degrees of freedom with on-site $(\bZ/2,\phi)$-symmetry. 
Assume that a set of internal degrees of freedom $\fR$ has a subset $\fI$ consisting of $\phi$-trivial representation, which satisfies the assumption of \cref{lem:equivariant.spectrum}. Here, we say that a $\phi$-twisted representation is $\phi$-trivial if it factors through $\phi \colon G \to \bZ/2$. Then, by the same argument as the above paragraph, 
\begin{align*}
    \Hom (\sX/G_0,\IP_{\fI,d})^{\bZ/2}  = \Hom (\sX,\IP_{\fI,d})^G \to \Hom (\sX,\IP_{\fR,d})^G
\end{align*}
is a weak equivalence for any free $G$-space $\sX$, where $G_0\coloneqq \mathop{\mathrm{Ker}} \phi$. Thus, we get the isomorphism 
\begin{align*} 
     \prescript{\phi}{}{\IP}^d(\sX/G) = [\sX/G_0 , \IP_{\fJ,d}]^{\bZ/2} \cong [\sX , \IP_{\fR,d}]^{G} = \prescript{\phi}{}{\IP}_G^d(\sX) .
\end{align*}
In particular, for any free $G$-manifold $\sM$, the homotopy sets of smooth sections are compared as
\begin{align*} 
    \prescript{\phi}{}{\sIP}_{\fJ,d}[\sM/G] \cong \prescript{\phi}{}{\rIP}^d(\sM/G) \cong \prescript{\phi}{}{\IP}_{\fR,d}^G(\sM) \cong \prescript{\phi}{}{\sIP}_{\fR,d}^G[\sM]. 
\end{align*}

The above observation also relates the equivariant IP cohomology with the associated Borel equivariant cohomology. 
\begin{defn}\label{defn:G.index}
    Let $\fR$ be a set of internal degrees of freedom with on-site $(G,\phi)$-symmetry that has a subset $\fI$ consisting of $\phi$-trivial representations such that the prime factors of $\{ \dim \sH_\lambda \}_{\lambda \in \fR}$ and those of $\{ \dim \sH_\lambda \}_{\lambda \in \fI}$ generate the same multiplicative subgroup of $\bQ^\times$. For any $G$-manifold $\sM$, the group homomorphism $\mathrm{Ind}_{G,\phi, \sM}^d$ is defined by 
\begin{align*} 
    \mathrm{Ind}_{G,\phi, \sM}^d \colon \prescript{\phi}{}{\sIP}_{\fR,d}^G[\sM] \cong \prescript{\phi}{}{\rIP}^d_G(\sM) \xrightarrow{\mathrm{pr}_{\sM}^*} \prescript{\phi}{}{\rIP}^d_G (\sM \times EG) \cong  \prescript{\phi}{}{\rIP}^d(\sM \times_{G} EG). 
\end{align*}
In particular, in the case of $\sM =\pt$, we obtain the group homomorphism
\begin{align*}
    \mathrm{Ind}_{G,\phi}^d \colon {}&{} \prescript{\phi}{}{\rIP}^d_G(\pt) \to \prescript{\phi}{}{\rIP}^d(BG),
\end{align*}
which we call the \emph{SPT index}. The fermionic version is also defined similarly. 
\end{defn}

\section{Applications to topological phases}\label{section:homotopy}
By \cref{cor:spectrum}, the functors
\begin{align*}
    &&\rIP^d(\sX )\coloneqq{}&{}  [\sX,\IP_d] , & \rIP^d(\sX , \sA )\coloneqq {}&{}  [(\sX, \sA),(\IP_d , \sfh)],&&\\
    &&\rfIP^d(\sX )\coloneqq{}&{}  [\sX,\fIP_d] , & \rfIP^d(\sX , \sA )\coloneqq {}&{} [(\sX, \sA),(\fIP_d , \sfh)],&&
\end{align*}   
form generalized cohomology theories whose domains are the category $\kTop^2$ of pairs of compactly generated Hausdorff spaces. 
If $\sX$ is a manifold, then each of these groups is isomorphic to the smooth homotopy sets of the sheaves $\sIP_d$ and $\sfIP_d$. 
In this section, we summarize what we have already known about the groups $\IP^d(\pt)$ and $\fIP^d(\pt)$, in anticipation of some results in the subsequent paper \cite{kubotaStableHomotopyTheory2025b}, and their applications to the determine the group of SPT phases.

\subsection{High degree homotopy groups}\label{subsection:homotopy.higher}
In low dimensions, $d =0, 1$, the homotopy types of the spaces $\IP_d$ are determined. 
In particular, the homotopy groups $\pi_n(\IP)=[S^n,\IP]$ of the $\Omega$-spectrum $\IP$ is determined for $n \geq -1$. 
The $\Omega$-spectra $\IP$ and $\fIP$ are $2$-truncated, namely, their homotopy groups of degree $n >2$ vanish. 

\subsubsection{Nonnegative degrees}
Hereafter, for an abelian group $G$, $K(G,n)$ denotes the Eilenberg--Mac Lane space whose unique non-trivial homotopy group is $\pi_n(K(G,n)) \cong G$. 
\begin{prp}\label{prp:0d}
    There are weak equivalences $\IP_0 \simeq K(\bZ,2)$ and $\fIP_0 \simeq K(\bZ,2) \times K(\bZ/2,0)$.
\end{prp}
\begin{proof}
    The sheaf $\sIP_0$ is weakly equivalent to the colimit of smooth manifolds $\bP\cA_{\Lambda}$, where $\bP\cA_{\Lambda}$ is the space of rank $1$ projections in the matrix algebra $\cA_{\Lambda}$. Since $\bP\cA_{\Lambda}$ is the complex projective space and the map 
    \begin{align*} 
    \bP\cA_{\Lambda_1} \to \bP(\cA_{\Lambda_1 \boxtimes \Lambda_2}), \quad p \mapsto p \otimes p_0 
    \end{align*}
    induces an isomorphism of the second homotopy groups, we obtain $\sIP_0 \simeq K(\bZ,2)$. Similarly, $\sfIP_0$ is weakly equivalent to the colimit of the manifolds $\widehat{\bP}\widehat{\cA}_{\Lambda}$ of even rank $1$ projections in $\widehat{\cA}_{\Lambda}$, which is identified with two copies of complex projective spaces (consisting of even and odd lines respectively). Hence $\fIP_0 \simeq K(\bZ,2) \sqcup K(\bZ,2) = K(\bZ,2) \times K(\bZ/2,0)$.
\end{proof}
\cref{prp:0d} means that $\pi_2(\IP_0) \cong \bZ$ is the only non-trivial positive degree homotopy group. This isomorphism is given by the Berry curvature class. 
A remarkable consequence of \cref{cor:spectrum} is that, for any $d \in \bZ_{\geq 0}$, the spaces $\IP_d$ and $\fIP_d$ have the trivial $\pi_n$-group all but finitely many $n \in \bN$. 
The homotopy type of spaces of this kind is characterized by a finite set of the Postnikov $k$-invariants. A part of $k$-invariants of $\IP$ and $\fIP$ are discussed later.

\subsubsection{Degree \texorpdfstring{$-1$}{-1}}\label{subsection:1d}
We determine the homotopy groups of $1$-dimensional topological phases $\IP_1$ and $\fIP_1$. 
The topology of quantum spin chains, i.e., 1-dimensional quantum spin systems, has been studied in \cites{ogataClassificationGappedHamiltonians2019,kapustinClassificationInvertiblePhases2021,carvalhoClassificationSymmetryProtected2024} for the bosonic case and \cite{bourneClassificationSymmetryProtected2021} for the fermionic case. 
They conclude that any bosonic (invertible) gapped quantum spin chain can be connected to the trivial phase, while the fermionic invertible phases are classified by an $\bZ/2 $-valued index. 
Since we are dealing with a broader class of $1$-dimensional lattices, we reprove these results in our framework. 
The following argument has much overlap with the paper \cite{carvalhoClassificationSymmetryProtected2024}. 

Let $\Lambda \in \fL_{\fR,1}$. 
Continuing from \cref{subsection:bulk.boundary}, we use the notation $Y_L \coloneqq \pr_{\bR}^{-1}(\{x \leq 0 \} )$. 
Moreover, we write $Y_R \coloneqq \pr_{\bR}^{-1}(\{ x > 0\})$, which differs slightly from the previous notation. 
Note that the pair $(Y_L,Y_R)$ is linearly coarse transverse in the sense of \cref{defn:linearly.coarsely.transverse}, and hence \cref{prp:left.equivalent.state.unitary} is applicable.

\begin{prp}\label{prp:1d}
    The following hold:
    \begin{enumerate}
        \item The space $\pi_0(\IP_1 ) \cong 0$ holds.
        \item Suppose that the internal degrees of freedom $\fR$ contain $(\hat{\sH}_\lambda,\Omega_\lambda)$ such that $\dim \sH_\lambda^{(0)}=\dim \sH_\lambda^{(1)}$. Then the isomorphism $\pi_0(\fIP_1) \cong \bZ/2$ holds. 
    \end{enumerate} 
\end{prp}
\begin{proof}
    We first observe that $\pi_0(\IP_1) \cong \pi_0(\sIP_1) \cong 0$, which shows (1) by \cref{prp:0d}.  
    By \cref{prp:left.equivalent.state.unitary}, there is an almost local unitary $u \in \cA_{\Lambda}^{\al}$ such that the distinguished ground state of $u^*\sfH u$ decomposes into the tensor product as 
    \begin{align}
    \omega_{u^* \sfH u}= \omega_{\sfH} \circ \Ad (u)= \omega_L \otimes \omega_R,
    \label{eqn:1d.inner.split}
    \end{align}
    where $\omega_{\bullet}$ is a pure state on $\cA_{Y_{\bullet}}$ for $\bullet =L,R$. Now, let
    \begin{align*}
    \acute{\sfH} \coloneqq \acute{\sfH}_L \boxtimes \acute{\sfH}_R, \quad  \acute{\sfH}_L \coloneqq \Theta_{Y_L, \omega_R}(u^*\sfH_{\bm{x}} u), \quad \acute{\sfH}_R \coloneqq \Theta_{Y_R,\omega_L }(u ^* \sfH_{\bm{x}} u) 
    .
    \end{align*}
    By \cref{prp:cut.trivial.Hamiltonian}, $\acute{\sfH}_L$ and $\acute{\sfH}_R$ are gapped UAL Hamiltonians. 
    By \cref{rmk:almost.local.unitary}, $\sfH$ and $u^* \sfH u$ are connected by the smooth homotopy $\exp (-t \log u) \sfH \exp(t \log u)$. Moreover, the convex combination and the null-homotopy obtained by \cref{prp:Eilenberg.swindle} 
    gives a smooth homotopy $u^*\sfH u \simeq \acute{\sfH} \simeq \sfh$. 
    In conclusion, we obtain $[\sfH] =0 \in \pi_0(\sIP_1)$, that is, $\pi_0(\sIP_1) \cong 0$. 
    
    A similar argument shows that $\pi_0(\fIP_1 ) \cong \bZ/2$. 
    Following \cite{bourneClassificationSymmetryProtected2021}*{Definition 2.18}, we define
    \begin{align*}
    \cI^0 \colon \pi_0(\fIP_1) \to \bZ/2, \quad \cI^0(\sfH) = 
    \begin{cases}
        0 & \text{if $\pi(\widehat{\mathcal{A}}_{Y_L})'' \cong \cB(\sH)$, }\\
        1 & \text{if $\pi(\widehat{\mathcal{A}}_{Y_L})'' \cong \cB(\sH) \hotimes \mathbb{C}\ell_1$. }
    \end{cases}
    \end{align*}
    This index is a well-defined group homomorphism (\cite{bourneClassificationSymmetryProtected2021}*{Theorem 3.3}). Indeed, for a smooth path $\sfH \in \sfIP_1([0,1])$ connecting $\sfH_0 = \ev_0\sfH$ and $\sfH_1=\ev_1 \sfH$, we have an asymptotic equivalence of states 
    \[
    \omega_{\sfH_0}|_{\widehat{\cA}_{Y_L}} \sim_{\mathrm{asymp}} \omega_{\sfH_1} \circ \alpha \big( \Pi_{Y_L}(\sfG_{\sfH}) \,; 1 \big) |_{\widehat{\cA}_{Y_L}}
    \]
    by \cref{thm:automorphic.equivalence,lem:LGA.cone.decomposition}. 
    By \cite{bratteliOperatorAlgebrasQuantum1987}*{Corollary 2.6.11}, this implies that these states are quasi-equivalent.
    Therefore, the $\ast$-automorphism $\alpha \big( \Pi_{Y_L}(\sfG_{\sfH}) \,; 1 \big)^{-1}$ extends to the isomorphism of double commutant von Neumann algebras $\pi_{\omega_{\sfH_0}}(\widehat{\cA}_{Y_L})'' \cong \pi_{\omega_{\sfH_1}}(\widehat{\cA}_{Y_L})''$. 
    The surjectivity of $\cI^0$ follows from the subsequent \cref{exmp:FK}, in which the assumption on $\fR$ is used.  
    To see the injectivity, let $\sfH$ be an IG UAL Hamiltonian with $\cI^0(\sfH)=0$. Then \cref{prp:left.equivalent.state.unitary.fermion} applies to such $\sfH$, and hence the same argument as the above paragraph shows that $\sfH$ is smoothly homotopic to $\sfh$. 
\end{proof}

\begin{exmp}\label{exmp:FK}
    Fidkowski--Kitaev \cites{fidkowskiEffectsInteractionsTopological2010,fidkowskiTopologicalPhasesFermions2011} gives an example of $1$-dimensional fermionic UAL Hamiltonians with non-trivial index.
    Let $\fR$ be a set of fermionic internal degrees of freedom satisfying the assumption of \cref{prp:1d} (2). Then $\cA_{\lambda} \cong \cB(\sH_\lambda)$ is isomorphic to $\bC \ell_2 \otimes M_n(\bC)$, where $M_n(\bC)$ is trivially graded. Hence, by letting $\Lambda \coloneqq \bZ \times \{\lambda\} $, $\cA_{\Lambda}$ has a tensor component isomorphic to the infinite graded tensor product of $\bC \ell_2$, i.e., the CAR algebra $\cA_{\mathrm{CAR}}$. 
    The Fidkowski--Kitaev model is a Hamiltonian on $\cA_{\mathrm{CAR}}$. 
    It is regarded as an IG UAL Hamiltonian on $\cA_{\Lambda}$ by take the composite with a trivial Hamiltonian on the tensor complement $M_n(\bC) \subset \cA_{\lambda}$, given by a fixed rank $n-1$ projection $\sfh_{\lambda}' \in M_n(\bC)$.

    Let us consider a graded tensor decomposition $\cA_{\mathrm{CAR},\bm{x}} \cong \bC \ell_2 \cong \bC\ell_1 \hotimes \bC\ell_1=: \cB_{1,\bm{x}} \hotimes \cB_{2,\bm{x}}$ for each $\bm{x} \in \Lambda$. 
    Let $\sfH_{\bm{x}} \in \cB_{2,\bm{x}} \hotimes \cB_{1,\bm{x+1}} \cong M_2(\bC)$ be one of two even rank $1$ projections. 
    Then the family $ (\sfH _{\bm{x}} )_{\bm{x}} $ forms a gapped UAL Hamiltonian. 
    Moreover, it is invertible since we can apply \cref{prp:Eilenberg.swindle} to the composite system $\sfH \boxtimes \sfH$. Indeed, the local generator $\sfH_{\bm{x}} \boxtimes \sfH_{\bm{x}} \in \cB_{2,\bm{x}}^{\otimes 2} \hotimes \cB_{1,\bm{x+1}}^{\otimes 2}$ can be replaced homotopically with a local gapped Hamiltonian that is of the form $\sfH_{\bm{x},1} \boxtimes \sfH_{\bm{x},2}$, where $\sfH_{1,\bm{x}} \in \cB_{1,\bm{x+1}}^{\hotimes 2}$ and $\sfH_{2,\bm{x}} \in \cB_{2,\bm{x}}^{\hotimes 2}$ are gapped positive operators with $1$-dimensional kernels.   
    Finally, since $\omega_{\sfH}$ restricts to a pure state on $\cA_{Y_L} \hotimes \cB_{1,\bm{1}} = \mathop{\widehat{\bigotimes}}_{\bm{x} \in Y_L} (\cB_{2,\bm{x}} \hotimes \cB_{1,\bm{x+1}}) $, we obtain that
    \[ 
        \cB(\sH) \cong \pi_{\omega}(\cA_{Y_L} \hotimes \cB_{1,\bm{1}})'' \cong \pi_{\omega}(\cA_{Y_L})'' \hotimes \cB_{1,\bm{1}}.
    \]
    For the second isomorphism, we have used \cite{bourneClassificationSymmetryProtected2021}*{Lemma A.4}. This means that $\cI^0(\sfH)=1 $. 
\end{exmp}

\subsubsection{Postnikov truncation of the fermionic IP spectrum}\label{subsubsection:kinvariant}
\cref{prp:1d} concludes that $\IP_1 \simeq K(\bZ,3)$.  
We also determine the homotopy type of the space $\fIP_1$ of fermionic invertible phases. To this end, we use the following fact, which is proved in the subsequent paper.
\begin{thm}[{\cite{kubotaStableHomotopyTheory2025b}}]\label{rmk:KO}
    Araki's quasi-free second quantization \cite{arakiQuasifreeStatesCAR1970} gives rise to a morphism of $\Omega$-spectra 
    \[
    \mathrm{Q} \colon \Sigma^{-2}\mathit{KO} \to \fIP,
    \]
    which induces isomorphism of homotopy groups as $\pi_0(\mathit{KO}_{-1}) \cong \pi_0(\fIP_1) \cong \bZ/2$, $\pi_1 (\mathit{KO}_{-1}) \cong \pi_1(\fIP_1) \cong \bZ/2$, $\pi_2(\mathit{KO}_{-1}) \cong \pi_2(\fIP_1) \cong 0$, and $\pi_3( \mathit{KO}_{-1}) \cong \pi_3(\fIP_1) \cong \bZ$. 
\end{thm}
This theorem means that we have a weak equivalence 
\[  
    \tau_{[-1,\infty)} \fIP = \tau_{[-1,2]} \fIP \simeq \tau_{[-1,2]}\Sigma^{-2}\mathit{KO}.
\]
Here, for a space or a spectrum $\sX$, we write $\tau_{(-\infty,b]}\sX$ for the Postnikov $b$-coconnective truncation of $\sX$, $\tau_{[a,\infty)}\sX$ for the Postnikov $a$-connective cover of $\sX$, and $\tau_{[a,b]}\sX = \tau_{(-\infty,b]}\tau_{[a,\infty)}\sX$.  
We write
\begin{align*} 
    k^{l} \colon \tau_{(-\infty,l-2]}\sX \to \Sigma \tau_{[l-1,l-1]}\sX \simeq  \Sigma^lH\pi_{l-1}(\sX)
\end{align*}
for the $l$-th $k$-invariant, i.e., the morphism realizing $\tau_{(-\infty,l-1]}\sX$ as its homotopy fiber  (we refer to \cite{mayConciseCourseAlgebraic1999}*{Section 22.4}).

\begin{prp}\label{prp:Postnikov}
    The first $k$-invariant of $\tau_{[-1,\infty)}\fIP$ is
            \begin{align*}
            k^1 =\Sq^2 \colon \tau_{[-1,-1]}\fIP \simeq \Sigma^{-1}H\bZ/2 \to  \Sigma H\bZ/2 .
            \end{align*}
            Moreover, the first $k$-invariant of $\tau_{[1,\infty)}\fIP$ is 
            \begin{align*}
            k^3 = \beta \circ \Sq^2 \colon \tau_{[1,1]} \fIP \simeq  \Sigma H \bZ/2 \to \Sigma^3 H\bZ 
            \end{align*}
            where $\beta$ is the Bockstein homomorphism.
\end{prp}
\begin{proof}
    This follows from \cref{rmk:KO} and the corresponding facts in the KO-spectrum. See e.g.~\cite{thomasHomotopyClassificationMaps1964}*{Theorem 4.2} or \cite{fujiiKgroupsProjectiveSpaces1967}*{(1.3)}. 
\end{proof}

\begin{rmk}
    Let $i \colon \tau_{[1,1]}\fIP \to \tau_{[-1,1]}\fIP \simeq \mathrm{hofib} (\Sq^2)$ denotes the inclusion. Then, the second claim of \cref{prp:Postnikov} means that $k^3 \circ i = \beta \circ \Sq^2$, in which $k^3$ denotes the $k$-invariant of $\tau_{[-1,\infty)}\fIP$. 
    As is observed by Gaiotto--Johnson-Freyd in  \cite{gaiottoSymmetryProtectedTopological2019}*{Subsection 5.4} and Beardsley--Luecke--Morava in \cite{beardsleyBrauerWallGroupsTruncatedPicard2023}*{Theorem 4.14}, there are exactly two candidates of $k^3$, namely, morphisms $k \colon \mathrm{hofib}(\Sq^2) \to \Sigma^3 H\bZ$ satisfying the relation $k \circ i =\beta \circ \Sq^2$; however, the two possible choices do not affect the resulting $\Omega$-spectrum.  

    In the literature of topological phases, there are four $\Omega$-spectra having the $k$-invariants with $k^1=\Sq^2$ and $k^3 \circ i = \beta \circ \Sq^2$. 
    \begin{enumerate}
        \item The truncated connective KO-theory spectrum $\Sigma^{-1}\tau_{[1,4]}ko$, which is also denoted by $R^{-1}$ by Freed in \cites{freedLecturesTwistedTheoryOrientifolds2012,freedAnomaliesInvertibleField2014}. 
        \item The $\Omega$-spectrum $\mathit{cAlg}_{\bC}^\times$ associated to the Picard $2$-groupoid of invertible topological superalgebras (\cite{freedLecturesTwistedTheoryOrientifolds2012}). 
        \item The truncated Picard spectrum $\mathrm{pic}_0^3 \mathit{KU}$ of the complex K-theory spectrum \cite{beardsleyBrauerWallGroupsTruncatedPicard2023}.
        \item The truncated fermionic IP spectrum $\Sigma \tau_{[-1,\infty)}\fIP = \Sigma \tau_{[-1,2]}\fIP$. 
    \end{enumerate}
    According to \cite{gaiottoSymmetryProtectedTopological2019}*{Subsections 4.2, 5.3, 5.4} and \cite{johnson-freydTopologicalOrdersDimensions2022}, the extended supercohomology theory \cite{wangCompleteClassificationSymmetryprotected2018} should also be included to this list. 
    As is remarked above, they are all weakly equivalent as $\Omega$-spectra. 
    We also refer to \cite{freedLecturesTwistedTheoryOrientifolds2012}*{Theorem 1.52} and \cite{beardsleyBrauerWallGroupsTruncatedPicard2023}*{Theorem 5.27}. 
\end{rmk}

\begin{prp}[{cf.\ \cite{bourneClassificationSymmetryProtected2021}}]\label{prp:1d.fermion}
    For a space $\sX$, there is an isomorphism
    \[
        \rfIP^1(\sX) \cong \mathrm{H}^0(\sX\,; \bZ/2) \oplus (\mathrm{H}^1(\sX\,; \bZ/2) \times   \mathrm{H}^3(\sX\,; \bZ)), 
    \]
        where the group structure on the right hand side is given by the Wall group law
        \[
        (a_0,b_0,c_0)+(a_1,b_1,c_1)\coloneqq (a_0+a_1,b_0+b_1, c_0 + c_1 + \beta(b_0 \cup b_1)),
        \]
        where $\beta \colon \mathrm{H}^2(\sX \,; \bZ/2) \to \mathrm{H}^3(\sX \,; \bZ)$ denote the Bockstein homomorphism. 
\end{prp}
\begin{proof}
By the above discussion, we have a natural isomorphism 
\begin{align*}
    \rfIP^1(\sX) \cong (\tau_{[-1,\infty)}\rfIP)^1(\sX) \cong (\mathrm{pic}_0^3 KU)^0(\sX).
\end{align*}
The right hand side is determined in \cite{beardsleyBrauerWallGroupsTruncatedPicard2023}*{Proposition 4.18}.
\end{proof}


\begin{rmk}\label{rmk:supercohomology}
    Following \cite{gaiottoSymmetryProtectedTopological2019}*{Subsections 4.2, 5.3, 5.4} and \cite{johnson-freydTopologicalOrdersDimensions2022}, we relate our cohomology theory $\fIP$ with the notion of supercohomology \cites{guSymmetryprotectedTopologicalOrders2014,kitaevHomotopytheoreticApproachSPT2015,kapustinFermionicSPTPhases2017,wangCompleteClassificationSymmetryprotected2018} by recalling the work of Kock--Kristensen--Madsen \cite{kockCochainFunctorsGeneral1967}. If one has an $\Omega$-spectrum $E$ whose homotopy group $\pi_n(E)$ is trivial for all but finite $n \in \bZ$, it can be realized by a cochain complex. 
    For example, a model of $(\tau_{[0,\infty)} \fIP)^n(\sX)$ is given by the cochain complex
    \begin{align*}
    \mathrm{SC}_{\mathrm{res}}^n(\sX) \coloneqq C^{n}(\sX \,; \bZ ) \times  C^{n-2}(\sX \,; \bZ/2), \quad \delta_S^n \coloneqq 
    \begin{pmatrix}
        \delta^n & \beta \circ \Sq^2 \\ 
        0 & \delta^{n-2}
    \end{pmatrix}.
    \end{align*}
    This is the restricted supercohomology \cite{guSymmetryprotectedTopologicalOrders2014} is a cochain model of the truncated $\Omega$-spectrum $\tau_{[0,\infty)}\fIP$. 
    Here, a cochain-level realization of $\beta \circ \Sq^2$ is fixed, which commutes with $\delta$ but is not a homomorphism of abelian groups. 
    Now, by taking an additivity constraint map $d \colon C^{n-2}(\sX \,; \bZ /2) \times C^{n-2}(\sX \,; \bZ /2) \to C^{n}(\sX \,; \bZ )$ such that $\delta d(y_1,y_2) - d(\delta y_1,y_2) - d(y_1, \delta y_2) = \beta\Sq^2(y_1+y_2) - \beta\Sq^2(y_1) - \beta\Sq^2(y_2)$, a (non-associative) summation on $\mathrm{SC}_{\mathrm{res}}^n(\sX)$ is imposed as
    \begin{align*}
    (x_1,y_1) \dot{+} (x_2,y_2) \coloneqq (x_1 +x_2+ d(y_1,y_2),y_1+y_2 ).
    \end{align*}
    Then the set $\mathrm{SC}_{\mathrm{res}}^n(\sX)$ is no longer a group but forms a loop in the sense of \cite{kockCochainFunctorsGeneral1967}*{Part II, p.\ 151}, and $\delta_S^n$ is a morphism of loops. 
    The associated cohomology $\mathrm{SH}^n(\sX) \coloneqq \ker \delta_S / \sim $ is defined as \cite{kockCochainFunctorsGeneral1967}*{Part II, Section 2}, which becomes an abelian group isomorphic to $(\tau_{[0,\infty)} \fIP)^n(\sX)$. 
    The morphisms $\beta \circ \Sq^2$ and $d$ with the required properties are given by Steenrod's cup-$i$ product \cite{steenrodProductsCocyclesExtensions1947} as 
    \[
    \beta \circ \Sq^2(x) = \beta (x \cup_2 x), \quad d(y_1,y_2) = \beta (y_1 \cup_1 y_2). 
    \]
    Indeed, this reproves Wall's group law in \cref{prp:1d.fermion}.

    The extended supercohomology group \cite{wangCompleteClassificationSymmetryprotected2018}, a cochain model of $\tau_{[-1,\infty)}\fIP$, is realized in the same way. 
    By fixing a chain level realization $(\beta \circ \Sq^2,j)$ of the $k$-invariant $k^3$, let
    \begin{align*}
    \mathrm{SC}_{\mathrm{ext}}^n(\sX) \coloneqq 
      C^{n}(\sX \, ;\bZ) \times  C^{n-2}(\sX \, ;\bZ/2) \times  C^{n-3}(\sX \, ;\bZ/2), \quad \delta_{S}^n \coloneqq 
    \begin{pmatrix}
        \delta^n & \beta \circ \Sq^2 & j \\
        0 & \delta^{n-2} & \Sq^2 \\
        0 & 0 & \delta^{n-3}
    \end{pmatrix}
    \end{align*}
    on which the summation is imposed by 
    \begin{align*}
    (x_1,y_1,z_1) \dot{+} (x_2,y_2,z_2) \coloneqq (x_1 +x_2+ d(y_1,z_1,y_2,z_2) ,y_1+y_2 + z_1 \cup_1 z_2,z_1+z_2 ).
    \end{align*}
    Here, $d$ is a fixed choice of the the additivity constraint for $(\beta \circ \Sq^2,j)$ with respect to the summation of the domain twisted by $z_1 \cup_1 z_2$. 
    The problem of determining a cochain-level realizations of $(\beta \circ \Sq^2,j)$ and $d$ are considered in \cite{wangCompleteClassificationSymmetryprotected2018}. 
\end{rmk}

\subsubsection{Degree \texorpdfstring{$-2$}{-2}}
Applying the functor $\Omega^\infty \Sigma^d $ taking the space at level $d$, the $k$-invariants in \cref{prp:Postnikov} are $\Sq^2$ acting on $(H\bZ/2)^{d-1}$ and $\beta \circ \Sq^2$ acting on $(H\bZ/2)^{d}$, which vanish if $d=1$. Therefore, we obtain that $\tau_{[0,1]}\fIP_1 \simeq K(\bZ/2,0) \times K(\bZ/2,1)$. 
Moreover, since the differentials of the Atiyah--Hirzebruch spectral sequences computing $\rfIP^d(\sX)$ is given by these $k$-invariants by \cite{maunderSpectralSequenceExtraordinary1963}, and hence has no non-trivial higher differentials if $d=1$. Hence there is a short exact sequence
\[
    0 \to \mathrm{H}^3(\sX \,; \bZ) \to \rfIP^1(\sX) \to \mathrm{H}^0(\sX \,; \bZ/2) \oplus \mathrm{H}^1(\sX \,; \bZ/2) \to 0,
\]
which is consistent with \cref{prp:1d.fermion}. 
We apply this computation to $d=2$ after taking the connective cover. 

\begin{prp}\label{prp:2d}
The following hold:
    \begin{enumerate}
        \item The space $\IP_2$ is weakly equivalent to $K(\bZ,4) \times K(\pi_0(\IP_2),0)$.
        \item The space $\fIP_2$ is weakly equivalent to 
        \[
        K(\pi_0(\fIP_2), 0) \times K(\bZ/2 , 1) \times \mathrm{hofib}(\beta \circ \Sq^2 \colon K(\bZ/2,2) \to K(\bZ,5) ). 
        \]
    \end{enumerate} 
\end{prp}
We remark that the above weak equivalence may not preserve the H-space structure, and hence it does not say anything about the group structure of the homotopy set $\rfIP^2(\sX)$.
\begin{proof}
    Since the connected components of $\IP_2$ are weakly equivalent to each other, we have $\IP_2 \simeq \tau_{[1,\infty)}\IP_2 \times K(\pi_0(\IP_2),0)$. Similarly, we also have $\fIP_2 \simeq \tau_{[1,\infty)}\fIP_2 \times K(\pi_0(\fIP_2),0)$.
    
    Consider the morphism
    \[
    \Sigma^{-1}H\bZ/2 \simeq \tau_{[-1,-1]}\fIP \to \Sigma \tau_{[0,\infty)} \fIP \simeq \Sigma \mathrm{hofib}(\beta \circ \Sq^2 \colon H\bZ/2 \to \Sigma^3 H\bZ )
    \]
    realizing $\tau_{[-1,\infty)}\fIP_2$ as its homotopy fiber. According to the computation in \cite{beardsleyBrauerWallGroupsTruncatedPicard2023}*{Proposition 4.10}, it is trivial at the spaces of level $2$. This concludes that $\tau_{[1,\infty)}\fIP_2 \simeq \tau_{[1,1]}\fIP_2 \times \tau_{[2,\infty)} \fIP_2$. 
    Since the $k$-invariant of $\tau_{[2,\infty)} \fIP_2$ is $\beta \circ \Sq^2$, we obtain the weak equivalence in (2). 
\end{proof}
In comparison to \cref{rmk:supercohomology}, a cochain model of the functor $\rfIP^2$ is given by the set of equivalence classes of triples $(\kappa,b,a) \in C^5(\sX\,; \bZ) \times C^2(\sX \,; \bZ/2) \times C^2(\sX \,; \bZ/2)$ with the twisted cocycle conditions $\delta \kappa =\beta \circ \Sq^2(b) $, $\delta b = 0$, and $\delta a=0$.

\subsection{The Chern--Dold character}\label{subsection:Chern.Dold}
A remarkable consequence of \cref{cor:spectrum} is the computation of the $\bQ$-coefficient cohomology group $\rIP^d(\sX) \otimes_{\bZ}\bQ$. 
For a generalized cohomology functor $\mathrm{E}^*$, the Chern--Dold character (we refer to \cite{rudyakThomSpectraOrientability1998}*{Theorem-Definition II.7.13}) gives an isomorphism 
    \begin{align*} 
    \ch \colon \mathrm{E}^d(\sX) \otimes \bQ \to \mathrm{H}^d(\sX,\sA \, ; \mathrm{E}_\ast \otimes \bQ).
    \end{align*}
Apply this to $\rIP$ and $\rfIP$, we obtain the following isomorphisms. 
\begin{prp}\label{prp:top.degree}
    There are isomorphisms
\begin{align*}
    \rIP^d(\sX,\sA) \otimes_{\bZ} \bQ \cong {}& \prod_{p+q=d}\mathrm{H}^{p}(\sX,\sA \, ; \pi_{-q}(\IP) \otimes_{\bZ} \bQ), \\
    \rfIP^d(\sX,\sA) \otimes_{\bZ} \bQ \cong {}& \prod_{p+q=d}\mathrm{H}^{p}(\sX,\sA \, ; \pi_{-q}(\IP) \otimes_{\bZ} \bQ).    
\end{align*}
\end{prp}
Since the $\Omega$-spectra $\IP$ and $\fIP$ are $2$-truncated, the `top terms' of the Chern--Dold characters are extracted as direct summands. 
\begin{defn}
    We define $\ch_{\mathrm{top}} \colon \rIP^d(\sX,\sA) \to \mathrm{H}^{d+2}(\sX,\sA\,; \bQ)$ and  $\ch_{\mathrm{top}} \colon \rfIP^d(\sX,\sA) \to \mathrm{H}^{d+2}(\sX,\sA\,; \bQ)$ by the composition of the Chern--Dold character and the projection.
\end{defn}
It will be proved in \cite{kubotaStableHomotopyTheory2025b} that this character is identical to the higher order Berry curvature class formulated by Kapustin--Sopenko \cites{kapustinLocalNoetherTheorem2022,artymowiczQuantizationHigherBerry2023}. A fundamental question in higher order Berry curvature class is the existence of an integral lift, i.e., a morphism of cohomology functors $\rIP^d \to \mathrm{H}^{d+2}(\blank\,; \bZ)$. It is proved in \cite{artymowiczQuantizationHigherBerry2023} that there is an integral lift for $d \leq 2$. The discussion of the previous subsection reproves this fact. 

\begin{prp}\label{prp:integral.lift}
    For $d \leq 2$, the Postnikov approximation $K(\bZ,d+2) \simeq \tau_{[d+2,\infty)}\IP_d \to \IP_d$ splits. That is, for any CW-complex $\sX$, we have
\begin{align*} 
    \rIP^{d}(\sX) \cong \mathrm{H}^{d+2}(\sX  \, ; \bZ) \oplus (\text{lower terms}).
\end{align*}
\end{prp}
\begin{proof}
     In $d=0,1$, it is already verified in \cref{prp:0d,prp:1d} that the map $K(\bZ,d+2) \to \IP_d$ is a weak equivalence. In $d=2$, the map $\tau_{[1,\infty)}\IP_2 \to \IP_2$ is the inclusion $K(\bZ,4) \to K(\pi_0(\IP_2),0) \times K(\bZ,4)$, which splits since $\IP_2$ is an H-group. 
\end{proof}

There would be a certain validity, based on physics, in expecting that the $\Omega$-spectrum $\IP$ has a relation to the Anderson dual $I_{\mathbb{Z}}\mathit{MSO}$ of the Thom spectrum, which gives a topological classification of bosonic invertible (geometric) quantum field theory \cites{freedReflectionPositivityInvertible2021,yamashitaDifferentialModelsAnderson2023,gradyDeformationClassesInvertible2023}.

\begin{lem}
    The morphism $(H\bZ)_n \to (I_\bZ \mathit{MSO})_n$ splits if $n = 5,6$, although does not split if $n \geq 8$.  
\end{lem}
\begin{proof}
    When $n= 5$, the only non-trivial homotopy groups of $(I_\bZ \mathit{MSO})_3$ are $\pi_1((I_\bZ \mathit{MSO})_5) \cong \bZ$ and $\pi_5((I_\bZ \mathit{MSO})_5) \cong \bZ$. However, since $[K(\bZ ,1), K(\bZ, 6)] \cong 0$, the $k$-invariant can not be non-trivial. Similarly, when $n=6$, the only non-trivial homotopy groups are $\pi_0((I_\bZ \mathit{MSO})_6) \cong \bZ/2$, $\pi_2((I_\bZ \mathit{MSO})_4) \cong \bZ$ and $\pi_6((I_\bZ \mathit{MSO})_6) \cong \bZ$. Since $[K(\bZ/2,0),K(\bZ,7)] \cong 0$ and $[K(\bZ,2),K(\bZ,7)] \cong \mathrm{H}^7(\bC \bP^\infty \,; \bZ) \cong 0 $, we have $[\tau_{\leq 2}(I_\bZ \mathit{MSO})_6,K(\bZ ,7)] \cong 0$, and hence the $k$-invariant can not be non-trivial.

    Finally, we show that $\mathrm{H}^8(X;\, \bZ) \to \mathrm{I}_{\bZ}\mathrm{MSO}^8(X)$ is not injective for some $X$. It is known that, for $X=K(\bZ/3,2)$, the map $\mathrm{MSO}_7(X) \to H_7(X\,;\bZ) \cong \bZ/3$ is zero (cf.\ \cite{rudyakThomSpectraOrientability1998}*{Theorem IV.7.35}). By comparing the UCT exact sequences of Anderson duals
    \begin{align*}
        \xymatrix{
        0 \ar[r] & \mathrm{Ext}_\bZ^1(\mathrm{H}_7(X),\bZ) \ar[r] \ar[d] &  \mathrm{H}^8(X\,; \bZ) \ar[d] \ar[r] & \Hom(\mathrm{H}_8(X),\bZ) \ar[r] \ar[d] & 0 \\
        0 \ar[r] & \mathrm{Ext}_\bZ^1(\mathrm{MSO}_7(X),\bZ)  \ar[r] &  \mathrm{I}_\bZ\mathrm{MSO}^8(X) \ar[r]  & \Hom(\mathrm{MSO}_8(X),\bZ) \ar[r]  & 0 ,
        }
    \end{align*}
    the vanishing of the left vertical map shows the desired non-injectivity.
\end{proof}

This leads us to the following conclusion we present as a conjecture.
\begin{conj}\label{conj:quantization.Chern}
    The top Chern--Dold character $\ch_{\mathrm{top}} \colon \rIP^d \to \mathrm{H}\bQ^{d+2}$ has an integral lift for $d=3,4$, but does not have such a lift for $d \geq 6$.
\end{conj}

\subsection{Symmetry-protected topological phases}
When the internal degrees of freedom $\fR$ has a on-site symmetry of $(G,\phi)$, we write the associated $G$-fixed point spectrum by 
\begin{align*}
    \prescript{\phi}{}{\IP}_{\fR,d}^{G} \coloneqq |\Sing \bsIP_{\fR,d}^G|.
\end{align*}
Similarly, we write $\prescript{\phi}{}{\fIP}_{\fR,d}^{G}$ for the $G$-fixed point of the fermionic IP spectrum. Our main interest is the following group of SPT phases. 
\begin{defn}
Let $G$ be a compact Lie group, let $\phi \colon G \to \bZ/2$, and let $\fR$ be a set of (bosonic or fermionic) internal degrees of freedom with on-site $(G,\phi)$-symmetry. 
The groups of bosonic and fermionic symmetry-protected topological phases in dimension $d$ are defined by 
\begin{align*}
    \mathrm{SPT}_{\fR, d}(G,\phi ) \coloneqq {}&{}  \Ker (\pi_0(\IP_{\fR,d}^G) \to \pi_0(\IP_{\fR,d})), \\
    \mathrm{fSPT}_{\fR, d}(G,\phi) \coloneqq {}&{}  \Ker (\pi_0(\fIP_{\fR,d}^G) \to \pi_0(\fIP_{\fR,d})). 
\end{align*}
\end{defn}

Let $\Rep(G,\phi)$ (resp.\ $\widehat{\Rep}(G,\phi)$) denote the set of isomorphism classes of finite dimensional $\phi$-twisted unitary representations (resp.\ $\bZ/2$-graded $\phi$-twisted unitary representations) of $G$. 
For $[\tau] \in \mathrm{H}^2(G \,; \bT_\phi)$ , we call a pair $(v,\sV)$ a $(\phi , \tau )$-twisted representation if $\sV$ is a Hilbert space, $v$ is a continuous map from $G$ to the group of linear or antilinear unitaries on $\sV$ such that $v(g)$ is linear or antilinear if $\phi(g)$ is $0$ or $1$ and $v(g)v(h)=\tau(g,h)v(gh)$. 
Furthermore, for $c \in \mathrm{H}^1(G\,; \bZ/2) \cong \Hom (G,\bZ/2)$, the above pair $(v,\sV)$ is a $(\phi,c,\tau)$-twisted representation if $\sV$ is $\bZ/2$-graded and $v(g)$ is even or odd if $c(g)$ is $0$ or $1$. See e.g.\ \cites{freedWignerTheorem2012,freedTwistedEquivariantMatter2013,kubotaNotesTwistedEquivariant2016}.

\begin{defn}\label{defn:IDF.ample}
We say that a bosonic internal degree of freedom $\fR$ is 
\begin{itemize}
\item \emph{$0$-ample} if the smallest subset of $\Rep(G , \phi )$ closed under the subrepresentation and the tensor product containing $\{ \pi_{\lambda} \mid \lambda \in \fR \}$ is $\Rep(G , \phi )$, and
\item \emph{$1$-ample} if it is $0$-ample and, for any $[\tau] \in \mathrm{H}^2(G\,; \bT_\phi)$, there is a $(\phi,\tau)$-twisted representation $(\sV ,v)$ and $\lambda_1,\cdots,\lambda_l \in \fR$ such that there is a unital $G$-equivariant $\ast$-homomorphism $\cB(\sV) \to \cA_{\lambda_1} \otimes \cdots \otimes \cA_{\lambda_l}$. 
\end{itemize}
The $1$-ampleness of fermionic internal degrees of freedom is defined similarly by replacing $\tau$ with a pair $(c,\tau) \in \mathrm{H}^1(G\,; \bZ/2) \oplus \mathrm{H}^2(G\,; \mathbb{T}_\phi)$ and $(\sV ,v)$ with any $(\phi,c,\tau)$-twisted unitary representation.
\end{defn}

Unless the discussion depends on a particular choice of such $\fR$, we abbreviate it.
By \cref{cor:spectrum}, the spaces $\prescript{\phi}{}{\IP}_{d}^{G}$, $\prescript{\phi}{}{\IP}_{d}^{G}$ form $\Omega$-spectra. Their homotopy groups are determined in the same way as the non-equivariant case.

\begin{prp}\label{prp:0d.G}
    Assume that $\fR$ is $0$-ample. The following isomorphism holds.
\begin{enumerate}
    \item If $\phi$ is trivial, then 
    \begin{align*}
        \pi_n(\IP_0^G) \cong
        \begin{cases} 
            0 & \text{$n \geq 3$, $n=1$,}\\
            \bZ & \text{$n=2$, }\\
            \mathrm{Line}(G) & \text{$n=0$,}
        \end{cases}
    \quad 
    \pi_n(\fIP_0^G) \cong
        \begin{cases} 
            0 & \text{$n \geq 3$, $n=1$,}\\
            \bZ & \text{$n=2$, }\\
            \mathrm{Line}(G) \times \bZ/2& \text{$n=0$.}
        \end{cases}
    \end{align*}
    \item If $\phi$ is non-trivial, then
    \begin{align*}
        \pi_n(\prescript{\phi}{}{\IP}_0^G) \cong
        \begin{cases} 
            0 & \text{$n \geq 2$,}\\
            \bZ/2 & \text{$n=1$, }\\
            \mathrm{Line}(G,\phi) & \text{$n=0$,}
        \end{cases}
    \quad 
    \pi_n(\prescript{\phi}{}{\fIP}_0^G) \cong
        \begin{cases} 
            0 & \text{$n \geq 2$,}\\
            \bZ /2 & \text{$n=1$, }\\
            \mathrm{Line}(G,\phi) \times \bZ/2 & \text{$n=0$.}
        \end{cases}
    \end{align*}
\end{enumerate}
Here, $\mathrm{Line}(G,\phi)$ denotes the set of isomorphism classes of $1$-dimensional $\phi$-twisted (resp.\ $\bZ/2$-graded) representations of $G$. 
\end{prp}

\begin{proof}
    It is verified in the same way as \cref{prp:0d} that 
    \begin{align*} 
    \IP_0^G = \colim (\bP \cA_{\Lambda})^G, \quad  (\bP \cA_{\Lambda})^G = \bigsqcup_{\pi \in \mathrm{Line}(G)} \bP \cB(\sH_{\Lambda,\pi}), 
    \end{align*}
    where $\sH_{\Lambda,\pi}$ denotes the subrepresentation of $\bigotimes_{\lambda \in \Lambda} \sH_\lambda$ corresponding to the irreducible representation $\pi$. Since each $\bP \cB(\sH_{\Lambda,\pi})$ is a complex projective space, this shows the bosonic part of (1). The fermionic part is verified in the same way. 
    
    When $\phi$ is non-trivial, the group $G$ acts on each $\cB(\sH_{\Lambda,\pi})$ by complex conjugation, and hence the space $(\bP\cB(\sH_{\Lambda,\pi}))^G$ of $G$-invariant rank $1$ projections is identified with a real projective space. 
    The colimit has the homotopy type of $\bR\bP^\infty =B\bZ/2$, which proves (2). The fermionic version is proved in the same way. 
    Indeed, the spaces $\fIP_0^G$ and $\prescript{\phi}{}{\fIP}_0^G$ are the two copies, corresponding to the $\bZ/2$-grading, of the complex and the real projective spaces respectively.
\end{proof}

\begin{prp}\label{prp:Brauer.1d}
Let $\mathrm{Br} (G,\phi)$ (resp.\ $\widehat{\mathrm{Br}}(G,\phi)$) denote the Brauer group, i.e., the set of isomorphism classes of infinite dimensional type I factors (resp.\ $\bZ/2$-graded type I factors) with $(G,\phi)$-actions, whose group structure is imposed by the tensor product (resp.\ the graded tensor product). 
The index maps
\begin{align*}
    \cI_{G,\phi}^1 \colon \pi_0(\prescript{\phi}{}{\IP}_1^G) \to \mathrm{Br}(G,\phi), \quad \cI_{G,\phi}^1 \colon \pi_0(\prescript{\phi}{}{\fIP}_1^G) \to \widehat{\mathrm{Br}}(G,\phi),
\end{align*}
given by $\cI_{G,\phi}^1(\sfH)=\pi_{\omega_{\sfH}}(\cA_{Y_L})''$ (cf.\ \cref{prp:1d} (2)) is well-defined. Moreover, it is injective if $\fR$ is $0$-ample and is surjective if $\fR$ is $1$-ample.  
\end{prp}

\begin{proof}
    The well-definedness (i.e., the homotopy invariance) is proved in the same way as \cref{prp:1d} (2). 
    The surjectivity follows from the construction of the equivariant Fidkowski--Kitaev model given in the following \cref{exmp:equivariant.FK}, that is also given in \cite{carvalhoClassificationSymmetryProtected2024}*{Lemma 3.14}.
    The injectivity is also almost the same as \cref{prp:1d}, but we must care for the proof of the equivariant version of \cref{prp:left.equivalent.state.unitary,prp:left.equivalent.state.unitary.fermion}. 

    In the equivariant case, $\cI_{G,\phi}([\sfH] ) =0 $ means that the GNS Hilbert space $(\sH,\pi)$ for $\omega_{\sfH}$ decomposes to the tensor product of unitary representations $\sH_L$ and $\sH_R$ of $G$. 
    A point is that, even if a unit vector $\Omega \in \sH=\sH_L \otimes \sH_R$, with the Schmidt decomposition $\Omega = \sum_i \lambda_i (\xi_i \otimes \eta_i)$, is $G$-invariant in total, a split vector $\xi_1 \otimes \eta_1 \in \sH_L \otimes \sH_R$ can not be taken $G$-invariant in general. The same issue is dealt with in \cite{carvalhoClassificationSymmetryProtected2024}*{Proposition 5.2}. 
    In general, there are finite dimensional irreducible subrepresentations $\sE_L \leq \sH_L$ and $\sE_R\leq \sH_R$ such that $\sE_L \otimes \sE_R$ has a $G$-invariant unit vector $\Omega_0$ that represents an $G$-equivariant isomorphism of $\sE_L^* \cong \sE_R$. 
    By assumption of $0$-ampleness, for $\bullet = L,R$, there is $\lambda_{1,\bullet}, \cdots, \lambda_{n,\bullet} \in \fR $ such that $\sH_{\lambda_{1,\bullet}} \otimes \cdots \otimes \sH_{\lambda_{k,\bullet}}$ has a direct summand $\sE_{\bullet}'$ isomorphic to $\sE_{\bullet}^*$. Let $\sfH_0$ be the $G$-invariant IG UAL Hamiltonian supported on $\{ \bm{1}, \bm{-1} \} \subset \bR$ whose ground state is given by the unit vector
    \[
    \Omega_0^* \in \sE_L^* \otimes \sE_R^* \cong \sE_L' \otimes \sE_R' \subset (\sH_{\lambda_{1,L}} \otimes \cdots \otimes \sH_{\lambda_{k,L}})_{\bm{-1}} \otimes (\sH_{\lambda_{1,R}} \otimes \cdots \otimes \sH_{\lambda_{k,R}})_{\bm{1}}.
    \]
    Then, the ground state of $\sfH \boxtimes \sfH_0$ is asymptotically equivalent to the split vector $\Omega_0 \otimes \Omega_0^* \in (\sE_L \otimes \sE_L') \otimes (\sE_R \otimes \sE_R')$ under the identification $\sE_L' \cong \sE_L^* \cong \sE_R$. 
    Applying the remaining argument of the proof of the proof of \cref{prp:1d} (2), we get a $G$-invariant smooth null-homotopy of $\sfH \boxtimes \sfH_0$. Finally, again by \cref{prp:Eilenberg.swindle}, the tensor component $\sfH_0$ is also null-homotopic in $\prescript{\phi}{}{\sIP}_1^G$, and hence that the same is true for $\sfH$. 
\end{proof}

\begin{exmp}\label{exmp:equivariant.FK}
    The equivariant version of \cref{exmp:FK} is defined in the following way. Here, we discuss the fermionic case, but the same construction also works for the bosonic case.
    
    Let $\cB$ be a $\bZ/2$-graded finite dimensional C*-algebra that is graded simple (in other words, a finite dimensional graded type I factor), equipped with a $\phi$-twisted action of $G$.   
    Then there is another $\bZ/2$-graded simple finite dimensional C*-algebra $\check{\cB}$ such that $\cA_{\lambda} \coloneqq \cB \hotimes \check{\cB}$ is isomorphic to $\cB(\sK)$ for some $\phi$-linear unitary representation $(\sK,\pi)$ that contains a $G$-invariant unit vector $\Omega$. 
    Let $\fR$ be a set of fermionic internal degrees of freedom with on-site $(G,\phi)$-symmetry containing $\lambda \coloneqq (\sK,\pi,\Omega)$.
    Let $\Lambda \coloneqq \bZ \times \{\lambda\} $, $\cA_{\Lambda}$ and define an IG UAL Hamiltonian on $\cA_{\Lambda}$. 

    In the same way as \cref{exmp:FK}, consider the graded tensor decomposition of the on-site observable algebra as $\cA_{\bm{x}} \cong \cB_{\bm{x}} \hotimes \check{\cB}_{\bm{x}}$ for each $\bm{x} \in \Lambda$ and take $\sfH_{\sK,\bm{x}} \in \cB_{\bm{x}} \hotimes \check{\cB}_{\bm{x+1}} \cong \cB(\sK)$ to be the local gapped Hamiltonian whose unique ground state is $\Omega$. 
    Then the family $\sfH_{\cB} \coloneqq (\sfH _{\cB,\bm{x}} )_{\bm{x}} $ forms a gapped UAL Hamiltonian with the homotopy inverse $\sfH_{\check{\cB}}$.   
    Now, since $\check{\cB} \hotimes \cB(\sH_L)$ is isomorphic to a bounded operator algebra, we get 
    \[
    \cI^1_{G,\phi}(\sfH) + [\check{\cB}] =\cI^1_{G,\phi}(\sfH) - [\cB] = 0, \quad \text{i.e., } \quad \cI^1_{G,\phi}(\sfH)= [\cB]. 
    \]
    This shows the surjectivity of $\cI^1_{G,\phi}$ under the assumption that $\fR$ is $1$-ample. 
\end{exmp}

\begin{rmk}\label{rmk:equivariant.cohomology}
    An element of $\mathrm{Line}(G,\phi)$ is represented by a continuous map $\pi \colon G \to \bT$ satisfying the $1$-cocycle relation $\pi(g)\prescript{\phi(g)}{}{\pi(h)} = \pi(gh)$. 
    That is, $\mathrm{Line}(G)=\Hom(G,\bT) \cong \mathrm{H}^1(G\,; \bT)$ and $\mathrm{Line}(G,\phi) \cong \mathrm{H}^1(G \,; \bT_\phi )$, where $\bT_\phi$ denotes the group $\bT$ on which $G$ acts by complex conjugation through $\phi$. 

    Similarly, the equivariant Dixmier--Douady theory given in \cite{kumjianBrauerGroupLocally1998} (see \cite{kubotaNotesTwistedEquivariant2016}*{Section 2} for $\phi$-twisted case) states that $\mathrm{Br}(G,\phi)$ is isomorphic to the group $\mathrm{Ext}(G,\mathbb{T}_{\phi})$ of (continuous) extensions of $G$ by the $G$-module $\mathbb{T}_\phi$. In the fermionic case, Wall \cite{wallGradedBrauerGroups1964} describes the graded Brauer group as
    \[
    \widehat{\mathrm{Br}}(G,\phi) \cong \bZ/2 \oplus_{} \Hom(G,\bZ/2) \oplus \mathrm{Ext}(G,\mathbb{T}_\phi),
    \]
    where the right hand side is equipped with the structure of an abelian group via the Wall group law (see also \cite{bourneClassificationSymmetryProtected2021}). For example, when $G=\bZ/2$ and $\phi=\id$, the right hand side is $\bZ/8$ and is generated by real Clifford algebras.
    The general statement scoping the $\phi$-twisted version is found in \cite{kubotaNotesTwistedEquivariant2016}*{Section 2}. The extension group $\mathrm{Ext}(G,\mathbb{T}_\phi)$ is isomorphic to the second cohomology group $\mathrm{H}^2(G \,; \bT_{\phi})$.

    The cohomology groups $\mathrm{H}^n(G \,; \bT_{\phi})$ appeared above specifically refer to the `continuous' cohomology group in the sense of Segal and Wigner \cites{segalCohomologyTopologicalGroups1970,wignerAlgebraicCohomologyTopological1970} (\cite{tuGroupoidCohomologyExtensions2006} is also a nice reference), that is, the cohomology group of the simplicial sheaf of continuous $\bT$-valued functions over the simplicial topological space $B_{\bullet} G$ twisted by $\phi$. 
\end{rmk}

\begin{rmk}\label{rmk:comparison.equivariant.cohomology}
    \cref{prp:0d.G,prp:Brauer.1d,rmk:equivariant.cohomology} shows that the domain and the range of the SPT index in \cref{defn:G.index} are isomorphic. 
    More strongly, we can show that these isomorphisms are given by $\mathrm{Ind}_{G,\phi}^d$ when $d \leq 1$. Indeed, as a topological analogue of \cref{prp:0d.G,prp:Brauer.1d}, we can define isomorphisms
    \begin{align*}
        \cI_{\sM,\phi}^0 \colon \prescript{\phi}{}{\rIP}^0(\sM) \to \mathrm{Line}(\sM,\phi), \quad \cI_{\sM,\phi}^1 \colon \prescript{\phi}{}{\rIP}^1(\sM) \to \mathrm{Br}(\sM,\phi)
    \end{align*}
    for any $\sM \in \Man $ and $\phi \in \mathrm{H}^1(\sM\,;\bZ/2)$. Here, $\mathrm{Line}(\sM,\phi)$ (resp.\ $\mathrm{Br}(\sM,\phi)$) denotes the group of isomorphisms classes of complex line bundles (resp.\ bundles of type I factors) on $\sM_\phi$ on which the deck transformation by $\bZ/2$ acts antilinearly. 
    Indeed, $\cI_{\sM,\phi}^0([\sfH])$ is the pull-back of the tautological line bundle over $\colim_{\Lambda}  \bP\cA_{\Lambda} \simeq \IP_0$. Similarly, $\cI_{\sM,\phi}^1$ is given by
    \[
    \cI_{\sM,\phi}^1([\sfH]) \coloneqq \Big[ \bigsqcup_{p \in \sM}\pi_{\omega_{\sfH(p)}}(\cA_{Y_L})''\Big] \in \mathrm{Br}(\sM,\phi).
    \]
    In fact, as detailed in the forthcoming paper \cite{kubotaDixmierDouadyTheoryHigher}, this forms a bundle of type I factors. 
    For a general space like $BG$, we can define $\cI_{BG, \phi}$ by taking the colimit. 
    They eventually give the following commutative diagrams 
    \begin{align*}
        \xymatrix{
        \prescript{\phi}{}{\rIP}_G^0(\pt) \ar[r]^{\cI_{G,\phi}^0}_{\cong} \ar[d]^{\mathrm{Ind}_G^0} & \mathrm{Line}(G,\phi) \ar[d]^{\pr_{EG}^*}_{\cong} \\
        \prescript{\phi}{}{\rIP}^0(BG) \ar[r]^{\cI_{BG,\phi}^0}_{\cong}  & \mathrm{Line}(BG,\phi),  \\
        }
        \quad 
        \xymatrix{
        \prescript{\phi}{}{\rIP}_G^1(\pt) \ar[r]^{\cI_{G,\phi}^1}_{\cong} \ar[d]^{\mathrm{Ind}_G^1} & \mathrm{Br}(G,\phi) \ar[d]^{\pr_{EG}^*}_{\cong} \\
        \prescript{\phi}{}{\rIP}^1(BG) \ar[r]^{\cI_{BG,\phi}^1}_{\cong}  & \mathrm{Br}(BG,\phi). \\
        }
    \end{align*}
    The same holds for the fermionic versions.
\end{rmk}

\begin{rmk}\label{cor:equivariant.DW.cohomology}
    It will be proved in the subsequent paper \cite{kubotaStableHomotopyTheory2025b} that the lattice Dijkgraaf--Witten model given in \cites{chenSymmetryProtectedTopological2013,ogataClassificationPureStates2021} is generalized to a morphism of twisted $G$-equivariant cohomology functors 
    \[
    \prescript{\phi}{}{\mathrm{DW}}_{d}^G \colon \mathrm{H}^{d+1}_G(\sM  \,; \underline{\bT}{}_\phi) \to \prescript{\phi}{}{\rIP}^{d+2}_G(\sM)
    \]
    such that the following diagram commutes;
   \begin{align*}
   \xymatrix{
        \mathrm{H}^{d+1}_G(\sM\, ;\underline{\bT}{}_{\phi}) \ar[r]^{\DW_d^G} \ar[d]^{\mathrm{pr}_{EG}^*}_{\cong }  &  \prescript{\phi}{}{\rIP}_G^d(\sM ) \ar[d]^{\mathrm{Ind}_{G,\phi}} \\
        \mathrm{H}^{d+1}(\sM \times_G EG \, ; \underline{\bT}{}_{\phi})\ar[r]^{\DW_d } & \prescript{\phi}{}{\rIP}^{d}(\sM \times_G EG). 
   }
   \end{align*}
   In particular, when $d=2$, the map $\DW_d$ below is an isomorphism by \cref{prp:2d}. This means that  $\mathrm{Ind}_{G,\phi}$ has a left inverse, and hence $\prescript{\phi}{}{\rIP}_G^d(\sM )$ contains $ \mathrm{H}^{d+1}_G(\sM\, ;\underline{\bT}{}_{\phi})$ as a direct summand. 
\end{rmk}

\begin{exmp}\label{exmp:list}
    The results on invertible SPT phases listed here are understood as a consequence of \cref{cor:spectrum} and the discussion in this section. Here, the set of internal degrees of freedom $\fR$, which is abbreviated, is assumed to be $0$-ample, $1$-ample, and to satisfy the assumption of \cref{lem:equivariant.spectrum}. 
    \begin{enumerate}
        \item When $d =1$, the SPT index is a group isomorphism
        \begin{align*}
        \mathrm{Ind}_G^1 \colon \mathrm{SPT}_1(G) = \rIP^1_G(\pt) \to  \rIP^1(BG) \cong \mathrm{H}^3(BG \, ;\bZ)
        \end{align*}
        by \cref{rmk:comparison.equivariant.cohomology}.  
        Such a homomorphism is defined by Ogata \cite{ogataClassificationGappedHamiltonians2019} and Kapustin--Sopenko--Yang \cite{kapustinClassificationInvertiblePhases2021} when $G$ is finite, and extended to general compact group by Kapustin--Sopenko \cite{kapustinAnomalousSymmetriesQuantum2024} and Carvalho--de Roeck--Jappens \cite{carvalhoClassificationSymmetryProtected2024}. 
        The isomorphism is proved in \cite{ogataClassificationPureStates2021} for finite $G$, and in \cite{carvalhoClassificationSymmetryProtected2024} for general $G$. 
        \item The SPT index for a $G$-Thouless pump, i.e., a based loop of $G$-invariant IG UAL Hamiltonians, is a group homomorphism
        \[
        \mathrm{Ind}_{G,(S^1,\pt)}^1 \colon \rIP^1_G(S^1,\pt) \to \rIP^1(S^1 \wedge BG_+) \cong \rIP^0(BG) \cong \mathrm{H}^2(BG\,; \bZ). 
        \]
        Such a topological invariant is defined by 
        Bachmann--De Roeck--Fraass--Jappens \cite{bachmannClassificationGchargeThouless2023}.
        The map $\mathrm{Ind}_{G,(S^1,\pt)}^1$ is identical to the composition
        \begin{align*} 
        \mathrm{Ind}_G^0 \circ \vartheta_1 \colon \rIP_G^1(S^1,\pt) \to \rIP_G^0(\pt) \to \rIP^0(BG) \cong \mathrm{H}^2(BG \, ; \bZ ),
        \end{align*}
        and hence is an isomorphism by \cref{cor:fermionic.equivariant.spectrum} and \cref{rmk:comparison.equivariant.cohomology}. 
        \item The $2$-dimensional $G$-Thouless pump is considered similarly. The topological invariant is obtained by  
        \begin{align*} 
        \mathrm{Ind}_{G,(S^1,\pt)}^2 \colon \rIP_G^2(S^1,\pt) \to \rIP^2(S^1 \wedge BG_+) \cong \rIP^1(BG) \cong \mathrm{H}^3(BG\,;\bZ). 
        \end{align*}
        It is identical to the composition $\mathrm{Ind}_G^1 \circ \vartheta_2$, and hence is an isomorphism by \cref{cor:fermionic.equivariant.spectrum} and \cref{rmk:comparison.equivariant.cohomology}.
        \item The above (1)--(3) can be generalized to the case that $\phi$ is non-trivial. We obtain topological invariants
        \begin{align*}
            \mathrm{Ind}_{G,\phi}^1 \colon {}&{}\prescript{\phi}{}{\rIP}^1_G(\pt) \to   \mathrm{H}^3(BG \, ;\bZ_{\phi}),\\
            \mathrm{Ind}_{G,\phi,(S^1,\pt)}^d \colon {}&{} \prescript{\phi}{}{\rIP}^d_G(S^1,\pt) \cong \prescript{\phi}{}{\rIP}^{d-1}_G (\pt ) \to \mathrm{H}^{d+1}(BG\,; \bZ_\phi) 
        \end{align*}
        for $d=1,2$, which are isomorphisms by \cref{rmk:comparison.equivariant.cohomology}. A typical example is the case of $G=\bZ/2$ and $\phi = \id$, namely, the SPT phases protected by the time-reversal symmetry. In this case, the SPT index gives an isomorphism $\mathrm{SPT}(G,\phi)\to \bZ/2$.
        Such a topological invariant is defined by Ogata \cite{ogataValuedIndexSymmetryprotected2021}.  
        \item The fermionic version of (1)-(4) is considered as well. For example, the fermionic SPT index is a group isomorphism
            \begin{align*}
            \mathrm{Ind}_{G,\phi}^1 \colon \prescript{\phi}{}{\rfIP}^1_G(\pt) \to \prescript{\phi}{}{\rfIP}^1(BG) \cong \bZ/2 \oplus \mathrm{H}^1(G\,; \bZ/2) \oplus \mathrm{H}^3(G\,; \bZ_\phi )
        \end{align*}
        by \cref{rmk:comparison.equivariant.cohomology}. 
        Such an invariant is studied by Bourne--Ogata \cite{bourneClassificationSymmetryProtected2021}. In particular, when $G=\bZ/2$ and $\phi=\id$, it realizes the $8$-fold classification of the Majorana fermion by Fidkowski--Kitaev \cite{fidkowskiTopologicalPhasesFermions2011}. 
        Tasaki \cite{tasakiRigorousIndexTheory2023} also gives an explicit description of the topological invariant in the presence of charge-conjugation and particle-hole symmetries. 
        \item When $d =2$, the SPT index restricts to a group homomorphism
        \begin{align*}
        \mathrm{Ind}_{G,\phi}^2 \colon {}&{} \mathrm{SPT}_2(G,\phi)  \to \prescript{\phi}{}{\rIP}^2(BG, \pt) \cong \mathrm{H}^4(BG\,; \bZ_{\phi}),
        \end{align*}
        where the isomorphism at the right hand side is given by \cref{prp:2d}. Such a topological invariant is defined by Ogata \cite{ogataValuedIndexSymmetryprotected2021} and Sopenko \cite{sopenkoIndexTwodimensionalSPT2021}. Moreover, we obtain that it is split surjective by \cref{cor:equivariant.DW.cohomology}. 

        The fermionic case is defined similarly. If $\phi$ is trivial, then the SPT index is
        \begin{align*}
            \mathrm{Ind}_{G}^2 \colon{}&{} \mathrm{fSPT}_2(G)  \to \rfIP^2(BG, \pt) \cong \mathrm{SH}^4(BG).
        \end{align*}
        Such a topological invariant is studied by Ogata  \cites{ogataInvariantSymmetryProtected2022,ogata2DFermionicSPT2023}.
        The invariant constructed there takes value in a set of group cochains, which is larger than the supercohomology group. 
        Our SPT index takes value precisely in the group originally proposed by \cite{guSymmetryprotectedTopologicalOrders2014} (cf.\ the remark below \cref{prp:2d}). 
        \item In \cref{exmp:reflection.phase,exmp:reflection.onsite}, we deal with the invertible phase of reflection invariant systems in dimension $1$. The former example recovers the result of Ogata \cite{ogataIndexSymmetryProtected2021}.   
        \item Let $d$ be arbitrary and let $G=U(1)$. Then, for $\sM \in \Man $ with trivial $G$-action, the rational SPT index takes value in
        \begin{align*}
        \ch_{\mathrm{top}} \circ \mathrm{Ind}_{U(1)}^d \colon \rIP^d_{U(1)}(\sM) \to \rIP^d(\sM \times BU(1)) \to  \mathrm{H}^{d+2}(\sM \times BU(1) \, ;\bQ).
        \end{align*}
        The cohomology group at the right hand side is isomorphic to $\bigoplus_{k \geq 0}\mathrm{H}^{d+2-2k}(\sM\,; \bQ)$, which looks similar to the rational connective K-group of $\sM$. 
        When $d=2$, this map is related to the magnetic flux insertion in \cite{kapustinLocalNoetherTheorem2022}*{Section IV.D.2}, which is considered in  \cites{bachmannQuantizationConductanceGapped2018,bachmannManybodyIndexQuantum2020} as a description of the integer quantum Hall effect in terms of interacting Hamiltonians. 
    \end{enumerate}
\end{exmp}

We remark that our framework does not cover recent developments on the Lieb--Mattis--Schulz theorem \cites{ogataLiebSchultzMattisTypeTheorems2019,ogataGeneralLiebSchultzMattisType2021} since we do not allow on-site symmetry to be a projective representation.

\section{Invertible gapped localization flows}\label{section:localizing.path}
This section gives an explicit model of the generalized homology functor $\rIP_n \colon \kTop \to \mathsf{Ab}$ associated with the $\Omega$-spectrum $\IP_*$. 
For a proper metric space $X$ that admits a bi-Lipschitz embedding into a Euclidean space, we define an $\Omega$-spectrum sheaf $\sIP_{\loc ,d}(X \midbar \blank)$ of smooth families of certain semi-infinite paths of IG UAL Hamiltonians, which we call invertible gapped localization flows, placed on $X$. 
We then show that $\rIP_{\loc,n}(X)\coloneqq \pi_n(\sIP_{\loc}(X \midbar \blank))$ satisfies the Eilenberg--Steenrod axiom of generalized homology theory. 
Finally, we prove that the resulting homology functor $\rIP_{\loc,n}$ is naturally isomorphic to the abstractly defined homology functor $\rIP_n(X)$. 

\subsection{Lipschitz continuous coarse geometry of proper metric spaces}\label{subsection:localizing.path}
Consider placing a quantum spin system on a general proper metric space beyond $\bR^d$. 
The notion of growth rate for non-discrete metric spaces is defined in \cite{blockAperiodicTilingsPositive1992} in terms of its quasilattice (we refer to \cite{alvarezlopezGenericCoarseGeometry2018}*{Definition 4.2}). 
For $R>0$ and a monotonously increasing function $Q \colon \bR_{>0} \to \bR_{>0}$, a subset $L$ of a  proper metric space $X$ is said to be an $(R,Q)$-quasilattice if it is relatively dense and $\# (B_r(\bm{x}) \cap L) \leq Q(r)$ for any $r >0$ and $\bm{x} \in X$. 
A proper metric space $X$ having a quasilattice is said to have bounded geometry. 
We say that $X$ has polynomial growth if the function $Q$ can be chosen as a polynomial. 
We remark that any two quasilattices $L_1, L_2$ of $X$ are quasi-isometric (\cite{alvarezlopezGenericCoarseGeometry2018}*{Proposition 4.1}). 

Let $\iota \colon \Lambda \to X$ be a weakly uniformly discrete map in the sense of \cref{subsection:coarse.lattice}, which induces a pseudo-metric on $\Lambda$.
Contrary to what one might expect, the polynomial growth property of $X$ does not ensure that $\Lambda$ also has polynomial growth. 
For example, the disjoint union of open balls with increasing dimension $\bigsqcup_{n \in \bN} B_{2}^{\bR^n}(\bm{x}_n)$, where the centers $\bm{x}_n$ are placed at regular intervals of distance $\gg 1$, has a $(2, Q)$-quasilattice $L \coloneqq \bigsqcup_{n \in \bN} \{\bm{x}_n\}$ that is quasi-isometric to $\bN$, and hence has polynomial growth.
However, if $\Lambda$ is the disjoint union of the standard bases $\{\bm{e}_1^{(n)},\cdots,\bm{e}_n^{(n)}\} \subset B_2^{\bR^n}(\bm{0}) \cong B_{2}^{\bR^n}(\bm{x}_n)$, then it is $(1/4,1)$-uniformly discrete in $X$, but it does not admit a growth function; $\sup_{\bm{x} \in \Lambda}\# B_{2}(\bm{x})= \infty$.
The problem is caused by our assumption on the existence of a $(R,Q)$-quasilattice for \emph{some} $R>0$. 
Indeed, the following holds. 
\begin{lem}\label{lem:equi.polynomial.growth.inherit}
    Let $X$ be a proper metric space that has an $(R,Q)$-quasilattice $L$, where $Q$ is a polynomial. Let $\iota \colon \Lambda \to X$ be an $(R,N)$-weakly uniformly discrete map by the same $R>0$. Then, $\Lambda$ also has polynomial growth with respect to the pseudo-metric induced by $\iota$.  
\end{lem}
\begin{proof}
    Take a map $\pi \colon \Lambda \to L$ such that $\rmd(\iota(\bm{x}),\pi(\bm{x}))$ attains $\mathrm{dist}(\iota(\bm{x}),L)$ for any $\bm{x} \in \Lambda$. Since $L$ is $R$-relatively dense, we have $\pi^{-1}(\bm{x}) \subset  \iota^{-1} (B_R(\iota(\bm{x}))) $. By assumption on $\Lambda$, this shows $\# \pi^{-1}(\bm{x}) \leq N$. 
     Now, the growth rate function for $\Lambda$ is bounded above by a polynomial $N \cdot Q(r+R)$. 
\end{proof}

This observation leads us to the following refined definition of the polynomial growth property of a proper metric space. 
Here, we say that a family of sets $\{ \Lambda_p\}_{p \in \sM}$ is a lower-hemicontinuous family over $X$ if the sets $\Lambda_p$ are all contained in a single set $\bigcup_p \Lambda_p$ equipped with a map $\bigcup_p\iota_p\colon \bigcup_p \Lambda_p \to X$ and, for any $p \in \sM$, there is an open neighborhood $\sU \ni p$ such that $\Lambda_{p}\subset \Lambda_{q}$ if $q \in \sU$.

\begin{defn}\label{defn:equi.coarse}
Let $X$ be a  proper metric space. 
    \begin{enumerate}
        \item We say that $X$ has \emph{equi-bounded geometry} if there is $R>0$ and a monotonously increasing function $Q \colon \bR_{\geq 0} \to \bR_{\geq 0}$ such that there is an $(s^{-1}R,Q(s \cdot \blank))$-quasilattice for \emph{any} $s\in [1,\infty)$. 
        \item We say that $X$ has \emph{equi-polynomial growth} if it is equi-bounded geometry and the growth rate function $Q$ is taken to be a polynomial. 
        \item We say that a lower-hemicontinuous family $\mathbbl{\Lambda}\coloneqq\{ \Lambda _s \}_{s \in [1,\infty)} $ of discrete proper metric spaces over $X$ is $(r,N)$-\emph{weakly equi-uniformly discrete} in $X$ if each $\iota_s \colon \Lambda_s \to X$ is $(s^{-1}r,N)$-uniformly discrete for any $s \in [1,\infty)$. 
        \item We say that a family $\mathbbl{\Lambda}\coloneqq\{ \Lambda _s \}_{s \in [1,\infty)}$ of discrete proper metric spaces has \emph{equi-polynomial growth} if there is $\kappa_{\mathbbl{\Lambda}} >0$, $l_{\mathbbl{\Lambda}} >0$ such that $Q_{\mathbbl{\Lambda}}(r) \coloneqq \kappa_{\mathbbl{\Lambda}} (1+r)^{l_{\mathbbl{\Lambda}}}$ satisfies $\# B_{r}^{\Lambda_s}(\bm{x}) \leq Q_{\mathbbl{\Lambda}}(s^{-1}r)$ for any $s\in [1,\infty)$ and $\bm{x} \in \Lambda_s$. 
        \item We say that a family of maps $F_s \colon \Lambda_s \to X$ is \emph{equi-linearly proper} if there is a function $\varphi(r)=m(1+r)$ such that $\mathrm{diam} (F_s^{-1}(B_{r}(\bm{x}))) \leq\varphi(sr)$ for any $s \in [1,\infty)$. 
    \end{enumerate}
\end{defn}
For a proper metric space $X$ and $s \in[1,\infty)$, we write $\rmd_s$ for the rescaled metric on $X$;
\[
    \rmd_s(\bm{x},\bm{y}) \coloneqq s \cdot \rmd(\bm{x},\bm{y}). 
\]
Let $\Box X$ denote the box space of the rescaled copies of $X$, i.e., the disjoint union $\bigsqcup _{s \in [1,\infty)} (X,\rmd_{s})$ placed so that the different components are infinitely distant. 
By definition, $X$ has uniformly bounded geometry (resp.\ uniformly polynomial growth) if and only if $\Box X$ has bounded geometry (resp.\ polynomial growth). 
Indeed, if $\{L_s\}_{s \in [1,\infty)}$ is a family of $(s^{-1}R,Q(s\cdot \blank))$-quasilattices of $X$, then $\bigsqcup_{s \in [1,\infty)} L_{s} \subset \Box X$ is a quasilattice of $\Box X$. 
Similarly, $F_s$ is equi-linearly proper if and only if $\bigsqcup_{s \in [1,\infty)}F_{s} \colon \bigsqcup_{s} \Lambda_{s} \to \Box X$ is linearly proper.
In particular, if a Lipschitz continuous map $F \colon X \to Y$ is linearly proper in the sense of \cref{defn:polynomially.proper} and if $\mathbbl{\Lambda}$ is weakly equi-uniformly discrete in $X$, then the family $F|_{\Lambda_s} \colon \Lambda_s \to Y$ is equi-linearly proper.

By \cref{lem:equi.polynomial.growth.inherit}, we obtain the following implication of equi-polynomial growth property. 
\begin{lem}\label{lem:equi.polynomial.growth.inherit.2}
    Let $X$ be a  proper metric space with equi-polynomial growth in the sense of \cref{defn:equi.coarse} (1). Let $\mathbbl{\Lambda}$ be weakly equi-uniformly discrete in $X$. Then, $\mathbbl{\Lambda}$ has equi-polynomial growth in the sense of \cref{defn:equi.coarse} (3).
\end{lem}

\begin{exmp}\label{exmp:lattice.Rl}
    If $X$ is a closed subspace of $\bR^{l_X}$, then its $(s,Q_{\mathbbl{\Lambda}}(s)) $-quasilattice $\mathbbl{\Lambda}=\{ \Lambda_s\}_{s \in [1,\infty)}$ is obtained by sampling a point from each non-empty intersection $X \cap \prod_{i=1}^{l_X} [cn_i,cn_i+c]$, where $c \coloneqq s/\sqrt{l_X}$. 
    Its growth rate function is taken to be $Q_{\mathbbl{\Lambda}}(sr) \leq \kappa (1+sr)^{l_X}$ for some $\kappa >0$. That is, $X$ has equi-bounded geometry and equi-polynomial growth. 

    Moreover, since each $\Lambda_s$ is identified with a subset of the standard lattice, it has a brick $(\bB_s, B_s)$ in the sense of \cref{defn:brick}, whose growth rate introduced in \cref{lem:R(rho.x).growth} is taken to be $\kappa_{\bB}(1+sr)^{2l_{\mathbbl{\Lambda}}}$ uniformly on $s \in [1,\infty) $. 
    For this reason, we restrict the class of metric spaces $X$ that we consider in this paper to the closed subspaces of a Euclidean space. 
\end{exmp}

\begin{defn}\label{defn:MCW}     
    Let $\sfE$ denote the category whose objects are proper metric spaces $(X,\rmd)$ that admits a bi-Lipschitz embedding to some Euclidean space, and whose morphisms are linearly proper large-scale Lipschitz maps. 
    Let $\ECW$ denote its subcategory whose objects are pairs $(X,\rmd)$, where $X$ is a CW-complex equipped with a proper metric $\rmd$ such that 
    \begin{enumerate}
        \item there is $\cR_X>0$ such that $D \cap N_{\cR_X}(\partial D) \subsetneq D$ for any cell $D$ of $X$, and
        \item there is a bi-Lipschitz homeomorphism of $X$ and a closed subspace of $\bR^l$ for some $l>0$,
    \end{enumerate}
    and whose morphisms are linearly proper Lipschitz continuous maps. We write $\ECW_{\fin} \subset \ECW$ for the full subcategory consisting of finite CW-complexes. 
\end{defn}

\begin{rmk}\label{rmk:embed.CWfin.EucCW}
    A finite CW complex $X \in \CW_{\fin}$ can be embedded into a Euclidean space $\bR^n$ by inductively embedding each newly attached cell into additional dimensions (see e.g.\ \cite{hatcherAlgebraicTopology2002}*{Corollary A.10}). This embedding induces a metric on $X$ from that of $\bR^n$. 
    More strongly, this construction shows that $X$ is a Euclidean neighborhood retract (ENR), i.e., the above embedding  $\iota_X \colon X \to \bR^n$ satisfies that $\iota_X(X) \subset \bR^n$ has a deformation retract neighborhood $(U_X,r_X)$. Moreover, the deformation retract $r_X \colon U_X \times [0,1] \to U_X$ can be taken to be Lipschitz continuous. 
    This shows that, for $X,Y \in \CW_{\fin}$ equipped with an embedding as above, any continuous map $F \colon X \to Y$ is approximated by a Lipschitz continuous map $r_{Y,1} \circ F' \colon X \to U_Y \to Y$, which shows that the Lipschitz homotopy set $[X,Y]_{\mathrm{Lip}}$ is naturally isomorphic to the ordinary homotopy sets $[X,Y]$. 
    (The same fact also follows from \cite{mitsuishiGoodCoveringsAlexandrov2019}*{Corollary 1.3}.)
    This shows that there is a full subcategory $\CW_{\fin}^{\iota} \subset \ECW_{\fin}$ such that the forgetful functor $\CW_{\fin}^{\iota} \to \CW_{\fin}$ is essentially surjective and induces isomorphisms of homotopy sets of morphisms.
\end{rmk}

\begin{defn}\label{defn:cone.suspension}
    Let $X \in \ECW$. We call $\mathbf{C}X \coloneqq X \times \bR_{\leq 0}$ and $\boldsymbol{\Sigma} X \coloneqq X \times \bR$, equipped with the canonical CW-complex structure and the $\ell^2$-product metric, the \emph{cone} and the \emph{suspension} of the category $\ECW$ respectively. 
\end{defn}

For a linear function $\varphi (r)=m(1+r)$, the subspace 
\[
    \mathbf{I}_\varphi X \coloneqq \{ (\bm{x},y) \in X \times \bR \mid 0 \leq y \leq \varphi(\mathrm{d}(\bm{x},\bm{x}_0)) \}
\]
is also contained in the class $\ECW$ (cf.\  \eqref{eqn:coarse.homotopy}). 
For Lipschitz continuous maps $F_0, F_1 \colon X \to Y$, their \emph{Lipschitz continuous coarse homotopy} is given by a Lipschitz continuous map $\widetilde{F} \colon \mathbf{I}_{\varphi} X \to Y$ for some $\varphi$ such that $\widetilde{F} \circ \iota_i = F_i$. 

\begin{defn}[\cite{roeIndexTheoryCoarse1996}*{Definitions 9.3, 9.10}]\label{defn:coarse.flasque.scaleable}
    Let $X \in  \ECW$. 
    \begin{enumerate}
        \item We say that $X$ is \emph{topologically flasque} if there is a linearly proper Lipschitz continuous map $T_X \colon X \times [0,1] \to X$ such that $T_{X,0}=\id$, $T_{X,1}$ is $1$-Lipschitz, and for any bounded subset $B \subset X$ there is $n \in \bZ_{\geq 1}$ such that $\Im (T_{X,1}^{n}) \cap B = \emptyset$. 
        \item We say that $X$ is \emph{scaleable} if there is a linearly proper Lipschitz continuous map $S_X \colon X \to X $ such that $\rmd(S_{X}(\bm{x}),S_{X}(\bm{y})) \leq \rmd(\bm{x},\bm{y})/2$ for any $\bm{x},\bm{y} \in X$ and there is a coarse homotopy $\widetilde{S}_X \colon \mathbf{I}_\varphi X \to X$ of $\id_X$ and $S_X$ given by a linearly proper Lipschitz continuous map. 
    \end{enumerate}
\end{defn}
Note that our definition of scaleability is slightly different from the original one by Roe. 
The cone $\mathbf{C} X$ in \cref{defn:cone.suspension} is a motivating example of flasque metric spaces by $T_{X,t}(\bm{x},v)=(\bm{x},v+t)$. 
For a closed subspace $X \subset S^n$, its Euclidean cone $\cO X = \{ \bm{x} \in \bR^{n+1} \mid \bm{x}/\| \bm{x}\| \in X \}$ is scaleable by $S_{X} (\bm{x}) = \bm{x}/2$ (cf.\ \cref{exmp:Lipschitz.homotopy}). 
In particular, $\bR^d \times \bR_{\geq 0}$ is topologically flasque and $\bR^d $ is scaleable.

\subsection{The sheaf of IG localization flows of UAL Hamiltonians}
We first introduce the sheaves of smooth families of gapped and IG UAL Hamiltonians on a proper metric space $X$ in the class $\sfE$. This part is a straightforward generalization of \cref{subsection:sheaf.lattice}. 
\begin{defn}\label{defn:gapped.Hamiltonian.X}
Let $X \in \sfE$. 
\begin{enumerate}
    \item Let $\fL_X$ denote the set of subsets $\Lambda \subset X \times \bR^\infty \times \fR \times \bN$ satisfying the following conditions (cf.\ \cref{defn:lattice.set}):
    \begin{enumerate}
    \item[(i)] The set $\Lambda$ is contained in $X \times \bR^{l_\Lambda} \times \fR \times [0,n_\Lambda]$ for some $l_{\Lambda} \in \bN$ and $n_{\Lambda} \in \bN$. 
    \item[(ii)] The projection $\pr_{X \times \bR^{l_{\Lambda}}} \colon \Lambda \to X \times \bR^{l_{\Lambda}}$ is weakly uniformly discrete in the sense of \cref{subsection:coarse.lattice}, which imposes the pseudo-metric $\rmd$ on $\Lambda$.
    \item[(iii)] The map $\pr_{X} \colon \Lambda \to X$ is linearly proper in the sense of \cref{defn:polynomially.proper}. 
\end{enumerate}
    \item The sheaf $\sGP(X) = \sGP(X \midbar \blank)$ is defined by 
\begin{align*} 
    \sGP(X \midbar \blank) \coloneqq \colim_{\Lambda \in \fL_X} \sGP(\Lambda \midbar \blank). 
\end{align*}
    \item The composite operation $\blank \boxtimes \blank $ of lattices and gapped UAL Hamiltonians are defined in the same way as \eqref{eqn:composition.lattice} and \eqref{eqn:composition.system} respectively. 
    By the same argument as \cref{lem:H-monoid}, it follows that this $\boxtimes$ makes $\sGP(X)$ a local commutative H-monoid with a weakly strict unit in the sense of \cref{defn:local.commutative.Hmonoid}. 
    We define the sheaves $\sIP(X) \coloneqq \sGP(X)^\times $, and its refinement $\bsIP (X)$ in the sense of \cref{lem:bold.invertible.sheaf} (cf.\ \cref{eqn:IP.bold}). 
\end{enumerate}
\end{defn}

In the same way as \eqref{eqn:push.quantum.system}, we define the push-forward of quantum spin systems with respect to a linearly proper large-scale Lipschitz map. 
\begin{lem}\label{lem:MCW.induced.hom}
    Let $X,Y \in \sfE$ and let $F \colon X \to Y$ be a linearly proper large-scale Lipschitz map (that is not necessarily continuous). By using a fixed choice of a bi-Lipschitz embedding $\iota_X \colon X \to \bR^{l_X}$, let $\widetilde{F} \coloneqq F \times \iota_X \colon X \to Y \times \bR^{l_X}$. Then,  
    \begin{align*}
    F_*  \colon \sGP(X \midbar \sM) = \colim_{\Lambda \in \fL_X} \sGP(\Lambda \midbar \sM) \xrightarrow{\colim (\widetilde{F}|_{\Lambda})_*} \colim_{\Lambda \in \fL_X} \sGP(\widetilde{F}(\Lambda) \midbar \sM) \subset \sGP(Y \midbar \sM)
    \end{align*}
    gives a morphism of sheaves. Moreover, if two linearly proper large-scale Lipschitz maps $F_0$ and $F_1$ are near, then $F_{0,*}$ and $F_{1,*}$ are smoothly homotopic. 
\end{lem}
\begin{proof}
    It is proved in the same way as \cref{lem:induced.map.GP,lem:flip.homotopy.use}. It suffices to verify that any $\Lambda \in \fL_X$ is sent to $\widetilde{F}(\Lambda) \in \fL_Y$. The condition (i) is clear from the definition. The condition (ii) is valid since the composition of projections $\widetilde{F}(\Lambda) \to Y \times \bR^{l_X + l_{\Lambda}} \to \bR^{l_X +l_{\Lambda}}$, which coincides with $\pr_{X \times \bR^{l_{\Lambda}}}$, is weakly uniformly discrete. 
    Finally, the condition (iii) follows from the linear properness of $\pr_Y \circ \widetilde{F} = F \circ \pr_{X}$.
\end{proof}
\begin{lem}
    If $F_0, F_1 \colon X \to Y$ are connected by a linearly proper large-scale Lipschitz homotopy, then the induced morphisms $F_{0,*},F_{1,*} \colon \sIP(X) \to \sIP(Y)$ are smoothly homotopic.  
    In particular, for linearly proper large-scale Lipschitz maps $F \colon X \to Y$ and $G \colon Y \to Z$, the morphisms $(G \circ F)_*$ and $G_* \circ F_*$ are smoothly homotopic. That is, the assignment $\sIP(\blank ) \colon \sfE \to \Sh(\Man)$ is functorial up to homotopy. 
\end{lem}
\begin{proof}
    This lemma is proved in the same way as \cref{cor:coarse.homotopy.Hamiltonian}. 
    We only remark is that, even for an IG UAL Hamiltonian on $\Lambda \in \fL_X$, the truncated Kitaev pump given in \cref{lem:truncated.pump} is well-defined as a smooth path of IG UAL Hamiltonians on $\Lambda \times \bZ_{\geq 0}$. 
\end{proof}
This lemma shows that if $F \colon X \to Y$ is a quasi-isometry, then $F_*$ induces a weak equivalence of the sheaves $\sGP(X )$ and $\sGP(Y)$.
In fact, more strongly, $\sGP(X )$ and $\sGP(Y)$ are isomorphic as sheaves. It is verified by using a bijection of lattices instead of $\widetilde{F}$.
In particular, for a quasilattice $L \subset X$, the sheaf $\sGP(X )$ is isomorphic to $\sGP(L)$.

Next, we give a refined notion of topological phases over a proper metric space $X$ in the class $\ECW$, which reflects the local topology of $X$. 
\begin{defn}\label{defn:localizing.path}
Let $X \in \ECW$. 
\begin{enumerate}
    \item Let $\fL_X^{\loc}$ denote the set of lower-hemicontinuous family $\mathbbl{\Lambda}\coloneqq \{ \Lambda_s \}_{s \in [1,\infty)}$ of discrete subspaces of $X \times \bR^\infty \times \fR \times \bN^2$ such that the following hold:
    \begin{enumerate}
    \item[(i)] There is $l_{\mathbbl{\Lambda}} \in \bN$ and $n_{\mathbbl{\Lambda}} \in \bN$ such that, for any $s \in [1,\infty)$, there is $m_{\Lambda_s}\in \bN$ such that $\Lambda_s$ is contained in $X \times \bR^{l_{\mathbbl{\Lambda}}} \times \fR \times [0,m_{\Lambda_s}] \times [0,n_{\mathbbl{\Lambda}}]$. 
    \item[(ii)] The family of projections $\pr_{X \times \bR^{l_{\Lambda}}} \colon \Lambda_s \to X \times \bR^{l_{\Lambda}}$ is weakly equi-uniformly discrete in the sense of \cref{defn:equi.coarse} (3).
    \item[(iii)] The family of projections $\pr_{X} \colon \Lambda_s \to X$ is equi-linearly proper in the sense of \cref{defn:equi.coarse} (5). 
\end{enumerate}
    \item A smooth family of \emph{gapped localization flows} of UAL Hamiltonians on $\mathbbl{\Lambda} \in \fL_X^{\loc}$ parametrized by $\sM \in \Man$ is a family $\sfH(p \, ;s) \in \fH_{\Lambda_s}^{\al}$ such that, on each relatively compact chart $\sU$ of $\sM$, the following localizatin $C^k$-norms $\vvert \sfH \vvert _{\loc , \sU, C^k, f}$ are finite for any $f \in \cF$;
        \begin{align}
        \begin{split}
            \vvert \sfH \vvert _{\loc , \sU, C^k, f} \coloneqq   {}&{} \sup_{s \in [1,\infty)} \vvert \sfH  \vvert_{(s), \sU,C^k,f} , \\
            \vvert \sfH  \vvert_{(s), \sU,C^k,f} \coloneqq {}&{} \sup_{l\leq k}\vvert \partial_s^l \sfH(\blank \,; s)\vvert_{\sU,C^{k-l},f} \quad \text{with respect to the metric $\rmd_s$} \\
            ={}&{}
            \sup_{p \in \sU} \sup_{\bm{x} \in \Lambda_s } \sup_{i +|I| \leq k} \sup_{r >0} \Big(  f(sr)^{-1} \cdot  \| \partial_s^{i} \partial^{I} \sfH_{\bm{x}}(p\,; s) - \Pi_{\bm{x},r}(\partial_s^{i} \partial^{I} \sfH_{\bm{x}}(p\,; s)) \| \Big).
        \end{split}\label{eqn:localizing.norm}
        \end{align}
    \item We write $\sGP_{\loc}(\mathbbl{\Lambda}) = \sGP_{\loc}(\mathbbl{\Lambda} \midbar \blank)$ for the sheaf of smooth families of gapped localization flows supported on $\mathbbl{\Lambda}$ and let
    \begin{align*}
    \sGP_{\loc}(X)=\sGP_{\loc}(X \midbar \blank) \coloneqq \colim_{\mathbbl{\Lambda} \in \fL_X^{\loc}} \sGP_{\loc}(\mathbbl{\Lambda} \midbar \blank). 
    \end{align*}
    By regarding it as a local commutative H-monoid via the product in \cref{defn:gapped.Hamiltonian.X} (3), we define the subsheaf $\sIP_\loc(X)$ of IG localization flows and its refinement $\bsIP_{\loc}(X)$ in the sense of \cref{lem:bold.invertible.sheaf}. 
    Its realization denotes $\IP_{\loc}(X) \coloneqq | \Sing \bsIP_{\loc}(X)|$.
    \item We also define the degree-shift sheaf $\sIP_{\loc,d}(X)$ and its realization $\IP_{\loc,d}(X)$ by
    \begin{align*}
        \sIP_{\loc,d}(X) \coloneqq  \sIP_{\loc}(\boldsymbol{\Sigma}^dX), \quad \IP_{\loc,d}(X) \coloneqq |\Sing \bsIP_{\loc,d}(X)|.
    \end{align*}
\end{enumerate}
\end{defn}
Here, the reason for reinterpreting the $\bN$ infinite internal degrees of freedom as $\bN^2$ in (1) becomes clear during the proof of \cref{thm:scaleable}.

    Unless $X$ is discrete, a localization flow $\{ \sfH(p\,; s)\}_{s \in [1,\infty)}$ can not be expected to be supported on the same lattice. 
    The typical situation is that, the support of $\sfH(p\,; s)$ is an increasing sequence of lattices $ \{ \Lambda_s \}_{s \in [1,\infty)}$ that becomes finer as $s \to \infty$.   

We use the typescript font $\mathtt{ev}_s \colon \sIP_{\loc}(X) \to \sIP(X)$ for the evaluation of the localization parameter.
\begin{rmk}\label{rmk:parameter.shift}
    For a smooth family $\sfH(p\,; s)$ of localization flows parametrized by $p \in \sM$, its parameter shift 
\begin{align*}
    (\sigma_u \sfH) (p\,;s) \coloneqq \sfH (p \,; s+u)
\end{align*}
is also a smooth family of localization flows. 
\end{rmk}

\begin{prp}\label{prp:localizing.induced.hom}
    Let $X, Y \in \ECW$ and let $F \colon X  \to Y$ be a linearly proper Lipschitz continuous map. Then the induced map $F_* \colon \sGP(X \midbar  \blank \times \bR_{\geq 1}) \to \sGP(Y \midbar \blank \times \bR_{\geq 1})$ defined in \cref{lem:MCW.induced.hom} restricts to the morphism of subsheaves of localization flows as
    \begin{align*}
    F_* \colon \sGP_{\loc}(X ) \to \sGP_{\loc}(Y), \quad F_* \colon \sIP_{\loc}(X ) \to \sIP_{\loc}(Y). 
    \end{align*}
    Moreover, if two linearly proper Lipschitz continuous maps $F_0, F_1 $ satisfies $\pr_{\bR^{d+\infty}} \circ F_0 = \pr_{\bR^{d+\infty}} \circ F_1$, then $F_{0,*}$ and $F_{1,*}$ are smoothly homotopic.  
\end{prp}
\begin{proof}
    First, we observe that the image $\widetilde{F}(\mathbbl{\Lambda})$ of $\mathbbl{\Lambda} \in \fL_X^{\loc}$ is also contained in $\fL_{Y}^\loc$. 
    Indeed, the condition (i) of \cref{defn:localizing.path} (1) is clear. The condition (ii) is verified in the same way as \cref{lem:MCW.induced.hom}. 
    To see (iii), observe that the composition of an equi-linearly proper family $\pr_X \colon \Lambda_s \to X$ and a linearly proper map $F \colon X \to Y$ is again equi-linearly proper.

    The remaining task is to show $F_*\bsfH \in \sGP_{\loc}(Y \midbar \sM)$ for  $\sfH \in \sGP_{\loc}(X\midbar \sM)$. 
    This follows from \cref{lem:induced.map.GP}. 
    Indeed, since $F$ is Lipschitz continuous, the estimate \eqref{eqn:f-norm.push} holds with $\beta=0$. Hence we have $\vvert F_*\sfH \vvert_{(s),\sU,C^k,\mu} \leq 2c_{0,\mu,\alpha} \vvert \sfH \vvert_{(s),\sU,C^k,\mu_1}$ independent of $s \in \bR_{\geq 1}$. 
    
    The second claim follows from \cref{lem:flip.homotopy} in the same way as \cref{lem:flip.homotopy.use}. 
\end{proof}

In parallel to \cref{defn:localizing.path} (2), for $\mathbbl{\Lambda} \in \fL_X^{\loc}$, let 
\[
    C^\infty(\sM,\fLDer ^{\al}_{\mathbbl{\Lambda}}) \coloneqq  \{ (p,s) \mapsto \sfG(p\,; s) \in \fD_{\Lambda_s}^{\al} \mid \text{$\vvert \sfG \vvert_{\loc,\sU,C^k,f} <\infty$ for any $f \in \cF$, $k \in \bN$ and $\sU$ } \},
\]
where the almost local $C^k$-norms are defined in the same way as \eqref{eqn:localizing.norm}. The notion of differential forms taking value in $\fLDer ^{\al}_{\mathbbl{\Lambda}}$ is also defined similarly. 
\begin{lem}\label{lem:auto.localizing.path}
    Let $\mathbbl{\Lambda} \in \fL_X^{\loc}$. The following holds;
    \begin{enumerate}
        \item Let $\sfG \in \Omega^1(\sM,\fLDer^{\al}_{\mathbbl{\Lambda}})$ and $ \sfG'=\sfG'_t dt \in \Omega^1( \sM \times [0,1] , \fLDer^{\al}_{\mathbbl{\Lambda}})$. Then we have $\alpha_p(\sfG'\,; t)(\sfG(p)) \in \Omega^1(\sM, \fLDer_{\mathbbl{\Lambda}}^{\al})$. Similarly, for $\sfH \in \sGP_{\loc}(\mathbbl{\Lambda} \midbar \sM)$, we have $\alpha_p(\sfG'\,; t)(\sfH(p)) \in \sGP_{\loc}(\mathbbl{\Lambda} \midbar \sM)$. 
        \item If $\sfH \in \sGP_{\loc} (\mathbbl{\Lambda} \midbar  \sM)$, then the adiabatic connection $1$-form (\cref{defn:automorphic.connection}) defines $\sfG_{\sfH} \in \Omega^1(\sM, \fLDer ^{\al}_{\mathbbl{\Lambda}})$. 
    \end{enumerate}
\end{lem}
\begin{proof}
    The claims follow from \cref{prp:Lieb.Robinson}, in the same way as \cref{lem:smooth,lem:adiabatic.well-defined}. 
    The point is \cref{lem:equi.polynomial.growth.inherit.2,exmp:lattice.Rl}: A weakly equi-uniformly discrete subset $\mathbbl{\Lambda}$ of $X \times \bR^{l_{\mathbbl{\Lambda}}} \subset \bR^{l_X+l_{\mathbbl{\Lambda}}}$ has equi-polynomial growth and each $\Lambda_s$ has a brick whose growth rate is bounded by $\kappa_{\bB}(1+sr)^{2l_X+2l_{\mathbbl{\Lambda}}}$, where $l_X$ is the dimension of a Euclidean space in which $X$ admits a bi-Lipschitz embedding.
    For example, in (1), we have 
    \begin{align*}
     \|  \alpha_c(\sfG\,; t) (\sfG_{\bm{x}}' )\|_{(s),\bm{x},\nu,\mu} \leq \Upsilon_{(s),\sfG,\nu,\mu}(1) \cdot \vvert \sfG \vvert_{(s),\nu_1,\mu_1} .
    \end{align*}
    Here, for $s \in [1,\infty)$, the constant $\Upsilon_{(s),\sfG,\nu,\mu}(1)$ is the one defined in \cref{prp:Lieb.Robinson} (2) for $\mathtt{ev}_s \sfG$ with respect to the metric $\rmd_s$ of $X$, which is uniformly bounded since it depends only on the almost local norms of $\mathtt{ev}_s\sfG $.  
\end{proof}

\subsection{Localized Kitaev pump}
We define the Kitaev pump of IG localization flows.
\begin{lem}\label{lem:lem:flip.middle.path}
    Let $\sfH \in \sGP(\Lambda \midbar \sM)$ be a null-homotopic family of IG UAL Hamiltonians. 
    Then an arbitrary smooth homotopy $\widetilde{\sfH} \in \cP_{\sfH \boxtimes \sfH,\sfh\boxtimes \sfh} \sGP(\Lambda \boxtimes \Lambda \midbar \sM )$ of $\sfH \boxtimes \sfH$ and $\sfh \boxtimes \sfh$ is smoothly homotopic to its flip $\flip(\widetilde{\sfH})$ in the sheaf $\cP_{\sfH \boxtimes \sfH,\sfh\boxtimes \sfh} \sGP(\Lambda \boxtimes \Lambda \midbar \sM )$. 
\end{lem}
\begin{proof}
First, by a smooth reparametrization, we may assume that $\widetilde{\sfH}$ is constant on $\sM \times [0,1/3]$ and $\sM \times [2/3,1]$. 
From the smooth homotopy $\sfH' \in \cP\sGP(\Lambda \midbar \sM )$ connecting $\sfH$ and $\sfh$, we get a smooth family of LG automorphisms 
\[
    \alpha_{p,t} \coloneqq \alpha_p(\sfG_{\sfH'} \,; t)^{-1} \colon \sM \times [0,1] \to \cG_{\Lambda} .
\]
By \cref{thm:automorphic.equivalence}, it satisfies $\alpha_{p,0} = \id$ and $\omega_{\sfH(p)} \circ \alpha_{p,1}= \omega_0$.

Let $\widetilde{\flip}{}^s_*$ be the smooth homotopy of local automorphisms defined in \cref{lem:flip.homotopy}. 
Consider the smooth family
\begin{align*} 
    \mathtt{FL}_s \widetilde{\sfH}(p,t) \coloneqq (\alpha_{\min \{ 1, 3-3t\}} \otimes \alpha_{\min \{ 1, 3-3t\}}) \circ \widetilde{\flip}{}^{s  \max \{ 1,3t\}}_* \circ (\alpha_{\min \{ 1, 3-3t\}}^{-1} \otimes \alpha_{\min \{ 1, 3-3t\}}^{-1})(\widetilde{\sfH} (p,t))
\end{align*}
defined on $\sM \times [1/3,1]_t \times [0,1]_s$.  
Then we have 
\begin{align*}
    \mathtt{FL}_{1}\widetilde{\sfH}(p,t) ={}&{} (\alpha_{\min \{ 1, 3-3t\}}\otimes \alpha_{\min \{ 1, 3-3t\}}) \circ \flip \circ (\alpha_{\min \{ 1, 3-3t\}}^{-1} \otimes \alpha_{\min \{ 1, 3-3t\}}^{-1})(\widetilde{\sfH} (p,t))\\
    ={}&{} \flip \widetilde{\sfH}(p,t)
\end{align*}
and, by $\widetilde{\flip}{}^s(\sfh) = \sfh$, 
\begin{align*}
    \omega_{\mathtt{FL}_{1}\widetilde{\sfH}(p,1/3)} ={}&{} (\omega_{\sfH} \otimes \omega_{\sfH}) \circ (\alpha_{1}^{-1}\otimes \alpha_{1}^{-1}) \circ \widetilde{\flip}{}^s_* \circ (\alpha_{1} \otimes \alpha_{1}) \\
    ={}&{} (\omega_0 \otimes \omega_0) \circ \widetilde{\flip}{}^s_* \circ (\alpha_{1} \otimes \alpha_{1}) = \omega_{\widetilde{\sfH}(p,1/3)}. 
\end{align*}
Therefore, it extends to a smooth homotopy $\mathtt{FL}_{s} \widetilde{\sfH}$ on $\sM \times [0,1/3]_t \times [0,1]_s$ by the convex combination as
\begin{align*}
    \mathtt{FL}_{s} \widetilde{\sfH} (p,t) \coloneqq \max \{ 0,1-3t\} \cdot  \sfH \boxtimes \sfH + \min \{ 1,3t \} \cdot  \mathtt{FL}_s \widetilde{\sfH}(p,1/3).
\end{align*}
By smoothing $\mathtt{FL}_s \widetilde{\sfH}$ on $\sM \times [0,1]_t \times [0,1]_s$ as in  \cref{rmk:smoothing.path}, it gives a smooth homotopy of $\mathtt{FL}_0\widetilde{\sfH}=\widetilde{\sfH}$ and $\mathtt{FL}_1\widetilde{\sfH}= \flip \widetilde{\sfH}$ as desired. 
\end{proof}

For $\mathbbl{\Lambda} \in \fL_{X}^\loc$, we define $\mathbbl{\Sigma}\mathbbl{\Lambda} = \{ \Sigma\Lambda_s \}_{s \in [1,\infty)} \in \fL_{\mathbbm{\Sigma}X}^\loc$ by
    \begin{align*}
    \Sigma\Lambda_s \coloneqq \Lambda_s \times \frac{1}{5^{\lceil s \rceil -1 } } \bZ  \subset (X \times \bR^{l_{\mathbbl{\Lambda}}} \times \fR \times \bN ) \times \bR \cong \boldsymbol{\Sigma}X \times \bR^{l_{\mathbbl{\Lambda}}} \times \fR \times \bN .
    \end{align*}
    Here, $\lceil s \rceil$ denotes the minimal integer that is not less than $s$. 
For an interval $I \subset \bR$, each of two ends is either open or closed, we write $\mathbbl{\Sigma}^I\mathbbl{\Lambda}$ for the family of lattices $\Sigma^I\Lambda_s \coloneqq \Sigma\Lambda_s \cap I$ of $X \times I$. 

\begin{lem}\label{lem:spatial.localizing.path}
    Let $(\sfH, \check{\sfH},\overline{\sfH}) \in \bsIP_{\loc }(X\midbar \sM)$ which is supported on $\mathbbl{\Lambda} \in \fL_X^{\loc}$.
    Then, there are smooth families of IG localization flows $\sfH^{[0,1)} $, $\check{\sfH}^{[1,2)} $ and $\overline{\sfH}{}^{[0,2)} $, which is supported on $\mathbbl{\Sigma}^{[0,1)}\mathbbl{\Lambda}$, $\mathbbl{\Sigma}^{[1,2)}\mathbbl{\Lambda}$, and $\mathbbl{\Sigma}^{[0,2)}\mathbbl{\Lambda}$ respectively, such that 
    \begin{enumerate}
        \item $\overline{\sfH}{}^{[0,2)}(p,0\,; s) = \sfH^{[0,1)}(p\,; s) \boxtimes \check{\sfH}{}^{[1,2)} (p\,; s)$, $\overline{\sfH}{}^{[0,2)}(p,1\,; s) = \sfh$ and
        \item $\overline{\sfH}{}^{[0,2)} (p,t\,; 1) = \overline{\sfH}(p,t)$ under the canonical identification of $\Lambda_1^{[0,2)}$ and $\Lambda_1 \sqcup \Lambda_1$. 
    \end{enumerate}
\end{lem}
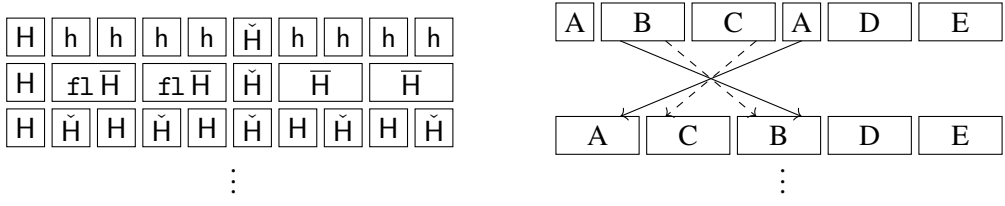
\begin{figure}[t]
    \centering
    \begin{minipage}[b]{0.45\hsize}
        \begin{tikzpicture}
        \draw (0,0) rectangle (0.5,0.5);
        \draw (0.6,0) rectangle (1.1,0.5);
        \draw (1.2,0) rectangle (1.7,0.5);
        \draw (1.8,0) rectangle (2.3,0.5);
        \draw (2.4,0) rectangle (2.9,0.5);
        \draw (3.0,0) rectangle (3.5,0.5);
        \draw (3.6,0) rectangle (4.1,0.5);
        \draw (4.2,0) rectangle (4.7,0.5);
        \draw (4.8,0) rectangle (5.3,0.5);
        \draw (5.4,0) rectangle (5.9,0.5);
        \node at (0.25,0.25) { $\sfH$};
        \node at (0.85,0.25) { $\sfh$};
        \node at (1.45,0.25) { $\sfh$};
        \node at (2.05,0.25) { $\sfh$};
        \node at (2.65,0.25) { $\sfh$};
        \node at (3.25,0.25) { $\check{\sfH}$};
        \node at (3.85,0.25) { $\sfh$};
        \node at (4.45,0.25) { $\sfh$};
        \node at (5.05,0.25) { $\sfh$};
        \node at (5.65,0.25) { $\sfh$};
        \draw (0,-0.1) rectangle (0.5,-0.6);
        \draw (0.6,-0.1) rectangle (1.7,-0.6);
        \draw (1.8,-0.1) rectangle (2.9,-0.6);
        \draw (3,-0.1) rectangle (3.5,-0.6);
        \draw (3.6,-0.1) rectangle (4.7,-0.6);
        \draw (4.8,-0.1) rectangle (5.9,-0.6);
        \node at (0.25,-0.35) { $\sfH$};
        \node at (1.15,-0.35) { $\flip \overline{\sfH}$};
        \node at (2.35,-0.35) { $\flip \overline{\sfH}$};
        \node at (3.25,-0.35) { $\check{\sfH}$};
        \node at (4.15,-0.35) { $\overline{\sfH}$};
        \node at (5.35,-0.35) { $\overline{\sfH}$};
        \draw (0,-0.7) rectangle (0.5,-1.2);
        \draw (0.6,-0.7) rectangle (1.1,-1.2);
        \draw (1.2,-0.7) rectangle (1.7,-1.2);
        \draw (1.8,-0.7) rectangle (2.3,-1.2);
        \draw (2.4,-0.7) rectangle (2.9,-1.2);
        \draw (3.0,-0.7) rectangle (3.5,-1.2);
        \draw (3.6,-0.7) rectangle (4.1,-1.2);
        \draw (4.2,-0.7) rectangle (4.7,-1.2);
        \draw (4.8,-0.7) rectangle (5.3,-1.2);
        \draw (5.4,-0.7) rectangle (5.9,-1.2);
        \node at (0.25,-0.95) { $\sfH$};
        \node at (0.85,-0.95) { $\check{\sfH}$};
        \node at (1.45,-0.95) { $\sfH$};
        \node at (2.05,-0.95) { $\check{\sfH}$};
        \node at (2.65,-0.95) { $\sfH$};
        \node at (3.25,-0.95) { $\check{\sfH}$};
        \node at (3.85,-0.95) { $\sfH$};
        \node at (4.45,-0.95) { $\check{\sfH}$};
        \node at (5.05,-0.95) { $\sfH$};
        \node at (5.65,-0.95) { $\check{\sfH}$};
        \node at (3.0,-1.55) {$\vdots$};
        \end{tikzpicture}
    \end{minipage}
    \begin{minipage}[b]{0.45\hsize}
        \begin{tikzpicture}
        \draw (0,0) rectangle (0.5,0.5);
        \draw (0.6,0) rectangle (1.7,0.5);
        \draw (1.8,0) rectangle (2.9,0.5);
        \draw (3,0) rectangle (3.5,0.5);
        \draw (3.6,0) rectangle (4.7,0.5);
        \draw (4.8,0) rectangle (5.9,0.5);
        \node at (0.25,0.25) { A};
        \node at (1.15,0.25) { B};
        \node at (2.35,0.25) { C};
        \node at (3.25,0.25) { A};
        \node at (4.15,0.25) { D};
        \node at (5.35,0.25) { E};
        \draw[->] (0.85,0) to (3.25,-1.0);
        \draw[->] (3.25,0) to (0.85,-1.0);
        \draw[dashed,->] (1.45,0) to (2.65,-1.0);
        \draw[dashed,->] (2.65,0) to (1.45,-1.0);
        \draw (0.0,-1.0) rectangle (1.1,-1.5);
        \draw (1.2,-1.0) rectangle (2.3,-1.5);
        \draw (2.4,-1.0) rectangle (3.5,-1.5);
        \draw (3.6,-1.0) rectangle (4.7,-1.5);
        \draw (4.8,-1.0) rectangle (5.9,-1.5);
        \node at (0.55,-1.25) {A};
        \node at (1.75,-1.25) {C};
        \node at (2.95,-1.25) {B};
        \node at (4.15,-1.25) {D};
        \node at (5.35,-1.25) {E};
        \node at (3.0,-1.75) {$\vdots$};
        \end{tikzpicture}
    \end{minipage}
    \caption{The picture of $\sfH^{[0,2)}$ (left) and the flip homotopy of $\overline{\sfH}{}^{[0,2)}$ at $s \in [3/2,2]$ (right). On each A, B, C, D, and E in the right picture, the smooth path $\overline{\sfH}$ is placed. }
    \label{fig:picture.local.pump}
\end{figure}
\begin{proof}
    First, we define $\sfH^{[0,1)}$ and $\check{\sfH}^{[1,2)}$. 
    When $\lceil s \rceil $ increases by $1$, the $[a,b)$-component of $\Lambda_s^{[a,b)}$ are subdivided into $1/5$. 
    We insert $\sfH \boxtimes \check{\sfH}$ at newly appeared $4$ copies of $\Lambda_s$ via the homotopy $\overline{\sfH}$. 
    More precisely, for $(p,s) \in \sM \times [n,n+1)$, let
    \begin{align*}
        \sfH^{[0,2)}(p\,; s) \coloneqq {}&{} \Big( \sfH(p\,; s) \boxtimes \flip \overline{\sfH}(p, s_n \,; s) \boxtimes \flip \overline{\sfH}(p, s_n \,; s) \boxtimes  \check{\sfH}(p\,; s) \boxtimes  \overline{\sfH} (p, s_n \,; s) \boxtimes \overline{\sfH}(p, s_n \,; s) \Big)^{\boxtimes 5^{n-1} },
    \end{align*}
    where $s_n \coloneqq \max \{ 0,2n+1-2s\}$,
    which are placed on $\Lambda_s^{\sqcup 2 \cdot 5^{n}} \cong \Lambda_s^{[0,2)}$. 
    We then define $\sfH^{[0,1)}$ and $\check{\sfH}{}^{[1,2)}$ by the restrictions of $\sfH^{[0,2)}$ to $\Sigma^{[0,1)}\Lambda_s$ and $\Sigma^{[1,2)}\Lambda_s$ respectively.    

    At $s \in [1,3/2]$, we define $\overline{\sfH}{}^{[0,2)}$ by using the leg notation on $\cA_{\Lambda_s^{[0,2)}} \cong \cA_{\Lambda_s}^{\otimes 10}$ as
    \begin{align*}
    \overline{\sfH}{}^{[0,2)} (p,t \,; s) \coloneqq 
    \overline{\sfH} (p,t \,; s)_{0,5} \boxtimes \overline{\sfH} (p,t_s \,; s)_{2,1} \boxtimes \overline{\sfH} (p,t_s \,; s)_{4,3} \boxtimes \overline{\sfH} (p,t_s \,; s)_{6,7} \boxtimes \overline{\sfH} (p,t_s \,; s)_{8,9},
    \end{align*}
    where $t_s = \min \{ 3-2s+t ,1\}$. At $s \in [3/2,2]$, we define a homotopy $\overline{\sfH}{}^{[0,2)} (p,t\,; s)$ of $\overline{\sfH}{}^{[0,2)} (p,t\,; 3/2)$ and 
    \begin{align*}
        \overline{\sfH}{}^{[0,2)} (p,t\,; 2) \coloneqq \overline{\sfH}{}^{[0,2)} (p,t \,; s)^{\boxtimes 5} = (\flip_{1,5} \otimes \flip_{2,4}) \big( \overline{\sfH}{}^{[0,2)} (p,t \,; 3/2) \big)
    \end{align*}
    by using the homotopy $\mathtt{FL}$  in \cref{lem:lem:flip.middle.path} as
    \begin{align*} 
    \overline{\sfH}{}^{[0,2)} (p,t\,; s) \coloneqq \big( (\mathtt{FL}_{2s-3})_{1,5} \otimes (\mathtt{FL}_{2s-3})_{2,4} \big) \big( \overline{\sfH}{}^{[0,2)} (p,t \,; 3/2) \big).
    \end{align*} 
    At $s \in \bR_{\geq 2}$, we extend $\overline{\sfH}{}^{[0,2)}(p,t\,; s)$ in a self-similar way. Finally, the above $\sfH^{[0,1)}$, $\check{\sfH}^{[1,2)}$, and $\overline{\sfH}{}^{[0,2)}$ are localization flows with respect to the metric induced from that of $X \times [0,w] \times \bR^{l_{\mathbbl{\Lambda}}}$.  
\end{proof}
We remark that $\sfH^{[0,1)}$ is of the form $\sfH \boxtimes  \sfH^{(0,1)}$, where $\sfH^{(0,1)}$ is supported on $\mathbbl{\Lambda}^{(a,b)} = \mathbbl{\Lambda}^{[a,b)} \setminus (\mathbbl{\Lambda} \times \{a \} )$. Similarly, $\check{\sfH}^{[1,2)}$ is also of the form $\check{\sfH} \boxtimes \check{\sfH}^{(1,2)}$, and $\check{\sfH}^{(1,2)}$ is the same as the reflection of $\sfH^{(0,1)}$ in the $\bR$-direction after a translation.

\begin{defn}\label{defn:localizing.Kitaev.pump}
    The localized Kitaev pump $\kappa_{X}^{\loc} \colon \bsIP_{\loc}(X) \to \Omega \sIP_{\loc}(\boldsymbol{\Sigma} X)$ is defined by
\begin{align*}
    \kappa_X^{\loc} (\bsfH)|_t = \left\{
        \begin{array}{rcccccccccccll}
            \cdots &\boxtimes & \sfh & \boxtimes & \sfh & \boxtimes & \sfh &\boxtimes & \sfh & \boxtimes &\cdots & t= 0, \\
            \cdots &\boxtimes & \multicolumn{3}{c}{\overline{\sfH}^{[-2,0)}|_{1-2t}}  & \boxtimes &\multicolumn{3}{c}{\overline{\sfH}{}^{[-0,2)}|_{1-2t}} & \boxtimes &\cdots & 0<t<1/2, \\
            \cdots &\boxtimes & \sfH{}^{[-2,-1)}  & \boxtimes & \check{\sfH}{}^{[-1,0)} & \boxtimes & \sfH{}^{[0,1)} & \boxtimes & \check{\sfH}{}^{[1,2)} & \boxtimes  & \cdots & t=1/2, \\            
            \cdots&\multicolumn{2}{c}
            {R_* \overline{\sfH}{}^{(-3,-1]}|_{2t-1}}
            & \boxtimes & \multicolumn{3}{c}
            {R_* \overline{\sfH}{}^{(-1,1]}|_{2t-1}} & \boxtimes &\multicolumn{2}{c}{R_* \overline{\sfH}^{(1,3]}|_{2t-1} }&  \cdots &  1/2<t<1, \\            
            \cdots &\boxtimes & \sfh & \boxtimes & \sfh & \boxtimes & \sfh &\boxtimes & \sfh & \boxtimes &  \cdots & t = 1,
        \end{array}
        \right.
\end{align*}
as is illustrated in \cref{fig:picture.kappa.local}. Here, $R \colon X \times [a,b) \to X \times (a,b]$ denotes the reflection in the second variable, i.e., $R(\bm{x},y) = (\bm{x},a+b-y)$.
\end{defn}

\begin{figure}[t]
    \centering
        \begin{tikzpicture}[scale=1.25]
        \fill[color = gray] (1.8,0.7) rectangle (1.9,-1.7);
        \fill[color = gray] (3.3,0.7) rectangle (3.4,-1.7);
        \fill[color = gray] (4.8,0.7) rectangle (4.9,-1.7);
        \fill[color = gray] (6.3,0.7) rectangle (6.4,-1.7);
        \fill[color = gray] (7.8,0.7) rectangle (7.9,-1.7);
        \fill[color = gray] (9.3,0.7) rectangle (9.4,-1.7);
        \fill[color = gray] (10.8,0.7) rectangle (10.9,-1.7);
        \fill[color = gray] (12.3,0.7) rectangle (12.4,-1.7);
        \draw[fill=white] (1.8,0) rectangle (4.7,0.5);
        \node at (3.25,0.25) { $\overline{\sfH}^{[-4,-2)}|_{1-2t}$};
        \draw[fill=white] (4.8,0) rectangle (7.7,0.5);
        \node at (6.35,0.25) { $ \overline{\sfH}^{[-2,0)}|_{1-2t}$};
        \draw[fill=white] (7.8,0) rectangle (10.7,0.5);
        \node at (9.25,0.25) { $ \overline{\sfH}^{[0,2)}|_{1-2t}$};
        \draw[fill=white] (10.8,0) rectangle (13.7,0.5);
        \node at (12.25,0.25) { $\overline{\sfH}^{[2,4)}|_{1-2t}$};
        \draw[fill=white] (1.8,-0.25) rectangle (3.2,-0.75);
        \node at (2.5,-0.5) {\scalebox{0.95}{ $\sfH^{[-4,-3)}$}};
        \draw[fill=white] (3.3,-0.25) rectangle (4.7,-0.75);
        \node at (4.0,-0.5) {\scalebox{0.95}{ $\check{\sfH}^{[-3,-2)}$}};
        \draw[fill=white] (4.8,-0.25) rectangle (6.2,-0.75);
        \node at (5.5,-0.5) {\scalebox{0.95}{$\sfH^{[-2,-1)}$}};
        \draw[fill=white] (6.3,-0.25) rectangle (7.7,-0.75);
        \node at (7.0,-0.5) {\scalebox{0.95}{ $\check{\sfH}^{[-1,0)}$}};
        \draw[fill=white] (7.8,-0.25) rectangle (9.2,-0.75);
        \node at (8.5,-0.5) {\scalebox{0.95}{ $\sfH^{[0,1)}$}};
        \draw[fill=white] (9.3,-0.25) rectangle (10.7,-0.75);
        \node at (10.0,-0.5) {\scalebox{0.95}{ $ \check{\sfH}^{[1,2)}$}};
        \draw[fill=white] (10.8,-0.25) rectangle (12.2,-0.75);
        \node at (11.5,-0.5) {\scalebox{0.95}{ $\sfH^{[2,3)}$}};
        \draw[fill=white] (12.3,-0.25) rectangle (13.7,-0.75);
        \node at (13.0,-0.5) {\scalebox{0.95}{ $\check{\sfH}^{[3,4)}$}};
        \draw [fill=white] (1.7,-1.0) -- (3.2,-1.0) -- (3.2,-1.5) -- (1.7,-1.5);
        \draw[fill=white] (3.3,-1.0) rectangle (6.2,-1.5);
        \node at (4.75,-1.25) { $R_* \overline{\sfH}^{(-3,-1]}|_{2t-1}$};
        \draw[fill=white] (6.3,-1.0) rectangle (9.2,-1.5);
        \node at (7.75,-1.25) { $R_* \overline{\sfH}^{(-1,1]}|_{2t-1}$};
        \draw[fill=white] (9.3,-1.0) rectangle (12.2,-1.5);
        \node at (10.75,-1.25) { $R_* \overline{\sfH}^{(1,3]}|_{2t-1}$};
        \draw[fill=white] (14.2,-1.0) -- (12.3,-1.0) -- (12.3,-1.5) -- (14.2,-1.5);
        \end{tikzpicture}
    \caption{The picture of $\kappa_{X,\loc}(\bsfH)$ (here $\sfH|_t$ is an abbreviation for $\ev_t\sfH$). Background gray vertical lines represent the integer lattice $\bZ \subset \bR$.}
    \label{fig:picture.kappa.local}
\end{figure}

We give a concise summary of the in \cref{section:Kitaev} generalizes to the sheaf $\sIP_{\loc}(X)$. 
First, in the same way as the proof of \cref{thm:inverse.pump}, especially \eqref{eqn:loop.reverse.homotopy}, the above $\kappa_{X}^\loc$ lifts to 
\begin{align*} 
    \boldsymbol{\kappa}_{X}^{\loc} \colon \bsIP(X) \to \Omega \bsIP( \boldsymbol{\Sigma} X)
\end{align*}
by defining $\check{\kappa}_{X}^\loc(\bsfH)$ to be the reversed path of $\kappa_{X}^\loc(\bsfH)$. 

In the same way as \cref{lem:truncated.pump}, the truncated version of the localized Kitaev pump is also defined as an IG localization flow on $X \times \bR_{\geq 0}$. 
Consequently, we obtain the following Lipschitz homotopy invariance in the same way as \cref{cor:coarse.homotopy.Hamiltonian}.  
\begin{cor}\label{cor:localizing.coarse.homotopy}
    Let $F_0, F_1 \colon X \to Y$ be Lipschitz continuous maps. 
    If there is a large-scale Lipschitz homotopy $\widetilde{F} \colon \mathbf{I}_\varphi X \to Y$ of $F_0$ and $F_1$ that is also a Lipschitz continuous map, then the induced maps $F_{0,*}, F_{1,*} \colon \sIP_\loc(X) \to \sIP_\loc(Y)$ are smoothly homotopic.
\end{cor}

 Next, the bulk-boundary correspondence in \cref{prp:Eilenberg.swindle} is generalized to IG localization flows. Here, the scope of the theorem is also generalized by using the terminology in \cref{defn:coarse.flasque.scaleable}. 
\begin{prp}\label{prp:flasque}
    If $X$ is topologically flasque, then $\sIP_{\loc}(X)$ is weakly contractible.  
\end{prp}
\begin{proof}
The proof is given in the same way as \cref{prp:Eilenberg.swindle}. 
Let $\sM$ be a compact manifold and let $\sfH \in \sIP_{\loc}(X \midbar \sM)$ be a smooth family supported on $\mathbbl{\Lambda} \in \fL_X^{\loc}$. Let $\widetilde{T}_{X} \colon \mathbf{C} X  \to \mathbf{C} X$ defined by 
\begin{align*} 
    \widetilde{T}_X(\bm{x},v) = (T_{X,v-\lfloor v \rfloor } \circ T_{X,1}^{\lfloor v \rfloor}(\bm{x}) , v)
\end{align*}
and set
\begin{align*}
    \widetilde{\mathbbl{\Lambda}} \coloneqq \widetilde{T}_{X} \big(\mathbbl{\Lambda}^{[0,\infty)} \big)  \subset \big( X \times \bR^{l_{\mathbbl{\Lambda}}} \times \fR \times \bN \big) \times \bR \cong X \times  \times \bR^{l_{\mathbbl{\Lambda}}+1} \times \fR \times \bN.
\end{align*}
Then, by definition of $T_X$, the projection $\pr_X \colon \widetilde{\mathbbl{\Lambda}} \to X$ is linearly proper, and hence $\widetilde{\mathbbl{\Lambda}} \in \fL_X^{\loc}$. 
Now, the truncated version $\kappa_{X}^{\loc,R}$ of the localized Kitaev pump (\cref{defn:localizing.Kitaev.pump}) gives a localization flow 
\begin{align*}
    \widetilde{T}_{X,*} \big( \kappa_X^{\loc,R} (\sfH , \check{\sfH}, \overline{\sfH}) \big) \in \cP \sIP_{\loc}(X \midbar \sM )
\end{align*}
that gives a smooth homotopy of $\sfH$ and the trivial Hamiltonian. Since $\sfH$ was arbitrary, this shows that the group $\sIP_{\loc}[X \midbar \sM] $ is trivial.  
\end{proof}

For $X \in \ECW$, set $\mathbf{C}_L X = \{ (\bm{x},v) \in \boldsymbol{\Sigma}X \mid t \leq 0 \}$ and $\mathbf{C}_R X = \{ (\bm{x},v) \in \boldsymbol{\Sigma}X \mid t \geq 0 \}$. The latter space $\mathbf{C}_RX$ is identified with the cone $\mathbf{C}X$. Following \cref{subsection:bulk.boundary}, we define the following sheaves on $\Man$. 
\begin{itemize}
    \item Let $\bsIP_{\loc}(\boldsymbol{\Sigma}X \midbar \blank )_{R\star}$ denote the subsheaf of $\bsIP_{d+1}$ consisting of triples $(\sfH , \check{\sfH}, \overline{\sfH})$ such that the constant $K_{\loc,\mathbf{C}_RX,\bullet ,\nu,\mu}^{(k)}$ is finite for $\bullet=\sfH-\sfh,\check{\sfH}-\sfh, \overline{\sfH}-\sfh$ and for any $ f \in \cF$, $k \in \bN$ and relatively compact open chart $\sU$, where  
    \begin{align}
        K_{\loc,Y,\sfG,\nu,\mu}^{(k)}\coloneqq \sup_{s \in \bR_{\geq 1}} 
        \sup_{\bm{x} \in \Lambda_s \cap Y^c} 
        f_{\nu,\mu}(\mathrm{dist}_s(\bm{x}, Y )) \cdot 
        \| \sfG_{\bm{x}} (p,0\,; s)\|_{\sU,C^k}.
        \label{eqn:constant.K.local}
    \end{align}
    \item Let $\bsIP_{\loc}(\boldsymbol{\Sigma}X \midbar \blank )_{L}$ denote the subsheaf of $\bsIP_{d+1}$ consisting of smooth sections supported on $\mathbf{C}_LX$. 
    \item Let $\cS \bsIP_{\loc}(\boldsymbol{\Sigma}X)$ denote the subsheaf of $\cP \bsIP_{\loc}(\boldsymbol{\Sigma}X)$ whose section on $\sM \in \Man$ consists of smooth families of paths $\sfH \in \cP\bsIP_{\loc}(\boldsymbol{\Sigma}X\midbar \sM)$ such that $\ev_0 \sfH \in \bsIP_{\loc}(\boldsymbol{\Sigma}X\midbar \sM)_{R\star} $ and $\ev_1 \bsIP_{\loc}(\boldsymbol{\Sigma}X\midbar \sM)_L$. 
\end{itemize}

\begin{lem}\label{lem:switch.localization.flow}
    The following holds.
\begin{enumerate}
    \item The inclusion $i \colon \Omega \bsIP_{\loc}(\boldsymbol{\Sigma}X) \to \cS \bsIP_{\loc}(\boldsymbol{\Sigma}X) $ is a weak equivalence. 
    \item The inclusion $j_* \colon \bsIP_{\loc} (X ) \to \bsIP(\boldsymbol{\Sigma}X)_{L,R\star}$ is a weak equivalence.
    \item The inclusion $\iota \colon \bsIP_{\loc}(X)_{L,R\star} \to \cS \bsIP_{\loc}(\boldsymbol{\Sigma}X)$ is a weak equivalence. 
    \item The following diagram commutes: 
    \begin{align*}
    \xymatrix{
    \bsIP_{\loc}(X) \ar[rr]^{\boldsymbol{\kappa}_d^{\loc}} \ar[rd]_{\simeq}^{\iota \circ j_*} && \Omega \bsIP_{\loc} (\boldsymbol{\Sigma}X )\ar[ld]^{\simeq}_{i} \\
    &\cS \bsIP_{\loc}(\boldsymbol{\Sigma}X).&
    }
    \end{align*}
\end{enumerate}
\end{lem}
\begin{proof}
    The proofs are the same as \cref{lem:weak.equivalence.stick,lem:switch.constant.equivalence,lem:switch.loop.equivalence,lem:switch.diagram.commute}. 
\end{proof}

\begin{thm}\label{cor:spectrum.localizing}
    There localizing Kitaev pump $\boldsymbol{\kappa}_{X}^\loc \colon \bsIP_{\loc}(X) \to \Omega \bsIP_{\loc} (\boldsymbol{\Sigma} X)$ is a weak equivalence. The sequence $\{ \bsIP_{\loc ,d}(X), \boldsymbol{\kappa}_{\boldsymbol{\Sigma}^dX}^\loc \}$ forms an $\Omega$-spectrum object in $\Sh (\Man)$.
\end{thm}

At the end of this subsection, we introduce another fundamental property of the sheaf $\sIP_{\loc}$ derived from the coarse homotopy theory.

\begin{thm}\label{thm:scaleable}
Let $X \in \ECW$ be a scaleable metric space. Then, the morphism of sheaves 
\begin{align*} 
    \mathtt{ev}_1 \colon \sIP_{\loc}(X) \to \sIP(X) , \quad \mathtt{ev}_1(\sfH(p \,; s)) = \sfH(p\,; 1),
\end{align*}
is a weak equivalence. 
\end{thm}
This theorem applies, not only to $X= \bR^d$, but also to the metric cones $X=\cO Y$ for a closed subspace $Y \subset S^n$. Such a class of metric spaces plays a fundamental role in the recent developments of this research field such as \cites{eloklUniversalCoarseGeometry2024,artymowiczMathematicalTheoryTopological2024}. 
\begin{proof}
We first show that $\mathtt{ev}_1$ is a fibration. Let $\sfH \in \sIP_\loc(X \midbar \sM)$ and $\widetilde{\sfH} \in \cP \sIP(X \midbar \sM )$ such that $ \sfH (p\,;1 )  = \widetilde{\sfH} (p,0)$ for any $p \in \sM$. Then, they are glued to a family on $\sM \times  ([1,\infty) \times \{ 0\} \cup \{1\} \times [0,1] ) $, which extends to a smooth family 
\begin{align*}
    \widetilde{\sfH} (p,t \, ;s) \coloneqq 
    \begin{cases}
        \widetilde{\sfH} (p,t+1-s) & \text{ if $1 \leq s \leq t+1$}\\
        \sfH (p \,; s-t) & \text{ if $s \geq t+1$,}
    \end{cases}
\end{align*}
which defines a smooth path $ \widetilde{\sfH} \in \cP \sIP_{\loc}(X \midbar \sM)$.

Next, we show that $\mathtt{ev}_1$ is surjective.
By \cref{exmp:Lipschitz.homotopy} and \cref{cor:localizing.coarse.homotopy}, the coarse homotopy $\widetilde{S}_X$ in \cref{defn:coarse.flasque.scaleable} induces a homotopy
\begin{align*}
    \widetilde{S}_{X,*} \colon \sIP(X) \to \cP\sIP(X)
\end{align*}
such that $\ev_0 \circ \widetilde{S}_{X,*}=\id$ and $\ev_1 \circ \widetilde{S}_{X,*} = S_{X,*}$. 
Thus, for $\bsfH \in \bsIP(X \midbar \sM)$, the smooth family 
\begin{align*}
    \widetilde{\bsfH}(p \, ;s) \coloneqq S_{X,*}^{\lfloor s\rfloor -1}  \circ \big(\ev_{s-\lfloor s \rfloor} \widetilde{S}_{X,*} \big) (\bsfH (p)) 
\end{align*}
is a localization flow such that $\mathtt{ev}_1(\overline{\sfH}) =\sfH $ as desired.

Finally, we show that the fiber $\sIP_{\loc ,\circ}(X) \coloneqq \mathtt{ev}_1^{-1}(\{ \sfh \})$ is weakly contractible. 
This follows from the analogous Eilenberg swindle argument developed in coarse C*-algebra K-theory by Qiao--Roe~\cite{qiaoLocalizationAlgebraGuoliang2010}. 
Let $\mathbbl{\Lambda} \in \fL_X^{\loc}$ and let $\sfH \in \sIP_{\loc , \circ}(X \midbar \sM)$ be a smooth family of localization flows supported on $\mathbbl{\Lambda}$. 
Let us fix a bijection of the first component $\bN$ in the direct product $X \times \bR^\infty \times \fR \times \bN \times \bN$ with the disjoint union $\bN^{\sqcup \bN}$ so that $X \times \bR^{l_{\mathbbl{\Lambda}}} \times \fR \times \bN^2$ is decomposed to countable copies. 
By using the homotopy \cref{lem:flip.homotopy.use}, we may assume that $\mathbbl{\Lambda}$ is contained in the first copy.  
We define $\widetilde{\mathbbl{\Lambda}} \coloneqq \{ \widetilde{\Lambda}_s \}_{s \in [1,\infty)} \in \fL_X^{\loc}$ so that $\widetilde{\Lambda}_s$ is the disjoint union $\Lambda_s \sqcup S_{X}\Lambda_{s-1} \sqcup \cdots \sqcup S_{X}^{\lfloor s \rfloor}\Lambda_{s-\lfloor s \rfloor } $, whose $i$-th component is stored in the $i$-th component of $\bN^{\sqcup \bN}$. Then, the smooth family 
\begin{align*}
    \widetilde{\sfH}(p\,;s) \coloneqq \sfH(p\,;s) \boxtimes S_{X,*}\sfH(p\,;s -1)  \boxtimes S_{X,*}^2\sfH(p\,;s-2) \boxtimes \cdots \boxtimes S_{X,*}^{\lfloor s \rfloor}\sfH(p\,;s-\lfloor s \rfloor ) \in \sIP_{\loc}(\widetilde{\mathbbl{\Lambda}}\midbar \sM )
\end{align*}
defines a smooth family of localization flows. Again, by using the homotopy \cref{lem:flip.homotopy.use}, we obtain smooth homotopies
\begin{align*}
    \sfH \boxtimes \widetilde{\sfH} \simeq \sfH \boxtimes \sigma_1S_{X,*}\widetilde{\sfH} \simeq \widetilde{\sfH},
\end{align*}
which implies $[\sfH] = 0 \in \sIP_{\loc,\circ}[X \midbar \sM]$. Since $\sM$ and $\sfH$ were arbitrary, this shows the weak contractibility of the sheaf $\sIP_{\loc,\circ} (X)$.
\end{proof}

\begin{exmp}\label{exmp:models}
Examples of IG localization flows are listed below. Examples (2) and (3) will be discussed in detail in \cite{kubotaStableHomotopyTheory2025b}. 
\begin{enumerate}
    \item The $1$-dimensional Fidkowski--Kitaev modeln in \cref{exmp:FK,exmp:equivariant.FK} lifts to a localization flow. Indeed, there is a smooth homotopy 
    \[
    \Omega_{12} \otimes \Omega_{36} \otimes \Omega_{45} \simeq \Omega_{12} \otimes \Omega_{34} \otimes \Omega_{56} \in \cB_{\bm{x}-1} \hotimes (\hat{\cB}_{\bm{x}} \hotimes \cB_{\bm{x}} \hotimes  \hat{\cB}_{\bm{x}} \hotimes \cB_{\bm{x}}) \hotimes \hat{\cB}_{\bm{x}+1} 
    \]
    that connects $\sfH \boxtimes \sfh$ with the $\sfH$ placed on $\Lambda_2 = \frac{1}{2}\bZ$ with an appropriate rescaling. By iterating this homotopy in a self-similar way, we obtain an IG localization flow.
    \item The lattice Dijkgraaf--Witten model (\cref{cor:equivariant.DW.cohomology}) is ultimately generalized to a morphism of cohomology functors
    \[
    \mathrm{DW}_X \colon \mathrm{HH}^2(X \midbar \sM \,; \bZ) \to \rIP_{\loc}(X \midbar \sM)
    \]
    for any $X \in \ECW$. The original Dijkgraaf--Witten model corresponds to the case that $X=\bR^d$ and $\sM=BG$. Here, the domain is the bivariant homology group 
\[  
    \mathrm{HH}^{n}(X \midbar \sM \,; \bZ) \coloneqq \colim_{k \to \infty}[\Sigma^k \sM_+, X_+ \wedge K(\bZ, n+k)].
\]
    We briefly explain the mechanism by which the localization flow emerges.  
    Our Dijkgraaf--Witten model uses as input a double complex formed from the chain complex of the Vietoris--Rips complex $\mathrm{Rips}_r(X)$ of $X$ together with the \v{C}ech complex of $\sM$. The parameter $r>0$ plays the role of the interaction range in the resulting Hamiltonian. 
    Any ordinary chain $\sigma$ in $X$ can be refined, through repeated barycentric subdivision, to a chain $\sigma_r$ in $\mathrm{Rips}_r(X)$ for arbitrarily small values of $r>0$. When $r'<r$, a coboundary between $\sigma_r$ and $\sigma_{r'}$ on $\mathrm{Rips}_r(X)$ produces a smooth homotopy between the corresponding Hamiltonians. Taken together, these homotopies assemble into a localization flow.
    \item The free-fermionic analogue of the localization flow is defined similarly, just by replacing tensor products in the definitions to direct sums. It is almost the same as the K-homology group of $X$ realized by the localization algebra of G.\ Yu \cite{yuLocalizationAlgebrasCoarse1997}. Araki's quasi-free second quantization (cf.\ \cref{rmk:KO}) gives a morphism of bivariant cohomology functors
        \[
        \mathrm{Q} \colon \mathrm{KKO}^2(X \midbar \sM) \to \rfIP_{\loc}(X \midbar \sM),
        \]
        where the left hand side denotes Kasparov's KK-group \cite{kasparovOperatorFunctorExtensions1980}.
    For example, suppose $X$ is an $n$-dimensional complete spin manifold. In this case, its Dirac operator represents the K-theoretic fundamental class $[D_X] \in \KO_n(X) \cong \mathrm{KKO}^{-n}(X \midbar \pt)$, which is also represented in a non-trivial way by a free-fermionic localization flow. 
    Its quasi-free second quantization gives an IG localization flow of fermionic UAL Hamiltonians. 
\end{enumerate}    
\end{exmp}

\subsection{Generalized homology theory}\label{subsection:generalized.cohomology}
Let $X \in \ECW$ and let $A \leq X$ be a closed subcomplex. We write 
\begin{align*}
    \mathbf{C}(X,A) \coloneqq {}&{} X \cup \mathbf{C} A \subset \mathbf{C} X \subset \boldsymbol{\Sigma}X. 
\end{align*}
Note that $\mathbf{C}(X,A)$ is also contained in $\ECW $ via the inclusion $\mathbf{C}(X,A) \subset \mathbf{C} X \subset \mathbf{C} \bR^{l_X} \subset \bR^{l_X+1}$.
We write $\ECW^2 $ for the category of pairs $(X,A)$ of CW complexes such that $X \in \ECW$ and Lipschitz continuous maps of pairs. 
For $(X,A) \in \ECW^2$, its suspension $(\boldsymbol{\Sigma}X, \boldsymbol{\Sigma} A) $ is also contained in the class $\ECW^2$, which gives a functor $\boldsymbol{\Sigma} \colon \ECW^2 \to \ECW^2$. 

\begin{defn}
    For $(X,A) \in \ECW^2$, the sheaves $\sIP_{\loc}(X,A) = \sIP_{\loc}(X,A  \midbar \blank) $ and $\bsIP_{\loc}(X,A) = \bsIP_{\loc}(X,A \midbar \blank) $ are defined by
    \[
    \sIP_{\loc}(X,A  \midbar \sM) \coloneqq \sIP_{\loc}(\mathbf{C}(X,A)  \midbar \sM), \quad \bsIP_{\loc}(X,A  \midbar \sM) \coloneqq \bsIP_{\loc}(\mathbf{C}(X,A)  \midbar \sM) .
    \]
\end{defn}

The sequence of sheaves $\bsIP_{\loc ,n}(X,A) \coloneqq \bsIP_{\loc }(\boldsymbol{\Sigma}^nX,\boldsymbol{\Sigma}^n A)$ forms an $\Omega$-spectrum object via the localized Kitaev pump 
\begin{align}
    \boldsymbol{\kappa}_{X,A}^{\loc} \coloneqq \boldsymbol{\kappa}_{\mathbf{C}(X,A)}^\loc \colon \bsIP_{\loc}(X,A) \to \Omega \bsIP_{\loc}(\boldsymbol{\Sigma}X,\boldsymbol{\Sigma} A). \label{eqn:relative.localizing.pump}
\end{align}
As well as \cref{prp:localizing.induced.hom}, a Lipschitz continuous map of CW-pairs $F \colon (X,A) \to (Y,B)$ of $\ECW^2$ induces 
\begin{align*} 
    F_* \colon \bsIP_{\loc}(X,A) \to \bsIP_{\loc}(Y,B). 
\end{align*}
Moreover, by using \eqref{eqn:relative.localizing.pump}, it is proved in the same way as \cref{cor:localizing.coarse.homotopy} that if there is a Lipschitz homotopy of $F_0$ and $F_1$, then $F_{0,*}$ and $F_{1,*}$ are smoothly homotopic. 

\begin{thm}\label{thm:IP.bivariant}
The bifunctor $\rIP_{\loc}^d \colon \ECW^2 \times \CW_* \to \mathsf{Ab}$ given by
\begin{align*}
    \rIP_{\loc}^d(X,A \midbar \sY) = \rIP_{\loc, -d}(X,A \midbar \sY) \coloneqq {}&{}  [\Sigma^n \sY,\IP_{\loc,d+n}(X,A)],
\end{align*}
which does not depend on $n \in \bZ$ such that $n+d \geq 0$, forms a bivariant homology theory, which is finitely additive in the first variable. 
\end{thm}

For $(X,A) \in \ECW^2$, set $\mathbf{D}(X,A) \coloneqq \boldsymbol{\Sigma}A \cup X \subset \boldsymbol{\Sigma}X$, $\mathbf{C}_L (X,A) = \{ (\bm{x},v) \in \mathbf{D}(X,A) \mid t \leq 0 \}$ and $\mathbf{C}_R (X,A) = \{ (\bm{x},v) \in \mathbf{D}(X,A) \mid t \geq 0 \}$. The latter space $\mathbf{C}_L(X,A)$ is identified with the cone $\mathbf{C}(X,A)$. Following \cref{subsection:bulk.boundary}, we define the following sheaves on $\Man$. 
\begin{itemize}
    \item Let $\bsIP_{\loc}(\mathbf{D}(X,A) \midbar \blank )_{R\star}$ denote the subsheaf of $\bsIP_{d+1}$ consisting of smooth sections $\bsfH$ such that the constants $K_{\loc,\mathbf{C}_R(X,A),\bullet,f }^{(k)}$ in \eqref{eqn:constant.K.local} are finite for $\bullet = \sfH-\sfh, \check{\sfH}-\sfh, \overline{\sfH} - \sfh$, and for any $f \in \cF$, $k \in \bN$, and a relatively compact open chart $\sU$. 
    \item Let $\bsIP_{\loc}(\mathbf{D}(X,A) \midbar \blank )_{L}$ denote the subsheaf of $\bsIP_{d+1}$ consisting of smooth sections supported on $\mathbf{C}_L(X,A)$. 
    \item Let $\cS \bsIP_{\loc}(X,A)$ denote the subsheaf of $\cP \bsIP_{\loc}(\mathbf{D}(X,A))$ whose section on $\sM \in \Man$ consists of smooth families of paths $\bsfH \in \cP\bsIP_{\loc}(X,A\midbar \sM)$ such that $\ev_0 \bsfH \in \bsIP_{\loc}(\mathbf{D}(X,A)\midbar \sM)_{R\star} $ and $\ev_1\bsfH \in  \bsIP_{\loc}(\mathbf{D}(X,A)\midbar \sM)_L$. 
\end{itemize}

\begin{lem}\label{lem:IP.homology.longexact}
Let $(X,A) \in \ECW^2$. 
\begin{enumerate}
    \item The morphism $\ev_1\colon \cS \bsIP_{\loc}(X,A) \to \bsIP_{\loc}(X,A)$ is a fibration. 
    \item The inclusion $j_* \colon \bsIP_{\loc}(X) \to  \bsIP_{\loc}(\mathbf{D}(X,A))_{L,R\star}$ is a weak equivalence. 
    \item The inclusion $\iota \colon \bsIP_{\loc}(X)_{L,R\star} \to \cS \bsIP_{\loc}(\mathbf{D}(X,A))$ is a weak equivalence. 
    \item The morphism $i \circ \kappa_A^{\loc} \colon \bsIP_{\loc}(A) \to \Omega \bsIP_{\loc}(\boldsymbol{\Sigma} A) \subset \ev_1^{-1}(\sfh) \subset \cS \bsIP_{\loc}(\mathbf{D}(X,A))$ is a weak equivalence. Moreover, as morphisms to $\cS \bsIP_{\loc}(\mathbf{D}(X,A))$, there is a homotopy $i \circ \kappa_A^{\loc} \sim \iota \circ j_* \circ h_*$, where $h \colon A \to X$ denotes the inclusion. 
\end{enumerate}
\end{lem}
\begin{proof}
    The claims (1), (2), (3) are proved in the same way as \cref{lem:weak.equivalence.stick,lem:switch.constant.equivalence,lem:switch.loop.equivalence,lem:switch.diagram.commute} as well as \cref{lem:switch.localization.flow}. 
    For (3), notice that both $\mathbf{C}_L(X,A)$ and $\mathbf{C}_R(X,A)$ contains the common subspace $\boldsymbol{\Sigma}A$, and hence $\vartheta_d(\bsfH)$ is supported near $\boldsymbol{\Sigma}A$. 
    The weak equivalence in claim (4) follows from \cref{prp:flasque} and \cref{cor:spectrum.localizing}. The homotopy is proved in the same way as \cref{lem:switch.diagram.commute}.
\end{proof}

\begin{prp}\label{prp:IP.homology.longexact}
    There is a long exact sequence 
    \begin{align*}
        \cdots \to \rIP_{\loc,d+1}(X,A \midbar \sY) \xrightarrow{\delta} \rIP_{\loc,d}(A \midbar \sY) \to \rIP_{\loc,d}(X \midbar \sY) \to \rIP_{\loc,d}(X,A \midbar \sY) \xrightarrow{\delta}  \rIP_{\loc,d-1}(A \midbar \sY) \to  \cdots   
    \end{align*}
    for any CW-complex $\sY$.
\end{prp}
\begin{proof}
    By \cref{lem:IP.homology.longexact}, the homotopy long exact sequence associated with the fibration 
    \[ 
        \ev_1^{-1}(\sfh) \to \cS \bsIP_{\loc}(\mathbf{D}(X,A)) \xrightarrow{\ev_1} \bsIP_{\loc}(\mathbf{C}_R(X,A))
    \]
    gives the desired long exact sequence. 
\end{proof}

\begin{lem}\label{lem:cut.project}
    Let $P \colon X \to A$ be a retract, which is not necessarily continuous but satisfies that $\rmd (\bm{x},P(\bm{x}) ) = \mathrm{dist}(\bm{x},A)$ for any $\bm{x} \in X$. Let $\sfH \in \mathscr{GP}(X \midbar \sM)$ such that $K_{\loc, A,\sfH-\sfh,f}^{(k)}$ for any $f \in \cF$ and $k \in \bN$. 
    Then, $P_*\sfH$ is a smooth family of IG localization flows on $A$.
\end{lem}
\begin{proof}
    It suffices to show that 
    \begin{align*}
    \sup_{\bm{x} \in \Lambda_X}\sup_{r>0} f(sr)^{-1} \cdot \| \mathtt{ev}_s(\sfH_{\bm{x}} - \Pi_{P^{-1}(B_{r}(P(\bm{x})))}(\sfH_{\bm{x}}) )\|_{\sU,C^k} <\infty.
    \end{align*}
    For $\bm{x},\bm{y} \in \Lambda_X$, we have
    \begin{align*}
        \rmd(P(\bm{x}),P(\bm{y})) \leq {}&{} \rmd(\bm{x},\bm{y}) + \rmd(\bm{x},P(\bm{x})) + \rmd(\bm{y},P(\bm{y})) \\
        \leq{}&{} \rmd(\bm{x},\bm{y}) + \mathrm{dist}(\bm{x},A) + \mathrm{dist}(\bm{y},A)\\
        \leq{}&{} 2\rmd(\bm{x},\bm{y}) + 2\mathrm{dist}(\bm{x},A) ,
    \end{align*}
    that is, $P^{-1}(B_r(P(\bm{x}))) \supset B_{r/2 - \mathrm{dist}(\bm{x},A) }$. Therefore, 
\begin{align*}
    {}&{} \sup_{s \geq 1}\sup_{\bm{x} \in \Lambda}\sup_{r>0} f_{\mu}(sr)^{-1} \cdot \| \mathtt{ev}_s(\sfH_{\bm{x}} - \Pi_{P^{-1}(B_{sr}(P(\bm{x})))}(\sfH_{\bm{x}})) \|_{\sU,C^k} \\
    \leq {}&{}
    \sup_{s \geq 1}\sup_{\bm{x} \in \Lambda}\sup_{r>0} f_{\mu}(sr)^{-1} \cdot  \|  \mathtt{ev}_s(\sfH_{\bm{x}} - \Pi_{sr/2-\mathrm{dist}_s(\bm{x},A)}(\sfH_{\bm{x}}))\|_{\sU,C^k}\\
    \leq {}&{}
    \begin{cases}
        \sup_{s \geq 1}\sup_{r>0} f_{\mu}(sr)^{-1} \cdot 2 f_{\mu_1}(s\mathrm{dist}(\bm{x},A)) \cdot K_{\loc,A,\sfH-\sfh,\mu_1}^{(k)} & \text{ if $r \leq 4 \mathrm{dist}(\bm{x},A)$}, \\
        \sup_{s \geq 1}\sup_{r>0} f_{\mu}(sr)^{-1} \cdot  f_{\mu_1}(sr/2 - s\mathrm{dist}(\bm{x},A))^{-1} \cdot \vvert \sfH \vvert_{\sU,C^k,\mu_1}& \text{ if $r \geq 4 \mathrm{dist}(\bm{x},A)$}.
    \end{cases}
\end{align*}
    The first case is bounded by $c_{0,\mu,4} \cdot K_{\loc,A,\sfH-\sfh,\mu_1}^{(k)}$. For the second case, in which $r/2 - \mathrm{dist}(\bm{x},A) \geq r/4$, it is bounded by $c_{0,\mu,4} \cdot \vvert \sfH \vvert_{\sU,C^k,\mu_1}$. 
    This proves the lemma.
\end{proof}

\begin{prp}\label{prp:excision}
    Let $X \in \ECW$ and let $A, B \leq X$ be subcomplexes such that $X=A \cup B$. 
    Then the morphism $\iota_* \colon \sIP_{\loc}(B,A \cap B) \to \sIP_{\loc}(X,A)$ induced from the inclusion $\iota \colon (B,A \cap B) \to (X,A)$ is a weak equivalence. 
\end{prp}
\begin{proof}
    Let $\bsfH \in \bsIP_{\loc}(X,A \midbar \sM)$. 
    As is considered in \cref{prp:flasque}, the truncated Kitaev pump $\kappa_X^{\loc,R} (\bsfH)$ gives a smooth null-homotopy of $\sfH$ in $\sIP_{\loc}(\mathbf{C} \mathbf{C}(X,A) ,\sM)$. 
    Let 
\[
    P \colon \mathbf{C}\mathbf{C} (B,A \cap B) \to B \cup \mathbf{C}^2(A \cap B)
\]
    be a (not necessarily continuous) map such that $\rmd( \bm{x}, P(\bm{x})) = \mathrm{dist}(\bm{x}, B \cup \mathbf{C}^2(A \cap B))$.
    By \cref{cor:Lieb.Robinson.approx.Ck,lem:cut.project}, the adiabatic connection along this homotopy gives new localization flows
    \begin{align*}
    \fT_{X,A}^t(\bsfH) \coloneqq \widetilde{T}_{X,*} \circ P_*  \circ \alpha \Big( \Pi_{\mathbf{C}^2 B^c}^{\star} \Big( \sfG_{\kappa_X^{\loc,R} (\bsfH)} \Big) \,; t \Big)^{-1} (\sfH),\\
    \check{\fT}_{X,A}^t(\bsfH) \coloneqq \widetilde{T}_{X,*} \circ P_*   \circ \alpha \Big( \Pi_{\mathbf{C}^2 B^c}^{\star} \Big( \sfG_{\check{\kappa}_X^{\loc,R} (\bsfH)} \Big) \,; t \Big)^{-1} (\check{\sfH}),\\
    \overline{\fT}_{X,A}^t(\bsfH) \coloneqq \widetilde{T}_{X,*} \circ P_*  \circ \alpha \Big( \Pi_{\mathbf{C}^2 B^c}^{\star} \Big( \sfG_{\bar{\kappa}_X^{\loc,R} (\bsfH)} \Big) \,; t \Big)^{-1} (\overline{\sfH}),
    \end{align*}
    which are all supported on $\mathbf{C}(X,A)$. 
    The distinguished ground state of $\fT_{X,A}^1(\bsfH) $, $\check{\fT}_{X,A}^1(\bsfH) $, and $\overline{\fT}_{X,A}^1(\bsfH) $ restricts to the trivial ground state $\omega_0$ on $\cA_{\mathbf{C} B^c}$. Therefore, the convex combination gives a smooth homotopy of triples
    \begin{align*}
        \bsfH= {}&{}(\fT_{X,A}^0(\bsfH), \check{\fT}_{X,A}^0(\bsfH) , \overline{\fT}_{X,A}^0(\bsfH)) \\ 
        \simeq {}&{} (\fT_{X,A}^1(\bsfH), \check{\fT}{}_{X,A}^1(\bsfH) , \overline{\fT}{}_{X,A}^1(\bsfH))\\
        \simeq {}&{}
        (\Theta_{\mathbf{C}B^c , \omega_0}\fT_{X,A}^1(\bsfH), \Theta_{\mathbf{C}B^c , \omega_0}\check{\fT}{}_{X,A}^1(\bsfH) , \Theta_{\mathbf{C}B^c , \omega_0}\overline{\fT}{}_{X,A}^1(\bsfH)) \in \bsIP_{\loc}(\mathbf{C}(B,A \cap B) \midbar \sM).
    \end{align*}
    This shows the surjectivity of $\iota_*$. The injectivity is proved similarly. Indeed, if $\bsfH(p)=\sfh$ for $p \in \sA \subset \sM$, then $\fT_{X,A}^t(\bsfH)(p) = \sfh$, $\check{\fT}_{X,A}^t(\bsfH)(p)=\sfh$, and $\overline{\fT}_{X,A}^t(\bsfH) (p) = \sfh$ holds.   
\end{proof}

\begin{proof}[Proof of \cref{thm:IP.bivariant}]
    By \cref{cor:spectrum.localizing}, this functor is cohomological with respect to the variable $\sY$.
    We show that, once we fix $\sY \in \CW_*$, the assignment $(X,A) \mapsto \rIP_{\loc}^d(X,A \midbar \sY)$ satisfies the Eilenberg-Steenrod axiom of homology theory. First, the exactness axiom is proved in  \cref{prp:IP.homology.longexact}. Next, \cref{prp:excision} shows that, for any $X \in \ECW$ and $A,B \leq X$ such that $A \cup B =X$, there is an isomorphism
    \begin{align*}
    \rIP_{\loc,d}(X,A \midbar \sM) \cong \rIP_{\loc,d}(B, A \cap B \midbar \sM),
    \end{align*}
    which shows the excision axiom. Finally, the finite additivity axiom follows from the excision axiom.
\end{proof}

\begin{cor}\label{cor:mayer-vietoris}
    The bivariant homology functor $\rIP_{\loc,d}(X \midbar \sY)$ has the following Mayer--Vietoris exact sequence for subcomplexes $A,B \leq X$ such that $X=A \cup B$;
    \begin{align*}
        \cdots \xrightarrow{\delta_{\mathrm{MV}}}{}&{} \rIP_{\loc,d+1}(A \cap B \midbar \sY) \to \rIP_{\loc,d+1}(A \midbar \sY) \oplus \rIP_{\loc,d+1}(B  \midbar \sY)  \to \rIP_{\loc,d+1}(X \midbar \sY) \\
         {}&{}\xrightarrow{\delta_{\mathrm{MV}}}  \rIP_{\loc,d}(A\cap B \midbar \sY) \to \rIP_{\loc,d}(A \midbar \sY) \oplus \rIP_{\loc,d}(B \midbar \sY) \to \rIP_{\loc,d}(X \midbar \sY)\xrightarrow{\delta_{\mathrm{MV}}} \cdots   .
    \end{align*}    
\end{cor}

We further discuss the additivity axiom of $\sIP_{\loc}(X)$. 
Indeed, the countable additivity also holds in a weaker sense, under an assumption of metrical control. 

\begin{prp}\label{prp:additivity}
    Let $X = \bigsqcup_{i \in I} X_i \in \ECW$ such that $\mathbbl{\Delta} \coloneqq \inf_{i} \mathrm{dist}(X_i,X_i^c) >0$. Then, there is a weak equivalence of $\sIP_{\loc}(X)$ and the `uniform' direct product $\prod_{i \in I}^u \sIP(X_i)$ defined by
    \begin{align*}
     {\prod_{i \in I}}^u \  \sIP(X_i \mid \sM)\coloneqq \bigg\{ (\sfH^{(i)})_{i \in I} \in \prod_{i \in I} \sIP_{\loc}(X_i \midbar \blank) \mid \sup_i \vvert \sfH^{(i)} \vvert_{\sU,C^k,\loc,f} <\infty \text{ for any $\sU,k,f$}\bigg\}. 
    \end{align*}
\end{prp}
\begin{proof}
    We show that the canonical inclusion $c \colon \prod_i^u \sIP_{\loc}(X_i) \to \sIP_{\loc}(X)$ is a weak equivalence. To see this, we show that the homotopy fiber sheaf given by
    \[
    \mathop{\mathrm{hofib}} (c) (\sM)  \coloneqq  \{ \sfH \in \cP \sIP_{\loc}(X \midbar \sM) \mid \ev_0\sfH=\sfh, \  \ev_1 \sfH \in \Im c \}
    \]
    is contractible. This will show that $c$ induces isomorphisms of homotopy groups of degrees $\geq 1$. For the $\pi_0$-groups, compare the $\pi_1$-groups of $\sIP_{\loc}(X \times \bR)$ and $\prod_{i }^u \sIP_{\loc}(X_i \times \bR)$. 

    We take $\sfH \in \mathop{\mathrm{hofib}}(c)(\sM)$ and show that it is smoothly null-homotopic. For a UAL derivation $\sfG$, we write $\sfG^{\Pi} \coloneqq \prod_{i \in I}\Pi_{X_i} (\sfG)$. Then, since  
    \begin{align*}
        \| \sfG_{\bm{x}} - \sfG^{\Pi}_{\bm{x}} \|_{(s),\sU,C^k,\bm{x},\nu,\mu} \leq {}&{}
        \| \sfG_{\bm{x}} - \Pi_{\bm{x},\mathbbl{\Delta}}(\sfG_{\bm{x}}) \|_{(s),\sU,C^k,\bm{x},\nu,\mu}\\
        \leq {}&{}
        2f_{\nu,\mu}(s\mathbbl{\Delta})^{-1}
        f_{\nu,\mu_1}(s\mathbbl{\Delta}) 
        \cdot \vvert \sfG \vvert_{(s),\sU,C^k,\nu,\mu_1},  
    \end{align*}
    we have $\vvert \sfG - \sfG^{\Pi}\vvert _{(s),\sU,C^k,\nu,\mu} \leq 2f_{\nu,\mu}(s\mathbbl{\Delta})^{-1}
        f_{\nu,\mu_1}(s\mathbbl{\Delta}) \cdot \vvert \sfG\vvert_{(s),\sU,C^k,\nu,\mu_1}$. Applying this to $\sfG=\sfG_{\sfH}$, by \cref{prp:Lieb.Robinson.Ck} we obtain that 
    \begin{align}
    \begin{split}
         {}&{} \vvert  \alpha_{(p\,; s)}(\sfG_{\sfH} \,; t)  (\sfh) - \alpha_{(p\,; s)} (\sfG_{\sfH}^{\Pi} \,; t)(\sfh) \vvert_{(s),\sU,C^k,\nu,\mu} \\
        \leq {}&{}  \Upsilon_{\sfG_{\sfH},\sfG_{\sfH}^{\Pi},\nu_k,\mu_{{3k+1}}}^{\mathrm{rel},(k)}(t) \cdot \vvert \sfG_{\sfH} - \sfG_{\sfH}^\Pi \vvert _{(s),\sU,C^k,\nu_{2k}, \mu_{6k+2}} \\
        \leq {}&{} 2 \Upsilon_{\sfG_{\sfH},\sfG_{\sfH}^{\Pi},\nu_k,\mu_{{3k+1}}}^{\mathrm{rel},(k)}(t) \cdot 
        \frac{
        f_{\nu_{2k},\mu_{6k+3}}(s\mathbbl{\Delta})
        }{
        f_{\nu_{2k},\mu_{6k+2}}(s\mathbbl{\Delta})
        }
        \cdot \vvert \sfG_{\sfH}\vvert_{(s),\sU,C^k,\nu_{2k},\mu_{6k+3}} =: \varrho_{(s),k,\nu,\mu}, 
    \end{split}\label{eqn:sumability.cutoff}
    \end{align}
    which decays as $s \to \infty$.

    Since the ground state of the IG UAL Hamiltonian $\alpha_{(p\,; s)}(\sfG_{\sfH} \,; 1)(\sfh)$ agrees with that of $\sfH(1) \in \Im (c)$, say $\omega_p$, the UAL Hamiltonian
    \[
    \sfH' (p\,; s) \coloneqq \prod_{i \in I}\Theta_{X_i^c, \omega_p} \Big(  \alpha_{(p\,; s)}(\sfG_{\sfH} \,; 1)(\sfh) \Big)
    \]
    is also gapped above the ground state $\omega$ and invertible. By \eqref{eqn:sumability.cutoff} and \eqref{eqn:redestribution.2}, we have 
    \begin{align}
    \begin{split}
    {}&{}\vvert \sfH'- \alpha(\sfG_{\sfH} \,; 1)(\sfh) \vvert_{(s),\sU,C^k,\nu,\mu} \\
    \leq {}&{} \Bigvvert  \prod_{i \in I}\Theta_{X_i^c, \omega_p} \Big(  \alpha(\sfG_{\sfH} \,; 1)(\sfh) \Big) - \prod_{i \in I}\Theta_{X_i^c, \omega_p} \Big(  \alpha(\sfG_{\sfH}^{\Pi} \,; 1)(\sfh) \Big) \Bigvvert_{(s),\sU,C^k,\nu,\mu} \\
    {}&{} + \bigvvert \alpha(\sfG_{\sfH} \,; 1)(\sfh)-\alpha(\sfG_{\sfH}^{\Pi} \,; 1)(\sfh)  \bigvvert_{(s),\sU,C^k,\nu,\mu}\\
    \leq {}&{}    \| \omega \|_{(s),C^k,\nu,\mu} \cdot g_{\nu,\mu} \cdot \varrho_{(s),k,\nu,\mu_1} +\varrho_{(s),k,\nu,\mu}
    \end{split}\label{eqn:summability.cutoff.2}
    \end{align}
     for sufficiently large $\mu <1$ so that $\| \omega\|_{(s),C^k,\nu,\mu}<\infty$, which decays as $s \to \infty$ as well.

    Finally, we construct a smooth homotopy of a given section $\sfH$ and a smooth section of the subsheaf $\cP_{\sfh} \prod_{i \in I}^u\sIP_{\loc}(X_i \midbar \sM) \subset \mathrm{hofib}(c)$, which is weakly contractible. This show that $\sfH$ is null-homotopic as follows. 
    First, we replace $\sfH$ with another path by homotopies 
    \begin{align}
    \begin{split}
    \sfH (p,t\,; s) \sim {}&{} 
    \begin{cases}
        \sfH(p,2t\,; s) & 0 \leq t \leq 1/2, \\
        (2-2t)\sfH(p,1\,; s) +(2t-1)\sfH'(p\,; s) & 1/2 \leq t \leq 1,
    \end{cases}
     \\
    \sim {}&{}  
    \begin{cases}
        \alpha_{(p\,; s)}(\sfG_{\sfH}\,; 2t)(\sfh) & 0 \leq t \leq 1/2, \\
        (2-2t)\alpha_{(p\,; s)}(\sfG_{\sfH}\,; 1)(\sfh) +(2t-1)\sfH'(p\,; s) & 1/2 \leq t \leq 1.
    \end{cases}
    \end{split}
    \label{eqn:summability.cutoff.homotopy}
    \end{align}
    By a parameter shift if necessary (cf.\ \cref{rmk:parameter.shift}), \eqref{eqn:sumability.cutoff} and \eqref{eqn:summability.cutoff.2} enables us to apply \cref{thm:gap.stability} to conclude that the convex combination 
    \[
    \begin{cases}
     (1-v)  \alpha_{(p\,; s)}(\sfG_{\sfH} \,; t)  (\sfh) +v \alpha_{(p\,; s)}(\sfG_{\sfH}^\Pi \,; 2t)(\sfh ) & 0 \leq t \leq 1/2, \\
     (1-v) \Big( (2-2t)\alpha_{(p\,; s)}(\sfG_{\sfH}\,; 1)(\sfh) +(2t-1)\sfH'(1)\Big) + v \alpha_{(p\,; s)}(\sfG_{\sfH}^{\Pi} \,; 1)(\sfh) & 1/2 \leq t \leq 1,
    \end{cases}
    \]
    is a smooth homotopy of IG localization flows with a uniform spectral gap $1/2$, which connects the right hand side of \eqref{eqn:summability.cutoff.homotopy} and the path
    \[
    \begin{cases}
        \alpha_{(p\,; s)}(\sfG_{\sfH}^{\Pi}\,; 2t)(\sfh) & 0 \leq t \leq 1/2, \\
        \alpha_{(p\,; s)}(\sfG_{\sfH}^{\Pi} \,; 1)(\sfh) & 1/2 \leq t \leq 1,
    \end{cases}
    \in \cP_{\sfh}  {\prod_{i \in I}}^u \ \sIP_{\loc}(X_i \midbar \sM). 
    \]
    This finishes the proof.
\end{proof}

\subsection{The Spanier--Whitehead duality}
Finally, we relate the generalized homology theory of IG localization flows with the one associated with the $\Omega$-spectrum $\IP_*$. 
Let us recall that, for an $\Omega$-spectrum $E_*$ and a finite CW-complex $X$, the $E$-homology group $E_n(X)$ is defined by the stable homotopy group $\pi_n^{\mathrm{st}}(X_+ \wedge E) \coloneqq \colim_{k \to \infty}\pi_{n+k}(X_+ \wedge E_k)$. 
The goal of this subsection is that there is a weak equivalence of $\Omega$-spectra 
\begin{align*}
    \IP_{\loc}(X) \simeq \Omega^{\infty} (X_+ \wedge \IP) \coloneqq \colim_{n \to \infty} \Omega^{k}(X_+ \wedge \IP_{k}),
\end{align*}
that is, there is a natural isomorphism of homology functors
\begin{align*} 
    \pi_n(\IP_{\loc }(X) ) \cong \pi_n^{\mathrm{st}}(X_+ \wedge \IP_*) = \colim_{k \to \infty} \pi_{n+k} (X_+ \wedge \IP_k). 
\end{align*}

To this end, we restrict the metric space $X$ to finite CW-complexes embedded to $\bR^l$. 
Then there is a deformation retract neighborhood $U_X$ of $X$. We may assume $U_X$ is the interior of an embedded manifold with boundary $\overline{U}_X$, which is also a deformation retract of $X$. 
We deal with the pair of topological spaces
\[
    \sD_X \coloneqq (\overline{U}_X, \partial \overline{U}_X) \in \CW^2.
\]
As is noted in \cref{subsubsection:definition.sheaf}, it also serves as an input of a sheaf on $\Man$.
The Dold--Puppe version of the Spanier--Whitehead duality (\cite{doldDualityTraceTransfer1983}, we also refer the reader to  \cite{rudyakThomSpectraOrientability1998}*{Example II.2.8 (b)}) shows that there is an isomorphism of homotopy groups 
\begin{align*} 
    \rIP^n(\sD_X) =[ \sD_X, \IP_{n}] \cong \pi_{l-n}^{\mathrm{st}}( X_+ \wedge \IP) =\rIP_{l-n}(X). 
\end{align*}
When $X$ is a closed submanifold of $\bR^n$, the pointed space $\sD_X$ coincides with the Thom space $\mathrm{Th}(\nu X)$ of the normal bundle of $X$, and the above isomorphism is nothing else than the Atiyah duality \cite{atiyahThomComplexes1961}. 

\begin{thm}\label{thm:Atiyah.duality}
Let $X \subset \bR^{l}$ be an embedded finite CW-complex, and let $\sD_X$ be as above. 
Then, there is a weak equivalence of sheaves
\begin{align*} 
    \sIP_{\loc}(X) \simeq \sHom (\sD_X , \sIP_{\loc,l} ) \simeq \sHom (\sD_X , \sIP_{l} ).  
\end{align*}
\end{thm}
The second weak equivalence follows from \cref{thm:scaleable}. We prove the first.

In the following, given a smooth family $\sfH$ of invertible localization flows, we construct a new Hamiltonian that has a ground state close to that of $\sfH(\bm{x})$ on a neighborhood of $\bm{x} \in \Lambda \cap \overline{U}_X$. 
This construction is similar to what is called a modulating Hamiltonian in \cite{yaoModulatingHamiltonianApproach2024}, but only the ground state is modulated in our construction.

\begin{lem}
    Let $\sH_1, \cdots,\sH_n$ be Hilbert spaces equipped with a fixed unit vector $\Omega_j \in \sH_j$, let $\xi \in \sH_1 \otimes \cdots \otimes \sH_n$ be another unit vector, and let $\varepsilon >0$ be sufficiently small. Assume that the vector state $\langle \blank \xi,\xi\rangle$ satisfies that 
    \begin{align*}
    \| \langle \blank \xi,\xi \rangle |_{\cB(\sH_j)} - \langle \blank \Omega_j,\Omega_j\rangle \| <\varepsilon/4n
    \end{align*}
    for any $j=1,\cdots,n$. Then we have $\| \langle \blank \xi,\xi \rangle- \langle \blank \Omega,\Omega\rangle \| <\varepsilon $, where $\Omega \coloneqq \Omega_1 \otimes \cdots \otimes \Omega_n$. 
\end{lem}
\begin{proof}
    For each $j=1,\cdots,n$, take a complete orthonormal system $e_1^j,e_2^j,\cdots$ of $\sH_j$ such that $e_1^j=\Omega_j$. Let us represent $\xi$ as
    \begin{align*} 
    \xi = \sum_{i_1,\cdots,i_n} \lambda_{i_1\cdots i_n} e_{i_1}^1 \otimes \cdots \otimes e_{i_n}^n. 
    \end{align*}
    Then, the assumption implies that 
    \begin{align*}
    \sum_{i_j\neq 1} |\lambda_{i_1\cdots i_n}|^2 = \langle (1 - \Omega_j \otimes \Omega_j^*) \xi,\xi\rangle < \varepsilon /4n
    \end{align*}
    for any $j=1,\cdots,n$, and hence
    \begin{align*}
     1-\lambda_{1,\cdots,1}^2 = \sum_{(i_1,\cdots,i_n) \neq (1,\cdots,1)} |\lambda_{i_1\cdots i_n}|^2 \leq \sum_{j=1}^n \sum_{i_j \neq 1} |\lambda_{i_1\cdots i_n}|^2 <\varepsilon/4. 
    \end{align*}
    Now, we have
    \begin{align*}
    {}&{} \|\langle \blank \xi,\xi\rangle - \langle \blank \Omega_1\otimes \cdots \otimes \Omega_n ,\Omega_1 \otimes  \cdots \otimes \Omega_n \rangle \| \leq   2\| \xi - \Omega_1 \otimes \cdots \otimes \Omega_n\| ^2\\
    \leq {}&{} 2\Big( (1 - \lambda_{1,\cdots,1})^2 + \sum_{(i_1,\cdots,i_n) \neq (1,\cdots,1)} |\lambda_{i_1\cdots i_n}|^2 \Big)
    = 4-4\lambda_{1,\cdots,1} \\
    \leq {}&{} 4(1-\sqrt{1-\varepsilon/4}) <\varepsilon
    \end{align*}
    for sufficiently small $\varepsilon >0$. 
\end{proof}

Without loss of generality, we may assume that $U_X$ is contained in the unit cube $[0,1]^{l}$. 
Fix the length scale $\epsilon = 1/m$ and mesh the unit cube with this width. 
For $\bm{k} \in (\epsilon \bZ)^{l}$, we consider the cubes
\begin{align*}
    Q_r(\bm{k}) = Q_{r}(k_1,\cdots,k_{l}) \coloneqq [k_1 -r , k_1+ \epsilon +r] \times \cdots \times [k_l -r , k_{l}+ \epsilon +r]
\end{align*}
and $Q(\bm{k}) \coloneqq Q_0(\bm{k})$. The length scale $\epsilon$ is chosen so that the meshed neighborhood 
\begin{align*}
    U_{X,2l\epsilon } \coloneqq \bigcup_{Q_{2l\epsilon}(\bm{k}) \cap U_X \neq \emptyset }Q(\bm{k})
\end{align*}
is a deformation retract neighborhood of $\overline{U}_X$. We may further assume that $\mathrm{dist}(\overline{U}_X, \partial [0,1]^l) \geq \epsilon$. 

\begin{lem}\label{lem:mesh.approximation}
    Let $\sfG$ be a smooth $1$-form on $[0,1]$ with coefficient in $\mathfrak{LD}_{\mathbbl{\Lambda}}^{\al}$. Then, there is a constant $J_{\sfG,\nu,\mu}$ depending only on $\mathbbl{\Lambda}$, $\nu$, $\mu$, and the almost local norms of $\sfG$, such that the following holds: Let $r_0, r_1 , t_0, t_1 \in [0,1]$ such that $t_0 <t_1$ and $r_1 -r_0 > \epsilon$. Then, for any $a \in \cA_{Q_{r_0}(\bm{k})}$ and $Y \supset Q_{r_1}(\bm{k}) $, we have
    \begin{align*}
        \big\| \alpha (\sfG \,; t_0,t_1)(a) - \Pi_{Y} \big( \alpha (\sfG \,; t_0,t_1)(a) \big) \big\| \leq {}&{} J_{\sfG , \nu,\mu}  \cdot f_{\nu,\mu} ( s (r_1-r_0)) \cdot  \| a \|, \\
        \big\| \alpha( \sfG \,; t_0,t_1)(a)-\alpha(\Theta_{Y , \omega_0} (\sfG)\,; t_0,t_1) (a)  \big\| \leq {}&{} J_{\sfG, \nu,\mu} \cdot f_{\nu,\mu} (s (r_1-r_0)) \cdot  \| a\|.
    \end{align*}
\end{lem}
\begin{proof}
Since the family $\pr_{\bR^l} \colon \Lambda_s \to \bR^l$ is equi-linearly proper, there is $C>0$ such that $\mathop{\mathrm{diam}}(Q_r(\bm{k}) ) = C \cdot (\epsilon + 2r)\cdot l^{1/2}$, and hence
\begin{align*}
    \# ( Q_r(\bm{k}) \cap \Lambda_s) \leq \kappa_{\mathbbl{\Lambda}} \cdot (1+ C(\epsilon +2r) \cdot l^{1/2} \cdot s)^{l} \leq \kappa_{\mathbbl{\Lambda}} l^{l/2}C^l \cdot (1+ (\epsilon +2r) s)^{l},
\end{align*}
and hence
\begin{align*}
     D_{\nu_2,\mu_1}(Q_{r_0}(\bm{k}),Y) \leq {}&{} D_{\nu_2,\mu_1}(Q_{r_0}(\bm{k}), Q_{r_1}(\bm{k})^c) \\
    \leq{}&{} 2^{l+1}\kappa_{\mathbbl{\Lambda}} \cdot \# (Q_{r_0}(\bm{k}) \cap \Lambda_s) \cdot  f_{\nu_1,\mu_1}(\mathrm{dist}_s(Q_{r_0}(\bm{k}),Q_{r_1}(\bm{k}))^c) \\
    \leq {}&{} 2^{l+1}\kappa_{\mathbbl{\Lambda}} \cdot \kappa_{\mathbbl{\Lambda}} l^{l/2}C^l  \cdot (1+ (\epsilon +2r_0) s)^{l} \cdot f_{\nu_1,\mu_1}(s(r_1-r_0)) \\
    \leq {}&{}  2^{l+1} \kappa_{\mathbbl{\Lambda}}^2 l^{l/2}C^l \cdot (1 + 2/\epsilon)^{l} \cdot  f_{\nu,\mu}(s(r_1-r_0) ). 
\end{align*}
Now the inequalities follow from \eqref{eqn:Lieb.Robinson.truncate} and \cref{prp:Lieb.Robinson} (1) by choosing the constant $J_{\sfG,\nu,\mu}$ as 
\begin{align*}
    J_{\sfG,\nu,\mu} \coloneqq 2^{l+2} \kappa_{\mathbbl{\Lambda}}^2 l^{l/2}C^l \cdot (1 + 2/\epsilon)^{l} \cdot  \exp\big( C_{\nu_2} L_{\nu_2,\mu_1}\vvert \sfG \vvert _{\nu_5,\mu_2} \big) \cdot \max \{ 1,C_{\nu_2} L_{\nu_2,\mu_1}\vvert \sfG \vvert _{\nu_5,\mu_2}\},
\end{align*}
which certainly depends only on $\nu$, $\mu$, $\mathbbl{\Lambda}$, and the almost local norms of $\sfG$. 
\end{proof}

Let $\sfH \in \sIP_{\loc}(\bR^l \midbar [0,1])$. We write $\omega_v \coloneqq \omega_{H(v)}$. When convenient, we use the same notation for the restriction of this state to any relevant subalgebra, with the meaning clear from the context.
For $k \in \epsilon \bZ$, let $Y_k \coloneqq \bR^{l -1} \times [k,\infty) $ and 
\begin{align*}
    \alpha_k \coloneqq \alpha \big(\sfG_{\sfH} \,; k , k+\epsilon \big), \quad   \gamma_{k} \coloneqq \alpha \big( \Theta_{Y_k ,\omega} (\sfG_{\sfH}) \,; k , k+\epsilon \big). 
\end{align*}
For $a,b \in \epsilon \bZ$, let 
    \begin{align*}
        \gamma_{[a,b]} \coloneqq \gamma_{b-\epsilon} \circ \gamma_{b-2\epsilon} \circ \cdots \circ \gamma_{a+\epsilon} \circ \gamma_a, \quad         \alpha_{[a,b]} \coloneqq \alpha_{b-\epsilon} \circ \alpha_{b-2\epsilon} \circ \cdots \circ \alpha_{a+\epsilon} \circ \alpha_a.
    \end{align*}
We are going to compare $\omega_0 \circ \gamma_{[0,1]}$ and $\omega_0 \circ \alpha _{[0,k_l]}$ on $\cA_{Q(\bm{k})}$.

\begin{lem}\label{lem:truncate.state.fix}
    Let $\sfH \in \sIP_{\loc} (\bR^l \midbar [0,1])$ satisfies that $\omega_{v} |_{\cA_{Q_{2\epsilon}(\bm{k})}} = \omega_0$ for any $v \in [k_l-\epsilon, k_l+2\epsilon]$. Then, for $k=k_l-\epsilon, k_l,k_l+\epsilon$ and $a \in \cA_{Q_{\epsilon}(\bm{k})}$, we have   
    \[
        |\omega_0 (\gamma_{k} (a)) - \omega_0( a )| \leq J_{\sfG_{\sfH},\nu,\mu} \cdot f_{\nu,\mu} (\epsilon s) \cdot \| a \|.
    \]
\end{lem}
\begin{proof}
    We begin by observing that, for any $k=k_l-\epsilon, k_l,k_l+\epsilon$, $t \in [0,\varepsilon]$, and $a \in \cA_{Q_{\epsilon}(\bm{k})}$, we have
    \[
    \omega_{0} \circ \alpha \big( \Theta_{Q_{2\epsilon}(\bm{k}) \cap Y_k , \omega }(\sfG_{\sfH}) \,; k,k+t \big)  (a) = \omega_0(a). 
    \]
    Indeed, this is verified as
    \begin{align*}
        {}&{} \frac{d}{dt}\omega_{0} \circ \alpha \big( \Theta_{Q_{2\epsilon}(\bm{k}) \cap Y_k , \omega}(\sfG_{\sfH}) \,; k,k+t \big)  (a) \\
        ={}&{} \omega_0 \big( \big[ \Theta_{Q_{2\epsilon}(\bm{k}) \cap Y_k , \omega} (\sfG_{\sfH}(k+t)) ,  \alpha \big( \Theta_{Q_{2\epsilon}(\bm{k}) \cap Y_k , \omega}(\sfG_{\sfH}) \,; k,k+t \big)  (a)  \big] \big)\\
        ={}&{} \omega_0 \big( \big[ \Theta_{Q_{2\epsilon}(\bm{k}) \cap Y_k, \omega} (\sfG_{\sfH}(k+t)) , (\id_{\cA_{Q_{2\epsilon}(\bm{k}) \cap Y_k }} \otimes \omega_0|_{_{\cA_{Q_{2\epsilon}(\bm{k}) \cap Y_k^c }}}) \circ  \alpha \big( \Theta_{Q_{2\epsilon}(\bm{k}) \cap Y_k, \omega}(\sfG_{\sfH}) \,; k,k+t \big)  (a)  \big] \big)\\
        ={}&{} \omega_{k+t} \big( \big[ \sfG_{\sfH}(k+t) , (\id_{\cA_{Q_{2\epsilon}(\bm{k}) \cap Y_k }} \otimes \omega_0|_{_{\cA_{Q_{2\epsilon}(\bm{k}) \cap Y_k^c }}}) \circ  \alpha \big( \Theta_{Q_{2\epsilon}(\bm{k}) \cap Y_k, \omega}(\sfG_{\sfH}) \,; k,k+t \big)  (a))  \big] \big) \\
        ={}&{} -\bigg(\frac{d}{dt}\bigg|_{t=0}\omega_{k+t}\bigg) \big( (\id_{\cA_{Q_{2\epsilon}(\bm{k}) \cap Y_k }} \otimes \omega_0|_{_{\cA_{Q_{2\epsilon}(\bm{k}) \cap Y_k^c }}}) \circ  \alpha \big( \Theta_{Q_{2\epsilon}(\bm{k}) \cap Y_k, \omega}(\sfG_{\sfH}) \,; k,k+t \big)  (a))  \big) =0 .
    \end{align*}
    In the last two equalities, we used \cref{thm:automorphic.equivalence} and the constancy of $\omega_{k+t}|_{\cA_{Q_{2\epsilon}(\bm{k})}}$.   
    
    This observation shows that the latter term in the right hand side of
    \begin{align*}
         {}&{} |\omega_0 (\gamma_{k}(a)) - \omega_0( a) | \\
        = {}&{}  \big| \omega_0 \big( \alpha \big( \Theta_{Y_{k},\omega}(\sfG_{\sfH}) \,; k,k+\epsilon)(a) -a \big) \big| \\
        \leq {}&{} \big| \omega_0 \big( \alpha \big( \Theta_{Y_{k},\omega}(\sfG_{\sfH}) \,; k,k+\epsilon)(a) -\alpha \big( \Theta_{Y_{k} \cap Q_{2\epsilon} (\bm{k}),\omega}(\sfG_{\sfH}) \,; k,k+\epsilon)(a) \big) \big| \\
        {}&{} \quad +  \big| \omega_0 \big( \alpha \big( \Theta_{Y_{k} \cap Q_{2\epsilon}(\bm{k}),\omega}(\sfG_{\sfH}) \,; k,k+\epsilon)(a) -a \big)\big|
    \end{align*}
    vanishes. By \cref{lem:mesh.approximation}, the first term is bounded above as desired. 
\end{proof}

\begin{lem}\label{lem:gradation.state.estimate}
    Let $\sfH \in \sIP_{\loc} (\bR^l \midbar [0,1])$ satisfies that $\omega_{v} |_{\cA_{Q_{2\epsilon}(\bm{k})}} = \omega_0$ for any $v \in [k_l-\epsilon, k_l+2\epsilon]$. Then, for any $a \in \cA_{Q(\bm{k})}$, we have 
\[
    |\omega_0 (\gamma_{[0,1]} (a)) - \omega_0(a)) | \leq  (3\epsilon^{-1} + 5 ) \cdot J_{\sfG_{\sfH},\nu,\mu} \cdot f_{\nu,\mu}(s \epsilon ) \cdot \| a\| .
\]
\end{lem}
\begin{proof}
    First, since $Y_k \supset Q_{2\epsilon}(\bm{k})$ for any $k =0,\cdots, k_l-2\epsilon$, we have
    \begin{align*}
        {}&{} \| \gamma_{[0,k_l-\epsilon]} (a) - \alpha_{[0,k_l-\epsilon ]}(a) \| \\
        \leq {}&{} \sum_{k\in [0 , k_l - 2\epsilon] \cap \epsilon\bZ }\| (\gamma_{k} - \alpha_{k}) \circ \alpha_{[0,k-1]}(a) \| \\
        \leq {}&{}\sum_{k\in [0 , k_l - 2\epsilon] \cap \epsilon\bZ } \Big(   \| (\gamma_{k} - \alpha_{k}) \circ \Pi_{Q_{\epsilon}(\bm{k})}(\alpha_{[0,k-1]}(a)) \| + 2 \cdot \| \alpha_{[0,k-1]}(a) - \Pi_{Q_{\epsilon}(\bm{k})}(\alpha_{[0,k-1]}(a)) \| \Big) \\
        \leq {}&{}\sum_{k\in [0 , k_l - 2\epsilon] \cap \epsilon\bZ } \Big(  J_{\sfG_{\sfH},\nu,\mu} \cdot f_{\nu,\mu}(\epsilon s) \cdot \| \Pi_{Q_{\epsilon}(\bm{k})}(\alpha_{[0,k-1]}(a)) \| + 2 \cdot J_{\sfG,\nu,\mu} \cdot f_{\nu,\mu}(\epsilon s)  \cdot  \|a \| \Big)\\
        \leq {}&{}3(k_l/\epsilon-2) \cdot J_{\sfG_{\sfH},\nu,\mu}\cdot f_{\nu,\mu}(\epsilon s)  \cdot  \|a \| ,
    \end{align*}
    and hence 
    \begin{align}
    \begin{split}
        |\omega_0 \circ \gamma_{[0,k_l-\epsilon]} (a) - \omega_0 (a)| \leq {}&{} |\omega_0 \circ (\gamma_{[0,k_l-\epsilon]}  - \alpha_{[0,k_l-\epsilon ]})(a)|\\
        \leq {}&{} 3(k_l/\epsilon-2) \cdot J_{\sfG_{\sfH},\nu,\mu}\cdot f_{\nu,\mu}(\epsilon s)  \cdot  \|a \|.
    \end{split}\label{eqn:gradation.estimate.1}
    \end{align}
    Second, 
    \begin{align}
    \begin{split}
        {}&{} |\omega_0 \circ \gamma_{[0,k_l+2\epsilon]} (a) -\omega_0 \circ \gamma_{[0,k_l-\epsilon]} (a)| \\
        \leq {}&{} \sum_{k=k_l-\epsilon,k_l, k_l+\epsilon} 2\| \gamma_{[0,k]} (a) - \Pi_{Q_{\epsilon}(\bm{k})} (\gamma_{[0,k]} (a) ) \| + |(\omega_0 \circ \gamma_{k} - \omega_0) (\Pi_{Q_{\epsilon}(\bm{k})} (\gamma_{[0,k]} (a) )| \\
        \leq {}&{} 9 \cdot J_{\sfG,\nu,\mu} \cdot f_{\nu,\mu}(\epsilon s) \cdot \| a \|.
    \end{split}\label{eqn:gradation.estimate.2}
    \end{align}
    Finally, 
    \begin{align*}
         {}&{}\| \gamma_{[0,1]} (a) - \gamma_{[0,k_l+2\epsilon]}(a) \| \\
        \leq {}&{} \big\| (\gamma_{[k_l+2\epsilon,1]} - \id)\big( \gamma_{[0,k_l+2\epsilon]}(a) - \Pi_{Q_{\epsilon}(\bm{k})}(\gamma_{[0,k_l+2\epsilon]}(a)) \big) \big\|  \leq 2J_{\sfG,\nu,\mu}  \cdot f(\epsilon s) \cdot \| a\|  ,
    \end{align*}
    and hence 
    \begin{align}
        |\omega_0 \circ \gamma_{[0,1]} (a) - \omega_0 \circ \gamma_{[0,k_l+2\epsilon]}(a)|\leq 2J_{\sfG,\nu,\mu}  \cdot f(\epsilon s) \cdot \| a\| . \label{eqn:gradation.estimate.3}
    \end{align}
    In summary, by \eqref{eqn:gradation.estimate.1}, \eqref{eqn:gradation.estimate.2}, and \eqref{eqn:gradation.estimate.3}, we obtain the desired estimate.
\end{proof}
Let $\bm{x}_0 \coloneqq (1/2,\cdots,1/2) \in \bR^l$. For $I \in 2^l$, we write $\bm{x}_I$ for the vertex of $[0,1]^l$ whose $i$-th coordinate is $1$ if and only if $i \in I$. We define the half spaces $Y_{L,i}$ and the conical regions $Z_{R,i}$ as
\[
    Y_{L,i} \coloneqq \{ \bm{x} \in \bR^d \mid x_i \leq 0\}, \quad Z_{R,i} \coloneqq \Big\{ \bm{x}_0 + \sum_{I_i =1} c_{I} (\bm{x}_I - \bm{x}_0) \mid c_I \geq 0 \Big\} \setminus [0,1]^l \subset \bR^l. 
\]
Then $Y_{L,i}$, $Z_{R,i}$ and $[0,1]^l$ intersects at their boundaries and their union covers $\bR^l$.
For $\bm{k}_1 \in (\epsilon \bZ \cap [0,1])^{l-1}$, set
\[
    W_{\bm{k}_1,r} \coloneqq \bigcup_{\substack{\pr_{\bR^{l-1}} (Q(\bm{k})) \subset Q_{r}(\bm{k}_1)  \\ Q_{r}(\bm{k}) \cap U_X = \emptyset}}Q(\bm{k}).
\]
Let $ \{ \rho_{\bm{k}} \}_{\bm{k}}$ be a smooth partition of unity on $[0,1]^{l-1}$, indexed by $\bm{k}_1 \in (\epsilon \bZ \cap [0,1])^l$, that each $\rho_{\bm{k}_1}$ is supported on the $(l-1)$-dimensional small cube $Q_{\epsilon}(\bm{k}_1)$. 

Let $\sfH \in \sIP_{\loc }(\bR^l \midbar \sM \times  \sD_X )$, that is, $\sfH$ is a smooth family of IG localization flows parametrized by $\sM \times [0,1]^l$ whose restriction to the complement $\sM \times \overline{U}{}_X^c$ is trivial.   
For $p_1 = (x_1,\cdots,x_{l-1}) \in \bR^{l-1}$, we write $\gamma_{[0,1]}^{p_1}$ for the LG automorphism in \cref{lem:gradation.state.estimate} associated to the smooth path $\sfH(p_1,\blank)$. 
By \cref{lem:gradation.state.estimate}, there is $s_0 \in \bR_{\geq 1}$ such that 
\begin{align*}
    \Gamma_X^{(1)} (\sfH) (p_1) \coloneqq \sum_{\bm{k}_1 \in (\epsilon \bZ \cap [0,1])^{l-1}} \rho_{\bm{k}_1} (p_1) \cdot  \Theta_{(Z_{R,l} \cup W_{\bm{k}_1,2\epsilon })^c} \big( \gamma_{[0,1]}^{p_1}(\sfh) \big) 
\end{align*}
is gapped for $s \geq s_0$. After the parameter shift by $s_0$ (\cref{rmk:parameter.shift}), this gives a smooth family of IG localization flows parametrized by $[0,1]^{l-1}$. Its value at $p_1$ trivial on $Y_{L,l}$, on $Z_{R,l}$, and on $Q(\bm{k})$ if $p_1 \in \pr_{\bR^{l-1}}(Q(\bm{k}))$ and $Q_{2\epsilon}(\bm{k}) \subset \overline{U}{}_{X}^c$. 

We can iterate this construction $l$ times. 
The $m$-th step $\Gamma_X^{(m)}(\sfH)$ is a $(l-m)$-parameter family of IG localization flows parametrized by $[0,1]^{l-m}$. Its value at $p_m$ is trivial on $Y_{L,i}$ and $Z_{R,i}$ for $l-m+1 \leq i \leq l$, and on $Q(\bm{k})$ if $p_m \in \pr_{\bR^{l-m}}(Q(\bm{k}))$ and $Q_{2m\epsilon }(\bm{k}) \subset \overline{U}{}^c_X$. In particular, $\Gamma_{X}^{(l)}(\sfH)$ is supported on
\[
    \bigg(\bigcup_{i=1}^l (Y_{L,i} \cup Z_{R,i}) \cup \bigcup_{Q_{2l\epsilon}(\bm{k}) \cap U_X = \emptyset } Q(\bm{k}) \bigg)^c = U_{X,2l\epsilon} \simeq X. 
\]
As a result, we obtain a morphism
\begin{align*}
    \Gamma_X^{(l)} \colon \sIP_{\loc}[\bR^l \midbar \sM \times \sD_X] \to \sIP_{\loc}[U_{X,2l\epsilon} \midbar \sM] \cong \sIP_{\loc}[X \midbar \sM], 
\end{align*}
that is natural with respect to the entry $\sM$. 
Moreover, if $A$ is a subcomplex of $X$ and $\sfH$ is supported on a deformation retract neighborhood $U_A$ of $A$, i.e., $\sfH (p) = \sfh$ for $p \in U_X \setminus U_A$, then $\Gamma_X (\sfH)$ is contained in $\sIP_{\loc}[A \midbar \sM]$. This means that the diagram
    \begin{align}
    \begin{split}
    \xymatrix{
        \sIP_{\loc}[\bR^l \midbar  \sM \times \sD_A] \ar[r]^{\hspace{3ex}\Gamma_A^{(l)}}  \ar[d]^{j^* } & \sIP_{\loc}[A \midbar \sM] \ar[d]^{\iota_*} \\
        \sIP_{\loc}[\bR^l \midbar \sM \times \sD_X] \ar[r]^{\hspace{3ex}\Gamma_X^{(l)}} &  \sIP_{\loc}[X \midbar \sM]
        }
    \end{split}\label{eqn:diagram.commute.duality}
    \end{align}
    commutes, where $j \colon \sD_X \to \sD_A$ denotes the quotient.

\begin{lem}\label{lem:duality.trivial}
    If $X$ is contractible, then $\Gamma_X^{(l)}$ is an isomorphism. 
\end{lem}
\begin{proof}
    Consider the diagram \eqref{eqn:diagram.commute.duality} for the inclusion $\pt \to X$. In this case, the vertical maps are isomorphisms by the homotopy invariance of both groups. 
    Moreover, 
\begin{align*}
    \Gamma_{\pt}^{(l)} \colon \sIP_{\loc}[\bR^l \midbar  \sM \times \sD_{\pt}] \to \sIP_{\loc}[\pt \midbar \sM]
\end{align*}
is nothing else than the $l$ times compositions of morphisms $\vartheta_{\loc}$, which are already proved to induce isomorphisms. 
\end{proof}

\begin{proof}[Proof of \cref{thm:Atiyah.duality}]
    By \eqref{eqn:diagram.commute.duality} and \cref{lem:duality.trivial}, we can apply the Mayer--Vietoris argument in \cref{cor:mayer-vietoris} for the morphism $\Gamma_X^{(l)}$ with respect to the cell attachment in the homological variable $X$. In conclusion, we obtain that 
\begin{align*}
    \Gamma_{X}^{(l)} \colon \sIP_{\loc}[\bR^l \midbar  \sM  \times  \sD_{X}] \to \sIP_{\loc}[X \midbar \sM]
\end{align*}    
is an isomorphism for any $\sM \in \Man$. 

Finally, we give rise to weak equivalence of sheaves. Let $\sHom(\sD_X, \sIP_{\loc,l})_{\spadesuit}$ denote the subsheaf of $\sHom(\sD_X, \sIP_{\loc,l})$ consisting of IG localization flows $\sfH$ such that the almost local norms $\vvert \sfH\vvert_{\loc,f}$ are small enough that $\Gamma_X^{(l)}(\sfH)$ is gapped for any $s \geq 1$ without any parameter shift. Consider the morphisms
\begin{align*}
    \sHom(\sD_X, \sIP_{\loc,l}) \hookleftarrow \sHom(\sD_X, \sIP_{\loc,l})_{\spadesuit} \xrightarrow{\Gamma_X^{(l)}} \sIP_{\loc}(X).
\end{align*}
The parameter shift homotopy in \eqref{rmk:parameter.shift} shows that the left morphism is a weak equivalence. 
Moreover, what is proved in the above paragraph is nothing but the weak equivalence of the right morphism. This finishes the proof.  
\end{proof}

\section{Invertible gapped phases with spatial symmetry}\label{section:assembly}
In this section, we study the topology of the invertible phases protected by the spatial symmetry given by a crystallographic group $\Gamma$. 
This class includes both translation symmetries and point group symmetries. 
We formulate the equivariant sheaves, spectra, and cohomology functors representing the spatial symmetries acting on the spaces of IG UAL Hamiltonians and IG localization flows. 
The former is a more natural object to consider, while the latter is easier to handle and aligns better with previous research. 
Their natural comparison is called the equivariant coarse assembly map. We show that it is split injective. 
We also discuss the provenance of this map. 
In fact, it coincides with the Davis--L\"{u}ck assembly map in equivariant algebraic topology, and hence is closely related to the Baum--Connes and the Farrell--Jones maps.

\subsection{Equivariant \texorpdfstring{$\Omega$}{Omega}-spectra of invertible phases with spatial symmetry}
Let $\Gamma$ be a crystallographic group, i.e., a countable group acting faithfully, properly and cocompactly on the Euclidean space $\bE_{\Gamma} \coloneqq \bR^m$. In other words, $\Gamma$ is a discrete cocompact subgroup of $\bR^m \rtimes O(m)$. 
Then, by the Bieberbach theorem (see e.g.\ \cite{charlapBieberbachGroupsFlat1986}*{Theorem 3.1}), the group $\Gamma$ has the finite index normal subgroup $N \coloneqq \Gamma \cap \bR^m$ that is isomorphic to $\mathbb{Z}^m$ acting on $\bE_{\Gamma}$ by translation. 

\begin{rmk}\label{rmk:classifying.space}
    We write $\cF_{\fin} $ for the class of finite subgroups of $\Gamma$. A $\Gamma$-$\cF_{\fin}$-CW-complex is the same thing as a proper $\Gamma$-CW-complex (\cite{davisSpacesCategoryAssembly1998}*{Section 6}). The $\Gamma$-space $E(\Gamma, \cF_{\fin})$ is the universal example of $\Gamma$-$\cF_{\fin}$-CW-complexes (cf.\ \cite{baumClassifyingSpaceProper1994}*{Section 2} and \cite{luckSurveyClassifyingSpaces2005}), i.e., any $\Gamma$-$\cF_{\fin}$-CW-complex $X$ has a $\Gamma$-equivariant map $X \to E(\Gamma , \cF_{\fin})$ that is unique up to homotopy. This $E(\Gamma,\cF_{\fin})$ is characterized by the property that the subspaces $E(\Gamma,\cF_{\fin})^K$ are all contractible for any $K \in \cF_{\fin}$ (\cite{luckSurveyClassifyingSpaces2005}*{Theorem 1.9}). 
    Therefore, when $\Gamma$ is a crystallographic group, the Euclidean space $\bE_{\Gamma}$ models $E(\Gamma,\cF_{\fin})$.
\end{rmk}

\begin{defn}\label{defn:MCW.equivariant}
    We write $\ECW_{\Gamma,\mathrm{p}}$ for the category whose object is a proper $\Gamma$-CW-complex $X$ equipped with a $\Gamma$-invariant proper metric $\rmd_X$ such that
    \begin{enumerate}
        \item $(X,\rmd_X) \in \ECW$ in the sense of \cref{defn:MCW}, and
        \item there is $w_X \geq 1 $ such that $w_X^{-1} \rmd_{\Gamma}(g,h) \leq \rmd(g\bm{x},h\bm{x}) \leq w_X \rmd_{\Gamma}(g,h)$ for any $g,h \in \Gamma$ and $\bm{x} \in X$. 
    \end{enumerate}
    A morphism between $X,Y \in \ECW_{\Gamma , \mathrm{p}}$ is a linearly proper Lipschitz continuous map that is $\Gamma$-equivariant. We write $\ECW_{\Gamma, \mathrm{fp}} $ for the full subcategory consisting of finite proper $\Gamma$-CW-complexes.
\end{defn}

\begin{rmk}
    In the above definition of $\ECW_{\Gamma ,\mathrm{p}}$, we do not assume for the embedding $X \hookrightarrow \bR^{l_X}$ to be equivariant. However, this relaxation does not allow the class of discrete groups that can be treated as the spatial symmetry $\Gamma$ to extend beyond crystallographic groups. 
    Indeed, \cref{defn:MCW.equivariant} requires the group $\Gamma$ to be quasi-isometrically embeddable into a Euclidean space. Such a group, which must have polynomial growth, is virtually nilpotent and has Shalom's property $H_{FD}$ (Gromov's polynomial growth theorem, \cites{gromovGroupsPolynomialGrowth1981,shalomHarmonicAnalysisCohomology2004,ozawaFunctionalAnalysisProof2018}). 
    By \cite{decornulierIsometricGroupActions2007}*{Theorem 1.4}, $\Gamma$ acts properly and cocompactly on a Euclidean space. That is, $\Gamma$ is an extension of a finite group by a crystallographic group.  
\end{rmk}

\begin{rmk}\label{rmk:equivariant.universal}
    We remark a refined version of the universality of the proper $\Gamma$-space $\bE_{\Gamma} =E(\Gamma,\cF_{\fin})$ in \cref{rmk:classifying.space}. For any $X \in \ECW_{\Gamma, \mathrm{p}}$, there is a Lipschitz continuous $\Gamma$-map $F \colon X \to \bE_{\Gamma}$. 
    Such $F$ is constructed explicitly in the following way. 
    Let $\Lambda $ be a $\Gamma$-invariant $(R,Q)$-quasilattice of $X$ and let $F' \colon \Lambda \to \bE_{\Gamma}$ be a $\Gamma$-equivariant map. By the assumptions $w_X>0$, $F'$ is automatically large-scale Lipschitz. 
    Let $\{ \rho_{\bm{x}} \}_{\bm{x} \in \Lambda}$ be the associated $\Gamma$-equivariant partition of unity subordinate to $\{ B_{R}(\bm{x}) \}_{\bm{x} \in \Lambda}$. 
    Since $X$ is bi-Lipschitz embeddable into an Euclidean space and $\{ \rho_{\bm{x}} \}_{\bm{x} \in \Lambda}$ is uniformly locally finite (any $\bm{x} \in X$ is contained in at most $N$ open balls), we may choose $\rho_{\bm{x}}$'s to be $A$-Lipschitz by the same $A>0$. 
    Then, $F(\bm{x}) \coloneqq \sum_{\bm{y} \in \Lambda} \rho_{\bm{y}}(\bm{x}) \cdot F'(\bm{y})$ gives a Lipschitz continuous map to $\bE_{\Gamma}$.
\end{rmk}

\begin{lem}\label{rmk:embedding.equivariant}
    Let $X \in \ECW_{\Gamma, \mathrm{p}}$. Then, there is a representation $\bV$ of $\Gamma/N$ and a Lipschitz continuous $\Gamma$-map $\iota_X \colon X\to \bV^{n_X}$ such that, for any $Y \in \ECW_{\Gamma,\mathrm{p}}$ and a $\Gamma$-equivariant linearly proper Lipschitz continuous map $F \colon X \to Y$, the product map $\widetilde{F}\coloneqq F \times \iota_X \colon X \to Y \times \bV^{n_X}$ is an injective linearly proper bi-Lipschitz embedding. In particular, $X$ is $\Gamma$-equivariantly embeddable into $\bE_{\Gamma}\times\bV^{n_X}$.  
\end{lem}
\begin{proof}
    Let $\bV$ be the group ring $\bR[\Gamma/N]$ regarded as the Euclidean space equipped with the $\ell^2$-distance, on which $\Gamma$ acts isometrically through the quotient $\Gamma \to \Gamma/N$. 
    Then, the finite ${\Gamma/N}$-CW-complex $X/N$ admits a ${\Gamma/N}$-equivariant bi-Lipschitz embedding into $\bV_{\Gamma/N}^{n_X}$ for some $n_X >0$. 
    Indeed, if $\iota_0 \colon X/N \to \bR^{n_X}$ is a non-equivariant bi-Lipschitz embedding, then 
    \[
    \iota(\bm{x}) \coloneqq (\iota_0(g^{-1}\bm{x}))_{g \in {\Gamma/N}} \in (\bR^{n_X})^{{\Gamma/N}} = \bV^{n_X}
    \]
    is an equivariant embedding. Let $\iota_X \coloneqq \iota \circ q$.
    Then, the product $F \times \iota_X \circ q \colon X \to Y \times \bV_{\Gamma/N}^{n_X}$ is injective, $\Gamma$-equivariant, and Lipschitz continuous. It is also bi-Lipschitz. Indeed, by taking $A \geq 1$ such that $\rmd(F(\bm{x}), F(\bm{y})) \leq A\rmd(\bm{x},\bm{y})$ and $g \in \Gamma$ attaining the minimum of $\rmd(g\bm{x},\bm{y})$, we have $(2A)^{-1} \leq 1$ and $\rmd(\iota_X(\bm{x}),\iota_X(\bm{y})) = \rmd(g\bm{x},\bm{y})$, and hence
    \begin{align*}
        {}&{} \rmd(F(\bm{x}) , F(\bm{y})) + \rmd(\iota_X(\bm{x}), \iota_X(\bm{y})) \\
        \geq {}&{}  (2A)^{-1} (\rmd(F(\bm{x}),F(g\bm{x})) - \rmd(F(g\bm{x}), F(\bm{y}))) + \rmd(g\bm{x},\bm{y})\\
        \geq {}&{} (2A)^{-1}w_X^{-1}w_Y^{-1} \rmd(\bm{x},g\bm{x}) + 2^{-1} \rmd(g\bm{x},\bm{y}) \geq \min \{ (2A)^{-1}w_X^{-1}w_Y^{-1}, 2^{-1}\} \rmd(\bm{x},\bm{y}). 
    \end{align*}
    The latter claim follows from \cref{rmk:equivariant.universal}. 
\end{proof}

\begin{rmk}
As a consequence of the equivariant bi-Lipschitz embeddability of $X$ into a Euclidean space, an equivariant version of \cref{rmk:embed.CWfin.EucCW} also holds; there is a full subcategory $\CW_{\Gamma, \mathrm{fp}}^{\iota} \subset \ECW_{\Gamma, \mathrm{fp}}$ such that the forgetful functor $\CW_{\Gamma,\mathrm{fp}}^\iota \to \CW_{\Gamma,\mathrm{fp}}$, where $\CW_{\Gamma,\mathrm{fp}}$ denotes the category of finite proper $\Gamma$-CW-complexes and continuous $\Gamma$-maps, is essentially surjective and induces isomorphisms of homotopy sets $[X,Y]_{\mathrm{Lip}}^\Gamma \cong [X,Y]^{\Gamma}$.
\end{rmk}

Let $\mathsf{Sub}_{\fin} (\Gamma)$ denote the category whose objects are finite subgroups of $\Gamma$ and whose morphisms are group homomorphisms of the form $\Ad (g^{-1}) \colon H \to g^{-1}Hg \subset K$ by $g \in \Gamma$. 
We describe a set of internal degrees of freedom on a proper $\Gamma$-space in terms of functors from this category. 
For finite groups, $H_1$ and $H_2$, let $\fR_{H_1}$ and $\fR_{H_2}$ be a set of internal degrees of freedom with on-site symmetry of $H_1$ and $H_2$, respectively. 
A morphism $\fR_{H_1} \to \fR_{H_2}$ is given by a triple $(\psi, \eta, v)$, where
\begin{itemize}
    \item $\psi \colon H_1 \to H_2$ is a group homomorphism, 
    \item $\eta \colon \fR_{H_1} \to \fR_{H_2}$ is a map, and 
    \item $v_\lambda \colon \sH_{\lambda} \to \sH_{\eta (\lambda)}$ is a family of unitaries such that $v_{\lambda} \Omega_{\lambda} = \Omega_{\eta \lambda}$ and $v_\lambda u_{\lambda}(g) v_{\lambda}^* = u_{\eta(\lambda)}(\psi(g))$. 
\end{itemize}
This morphism set forms the category $\mathsf{ID}$ of internal degrees of freedom with finite group symmetry. Both $\mathsf{Sub}_{\fin} (\Gamma)$ and $\mathsf{ID}$ has the forgetful functors to the category $\mathsf{Grp}_{\fin}$ of finite groups.  

\begin{defn}\label{defn:equivariant.coarse}
Let $\Gamma$ be a crystallographic group and let $X \in \ECW_{\Gamma,\mathrm{p}}$. Let $\bV$ be the representation of $\Gamma/N$ as in \cref{rmk:embedding.equivariant}.  
    \begin{enumerate}
        \item A set of internal degrees of freedom with spatial $\Gamma$-symmetry is a functor $\fR \colon \mathsf{Sub}_{\fin}(\Gamma) \to \mathsf{ID}$ that commutes with the forgetful functors to $\mathsf{Grp}_{\fin}$, i.e., $\mathsf{fg} \circ \fR=\mathsf{fg}$. More explicitly, each $\fR_K$ has the internal symmetry of $K$ and the morphism $(\psi_g, \eta_g,v_g) \colon \fR_H \to \fR_K$ induced by $g \in \Gamma$ satisfies $\psi_g = \Ad (g^{-1})$. By abuse of notation, we write 
        \begin{align*} 
        X \times \fR \coloneqq \bigsqcup_{\bm{x} \in X} \fR_{\Gamma_{\bm{x}}}.
        \end{align*}
        \item Let ${}_\star\fL_{\fR,X,\Gamma}$ denote the set of $\Gamma$-invariant subsets $\Lambda_\star$ of $X \times \bV^\infty \times \fR \times \bN$ satisfying (i), (ii) of \cref{defn:lattice.set,defn:gapped.Hamiltonian.X}, while $\fL_{\fR,X,\Gamma}$ denote the set of subsets $\Lambda$ of $X \times \bV^\infty \times \fR \times \bN$ satisfying (i), (ii), (iii) of \cref{defn:lattice.set,defn:gapped.Hamiltonian.X}. 
        For $\Lambda \in {}_\star\fL_{\fR,X,\Gamma}$, the C*-algebra $\cA_{\Lambda} \coloneqq \bigotimes_{\bm{x} \in \Lambda} \cA_{\lambda(\bm{x})}$ is equipped with a $\Gamma$-action 
        \begin{align*} 
        \beta_g \coloneqq \bigotimes_{\bm{x} \in \Lambda_\star}  \big( \Ad (v_{g , \lambda(\bm{x})}) \colon \Aut(\cA_{\bm{x}}) \to \Aut(\cA_{g\bm{x}}) \big) 
        \end{align*}
        by $\ast$-automorphisms preserving the almost local subalgebra $\cA_{\Lambda_\star}^{\al}$. Moreover, $\beta_g$ preserves the almost local norms $\vvert \blank \vvert_{f} $ on $\fDer ^{\al}_{\Lambda_{\star}}$ and $\fH_{\Lambda_{\star}}^{\al}$.  
        \item For $\Lambda \in {}_\star \fL_{\fR,X,\Gamma}$, we say that $\sfH \in \fH_{\Lambda_{\star}}^{\al}$ is \emph{almost properly supported} if there is a subset $\Lambda^\circ \leq \Lambda$ contained in the class $\fL_{\fR,X,\Gamma}$ such that the constant $K_{\Lambda, \sfH-\sfh,\nu,\mu}^{(k)} $ in \eqref{eqn:constant.K} is finite for any $f_{\nu,\mu} \in \cF_1$, $k \in \bN$ and $\sU$. 
        for any $f \in \cF$. We write ${}_\star \sGP(\Lambda \midbar \sM)$ for the set of smooth families of almost properly supported gapped UAL Hamiltonians. This forms a $\Gamma$-sheaf on $\Man$. We define the colimit $\Gamma$-sheaf 
        \begin{align*}
            {}_\star\sGP (X\midbar \blank) \coloneqq \colim_{\Lambda \in {}_\star \fL_{\fR,X,\Gamma}}{}_\star \sGP (\Lambda\midbar \blank ). 
        \end{align*}
        As we have repeated in this paper, as in \eqref{eqn:composition.lattice} and \eqref{eqn:composition.system}, the composite operation $\blank \boxtimes \blank$ makes ${}_\star\sGP(X\midbar \blank)$ a local commutative H-monoid with a weakly strict unit. We define the invertible subsheaf ${}_\star\sIP(X\midbar \blank)$ and its refinement ${}_\star\bsIP(X\midbar \blank)$ as in \cref{lem:bold.invertible.sheaf} ,both of which are $\Gamma$-sheaves.
        \item Let ${}_\star\fL_{\fR,X,\Gamma}^{\loc}$ denote the set of lower semi-continuous family ${}_\star\mathbbl{\Lambda} = \{ \Lambda_{\star,s} \}_{s \in [1,\infty)}$ of subsets of $X \times \bV^\infty \times \fR \times \bN^2$ such that each $\Lambda_s$ is $\Gamma$-invariant and ${}_\star \mathbbl{\Lambda}$ satisfies (i), (ii) of \cref{defn:localizing.path}. 
        \item We say that a gapped localization flow of UAL Hamiltonians $\sfH $ supported on $\mathbbl{\Lambda} \in {}_\star \fL_{\fR,X,\Gamma}^\loc$ is almost properly supported if there is a family $\mathbbl{\Lambda}^\circ \in \fL_{\fR,X,\Gamma}^\loc$ such that $\Lambda_s^\circ \leq \Lambda_{s}$ and the constant $K_{\loc,\Lambda, \sfH-\sfh,\nu,\mu}^{(k)} $ in \eqref{eqn:constant.K.local} is finite for any $f_{\nu,\mu} \in \cF_1$, $k \in \bN$ and $\sU$. 
        We write ${}_\star \sGP_{\loc}(\mathbbl{\Lambda} \midbar \sM)$ for the set of smooth families of almost properly supported localization flows. This forms a $\Gamma$-sheaf on $\Man$. As well as (3), for $X \in \ECW^\Gamma$, we define the $\Gamma$-sheaf 
        \[ 
        {}_{\star} \sGP_{\loc}(X\midbar \blank)\coloneqq  \colim_{\mathbbl{\Lambda} \in {}_\star \fL_{\fR,X,\Gamma}^\loc }{}_{\star} \sGP_{\loc}(\mathbbl{\Lambda}\midbar \sM),
        \]
        regard it as a local commutative H-monoid with a weakly strict unit, and define the invertible subsheaf ${}_{\star} \sIP_{\loc}(X\midbar \blank)$ and its refinement ${}_{\star} \bsIP_{\loc}(X\midbar \blank)$. 
        \item For a $\Gamma$-equivariant Lipschitz continuous map $F \colon X \to Y$, the push-forward $F_*$ is defined in the same way as \eqref{eqn:push.quantum.system}, \cref{lem:MCW.induced.hom,prp:localizing.induced.hom} as
        \begin{align*}
        F_* \coloneqq {}&{} \colim_{\Lambda \in \fL_{\fR,X,\Gamma}} \Big( \widetilde{F}_* \colon \pst\sGP(\Lambda \midbar \sM ) \to \pst\sGP(\widetilde{F}(\Lambda) \midbar \sM) \Big) \colon \sGP(X \midbar \sM) \to \pst\sGP(Y \midbar \sM), \\
       F_* \coloneqq {}&{}  \colim_{\mathbbl{\Lambda} \in \fL_{\fR,X,\Gamma}^{\loc}} \Big( \widetilde{F}_* \colon \pst\sGP_{\loc}(\mathbbl{\Lambda} \midbar \sM ) \to \pst\sGP_{\loc}(\widetilde{F}(\mathbbl{\Lambda}) \midbar \sM) \Big) \colon \sGP_{\loc}(X \midbar \sM) \to \pst\sGP_{\loc}(Y \midbar \sM).
        \end{align*}
        Here, the map $\iota_X$ in $\widetilde{F} \coloneqq F \times \iota_X $ is given in \cref{rmk:embedding.equivariant}. 
    \end{enumerate}
\end{defn}

We remark that the representation $\bV = \bR[\Gamma/N]$ includes a direct summand $\bR \cdot 1$, and hence the space $X \times \bV^\infty \times \fR \times \bN^2$ contains the subspace $X \times \bR^\infty \times \fR \times \bN^2$. 
This enables us to apply the arguments in \cref{section:Kitaev} involving the Kitaev pump. 

As a consequence, all the results in \cref{section:Kitaev,section:localizing.path}
 are generalized to $\Gamma$-equivariant versions.
\begin{thm}\label{thm:equivariant.stable.homotopy.IP}
    Let $X, Y \in \ECW_{\Gamma, \mathrm{p}}$. The following hold. 
    \begin{enumerate}
        \item The Kitaev pump $\kappa_X$ and the localized Kitaev pump $\kappa_X^\loc$, both are $\Gamma$-equivariant, makes both  ${}_\star \sIP(X\midbar \blank)$ and ${}_\star \sIP_{\loc}(X\midbar \blank)$ naive $\Gamma$-equivariant $\Omega$-spectrum object of $\Sh(\Man)$.
        \item If large-scale Lipschitz $\Gamma$-maps $F_0, F_1 \colon X \to Y$ are $\Gamma$-equivariantly large-scale Lipschitz homotopic, then the induced maps $F_{0,*}, F_{1,*} \colon \sIP^\Gamma(X) \to \sIP^\Gamma(Y)$ are smoothly homotopic.
        Similarly, if Lipschitz continuous $\Gamma$-maps $F_0, F_1 \colon X \to Y$ are $\Gamma$-equivariantly Lipschitz homotopic, then the induced maps $F_{0,*}, F_{1,*} \colon \sIP_{\loc}^\Gamma(X) \to \sIP_{\loc}^\Gamma(Y)$ are smoothly homotopic.
        \item The assignment $X \mapsto \rIP_{\loc,n}^\Gamma (X) \coloneqq \pi_n({}_\star \sIP_{\loc}^\Gamma(X))$ gives a functor $\ECW_{\Gamma,\mathrm{p}} \to \mathsf{Ab}$ that is excisive and exact (with respect to $\Gamma$-CW-subcomplexes). 
        \item If $\Gamma$ is a finite group and $X $ is $\Gamma$-equivariantly scaleable, then $\mathtt{ev}_1 \colon \sIP_{\loc}^\Gamma(X) \to \sIP^\Gamma(X)$ is a weak equivalence.
    \end{enumerate}
\end{thm}
\begin{proof}
    They are all proved in the same way as the corresponding statements in the non-equivariant setting. For (1), see \cref{lem:weak.equivalence.stick,lem:switch.constant.equivalence,lem:switch.loop.equivalence,lem:switch.diagram.commute} (and also \cref{lem:switch.localization.flow}). 
    For (2), see \cref{cor:coarse.homotopy.Hamiltonian,cor:localizing.coarse.homotopy}.
    For (3), see \cref{prp:excision,prp:IP.homology.longexact}. 
    Finally, for (4), see \cref{thm:scaleable}. 

    In the generalization to the equivariant setting, there is only one change: we can no longer perform truncation using \cref{lem:cut.diffused} since two $\Gamma$-invariant subspaces can never be coarsely transverse.  
    For this reason, we relax the definition, corresponding to the $\star$ notation, to accommodate almost properly supported Hamiltonians.
    In most cases, it is sufficient to avoid applying the truncation operation $\Theta_{Z}$. 
    However, one point requires some care: the proof of \cref{prp:coarse.homotopy,cor:coarse.homotopy.Hamiltonian}. 
    Here, the truncated Kitaev pump diffuses outside $\mathbf{I}_{2\varphi + r}\Lambda \setminus \mathbf{I}_\varphi \Lambda$, and if we do not truncate it, the proof will not work as is. 
    Therefore, we apply the flip homotopy $\widetilde{\flip}$ given in \cref{lem:flip.homotopy} to the observable algebras placed on $\widetilde{F}'(\bm{x},n), F_1'(\bm{x},n)$ for any $(\bm{x},n) \in \mathbf{I}_\varphi \Lambda$. 
    Since the two points $\widetilde{F}'(\bm{x},n)$ and $F_1'(\bm{x},n)$ are separated by a constant multiple of $|\varphi(x) - n|$, one might be concerned that the almost local norm increases during the flip homotopy. 
    However, because the Hamiltonian under consideration decays with distance from the $\mathbf{I}_{2\varphi +r}\Lambda \setminus \mathbf{I}_{\varphi}\Lambda$, this does not coarse a problem.
\end{proof}

Since $\Gamma$ is not compact, a $\Gamma$-manifold $\sM$ may not have the homotopy type of a $\Gamma$-CW-complex (unless the action is proper \cite{illmanExistenceUniquenessEquivariant2000}), and hence the set of smooth $\Gamma$-homotopy classes $\pst{\sIP}_{\loc}^\Gamma[X \midbar \sM]$ is no longer realized by the set of $\Gamma$-maps of topological spaces as $[\sM , \IP_{\loc}(X)]^\Gamma$. 
It is more relevant to treat the mapping sheaves 
\begin{align*}
    \pst\bsIP_d^N(X \midbar \sM \times \blank),\quad \pst\bsIP_{\loc,d}^N(X \midbar \sM \times \blank),
\end{align*}
regarded as $\Gamma/N$-sheaves. We define the associated $\Gamma/N$-equivariant realizations as
\begin{align*}
    \Hom_N(\sM,\IP_d)  \coloneqq {}&{} \big|\Sing_{\Gamma/N} \big(\pst\bsIP_d^N(X \midbar \sM \times \blank ))\big)\big|_{\Gamma/N}, \\
    \Hom_N(\sM,\IP_{\loc,d}) \coloneqq {}&{}  \big| \Sing_{\Gamma/N} \big( \pst\bsIP_{\loc,d}^N(X\midbar \sM \times \blank )\big) \big|_{\Gamma/N}.
\end{align*}

\begin{defn}\label{defn:spatial.equivariant.IP.cohomology}
For an intermediate subgroup $N \leq H \leq \Gamma$, we define the contravariant functors
\[  
    \rIP^H_{n-d}(X \midbar \sM \times \blank) \colon \kTop_{H/N} \to \mathsf{Ab}, \quad \rIP^H_{\loc,n-d}(X \midbar \sM \times \blank) \colon \kTop_{H/N} \to \mathsf{Ab}
\]
by using the $\Gamma/N$-realization (\cref{thm:equivariant.MW}) as
\begin{align*}
    \rIP^H_{n-d} (X \midbar \sM \times \sY ) \coloneqq {}&{} [ \Sigma^n\sY, \Hom_N(\sM,\IP_d) ]^{H/N} , \\
    \rIP_{\loc,n-d}^H(X \midbar \sM \times \sY) \coloneqq {}&{} [ \Sigma^n\sY, \Hom_N(\sM,\IP_{\loc,d})]^{H/N}.
\end{align*}
\end{defn}
This ambiguous notation is somewhat justified by the following isomorphism: When $\sX =\sN $ is a $H/N$-manifold, then we have isomorphisms
\begin{align*}
    \rIP^H_{-d} (X \midbar \sM \times \sN ) \cong {}&{} \pst\bsIP_d^H[X \midbar \sM \times \sN], \\
    \rIP_{\loc,-d}^H(X \midbar \sM \times \sN) \cong {}&{} \pst\bsIP_{\loc,d}^H[X\midbar \sM \times \sN ].
\end{align*}
Our primary interest is the case that the input to the cohomological variable is either $\sM=\pt$ or $\sM=\bE_{\Gamma}$. When $\Gamma$ is torsion-free and $\sM = \bE_{\Gamma} =E\Gamma$, the spaces $\Hom_\Gamma(\bE_{\Gamma}, \IP_d^\Gamma(X) )$ and $\Hom_\Gamma(\bE_{\Gamma}, \IP_{\loc,d}^\Gamma(X))$ is weakly equivalent to the homotopy fixed point of the spaces $\IP_d$ and $\IP_{\loc,d}$ respectively.

\begin{exmp}\label{exmp:MPS}
    The models of localization flows considered in \cref{exmp:models} are all $\Gamma$-invariant if so is the input data. For example, the Fidkowski--Kitaev model is a $1$-dimensional $\bZ$-invariant IG localization flow.
    This will be discussed in detail in \cite{kubotaStableHomotopyTheory2025b}. 
    
    Another example is the matrix product state \cites{fannesFinitelyCorrelatedStates1992,fannesFinitelyCorrelatedPure1994,perez-garciadandverstraetefrankandwolfmmandciracjiMatrixProductState2007}, a pure state on $\cA_{\bZ} = \bigotimes _{\bm{x}\in \bZ} \cA_{\bm{x}}$ with $\cA_{\bm{x}}=M_n(\bC)$ constructed from a class of $n$-tuple of matrices $\{A^i \}_{1 \leq i  \leq n } \subset \cB(\sV)$. 
    After numerous works like \cites{shiozakiHigherBerryCurvature2023,ohyamaHigherStructuresMatrix2024,ohyamaHigherBerryConnection2024,qiChartingSpaceGround2023,wenFlowHigherBerry2023}, a homotopy-theoretic formulation was finally given in \cite{beaudryClassifyingSpacePhases2025}. 
    In \cite{beaudryClassifyingSpacePhases2025}*{Theorem 6.6}, it is proved that the space $\cB$ of the above input data has the homotopy type $K(\bZ,2) \times K(\bZ,3)$.
    In the terminology of \cite{fannesFinitelyCorrelatedStates1992}, the MPS corresponds to an injective, ergodic, and purely generated state with correlation length $1$. Therefore, by \cite{fannesFinitelyCorrelatedStates1992}*{Theorem 6.4}, it is the unique gapped ground state of a UAL Hamiltonian called the parent Hamiltonian in the literature. 
    Such a state has the split property by \cite{matsuiBoundednessEntanglementEntropy2013}, and hence is invertible by \cref{lem:asymp.equal.state.invertible,prp:Eilenberg.swindle}. Therefore, we get a morphism 
    \begin{align*}
        \mathrm{MPS} \colon K(\bZ , 2) \times K(\bZ,3) \to \IP^{\bZ}_1. 
    \end{align*}
    It will be proved in \cref{thm:BCI} that the cohomology functor $\rIP_{\bZ}^1$ has a direct summand $\rIP_{\loc, \bZ}^1$ that is isomorphic to $H^2(\blank \,;\bZ) \oplus H^3(\blank \,;\bZ)$. In \cite{kubotaStableHomotopyTheory2025c}, we also show that $\mathrm{MPS}$ induces an isomorphism onto the direct summand $\rIP_{\loc,1}^{\bZ}$.  
\end{exmp}

\begin{exmp}\label{exmp:reflection.phase}
    Let $\Gamma= \bZ/2\bZ$ acts on $\bR$ by the reflection along the origin. In this case, we have a $\Gamma$-equivariant weak equivalence $\IP_* (\bR) \simeq \IP_{\loc,*}(\bR)$ by \cref{thm:equivariant.stable.homotopy.IP} (3). 
    Therefore, the long exact sequence of $\Gamma$-equivariant IP-homology groups
    \begin{align}
        \cdots \to \rIP_{\loc,1}^\Gamma(\bR,\pt) \to \rIP^\Gamma_{\loc,0}(\pt) \to \rIP_{\loc,0}^\Gamma(\bR) \to \rIP^\Gamma_{\loc,0}(\bR,\pt) \to \rIP_{\loc,-1}^\Gamma(\pt) \to \cdots \label{eqn:long.exact.exmp} 
    \end{align}
    is computable. First, by the excision axiom, we have
    \begin{align*} 
        \rIP^\Gamma_{\loc,*}(\bR,\pt) \cong \rIP^\Gamma_{\loc,*}(\bR \setminus [-1,1], \{ \pm 1 \}) \cong \rIP_{\loc,*}([1,\infty),\{1\}) \cong \rIP_{\loc,*-1}(\pt). 
    \end{align*}
    Second, since the origin $\pt$ is fixed by the $\Gamma$-action, we have $\rIP^\Gamma_{\loc,0}(\pt) \cong \rIP^\Gamma_{0}(\pt) \cong \mathrm{H}^2(\Gamma \,; \bZ) \cong \bZ/2$. Hence the long exact sequence \eqref{eqn:long.exact.exmp} is
    \begin{align*}
        \cdots \to 0 \to \bZ/2 \to \rIP_0^\Gamma(\bR) \to 0 \to \cdots  
    \end{align*}
    which concludes $\rIP_0^\Gamma(\bR) \cong \bZ/2$. This computation explains the topological origin of the $\bZ/2$-valued invariant given by Ogata in \cite{ogataIndexSymmetryProtected2021}. 
    We remark that this result differs from the case where $G=\bZ/2$ represents an on-site symmetry, in which we have $\rIP_{0,\loc}^{\bZ/2}(\bR) \cong H^3(\bZ/2 \,; \bZ) \cong 0$.
\end{exmp}

\begin{exmp}\label{exmp:reflection.onsite}    
    In addition to \cref{exmp:reflection.phase}, we further consider the on-site symmetry of $G$. Then, the long exact sequence corresponding to \eqref{eqn:long.exact.exmp} is computed by using the K\"unneth theorem as $\mathrm{IP}_{0,\loc}^{\Gamma \times G}(\pt ) \cong \mathrm{H}^2(\Gamma \times G\,; \bZ) \cong H^2(G\,; \bZ)  \oplus \bZ/2$, $\mathrm{IP}_{-1,\loc}^{\Gamma \times G}(\pt ) \cong \mathrm{H}^1(\Gamma \times G\,; \bZ) \cong H^1(G\,; \bZ)$, $\mathrm{IP}_{1,\loc}^{\Gamma \times G}(\bR,\pt ) \cong \mathrm{H}^2(G\,; \bZ)$, and $\mathrm{IP}_{0,\loc}^{\Gamma \times G}(\bR,\pt ) \cong \mathrm{H}^1(G\,; \bZ)$. That is, the sequence 
        \begin{align*}
        \cdots \to  \rIP^{\Gamma \times G}_{1,\loc}(\bR) \to \mathrm{H}^2(G\,; \bZ) \to \mathrm{H}^2(G\,; \bZ) \oplus \bZ/2 \to \rIP_{0,\loc}^{\Gamma \times G}(\bR) \to \mathrm{H}^1(G\,; \bZ) \to \mathrm{H}^1(G\,; \bZ)\to \cdots  
    \end{align*}
    is exact. 
    Again by the K\"{u}nneth theorem, the above morphisms $\mathrm{H}^2(G\,; \bZ) \to \mathrm{H}^2(G\,; \bZ)$ and $\mathrm{H}^2(G\,; \bZ) \to \mathrm{H}^2(G\,; \bZ)$ are multiplies by $2$. 
    Thus, the group $\IP_{0,\loc}^{\Gamma \times G} $ has $\bZ/2$ as a subgroup whatever $G$ is. If $G$ has $2$-torsion, then it also has other components coming from the non-triviality of the $2$-torsions of $\mathrm{H}^*(G\,; \bZ)$.

    However, if one forgets the reflection symmetry, then the inclusion $\pt \to \bR$ must induce the zero map. Indeed, it factors through a topologically flasque metric space $\bR_{\geq 0}$ as $\pt \to \bR_{\geq 0} \to \bR$ and $\bR_{\geq 0}$ is topologically flasque. This shows that the forgetful map $\rIP^0_{\Gamma \times G}(\bR) \to \rIP^0_G(\bR)$ vanishes. This means that a reflection invariant system can not have a non-trivial SPT index, which exhibits a similarity with a version of the Lieb--Mattis--Schultz theorem \cite{ogataGeneralLiebSchultzMattisType2021}.
\end{exmp}

\subsection{The covering isomorphism}
In this section, we prove the following theorem. 
\begin{thm}\label{thm:covering}
    Let $\Gamma$ be a crystallographic group and let $N \leq \Gamma$ be a torsion-free finite index normal subgroup. Then, for $X \in \ECW_{\Gamma, \mathrm{fp}}$, there is a weak equivalence of $\Gamma/N$-sheaves
    \begin{align*}
    \pst\sIP_{\loc}^N(X) \simeq  \pst\sIP_{\loc}(X/N ). 
    \end{align*}
\end{thm}
Here, we say that a morphism $\phi \colon \sF \to \sG$ of $G$-sheaves is a $G$-weak equivalence if and only if it restricts to weak equivalences of subsheaves $\phi \colon \sF^H \to \sG^H$ for any $H \leq G$. 
By the Elmendorf theorem \cite{elmendorfSystemsFixedPoint1983} (see also \cref{subsection:Gsheaf.category}), a $G$-weak equivalence induces an isomorphism $\phi \colon \sF^G(\sM) \to \sG^G(\sM)$ for any $G$-manifold $\sM$. 
Therefore, \cref{thm:covering} implies the isomorphisms
\[
    \pst\sIP_{\loc}^\Gamma[X \midbar \sM] \cong  \pst\sIP_{\loc}^{\Gamma/N}[X/N \midbar \sM].
\]

Let $X$ be a proper finite $\Gamma$-CW-complex on which $N$ acts freely. Let $\widetilde{\mathbbl{\Lambda}} = (\widetilde{\Lambda}_s) \in {}_{\star}\fL_{X,\Gamma}^{\loc}$. We write $\mathbbl{\Lambda} = (\Lambda_s)_{s} \in {}_{\star}\fL_{X/N, \Gamma/N}^{\loc}$. 
We fix $\mathbbm{r} >0$ so that the restriction of the covering $q \colon X \to X/N$ to any open ball $B_{\mathbbm{r}}(\bm{x})$ of radius $\mathbbm{r}$ is a homeomorphism onto the image. 
\begin{defn}
    We say that IG localization flows $\widetilde{\sfH} \in \sGP_{\loc}(\widetilde{\mathbbl{\Lambda}} \midbar \sM)$ and $\sfH \in \sGP_{\loc}(\mathbbl{\Lambda} \midbar \sM)$ are $\mathbbm{r}$-locally equivalent if, for any relatively compact open chart $\sU$ and $k \in \bN$, 
    \begin{align*}
        \sup_{\tilde{\bm{x}} \in \widetilde{\Lambda}_s} \big\| \Pi_{\tilde{\bm{x}},\mathbbm{r}} \big( \widetilde{\sfH}_{\tilde{\bm{x}}}\big) - 
        \Pi_{\bm{x},\mathbbm{r}} \big( \sfH_{\bm{x}}\big) \big\|_{(s),\sU,C^k,\bm{x},\nu,\mu} 
    \end{align*}
    is finite and converges to $0$ as $s \to \infty$, where $\bm{x}\coloneqq q(\tilde{\bm{x}})$. We write $\pst\widetilde{\sIP}{}^{N}_{\loc}(X)$ for the sheaf of smooth families of $\mathbbm{r}$-locally equivalent pairs. 
\end{defn}
We remark that, in the above definition, the operators $\Pi_{\bm{x},\mathbbm{r}}  =\Pi_{B_{\mathbbm{r}}(\bm{x})}$ are taken with respect to the initial metric $\rmd$ rather than the rescaled one $\rmd_s$. 
Similarly, we call a pair $(\widetilde{\sfG},\sfG)$ of localization flows of UAL derivations $\mathbbm{r}$-locally equivalent if the same condition is satisfied.

\begin{prp}\label{prp:locally.equivalent.propagation}
    Let $(\widetilde{\sfH},\sfH)$ be a $\mathbbm{r}$-locally equivalent pair of IG localization flows and let $(\widetilde{\sfG},\sfG)$ be a $\mathbbm{r}$-locally equivalent pair of localization flows of UAL derivations. Then the pair $(\alpha(\widetilde{\sfG} \,; t)(\widetilde{\sfH}) , \alpha(\sfG\,; t)(\sfH))$ is also $\mathbbm{r}$-locally equivalent. 
\end{prp}
\begin{proof}
Let $\widetilde{\sfG}_{\tilde{\bm{x}}}' \coloneqq \Pi_{\tilde{\bm{x}},\mathbbm{r}}(\widetilde{\sfG}_{\tilde{\bm{x}}})$ and $\sfG_{\bm{x}}' \coloneqq \Pi_{\bm{x},\mathbbm{r}}(\sfG_{\bm{x}})$. 
By assumptions, we have 
    \begin{align*}
        \vvert \widetilde{\sfG}' - \sfG' \vvert_{(s),\sU,C^k,\nu,\mu} \coloneqq \sup_{\tilde{\bm{x}} \in \widetilde{\Lambda}_s} \big\| \Pi_{\tilde{\bm{x}},\mathbbm{r}} \big( \widetilde{\sfG}_{\tilde{\bm{x}}}\big) - 
        \Pi_{\bm{x},\mathbbm{r}} \big( \sfG_{\bm{x}}\big) \big\|_{(s),\sU,C^k,\bm{x},\nu,\mu} \to 0
    \end{align*}
   as $s \to \infty$. Consider the estimate
    \begin{align}
    \begin{split}
        {}&{} \big\| \Pi_{\tilde{\bm{x}},\mathbbm{r}} \big( \alpha_p\big( \widetilde{\sfG}\,; t \big)
        \big( 
        \widetilde{\sfH}_{\tilde{\bm{x}}}
        \big)
        \big)  
        - 
        \Pi_{\bm{x},\mathbbm{r}} \big( \alpha_p(\sfG\,; t)(\sfH_{\bm{x}}) \big)\big\|_{(s),\sU,C^k,\tilde{\bm{x}},\nu,\mu}  \\
        \leq 
        {}&{} \big\| \alpha_p
        \big( \widetilde{\sfG}\,; t \big)
        \big( \widetilde{\sfH}_{\tilde{\bm{x}}} \big)
        -  
        \alpha_p \big(
        \Pi_{\bm{x},\mathbbm{r}}\widetilde{\sfG}' \,; t
        \big)
        \big(
        \Pi_{\tilde{\bm{x}},\mathbbm{r}}(\widetilde{\sfH}_{\tilde{\bm{x}}})
        \big)
        \big\|_{(s),\sU,C^k,\bm{x},\nu,\mu}  \\
        {}&{} + \big\| 
        \alpha_p \big( 
        \Pi_{\tilde{\bm{x}},\mathbbm{r}}
        \big( \widetilde{\sfG}' \big)  
        \,; t \big)
        (\Pi_{\tilde{\bm{x}},\mathbbm{r}}\big( \widetilde{\sfH}_{\tilde{\bm{x}}}\big)\big) 
        - 
        \alpha_p\big( \Pi_{\bm{x},\mathbbm{r}}(\sfG') \,; t\big)\big(\Pi_{\bm{x},\mathbbm{r}}\big( \sfH_{\bm{x}} \big) \big)\big\|_{(s),\sU,C^k,\bm{x},\nu,\mu} \\
        {}&{}+ \big\| \alpha_p(\sfG \,; t)(\sfH_{\bm{x}})  -  \alpha_p(\Pi_{\bm{x},\mathbbm{r}}(\sfG') \,; t)(\Pi_{\bm{x},\mathbbm{r}}(\sfH_{\bm{x}})) \big\|_{(s),\sU,C^k,\bm{x},\nu,\mu} .
    \end{split} \label{eqn:estimates.local.eq}
    \end{align}
    The first term is estimated as 
    \begin{align*}
        {}&{} \big\| \Pi_{\tilde{\bm{x}},\mathbbm{r}}\big( \alpha_p(\widetilde{\sfG}\,; t)(\widetilde{\sfH}_{\tilde{\bm{x}}}) \big)  -  \alpha_p(\Pi_{\bm{x},\mathbbm{r}}(\widetilde{\sfG}') \,; t)(\Pi_{\tilde{\bm{x}},\mathbbm{r}}(\widetilde{\sfH}_{\tilde{\bm{x}}})) \big\|_{(s),\sU,C^k,\bm{x},\nu,\mu} \\
        \leq {}&{}
        \big\| (\Pi_{\tilde{\bm{x}},\mathbbm{r}} - \Pi_{\tilde{\bm{x}},\mathbbm{r}/2})(\alpha_p(\widetilde{\sfG}\,; t)(\widetilde{\sfH}_{\tilde{\bm{x}}})) \big\|_{(s),\sU,C^k,\bm{x},\nu,\mu} \\
        {}&{}+
        \big\| 
        \Pi_{\tilde{\bm{x}},\mathbbm{r}/2}\big( 
        \alpha_p(\widetilde{\sfG}\,; t)
        - 
        \alpha_p(\widetilde{\sfG}'\,; t)\big)
        (\widetilde{\sfH}_{\tilde{\bm{x}}}) \big\|_{(s),\sU,C^k,\bm{x},\nu,\mu} \\
        {}&{}+
        \big\| \Pi_{\tilde{\bm{x}},\mathbbm{r}/2} \big( \alpha_p(\widetilde{\sfG}'\,; t)(\widetilde{\sfH}_{\tilde{\bm{x}}}-\Pi_{\tilde{\bm{x}},\mathbbm{r}}(\widetilde{\sfH}_{\tilde{\bm{x}}})) \big\|_{(s),\sU,C^k,\bm{x},\nu,\mu} \\
        {}&{} + \big\| (\id - \Pi_{\tilde{\bm{x}},\mathbbm{r}/2})(\alpha_p(\widetilde{\sfG} ' \,; t)(\Pi_{\tilde{\bm{x}}, \mathbbm{r}} \big( \widetilde{\sfH}_{\tilde{\bm{x}}})) \big) \big\|_{(s),\sU,C^k,\bm{x},\nu,\mu} .
    \end{align*}
    By \cref{cor:Lieb.Robinson.approx.Ck} (1) and \eqref{eqn:norm.comparison.truncated}, the first, the third, and the forth terms are bounded above by a constant multiple of $f_{\nu,\mu}(s\mathbbm{r}/2) \cdot \Upsilon^{(k)}_{\widetilde{\sfG},\nu,\mu_1}(t)$, $f_{\nu_k,\mu_{3k+1}}(s\mathbbm{r}/2) \cdot \Upsilon^{(k)}_{\widetilde{\sfG},\nu,\mu}(t)$, and $f_{\nu,\mu}(s\mathbbm{r}/2) \cdot \Upsilon^{(k)}_{\widetilde{\sfG}',\nu,\mu_1}(t)$,  respectively.
    By \cref{cor:Lieb.Robinson.approx.Ck} (1) (and the proof of (2)), the second term is also bounded above by a constant multiple of $f_{\nu,\mu}(s\mathbbm{r}/8) \cdot \upsilon(t)$, where $\upsilon(t)$ is a certain function on $\bR$ such that $\upsilon(t) \prec w_v(t)$. 

    The third term of \eqref{eqn:estimates.local.eq} is also estimated in the same way.  
    As for the second term of \eqref{eqn:estimates.local.eq}, we have
    \begin{align*}
        {}&{}\big\|  
        \alpha_p \big( \Pi_{\tilde{\bm{x}},\mathbbm{r}}(\widetilde{\sfG}')  \,; t \big)(\Pi_{\tilde{\bm{x}},\mathbbm{r}}\big( \widetilde{\sfH}_{\tilde{\bm{x}}}\big)\big) 
        - 
        \alpha_p\big( \Pi_{\bm{x},\mathbbm{r}}\sfG' \,; t\big)
        \big(\Pi_{\bm{x},\mathbbm{r}}\big( \sfH_{\bm{x}} \big) \big)
        \big\|_{(s),\sU,C^k,\bm{x},\nu,\mu}
        \\
        \leq {}&{} \Upsilon_{\Pi_{\tilde{\bm{x}},\mathbbm{r}}(\widetilde{\sfG}'), \Pi_{\bm{x},\mathbbm{r}}(\sfG') ,\nu,\mu }^{(k),\mathrm{rel}}(t) \cdot \vvert \Pi_{\tilde{\bm{x}},\mathbbm{r}}(\widetilde{\sfG}') - \Pi_{\bm{x},\mathbbm{r}}(\sfG') \vvert_{(s),\sU,C^k,\nu_k,\mu_{3k+1}} \cdot \vvert \widetilde{\sfH} \vvert_{(s),\sU,C^k,\nu_k,\mu_{3k+1}} \\
        {}&{} + \Upsilon^{(k)}_{\Pi_{\bm{x},\mathbbm{r}}(\sfG'),\nu,\mu}(t) \cdot \vvert \Pi_{\tilde{\bm{x}},\mathbbm{r}}(\widetilde{\sfH}_{\tilde{\bm{x}}})  - \Pi_{\bm{x},\mathbbm{r}}(\sfH_{\bm{x}}) \vvert_{(s),\sU,C^k,\nu,\mu}
    \end{align*}
    by \cref{prp:Lieb.Robinson.Ck} (1), (2). The both terms converges to $0$ as $s \to \infty$ by assumption. Indeed, we have 
    \begin{align*}
        {}&{} \vvert \Pi_{\tilde{\bm{x}},\mathbbm{r}}(\widetilde{\sfG}') - \Pi_{\bm{x},\mathbbm{r}}(\sfG') \vvert_{(s),\sU,C^k,\nu_k,\mu_{3k+1}} \\
        \coloneqq {}&{} \sup_{\bm{y} \in B_{\mathbbm{r}}(\bm{x})} \bigg\| \sum_{\tilde{\bm{y}} \in q^{-1}(\bm{y})} \Pi_{\tilde{\bm{x}},\mathbbm{r}}\Pi_{\tilde{\bm{y}},\mathbbm{r}}(\widetilde{\sfG}_{\tilde{\bm{y}}}) - \Pi_{\bm{x},\mathbbm{r}}\Pi_{\bm{y},\mathbbm{r}}(\sfG_{\bm{y}}) \bigg\|_{(s),\sU,C^k,\bm{y},\nu_k,\mu_{3k+1}} \\
        \leq {}&{} \sup_{\tilde{\bm{y}} \in B_{\mathbbm{r}}(\tilde{\bm{x}})} \| \Pi_{\tilde{\bm{x}},\mathbbm{r}}\Pi_{\tilde{\bm{y}},\mathbbm{r}}(\widetilde{\sfG}_{\tilde{\bm{y}}}) - \Pi_{\bm{x},\mathbbm{r}}\Pi_{\bm{y},\mathbbm{r}}(\sfG_{\bm{y}}) \|_{(s),\sU,C^k,\bm{y},\nu_k,\mu_{3k+1}}\\
        {}&{} + \sup_{\tilde{\bm{y}} \in B_{2\mathbbm{r}}(\tilde{\bm{x}}) \setminus B_{\mathbbm{r}}(\tilde{\bm{x}})} \| \Pi_{\tilde{\bm{x}},\mathbbm{r}}\Pi_{\tilde{\bm{y}},\mathbbm{r}}(\widetilde{\sfG}_{\tilde{\bm{y}}}) \|_{(s),\sU,C^k,\tilde{\bm{y}},\nu_k,\mu_{3k+1}}\\
        {}&{} + \sup_{\bm{y} \in B_{2\mathbbm{r}}(\bm{x}) \setminus B_{\mathbbm{r}}(\bm{x})} \| \Pi_{\bm{x},\mathbbm{r}}\Pi_{\bm{y},\mathbbm{r}}(\sfG_{\bm{y}}) \|_{(s),\sU,C^k,\bm{y},\nu_k,\mu_{3k+1}}.
    \end{align*}
    The first term converges to $0$ as $s \to \infty$ by assumption of $\mathbbm{r}$-local equivalence of $\widetilde{\sfG}$ and $\sfG$, and the second and the third terms converge to $0$ as $s \to \infty$ since $\widetilde{\sfG}$ is a localization flow.
\end{proof}

\begin{prp}\label{lem:locally.equivalent.adiabatic}
    Let $(\widetilde{\sfH}, \sfH)$ be an $\mathbbl{r}$-locally equivalent pair. Then the associated pair $(\sfG_{\widetilde{\sfH}},\sfG_{\sfH})$ is also $\mathbbl{r}$-locally equivalent. 
\end{prp}
\begin{proof}
    By \cref{prp:locally.equivalent.propagation}, we have that $(\tau_{\widetilde{\sfH},t}(d\widetilde{\sfH}), \tau_{\sfH,t}(\sfH))$ is an $\mathbbm{r}$-locally equivalent pair.  the integrand of the right hand side converges to $0$ as $s \to \infty$. 
    According to the estimates in \cref{prp:Lieb.Robinson.Ck,cor:Lieb.Robinson.approx.Ck}, the Lebesgue dominant convergence theorem can be applied to show that
    \begin{align*}
        {}&{} \| \Pi_{\tilde{\bm{x}},\mathbbm{r}}(\sfG_{\widetilde{\sfH},\tilde{\bm{x}}}) - \Pi_{\bm{x},\mathbbm{r}}(\sfG_{\sfH,\bm{x}})\|_{(s),\sU,C^k,\bm{x},\nu,\mu}  \\
        \leq {}&{} \int_{-\infty}^\infty \int_0^t \| \Pi_{\bm{x},\mathbbm{r}}(\tau_{\widetilde{\sfH} , u}(d \widetilde{\sfH}_{\bm{x}})) - \Pi_{\bm{x},\mathbbm{r}}(\tau_{\sfH , u}(d \sfH_{\bm{x}})) \|_{(s),\sU,C^k,\bm{x},\nu,\mu}    du \ w_v(t)dt 
    \end{align*}
    converges to $0$ as $s \to \infty$. 
\end{proof}

\begin{prp}\label{prp:covering.homology.theory}
    The following hold. 
    \begin{enumerate}
        \item The family $\{ \pst\widetilde{\sIP}{}_{\loc,d}^{N}, \kappa_{\loc,d}\}$ forms a $\Gamma/N$-spectrum. 
        \item The assignment $X \mapsto \pst\widetilde{\sIP}{}_{\loc}^{N}[X \midbar \sM]$ is a $\Gamma/N$-equivariant homology functor on $\ECW_{\Gamma, \mathrm{p}}$.
    \end{enumerate}
\end{prp}
\begin{proof}
    By \cref{prp:locally.equivalent.propagation,lem:locally.equivalent.adiabatic}, the adiabatic interpolation construction given in \cref{thm:interpolation.loop} preserves $\mathbbm{r}$-locally equivalent pairs. 
    Therefore, the arguments used in \cref{lem:weak.equivalence.stick,lem:switch.constant.equivalence,lem:switch.loop.equivalence,lem:switch.diagram.commute} and \cref{prp:excision,lem:IP.homology.longexact} are all applicable to the $\Gamma/N$-sheaf $\pst\widetilde{\sIP}{}_{\loc}^{N}$. 
\end{proof}

\begin{lem}\label{lem:covering.cell}
    Let $K$ be a finite subgroup of $\Gamma$, let $D \in \ECW_{K,\mathrm{fp}}$, and $X = \Gamma \times_K D$. Then the morphisms
    \begin{align*}
        \pst\sIP_{\loc}(X/N) \xleftarrow{e} \pst\widetilde{\sIP}{}_{\loc}^{N}(X) \xrightarrow{\tilde{e}} \pst\sIP_{\loc}^N(X),
    \end{align*}
    given by $\tilde{e}(\widetilde{\sfH},\sfH) \coloneqq \widetilde{\sfH}$ and $e(\widetilde{\sfH},\sfH) \coloneqq \sfH$, are both $\Gamma/N$-equivariantly weakly equivalent.
\end{lem}
\begin{proof}
    In this case, $X$ decomposes to the disjoint union of copies of $D$ as $\bigsqcup_{g K \in \Gamma/K} g \cdot D$. Hence there is a morphism $c \colon \pst\sIP_{\loc}(X/N) \to \pst\sIP_{\loc}^N(X)$ sending $\sfH \in \pst\sIP_{\loc}(X/N \midbar \sM)$ to the product of copies of it. It suffices to show that $c$ is a weak equivalence. This statement is proved in the same way as \cref{prp:additivity}, whose proof can immediately be made equivariant. 
\end{proof}

\begin{proof}[Proof of \cref{thm:covering}]
    By \cref{prp:covering.homology.theory,lem:covering.cell}, an iterated Mayer--Vietoris argument shows that the morphisms
    \begin{align*}
        \pst\sIP_{\loc}^{}(X/N) \leftarrow \pst\widetilde{\sIP}{}_{\loc}^{N}(X) \rightarrow \pst\sIP_{\loc}^N(X)
    \end{align*}
    are both $\Gamma/N$-equivariantly weakly equivalent for any $X \in \ECW_{\Gamma,\mathrm{fp}}$. 
\end{proof}

\subsection{Equivariant coarse assembly map}
The following definition is inspired from the work of G.~Yu \cite{yuCoarseBaumConnesConjecture1995}.  
\begin{defn}\label{defn:assembly.map}
    Let $\Gamma$ be a crystallographic group acting on the Euclidean space $\bE_{\Gamma}$. 
    We call the morphism of $\Omega$-spectra 
    \[
    \mu_{\rIP, \Gamma,N} \colon \IP_{\loc ,*}^\Gamma \to \IP_*^\Gamma
    \]
    induced by the forgetful maps of $\Gamma$-sheaves $\mathtt{ev}_{1} \colon \pst\bsIP_{\loc ,*}^\Gamma \to \pst\bsIP_*^\Gamma$ the \emph{equivariant coarse assembly map} of the equivariant IP cohomology theory.
\end{defn}

\begin{prp}\label{lem:proper.to.FC}
    Let $f \colon \sX \to \sY$ be a continuous map of $\Gamma$-CW-complexes that is $K$-equivariantly homotopy equivalent for any finite subgroup $K\leq \Gamma$. Then, for any $X \in \ECW_{\Gamma,\mathrm{p}}$, the induced map
    \[
        f^* \colon \rIP_{\loc}^\Gamma(X \midbar \sY) \to \rIP_{\loc}^\Gamma(X \midbar \sX)
    \]
    is an isomorphism. In particular,  
    \[
        \pr_{\sY}^* \colon \rIP_{\loc}^\Gamma(X \midbar \sY) \to \rIP_{\loc}^\Gamma(X \midbar \sY \times \mathbb{E}_{\Gamma})
    \]
    is an isomorphism. 
\end{prp}
\begin{proof}
    We begin with the case of $X \cong \Gamma \times _K D$ for some finite subgroup $K \leq \Gamma$ and a $K$-space $D$. By \cref{prp:additivity}, whose equivariant version is proved in the same way, we have 
    \begin{align}
    \mathrm{ind}_{K}^\Gamma \colon \rIP_{\loc}^K(D \midbar \blank) \to \rIP_{\loc}^\Gamma (\Gamma \times_K D \midbar \blank), \quad \mathrm{ind}_{K}^\Gamma[\sfH] = \bigg[ \bigsqcup_{gK \in \Gamma/K}\beta_g(\sfH) \bigg] \label{eqn:induction}
    \end{align}
    is a natural isomorphism. Hence, we get the commutative diagram as
    \[
    \xymatrix{
        {\rIP}_{\loc}^K(D \midbar \sY) \ar[r]^{f^*} \ar[d]^{\mathrm{ind}_H^\Gamma}_{\cong} &  \rIP_{\loc}^K(D  \midbar \sX) \ar[d]^{\mathrm{ind}_H^\Gamma}_{\cong} \\
        \rIP_{\loc}^\Gamma(\Gamma \times_K D  \midbar  \sY) \ar[r]^{f^*} &  \rIP_{\loc}^\Gamma(\Gamma \times_K D  \midbar  \sX).
    }
    \]
    The top horizontal map is an isomorphism by assumption. 
    Hence, the same is true for the bottom horizontal map.

    For general $X$, an iterated use of the five lemma for long exact sequences associated with the pair $(X_n,X_{n-1})$ shows the desired isomorphism.
\end{proof}

\begin{lem}\label{lem:BC.proper}
    Let $\sX$ be a proper $\Gamma$-CW-complex. Then the morphism
    \begin{align*}
    \mu_{\mathrm{IP}, \Gamma,n} \colon \rIP_{n,\loc}^\Gamma(\bE_{\Gamma}  \midbar  \sX) \to \rIP_{n}^\Gamma(\bE_{\Gamma}  \midbar  \sX)
    \end{align*}
    is an isomorphism. 
\end{lem}
\begin{proof}
    If $\sM = \Gamma \times_K \sD$ for a finite subgroup $K$ of $\Gamma$ and a $K$-manifold $\sD$, then the desired isomorphism follows from 
    \begin{align*}
    \xymatrix{
     \rIP_{n,\loc}^K(\bE_{\Gamma}  \midbar  \sD) \ar[r]^{\mu_{\mathrm{IP},K, n} }_{\cong} \ar[d]_{\cong }^{\mathrm{ind}_K^\Gamma} &  \rIP_{n}^K(\bE_{\Gamma} \midbar  \sD) \ar[d]_{\cong}^{\mathrm{ind}_K^\Gamma} \\ 
     \rIP_{n,\loc}^\Gamma(\bE_{\Gamma} \midbar  \Gamma \times_K \sD) \ar[r]^{\mu_{\mathrm{IP},\Gamma, n} }  &  \rIP_{n}^\Gamma(\bE_{\Gamma}  \midbar  \Gamma \times_K \sD ).
    }
    \end{align*}
    Here, $\mathrm{ind}_K^\Gamma$ is the induction morphism defined in \cref{lem:inducction.sheaves}, which makes sense since the sheaves $\sIP_{*}(X)$ and $\sIP_{\loc,*}(X)$ are indeed diffeological spaces. 
    For general $\sX$, the isomorphism follows from a standard Mayer--Vietoris argument. 
\end{proof}

\begin{thm}\label{thm:BCI}
Let $\Gamma$ be a discrete group acting properly on $\bR^d$. 
Then the equivariant coarse assembly map $\mu_{\rIP , \Gamma, d} \colon \rIP_{\loc ,d}^\Gamma \to \rIP_d ^\Gamma$ is split injective. 
\end{thm}
\begin{proof}
By \cref{lem:proper.to.FC,lem:BC.proper}, the theorem follows from the commutativity of the diagram
\begin{align*} 
    \xymatrix{
    \rIP^\Gamma_{\loc ,n}(\bE_{\Gamma}  \midbar \sM) \ar[r]^{\mu_{\rIP,\Gamma,n}} \ar[d]^{\cong }_{\pr_{\sM}^*} & \rIP_n^\Gamma (\bE_{\Gamma}  \midbar \sM) \ar[d]^{\pr_{\sM}^*} \\
    \rIP^\Gamma_{\loc ,n}(\bE_{\Gamma} \midbar \bE_{\Gamma} \times \sM) \ar[r]^{\mu_{\rIP, \Gamma, n} }_{\cong} & \rIP_n^\Gamma(\bE_{\Gamma}  \midbar \bE_{\Gamma} \times \sM),
    }
\end{align*}
which follows from the naturality of $\mu_{\rIP,\Gamma,n}$. 
\end{proof}

This shows that translation invariant invertible quantum spin systems has a weak topological phase. This statement is already proved by Jappens \cite{jappensSPTIndicesEmerging2024} for $d \leq 2$.
\begin{cor}\label{cor:translation.BCI}
Let $\Gamma \cong \bZ^d$ be the group of lattice translations. 
Then the group $\pi_{0}(\mathrm{IP} ^\Gamma(\bE_{\Gamma}))$ is isomorphic to the direct sum of $\binom{d}{k}$ copies of the group $\pi_{d-k}(\IP)$. 
\end{cor}
\begin{proof}
    The direct summand $\rIP_{\loc, n}^{\Gamma} (\bE_{\Gamma}) $ is isomorphic to $\rIP_{\loc ,n}(\bT^d)$, which can be calculated by the Eilenberg--Steenrod axiom proved in \cref{thm:IP.bivariant}. 
\end{proof}
The same statement holds for fermionic and on-site $(G,\phi)$-symmetric phases. 

\begin{cor}\label{cor:spectral.sequence}
    There is a spectral sequence that eventually converges to the direct summand $\IP_{\loc}^\Gamma(\bE_{\Gamma} \mid \pt)$ of $\IP^\Gamma(\bE_{\Gamma} \mid \pt)$ whose $E_2$-page is given by $E_2^{pq} \cong \mathrm{H}^p(\bE_{\Gamma}/\Gamma, \mathcal{IP}_{q}^\Gamma)$. Here, $\mathcal{IP}_q$ is the sheafification of the presheaf $U \mapsto \IP_{\loc}^\Gamma(\pr^{-1}(U))$ on $\bE_{\Gamma}/\Gamma$, where $\pr \colon \bE_{\Gamma} \to \bE_{\Gamma}/\Gamma$ is the projection. 
\end{cor}
Note that the germ of the above sheaf $\mathcal{IP}^\Gamma_q$ at $p \in \bE_{\Gamma}/\Gamma$ is $\IP^{\Gamma_p}(\{ p \} )$. In particular, if the $\Gamma$-action on $\bE_{\Gamma}$ is free, then the sheaf $\mathcal{IP}^\Gamma_q$ is indeed a local system. 

\begin{cor}\label{cor:pg}
Let $\Gamma$ be the crystallographic group of type \textsf{pg}, i.e., the fundamental group of the Klein bottle $\mathrm{Kl}$. Then the homotopy group $\rIP_n^\Gamma (\bE_{\Gamma}) = \pi_n(\sIP^\Gamma(\bE_{\Gamma}))$ contains a direct summand isomorphic to $\bZ$ if $n=2$, and $\bZ \oplus \bZ/2$ if $n=3$, respectively.
\end{cor}
\begin{proof}
    Let us embed the Klein bottle $\mathrm{Kl}$ into $\bR^5$ in the way that the normal bundle $\nu \mathrm{Kl}$ is isomorphic to $L \oplus \underline{\bR}^2$, where $L \to \mathrm{Kl}$ is the real line bundle corresponding to the double cover $\bT^2 \to \mathrm{Kl}$. For example, the embedding
    \[
        [x,y] \mapsto (\cos (2\pi x), \sin (2 \pi x), \cos (2 \pi y), \cos(\pi x)\sin (2 \pi y) , \sin(\pi x) \sin (2\pi y))
    \]
    satisfies this condition. Then the Spanier--Whitehead dual $\sD_{\mathrm{Kl}}$ is $\mathrm{Th}(\nu \mathrm{Kl}) \cong \Sigma^2 \mathrm{Th}(L)$.  
    By \cref{thm:Atiyah.duality,thm:covering,thm:BCI}, for $n \geq 2$, the group $\rIP_{n}^\Gamma(\bE_\Gamma)$ contains a direct summand 
    \begin{align*}
        \rIP_{\loc,n}^\Gamma(\bE_\Gamma) \cong {}&{} \rIP_{\loc}(\mathrm{Kl} \midbar S^n,\pt) \cong \rIP_{\loc}^5(\Sigma^n \sD_{\mathrm{Kl}}) \cong \rIP_{\loc}^{3-n}(\mathrm{Th}(L) ) \\
        \cong {}&{}  \mathrm{H}^{5-n}(\mathrm{Th}(L)\,; \bZ) \cong \mathrm{H}^{7-n}(\sD_{\mathrm{Kl}}\,; \bZ) \cong \mathrm{H}_{n-2}(\mathrm{Kl} ,; \bZ) .
    \end{align*}
    The right hand side is isomorphic to $\bZ$ if $n=2$, and $\bZ \oplus \bZ/2$ if $n=3$. 
\end{proof}

\subsection{Comparison with the Davis--L\"{u}ck assembly map}
In equivariant algebraic topology, a general definition of the assembly map is known for $\Omega$-spectra having the symmetry of a discrete group $\Gamma$ (more precisely, a covariant $\Or(\Gamma)$-spectrum).
This Davis--L\"{u}ck assembly map was introduced in \cite{davisSpacesCategoryAssembly1998}. 
In this subsection, we compare our equivariant coarse assembly map with the Davis--L\"{u}ck assembly map by using the comparison method given in \cite{davisSpacesCategoryAssembly1998}*{Section 6}.

As in the previous subsections, let $\Gamma$ be a crystallographic group acting on the Euclidean space $\bE_\Gamma \cong \bR^m$. 
For any subgroup $H \leq \Gamma$, the subgroup $N_H \coloneqq N \cap H$ is a finite index normal subgroup of $H$. The finite group $H/N_H$ acts on the quotient affine space $\bE_{\Gamma} / (N_H \otimes_{\mathbb{Z}}\mathbb{R}) $ isometrically, and hence has a unique fixed point. We write $\bE_H$ for the inverse image of this fixed point, an affine subspace of $\bE_\Gamma$ globally preserved by the $H$-action. Moreover, since $\bE_H \cong N_H \otimes_{\bZ}\bR$, the $H$-action on $\bE_H$ is proper and cocompact. 
\begin{defn}
We define the covariant functor $\mathbf{IP}^\Gamma \colon \Or(\Gamma) \to \Sp$ to the category $\Sp$ of $\Omega$-spectra in the following way. 
\begin{enumerate}
    \item For a subgroup $H \leq \Gamma$, let 
    \begin{align*}
        \mathbf{IP}^\Gamma(\Gamma/H) \coloneqq \IP^H_*(\mathbb{E}_H).
    \end{align*}
    \item For $g \in (\Gamma/K)^H \cong \Hom(\Gamma/H,\Gamma/K)$, we define $g_* \colon \mathbf{IP}^\Gamma(\Gamma/H) \to \mathbf{IP}^\Gamma(\Gamma/K)$ by
    \begin{align*}
        g_* \colon \IP_*^H(\bE_{H}) \xrightarrow{\beta_g} \IP_*^{gHg^{-1}}(g\bE_H) \cong \bigg( \prod_{[h] \in K/gHg^{-1}} \IP_*(hg\bE_{H}) \bigg)^K \to \IP_*^K(\bE_K).
    \end{align*}
\end{enumerate}
\end{defn}

In \cite{davisSpacesCategoryAssembly1998}, a covariant functor $\mathbf{E} \colon \Or(\Gamma) \to \Sp$ is called an $\Or(\Gamma)$-spectrum. 
Such a functor has an extension $\mathbf{E}_{\%} \colon \CW_\Gamma \to \Sp$ defined by
\begin{align*}
    \mathbf{E}_{\%}(X) \coloneqq \Hom_\Gamma(\blank , X)_+ \otimes _{\Or(\Gamma)} \mathbf{E} = \bigsqcup_{H \leq \Gamma} X^H_+ \wedge \mathbf{E}(\Gamma/H) /\sim,  
\end{align*}
where the equivalence relation $\sim$ is given by $(g^*x,y) \sim (x,gy)$ for any $g \colon \Gamma/H \to \Gamma/K$. 
In short, this tensor product $\blank \otimes_{\Or(\Gamma)} \blank $ is the coend of the contravariant functor $\Hom_\Gamma(\blank,X)$ and the covariant functor $\mathbf{E}$. The Davis--L\"{u}ck assembly map is the morphism
\begin{align*}
    \mu_{\mathbf{E}} \coloneqq \mathbf{E}_{\%}(\pr) \colon \mathbf{E}_{\%}(E(\Gamma,\mathcal{F}_{\mathrm{fin}})) \to \mathbf{E}_{\%}(\pt) \cong \mathbf{E}(\pt),
\end{align*}
where $E(\Gamma ,\cF_{\fin}) = \bE_{\Gamma}$ as is noted in \cref{rmk:classifying.space}. 

Applying this machinery to $\mathbf{IP}^\Gamma$, we get the Davis--L\"{u}ck assembly map 
\begin{align} 
    \mu_{\mathbf{IP}^\Gamma}^{\mathrm{DL}} \colon \mathbf{IP}_{\%}^\Gamma (E(\Gamma,\cF_{\fin})) \to \mathbf{IP}_{\%}^\Gamma(\pt) = \mathbf{IP}^\Gamma(\Gamma/\Gamma) = \IP^\Gamma_*(\bE_{\Gamma}).
    \label{eqn:assembly}
\end{align}
We prove that this morphism is identified with \cref{defn:assembly.map}. 

\begin{defn}
    Let $Z, X \in \ECW^\Gamma$ and let $\pi \colon X \to Z$ be a $\Gamma$-equivariant Lipschitz continuous map. We define the set $\pst\fL_{\fR,X,\Gamma }^{\piloc}$ and the sheaves $\pst\sGP_{\piloc }(X)$, $\pst\sIP_{\piloc }(X)$, and $\pst\bsIP_{\piloc }(X)$ in the same way as \cref{defn:equivariant.coarse}, by using the transversally rescaled metric 
    \[
    \rmd_{\pi,s}(\bm{x}, \bm{y} ) \coloneqq \rmd_{X \times \bR^l}(\bm{x},\bm{y}) + s \cdot \rmd_{Z \times \bR^l} (\pi(\bm{x}), \pi(\bm{y})). 
    \]
\end{defn}
By the same reason as \cref{thm:equivariant.stable.homotopy.IP} (1), the family $\{ \pst\sIP_{\piloc ,d}(X), \kappa_{d}^{\piloc}\}$ forms a naive $\Gamma$-equivariant $\Omega$-spectrum of sheaves.

\begin{defn}
    Let $X \in \ECW^\Gamma$. 
    Let $\pi \colon X \times \bE_{\Gamma} \to X$ denote the projection. We define the covariant functor $\mathbf{IP}_{\piloc}^\Gamma \colon \CW_\Gamma \to \Sp$ by
    \begin{align*}
        \mathbf{IP}_{\piloc}^\Gamma(X) \coloneqq \hocolim_{Z \leq X \times \bE_{\Gamma}} \IP_{\piloc}^\Gamma(Z), 
    \end{align*}
    where $Z$ runs over all finite subcomplexes of $X \times \bE_{\Gamma}$ equipped with the subspace metric. 
\end{defn}

A $\Gamma$-equivariant Lipschitz continuous map $F \colon X \to Y$ induces 
\begin{align*}
    F_* \coloneqq \hocolim_{Z \leq X \times \bE_{\Gamma}} \Big( F_* \colon \IP_{\piloc}^\Gamma(Z) \to \IP_{\piloc}^\Gamma(W) \Big) \colon \mathbf{IP}_{\piloc}^\Gamma(X) \to \mathbf{IP}_{\piloc}^\Gamma(Y).  
\end{align*}

\begin{lem}\label{lem:assembly.DL}
    The following hold. 
    \begin{enumerate}
        \item If two $\Gamma$-equivariant Lipschitz continuous maps $F_0, F_1 \colon X \to Y$ are connected by a $\Gamma$-equivariant large-scale Lipschitz homotopy that is Lipschitz continuous, then $F_{0,*}$ and $F_{1,*}$ are homotopic. 
        \item If $X = \Gamma/H$, then $\mathbf{IP}_{\piloc}^\Gamma(\Gamma/H)$ is weakly equivalent to $\IP_{\loc}^H(\bE_H)$.
        \item If $X$ is a proper finite $\Gamma$-CW-complex, then $\mathbf{IP}_{\piloc}^\Gamma(X)$ is weakly equivalent to $\IP_{\loc}^\Gamma(X)$. 
    \end{enumerate}
\end{lem}
\begin{proof}
    The claim (1) is verified in the same way as \cref{cor:localizing.coarse.homotopy}. The claim (2) follows from the weak equivalence
    \[
    \mathbf{IP}^{\Gamma}_{\piloc}(\Gamma/H) =\hocolim_{r \to \infty} \mathbf{IP}_{\piloc}^\Gamma(\Gamma \times _H N_r(\bE_H)) \simeq\hocolim_{r \to \infty} \IP^H (N_r(\bE_H)) \simeq\IP^H (\bE_H). 
    \]
    Here, the first weak equivalence is the induction $\mathrm{ind}_H^\Gamma$ given in \eqref{eqn:induction}, while the second weak equivalence comes from the fact that the inclusion $\bE_H \to N_r(\bE_H)$ is an $H$-equivariant large-scale Lipschitz equivalence.
    
    To see (3), take a Lipschitz continuous $\Gamma$-equivariant map $F\colon X \to  \bE_{\Gamma}$ by \cref{rmk:equivariant.universal}. Then, the direct product $X \times \bE_{\Gamma}$ is covered by an increasing sequence of $\Gamma$-CW-subcomplexes $X_n$ such that $N_r(\id \times F(X))$. 
    Now, the projection $\pi \colon X_n \to X$ is a Lipschitz continuous deformation retract of $\id \times F \colon X \to X_n$ (cf.\ \cref{rmk:embed.CWfin.EucCW}). 
    Hence the inclusion of the fixed point spectra $\IP_{\loc}^\Gamma(X) \to \IP_{\loc}^\Gamma(X_n)$ is a weak equivalence. 
    Since $\pst \bsIP_{\piloc}^\Gamma(X) = \pst \bsIP_{\loc}^\Gamma(X)$, this finishes the proof. 
\end{proof}

Recall \cref{rmk:classifying.space}, in which we state that $\Gamma$-$\cF_{\fin}$-CW-complex is the same as proper $\Gamma$-CW-complex.
\cref{lem:assembly.DL} (1), (2), together with \cref{thm:equivariant.stable.homotopy.IP} (2),  means that the functor $\mathbf{IP}_{\piloc}^\Gamma$ is $\Gamma$-homotopy invariant and $\cF_{\fin}$-excisive in the sense of \cite{davisSpacesCategoryAssembly1998}*{page 243}. 

In \cite{davisSpacesCategoryAssembly1998}*{Theorem 6.3}, Davis--L\"uck constructed another covariant functor $\mathbf{IP}_{\piloc}^{\Gamma,\%} \colon \CW_\Gamma \to \Sp$, which is explicitly defined by 
\begin{align*}
    \mathbf{IP}_{\piloc}^{\Gamma,\%}(X) \coloneqq \Hom_\Gamma(\blank,X)  \otimes_{\Or (\Gamma) \times \Delta}   B^{\mathrm{bar}} \big(? \downarrow (\Or (\Gamma) \times \Delta) \downarrow  ?? \big)  \otimes_{\Or (\Gamma) \times \Delta}  \mathbf{IP}_{\loc}^\Gamma(\blank ) 
\end{align*}
by using the notations of \cite{davisSpacesCategoryAssembly1998}*{Sections 1 and 3}, and the natural transforms 
\[
    \bA \colon \mathbf{IP}_{\piloc}^{\Gamma,\%} \to \mathbf{IP}_{\piloc}^\Gamma, \quad \bB \colon \mathbf{IP}_{\piloc}^{\Gamma,\%} \to (\mathbf{IP}_{\piloc}^\Gamma|_{\Or (\Gamma)})_{\%},
\]
such that $\bB_X$ is a weak equivalence for any $X \in \ECW_{\Gamma,\mathrm{fp}}$ and $\bA_{\Gamma/H}$ are weak equivalences for any $\Gamma/H \in \Or (\Gamma)$. 
The domain of our functor $\mathbf{IP}_{\piloc}^{\Gamma}$ is smaller than the whole category $\CW_{\Gamma}$ of $\Gamma$-CW-complexes, but as can be seen from the proof of \cite{davisSpacesCategoryAssembly1998}*{Theorem 6.3}, this does not affect the proof of these weak equivalences. 
These data give the commutative diagram 
\begin{align*}
    \xymatrix@C=8ex{
    \mathbf{IP}_{\piloc}^\Gamma(E(\Gamma,\cF_{\fin})) \ar[d]^{\pr_*}  &     \mathbf{IP}_{\piloc}^{\Gamma,\%}(E(\Gamma,\cF_{\fin})) \ar[r]^{\bB_{E(\Gamma,\cF_{\fin})} \hspace{4ex}} \ar[l]_{\bA_{E(\Gamma,\cF_{\fin})}} \ar[d]^{\pr_*} & 
    (\mathbf{IP}_{\piloc}^\Gamma|_{\Or (\Gamma)})_{\%}(E(\Gamma,\cF_{\fin})) \ar[d]^{\pr_*} \\
    \mathbf{IP}_{\piloc}^\Gamma(\pt)   &
    \mathbf{IP}_{\piloc}^{\Gamma,\%}(\pt) \ar[r]^{\bB_\pt } \ar[l]_{\bA_\pt } & 
    (\mathbf{IP}_{\piloc}^\Gamma|_{\Or (\Gamma)})_{\%}(\pt).
    }
\end{align*}
The right vertical map $\pr_*$ is the Davis--L\"uck assembly map. 
Moreover, \cite{davisSpacesCategoryAssembly1998}*{Theorem 6.3} also states that $\bA_X$ is a weak equivalence if there is a class of subgroups $\cF$, that is closed under taking subgroups and conjugations, such that $\mathbf{IP}_{\piloc}^\Gamma$ is $\cF$-excisive and $X$ is a $\Gamma$-$\cF$-CW-complex. 
We apply this fact to $\cF=\cF_{\fin}$ and $X=E(\Gamma, \cF_{\fin})$.

As is already mentioned in \cref{lem:assembly.DL} (3), our functor $\mathbf{IP}_{\piloc}^\Gamma$ is $\cF_{\fin }$-excisive. Moreover, the universal $\Gamma$-$\cF_{\fin}$-space $E(\Gamma, \cF_{\fin})$ is modeled by $\bE_{\Gamma}$ (\cref{rmk:classifying.space}). 
In summary, the above diagram is rewritten as
\begin{align*}
    \xymatrix@C=6ex{
    \mathbf{IP}^\Gamma(\bE_{\Gamma}) \ar[d]^{\mu_{\rIP}}  &     \mathbf{IP}_{\piloc}^{\Gamma,\%}(\bE_{\Gamma}) \ar[r]^{\bB_{\bE_{\Gamma}}}_{\simeq} \ar[l]_{\bA_{\bE_{\Gamma}}}^{\simeq} \ar[d] & 
    \mathbf{IP}_{\%}^\Gamma(\bE_{\Gamma}) \ar[d]^{\mu_{\mathbf{IP}}^{\mathrm{DL}}} \\
    \mathbf{IP}^\Gamma(\pt)   &
    \mathbf{IP}_{\piloc}^{\Gamma,\%}(\pt) \ar[r]^{\bB_{\pt}}_{\simeq} \ar[l]_{\bA_{\pt}}^{\simeq} & 
    \mathbf{IP}^\Gamma(\pt) .
    }
\end{align*}
This gives the desired comparison of the equivariant coarse and Davis--L\"{u}ck assembly maps.

In the end, we conclude the paper with an optimistic conjecture that our new assembly map is also an isomorphism as it is for the successful precedents in topological and algebraic K-theory, the Baum--Connes and the Farrell--Jones conjectures.

\begin{conj}
Let $\Gamma$ be a countable group acting properly and isometrically on the Euclidean space $\bE_{\Gamma}$. Then the assembly map \eqref{eqn:assembly} is a $\Gamma$-weak equivalence. 
\end{conj}

\appendix

\section{\texorpdfstring{$G$}{G}-simplicial sets, \texorpdfstring{$G$}{G}-sheaves, and geometric realization}\label{section:Gsheaf}
Let $G$ be a compact Lie group and let $\sF$ be a $G$-sheaf on $\Man$. We give a proof of \cref{thm:equivariant.MW}, stating that there is a topological $G$-space that remembers the $G$-equivariant smooth homotopy theory of the sheaf $\sF$. We first give a constructive proof in \cref{subsection:G.realization} following the line of \cite{madsenStableModuliSpace2007}, and then explain a more abstract homotopy theory behind it in \cref{subsection:sheaf.model,subsection:Gsheaf.category}.  

\subsection{Proof of \texorpdfstring{\cref{thm:equivariant.MW}}{Theorem 3.8}}\label{subsection:G.realization}
Let $\Delta$ denote the simplex category. For closed subgroups $H,K \leq G$, we write $\sHom_G(G/H,G/K)$ for the space of $G$-maps from $G/H$ to $G/K$, which is identified with the manifold $(G/K)^H$. 
\begin{defn}\label{defn:category.DeltaG}
    Let $\Delta_G$ denote the category where
    \begin{itemize}
        \item an object of $\Delta_G$ is a tuple $[n,\bm{H}]$, where $\bm{H}=\{H_i\}_{i=0}^n$ is a decreasing sequence of closed subgroups $H_0 \geq H_1\geq \cdots \geq H_n$ of $G$, and
        \item a morphism from $[n,\bm{H}]$ to $[m,\bm{K}]$ is given by a pair $(\iota,f)$, where $\iota \colon [n] \to [m]$ is a map of $\Delta$ and $f_i \colon \mathop{\mathrm{int}} \Delta_{i} \to \sHom_G (G/H_{i} , G/K_{\iota(i)})$ is a smooth map  such that 
        \begin{align*} 
        f \colon \Delta_{n}(G,\bm{H}) \to \Delta_p(G, \bm{K}), \quad f(t,gH_i) = f_i(t,gH_i) \text{ if $t \in \Delta_i$}
        \end{align*} 
        becomes a continuous map. 
    \end{itemize}
    Here, $\Delta_n(G,\bm{H})$ denote Illman's $G$-equivariant simplex \cite{illmanEquivariantTriangulationTheorem1983} defined by
    \begin{align*} 
    \Delta_n(G,\bm{H}) \coloneqq (\Delta_n \times G) /\sim , \quad \text{$(p,g_1) \sim (p,g_2)$ if $p \in \Delta_i$ and $g_1^{-1}g_2 \in H_i$, }
    \end{align*}
    where $\Delta_i$ is regarded as a subcomplex of $\Delta_n$ by the geometric realization of the standard inclusion $\{0 ,\cdots,i\} \subset \{0,\cdots,n\}$ in the category $\Delta$. Note that $\Delta[n,\bm{H}] \coloneqq \Delta_n (G,\bm{H})$ is a covariant functor from $\Delta_G $ to the category $\kTop_G$ of compactly generated Hausdorff $G$-spaces. 
\end{defn}
For a functor $\sfX \colon \Delta_G^{\mathrm{op}} \to \Set$, its geometric realization is a $G$-space defined by the coend 
    \begin{align*}
        |\sfX|_G \coloneqq \int^{\Delta_G} \Delta[n,\bm{H}] \times \sfX[n,\bm{H}].
    \end{align*}
Explicitly, it is the topological space  
\begin{align} 
    |\sfX |_G \coloneqq \bigsqcup_{[n,\bm{H}] \in \Delta_G}  \Delta_n (G,\bm{H}) \times \sfX[n,\bm{H}] \Big/ \sim, \label{eqn:Gsimplicial.realization1}
\end{align}
    where the equivalence relation is generated by 
\begin{align*}
    \Delta_n(G,\bm{H}) \times \sfX[n,\bm{H}] \ni \big( (p,gH_i),f^*\sfx \big)      \sim{}  \big( f(p,gH_i), \sfx \big) \in \Delta_m(G,\bm{K}) \times \sfX[m,\bm{K}]
\end{align*}
for any $(\iota,f) \in \Hom ([n,\bm{H}],[m,\bm{K}])$ and $(p,gK_i) \in \mathop{\mathrm{int}} \Delta_i \times G/K_{\iota(i)}$. 

We apply this construction to the following smooth $G$-singular simplicial set. 
\begin{defn}
    Let $\sF$ be a $G$-sheaf on $\Man$. Its $G$-singular simplicial set is a functor $\Sing_G \sF \colon \Delta_G^{\mathrm{op}} \to \Set$ defined by
    \begin{align*}
    \Sing_G \sF[n,\bm{H}] \coloneqq {}&{} \{ s \in \sF(\Delta_n^e) \mid \text{$s|_{\Delta_i^e} \in \sF(\Delta_i^e)^{H_i}$ for any $0 \leq i \leq n$} \}, \\
      (\Sing_G \sF) (\iota,f) (\sfs) \coloneqq {}& (f_i \times \iota_*)^* (\nu(\sfs)).
    \end{align*}
\end{defn}
The above definition of $(\Sing_G \sF) (\iota,f) $ makes sense since the $H_i$-invariance of $\sfs|_{\Delta_i^e}$ is equivalent to $\nu (\sfs|_{\Delta_i^e}) \in \sF(G/H_i \times \sM) \subset \sF(G \times \sM)$. 

    We write $\PSh(\Delta_G,\Set)$ for the category of presheaves, i.e., contravariant functors, from $\Delta_G$ to $\Set$. 
    The above assignments $\Sing_G \colon \Sh_G(\Man) \to \PSh(\Delta_G,\Set)$ and $|\blank |_G \colon \PSh(\Delta_G, \Set) \to \kTop_G$ are functorial in an obvious way. 

Since the $G$-fixed point set of $\Delta_n(G,\bm{H})$ is $\Delta_i$ such that $H_i=G$, we get the following identification. 
\begin{prp}\label{prp:fixed.point.Gsing}
    The $G$-fixed point set of $|\Sing_G \sF|_G$ is identical to the non-equivariant realization $|\Sing \sF^G|$ of the fixed point subsheaf. 
\end{prp}

\begin{lem}
    The functors $|\blank|_G \colon \PSh(\Delta_G,\Set) \leftrightarrows \kTop_G \colon \Sing_G$ form an adjoint pair. Here, for a $G$-space $Y$, $\Sing_GY$ is the abbreviation of $\Sing_G  C^0(\blank,Y)$.
\end{lem}
\begin{proof}
    The adjunction counit $\epsilon_Y \colon |\Sing_G Y|_G \to Y$ is given by the evaluation map
    \begin{align*} 
     \Delta_n(G,\bm{H}) \times \Sing_G Y [n,\bm{H}] \ni ((p,gH_i), \sfs ) \mapsto g \cdot \sfs(p) \in Y.
    \end{align*}
    The adjunction unit $\eta_{\sfX} \colon \sfX \to \Sing_G |\sfX|$ is given by 
    \begin{align*}
    \eta_{\sfX}[n,\bm{H}] (\sfx) = \id \times \sfx  \in \Hom (\Delta_n , \Delta_n(G,\bm{H}) \times \sfX[n,\bm{H}]) \subset \Hom (\Delta_n,|\sfX|).
    \end{align*}
    It is routine to show that they satisfy the unit-counit relation.
\end{proof}

\begin{lem}\label{lem:G.triangulation}
    Let $\sM$ be a smooth $G$-manifold. Then, there exists a continuous $G$-map $\chi_{\sM} \colon \sM \to |\Sing_G \sM|_G$ that is a right inverse of $\epsilon_\sM$. 
\end{lem}
\begin{proof}
    Let $\fC_G$ denote the poset of conjugacy classes of closed subgroups of $G$. 
    For a $G$-manifold $\sM$, the quotient $\sM/G$ has a map $\fc \colon \sM/G \to \fC_G$ given by $\fc(G \cdot m )= \langle \Stab_G(p) \rangle$, which is continuous with respect to the poset topology on $\fC_G$. That is, for any conjugacy class $\langle H\rangle \in \fC_G$, the subset 
    \begin{align*}
    (\sM/G)^{\langle H \rangle } \coloneqq \{ G \cdot p \in \sM/G \mid \langle \Stab_G(p) \rangle \leq \langle H\rangle  \} 
    \end{align*}
    is closed in $\sM/G$. 
    
    A smooth $G$-triangulation of a $G$-manifold $\sM$ in the sense of \cite{illmanEquivariantTriangulationTheorem1983} consists of the data $(u, \sfT)$ and $(\varphi_{\sigma})_\sigma$, where 
    \begin{itemize}
        \item $\sfT$ is an abstract simplicial complex, and $u \colon |\sfT| \to \sM/G$ is the barycentric subdivision of a triangulation $u_0 \colon |\sfT_0| \to \sM /G$ that is $\fC_G$-filtered in the sense of \cite{waasStratifiedHomotopyTheory2021}, i.e., $\sfT_0^{\langle H \rangle} \coloneqq \sfT_0 \cap (\sM/G)^{\langle H \rangle}$ is a closed subcomplex of $\sfT_0$ for any $\langle H \rangle \in \fC_G$, and smooth (i.e., smooth on the interior of each cell), and 
        \item a family $\{\varphi_\sigma\}_{\sigma \in \sfT}$ of continuous $G$-homeomorphisms $\varphi_\sigma \colon \Delta_n(G, \bm{H}) \to \pi^{-1}(\sigma)$ over $u_{\sigma}  \colon \Delta_n \to \pi^{-1}(u(|\sigma| ))$ whose restriction to the interior of each cell is smooth. Here, $n$ is the dimension of $\sigma$ and $H_0 \geq \cdots \geq H_n$ is a decreasing sequence of closed subgroups of $G$ such that $\langle H_i \rangle = \fc(f(p)) $ if $p \in \Delta_i$.
    \end{itemize}
    In \cite{illmanEquivariantTriangulationTheorem1983}, Illman proved that any $G$-manifold has a $G$-triangulation. It is not explicitly stated that such a $G$-triangulation can be chosen to be \emph{smooth}, but it is actually possible.  
    Indeed, the existence of $\fC_G$-filtered triangulation of $\sM/G$ by \cite{goreskyTriangulationStratifiedObjects1978} is proved in the smooth category (see also \cite{veronaTriangulationStratifiedFibre1979}*{Corollary 3.8}). Moreover, each $G$-map $\varphi_\sigma$ is uniquely recovered from its restriction $c_\sigma \coloneqq \varphi_{\sigma}|_{\Delta_n} \colon \Delta_n \to \pi^{-1}(u(|\sigma|))$ with the property that $c_\sigma(p) $ is $H_i$-invariant if $p \in \Delta_i$. 
    This $c_{\sigma}$ can be replaced with a smooth one by applying the Whitney approximation theorem inductively with respect to the dimension of the cells.

    By replacing $\varphi_\sigma$ if necessary, we may extend $c_{\sigma}$ to a smooth $G$-map $c_{\sigma}^e \colon \Delta_n^e \to \sM$ such that $\mathrm{Stab}_G(c_{\sigma}(p)) \supset H_i$. 
    Now, $\chi_\sM \colon \pi^{-1}(u(|\sigma|)) \to |\Sing_G\sM|_G$ is defined by 
    \begin{align*} 
    \chi_{\sM}(p) = ( u_{\sigma}^{-1}(p) , c_{\sigma}^e) \in \Delta_n(G,\bm{H}) \times C^\infty(\Delta_n^e,\sM).
    \end{align*}
    They are clearly assembled to the desired map.
\end{proof}

For a $G$-simplicial set $\mathsf{X}$, we define the quotient space $\sfX/G$ as a simplicial set by
    \begin{align*} 
    (\mathsf{X}/G)[n] \coloneqq \bigg( \bigsqcup_{H_0 \geq \cdots \geq H_n} \mathsf{X}[n,\bm{H}] \bigg)\bigg/ \sim,
    \end{align*}
    where the equivalence relation $\sim$ is generated by morphisms $\sfX [n,\bm{H}]\to \sfX[n,\bm{K}]$ induced by those of $\Delta_G$. 
    Then, the quotient map
    \begin{align*} 
    \Delta_n(G,\bm{H}) \times \sfX[n,\bm{H}] \to \Delta_n \times (\mathsf{X}/G)[n]
    \end{align*} 
    induces the continuous map $q \colon |\mathsf{X}|_G \to |\mathsf{X}/G|$, which descends to a homeomorphism $|\sfX|_G/G \to |\mathsf{X}/G|$. 
    Moreover, the family of subsets 
    \begin{align*} (\sfX/G)^{\langle H\rangle}[n] \coloneqq \bigg( \bigsqcup_{\langle H_0 \rangle \leq \langle H \rangle } \sfX[n,\bm{H}] \bigg)/\sim  \end{align*}
    makes $|\sfX/G|$ a $\fC_G$-filtered simplicial set. 
\begin{lem}\label{lem:covering.simplicial.approximation}
    Let $\sM$ be a $G$-manifold and let $F \colon \sM \to |\mathsf{X}|_G$ be a continuous $G$-map. Then there is a $G$-triangulation of $\sM$ and an another $G$-map $F'$ that is $G$-homotopic to $F$ and the induced map $F' \colon \sM /G \to |\mathsf{X}/G|$ is simplicial.
\end{lem}
\begin{proof}
    Pick a $\fC_G$-filtered triangulation $u \colon |\sfT_0| \to  \sM/G$ and let $\sfT$ be its barycentric subdivision.
    The composition $q \circ F$ factors through $F_0 \colon \sM/G \to |\mathsf{X}/G|$, which is a continuous map over $\fC_G$.
    Thus, by \cite{waasStratifiedHomotopyTheory2021}*{Theorem 1.3.44}, we may replace $F_0$ to a simplicial map $F_0'$ after a multiple barycentric subdivision of $\sfT$ if necessary.  
    This lifts to a $G$-map $F' \colon \sM \to |\mathsf{X}|_G$ by Palais' covering homotopy theorem (\cite{bredonIntroductionCompactTransformation1972}*{Theorem II.7.3}). 
    Finally, by an inductive replacement as in the proof of \cref{lem:G.triangulation}, $F'$ can be chosen so that the restrictions 
    \[
    F' \colon \mathop{\mathrm{int}}\pi^{-1}(u(|\sigma|)) \to \{ F'_0 (\sigma)\} \times \mathop{\mathrm{int}}\Delta[n,\bm{H}] \subset |\sfX|_G
    \]
    are smooth for each $\sigma \in \sfT$. Therefore, $\varphi_{\sigma} \coloneqq (F'|_{\pi^{-1}(u(|\sigma|))})^{-1}$ gives a smooth $G$-triangulation of $\sM$ such that $F'$ is $G$-simplicial. 
\end{proof}

The following is the equivariant version of \cite{madsenStableModuliSpace2007}*{Proposition A.1}, which is proved similarly. 
\begin{prp}\label{prp:Yoneda.Gsheaf}
    Let $\sM$ be a $G$-manifold. Then the composition 
    \begin{align*} 
    \xi \colon \sF(\sM)^G \xrightarrow{Y} \Hom_{\Sh_G(\Man)} (\sM , \sF)^G \xrightarrow{\chi_{\sM} \circ |\Sing_G (\blank )|_G } \Hom_{\kTop_G} (\sM, |\Sing_G (\sF )|_G )^G
    \end{align*}
    induces the isomorphism 
    \begin{align*}
    \sF[\sM]^G \cong [\sM, |\Sing_G (\sF )|_G ]^G.
    \end{align*}
\end{prp}
\begin{rmk}
    For any $G$-manifold, the Yoneda lemma shows that the embedding 
    \begin{align*}
    Y \colon \sF(\sM)^G \to \Hom_{\Sh_G(\Man)}(\sM,\sF)^G
    \end{align*}
    given by $Y(\sfs)(f) = f^*\sfs$ for $\sfs \in \sF(\sM)^G$ and $f \in C^\infty(\sN,\sM)$ is bijective with the inverse $\phi \mapsto \phi(\id_{\sM})$. Thus \cref{prp:Yoneda.Gsheaf} claims that $\chi_{\sM} \circ |\Sing_G (\blank)|_G$ induces isomorphisms of homotopy sets. 
\end{rmk}
\begin{proof}
We construct the inverse $\vartheta \colon [\sM, |\Sing_G (\sF )|_G ]^G \to \sF[\sM]^G$ of $\xi$. 
By \cref{lem:covering.simplicial.approximation}, a continuous $G$-map $F \colon \sM \to |\Sing_G \sF|_G$ is $G$-homotopic to the realization of a $G$-simplicial map $\sfT \to \Sing_G \sF$ for some $G$-triangulation $\sfT$ of $\sM$. 
In the same way as \cite{madsenStableModuliSpace2007}*{Proposition A.1}, pick a smooth $G$-map $h \colon \sM \times[0,1] \to \sM$ with $h_0=\id$ and 
\begin{itemize} 
\item $h_t$ sends each $G$-simplex $\varphi_{\sigma}(\Delta_n(G,\bm{H}))$ to itself,
\item $h_t$ sends each subset $c_\sigma^e(\Delta_n^e) \subset \sM$ to itself,
\item each $G$-simplex $\varphi_{\sigma}(\Delta_n(G,\bm{H}))$ has an $G$-invariant open neighborhood $V_\sigma \subset \sM$ with $h_1(V_\sigma) \subset \varphi_{\sigma}(\Delta_n(G,\bm{H}))$.
\end{itemize}
For $\sigma \in \sfT$ and $p \in \pi^{-1}(u(|\sigma|))$, set $F(p) =: (\varphi_\sigma(p),\sfs_\sigma)$, where $\sfs_\sigma \in \Sing_G \sF[n,\bm{H}]$ and $\varphi_\sigma \colon \sM \to \Delta_n$, and 
\begin{align*} 
\tilde{\sfs}_\sigma \coloneqq h_1^* \circ \varphi_\sigma^* (\sfs_{\sigma}) \in \sF(V_{\sigma})^G. 
\end{align*}
The family $\{ \tilde{\sfs}_{\sigma }$ satisfies the gluing condition, and hence is glued together to a global section $\sfs \in \sF(\sM)$. 
We define $\vartheta$ by $\vartheta ([F]) \coloneqq [\sfs]$. 

By definition, we have $\vartheta \circ \xi[\sfs] = h_1^*\sfs $, which is smoothly homotopic to $\sfs$. This shows $\vartheta \circ \xi = \id$. 
Conversely, let $F \colon \sM \to |\Sing_G \sF|_G$ be a continuous $G$-map. Without loss of generality, we may assume that it is $G$-simplicial with respect to a $G$-triangulation $\sfT$. 
Then, the $G$-simplicial map $F' \colon \sfT \to |\Sing_G \sF|_G$ given by $F'(\sigma) = h_1^* F(\sigma)$ represents the homotopy class $\xi \circ \vartheta [F]$. In the same way as \cite{madsenStableModuliSpace2007}*{Proposition A.1}, $F$ and $F'$ are connected by a path-valued $G$-simplicial map $\sfT \to |\cP \Sing_G \sF|_G$. By \cite{madsenStableModuliSpace2007}*{Lemma A.1.2}, this shows $\xi \circ \vartheta  = \id$. 
\end{proof}

\begin{cor}\label{cor:fixed.point.realization}
    For a closed subgroup $H \leq G$, there is a $H$-equivariant map $|\Sing_H \sF|_H \to |\Sing_G \sF|_G$ that induces a weak equivalence of fixed point subgroups $(|\Sing_H \sF|_H)^K \to (|\Sing_G \sF|_G)^K$ for any closed subgroup $K\leq H$. 
    In particular, the fixed point space $(|\Sing_G \sF|_G)^H$ of a $G$-realization is weakly equivalent to the non-equivariant realization $|\Sing \sF^H|$ of the $H$-fixed point sheaf.
\end{cor}
\begin{proof}
Let $\mathrm{ind}_H^G \colon \PSh(\Delta_H,\Set) \to \PSh(\Delta_G,\Set)$ be the functor induced from the canonical inclusion $\Delta_H \to \Delta_G$ sending $[n,\bm{L}]$ to itself. 
Since $\Delta_{n}(H,\bm{L}) \times_H G \cong \Delta_n(G,\bm{L})$, this induces a continuous $G$-map 
\begin{align*} 
    \mathop{\mathrm{ind}_H^G} \colon |\Sing_H \sF|_H \times_H G \to |\Sing_G \sF|_G
\end{align*}
For any $\sM \in \Man$, we have that $\sF^K(\sM) \cong \sF^H(\sM \times H/K) \cong \sF^G(\sM \times G/K)$ and the diagram 
\begin{align*}
\xymatrix{
 \sF[\sM \times H/K]^H \ar[r]^{\mathop{\mathrm{ind}_H^G}}_{\cong} \ar[d]^\xi_{\cong} & \sF[\sM \times G/K]^G \ar[d]^\xi_{\cong } \\
 [\sM \times H/K , |\Sing_H \sF|_H]^H \ar[r]^{\mathop{\mathrm{ind}_H^G}} & [\sM \times G/K, |\Sing_G \sF|_G]^G
}
\end{align*}    
commutes. Thus \cref{prp:Yoneda.Gsheaf} shows the desired weak equivalence. The latter claim follows from \cref{prp:fixed.point.Gsing}. 
\end{proof}

\subsection{Enrichment by the monoidal model category \texorpdfstring{$\Sh(\Man)$}{Sh(Man)}}\label{subsection:sheaf.model}
The following subsection provides a category-theoretic background behind our constructive proof of \cref{prp:Yoneda.Gsheaf}. 
As its basis, here we discuss the abstract homotopy theory of the category $\Sh(\Man)$. 
In conclusion, $\Sh(\Man)$ has the structure of a monoidal model category with nice properties required in \cite{guillouEnrichedModelCategories2020}. 
Therefore, it can play the role of a base category in enriched category theory. 
This observation allows us to give a model of the $(\infty,1)$-category and to consider the category of $\infty$-presheaves. 
This viewpoint will be applied to formulate the sheaf version of the Elmendorf theorem in \cref{prp:Elmendorf}. 
 
\begin{rmk}\label{lem:Sh.bicomplete}
    The category $\Sh(\Man)$ is bicomplete and locally presentable. We refer to \cite{maclaneSheavesGeometryLogic1994}*{Section III.6} and \cite{borceuxHandbookCategoricalAlgebra1994}*{Vol.\ 3, Proposition 3.4.16}. 
\end{rmk}

\begin{lem}\label{lem:Sh.closed.monoidal}
    The monoidal product $\sF \times \sG$ given by $(\sF \times \sG)(\sM) \coloneqq \sF(\sM) \times \sG(\sM)$ and the internal Hom sheaf $\sHom (\sF_1,\sF_2) \in \Sh (\Man)$ defined in \eqref{eqn:internal.hom} makes $\Sh_G(\Man)$ a closed symmetric monoidal category via the composition $\sHom (\sF_1,\sF_2) \times \sHom (\sF_2,\sF_3) \to \sHom (\sF_1,\sF_3)$ defined by  
    \begin{align*} 
    \sF_1(\sN) \xrightarrow{\phi_1} \sF_2(\sN \times \sM) \xrightarrow{\phi_2} \sF_3(\sN \times \sM \times \sM) \xrightarrow{\Delta_{\sM^*} } \sF_3(\sN \times \sM).
    \end{align*}
\end{lem}
\begin{proof}
    It suffices to show that the internal Hom functor is right adjoint to the tensor product, i.e.,  
    \begin{align*}
        \Hom_{\Sh(\Man)} (\sG_1, \sHom( \sF, \sG_2)) \cong \Hom_{\Sh(\Man)} (\sG_1 \times \sF , \sG_2).
    \end{align*}
    This isomorphism is given by the adjunction unit $\eta \colon \id \to \sHom (\sF, (\sF \times \blank ))$ defined by 
    \[
    \eta_{\sG}(\sM)(\sfs) \coloneqq \id_{\sF} \times \sfs \in \Hom (\sF,\Hom(\sM,\sF \times \sG))
    \]
    for $\sfs \in \sG(\sM)$, and the counit $\epsilon \colon \sF \times \sHom (\sF, \blank ) \to \id$ defined by 
    \[
    \epsilon_{\sG}(\sM)(\sfs \times \phi) \coloneqq \Delta_{\sM}^* \phi(\sfs) \in \sG(\sM)
    \]
    for $\sfs \times \phi \in \sF(\sM) \times \Hom (\sF, \Hom(\sM,\sG ) )$.
\end{proof}

For a simplicial set $\sfX$, its smooth realization defined in \cite{pavlovProjectiveModelStructures2022}*{Definition 3.3} by
\begin{align*}
    |\sfX|_{\mathrm{sm}} \coloneqq \colim_{x \in \Delta/\sfX} C^\infty ( \blank, \{\sfx \} \times \Delta_n^e),
\end{align*}
where $\Delta/\sfX$ denotes the category whose objects are simplices $\sfx$ of $\sfX$ and a morphism from an $m$-simplex $\sfx$ to an $n$-simplex $\sfy$ is given by a morphism $f \colon [m] \to [n]$ in $\Delta$ such that $f^*(\sfy) =\sfx$. 
More explicitly, a section $ |\sfX|_{\mathrm{sm}}(\sM)$ is given by a family of germs $(\sfx_p,\rho_p) \in \sfX_{n_p} \times C^\infty_p(\sM, \Delta_n^e)$ satisfying the following gluing property: For any $p \in \sM$, there is an open neighborhood $\sU_p$ and $\tilde{\rho}_p \in C^\infty(\sU_p,\Delta_{n_p}^e)$ such that $\sfx_q = f_{q,p}^*(\sfx_p) $ and $\rho_q = \tilde{\rho}_q$ for any $q \in \sU$, where $f_{q,p} \colon [n_q] \to [n_p]$ is uniquely characterized by $\Im \rho_q \subset f_{q,p}(\Delta_{n_q}^e)$.  
This smooth geometric realization gives a functor $|\blank |_{\mathrm{sm}} \colon \sSet \to \Sh(\Man)$.

\begin{lem}
    The smooth geometric realization and the smooth singular functors forms an adjoint pair
    \begin{align} 
    \xymatrix{
    |\blank|_{\mathrm{sm}} \colon \sSet \ar@<1ex>[r] \ar@{}[r]|{\bot \hspace{4ex}} & 
    \Sh(\Man)\ar@<1ex>[l] \colon \Sing .
    }\label{eqn:adjoint.singular.realization}
    \end{align}
\end{lem}
\begin{proof}
    This adjunction is given by the unit $\eta \colon \id \to \Sing (|\blank|_{\mathrm{sm}})$ defined by 
    \[
    \eta (\sfx) \coloneqq \{ \sfx \} \times  \id_{\Delta_n^e}\in C^\infty (\Delta_n^e , \{ \sfx \} \times \Delta_n^e) \subset C^\infty(\Delta_n^e, |\sfX|_{\mathrm{sm}})
    \]
    for $\sfx \in \sfX_n$, and the counit $\epsilon \colon |\Sing (\blank)|_{\mathrm{sm}} \to \id$ is the evaluation map defined by 
    \[
    \epsilon_{\sF}((\sfs_p,\rho_p)_{p \in \sM} ) \coloneqq (\text{the gluing of the germs $\rho_p^*\sfs \in \sF_p$}) \in \sF(\sM)
    \]
    for a smooth map $\varphi \colon \sM \to |\Sing \sF|_{\mathrm{sm}}$. It is verified in the same way as the ordinary realization-singular adjunction that they satisfy the conditions for unit and counit.   
\end{proof}

\begin{lem}[{\cite{pavlovProjectiveModelStructures2022}*{Theorem 7.4}}]\label{lem:Sh.model}
    There is a cofibrantly generated (\cite{hirschhornModelCategoriesTheir2003}*{Definition 11.1.1}) model structure on $\Sh(\Man)$ such that the adjoint pair \eqref{eqn:adjoint.singular.realization} is a Quillen equivalence. 
\end{lem}
\begin{proof}
    We apply Kan's criterion of the existence of right transfer model structure \cite{hirschhornModelCategoriesTheir2003}*{Theorem 11.3.2}. Let $\cI$ be the set of boundary inclusions $\partial \Delta_n \to \Delta_n$ and let $\cJ$ be the set of horn inclusions $\Lambda_{n,k} \to \Delta_n$, which are the generating cofibrations and generating acyclic cofibrations of $\sSet$. We write $|\cI|_{\mathrm{sm}}$ and $|\cJ|_{\mathrm{sm}}$ for their images by the geometric realization functor. 
    Since $\Sh(\Man)$ is locally presentable, it suffices to show that 
    any relative $|\cJ|_{\mathrm{sm}}$-cell complex is a weak equivalence.   
    Let $\kappa$ be a regular ordinal and if $\{ \sF_{\beta}, f_\beta \}_{\beta < \kappa}$ is a $\kappa$-sequence such that each $f_{\beta}$ is a pushout via the geometric realization of horn inclusions $|\Lambda_n^k|_{\mathrm{sm}} \to |\Delta_n|_{\mathrm{sm}}$. 
    Then, by definition of the colimit of sheaves on $\Man$, the compactness of $S^n$, and \cref{thm:MW}, we have 
    \[
        \pi_n\Big(\Sing \colim_{\beta < \kappa'} \sF_\beta \Big) \cong  \Big( \colim_{\beta < \kappa'} \sF_\beta \Big) [S^n] \cong \colim_{\beta < \kappa'} (\sF_\beta [S^n]) \cong \colim_{\beta < \kappa'} \pi_n (\Sing \sF_\beta)
    \]
    for any limit ordinal $\kappa' \leq \kappa$. 
    Hence, a transfinite induction shows that the transfinite composition $\Sing f_{0,\kappa} \colon \Sing \sF_0 \to \Sing (\colim_{\beta < \kappa} \sF_\beta)$ induces an isomorphism of homotopy groups, and hence is a weak equivalence. The model structure on $\Sh(\Man)$ is cofibrantly generated by definition (cf.\ \cite{hirschhornModelCategoriesTheir2003}*{Remark 11.1.3}). 
\end{proof}

\begin{rmk}\label{rmk:Sh.proper}
    A morphism $\phi \colon \sF \to \sG$ is a fibration with respect to this model structure if and only if $\Sing \phi \colon \Sing \sF \to \Sing \sG$ is a fibration, and hence if and only if $\phi$ is a fibration in the sense of \cref{paragraph:fibration}. In particular, the model structure on $\Sh(\Man)$ is right proper (\cite{hirschhornModelCategoriesTheir2003}*{Definition 13.1.1}). It is proved by Pavlov in \cite{pavlovProjectiveModelStructures2022}*{Theorem 7.4} that this model structure is proper. 
    
    This observation also shows that every object $\sF \in \Sh(\Man)$, which is sent to a Kan complex, is fibrant. Together with the cofibrancy of $\Sing \sF \in \sSet$, we obtain that the adjunction counit $\epsilon \colon |\Sing \sF|_{\mathrm{sm}} \to \sF$ is a weak equivalence by \cite{hoveyModelCategories1999}*{Corollary 1.3.16} (cf.\ \cref{thm:MW}).  
\end{rmk}

\begin{rmk}\label{rmk:Sh.monoidal.model}
    The model structure on $\Sh(\Man)$ is monoidal (\cite{hoveyModelCategories1999}*{Definition 4.2.6}). Since the monoidal unit $\mathbf{1} \in \Sh(\Man)$ is cofibrant, it suffices to show that 
    \[
    \phi^* \times\psi_* \colon \sHom (\sF_2, \sG_1) \to \sHom (\sF_1, \sG_1) \times_{\sHom(\sF_1,\sG_2)} \sHom (\sF_2, \sG_2)
    \]
    is a weak equivalence for any cofibration $\phi \colon \sF_1 \to \sF_2$ and any fibration $\psi \colon \sG_1 \to \sG_2$. This statement is proved in the same way as \cite{hoveyModelCategories1999}*{Proposition 4.2.11}, by using the Quillen equivalence \eqref{eqn:adjoint.singular.realization} instead of the analogous adjunction $\sSet \leftrightarrows \kTop$. 
\end{rmk}

\begin{rmk}\label{rmk:Sh.monoid}
    The weak equivalence of $\Sh(\Man)$ is closed under the transfinite composition. In particular, the model category $\Sh(\Man)$ satisfies the monoid axiom (\cite{guillouEnrichedModelCategories2020}*{Definition 4.26}).
\end{rmk}

By \cref{lem:Sh.bicomplete,lem:Sh.closed.monoidal,rmk:Sh.proper,rmk:Sh.monoidal.model,lem:Sh.model,rmk:Sh.monoid}, the category $\Sh(\Man)$ can serve as an enriching category. As well as $\sSet$ and $\kTop$, a $\Sh(\Man)$-enriched category gives a model of $(\infty,1)$-category. For example, the category $\Sh_G(\Man)$ of $G$-sheaves on $\Man$ is a $\Sh(\Man)$-enriched category by using the internal Hom sheaf $\sHom_G (\sF,\sG) \in \Sh(\Man)$ given by $\sHom_G (\sF,\sG)(\sM) \coloneqq \Hom_G(\sF,\Hom(\sM,\sG))$ as morphism objects.

\subsection{The Elmendorf theorem for \texorpdfstring{$G$}{G}-sheaves}\label{subsection:Gsheaf.category}
We are ready to formulate the Elmendorf theorem for $G$-sheaves on $\Man$.
As is explained at the end of the previous subsection, the category $\Sh_G(\Man)$ of $G$-sheaves is endowed with the structure of a $\Sh(\Man)$-enriched category.

\begin{defn}
    We use the same letter $\Or (G)$ for the $\Sh(\Man)$-enriched full subcategory of $\Sh_G(\Man)$ generated by homogeneous spaces $G/H$ via the Yoneda embedding $Y \colon \Man_G \to \Sh_G(\Man )$. We write $\PSh_{\infty}(\Or(G) , \Sh(\Man))$ for the category of contravariant $\Sh(\Man)$-enriched functors from $\Or(G)$ to $\Sh(\Man)$.   
\end{defn}

We remark that the morphism sheaf $\sHom_G (G/H,G/K)$ in the category $\Or(G)$ is the manifold $(G/K)^H$. For $G/H \in \Or(G)$, the functor $F_{G/H} \colon \Sh(\Man) \to \PSh_\infty(\Or(G), \Sh(\Man))$ is given by $F_{G/H}(\sF) (G/K)= \sF \times \sHom_G (G/K , G/H) $ (\cite{guillouEnrichedModelCategories2020}*{Definition 1.6}). This $F_{G/H}$ is left adjoint to $\ev_{G/H} \colon \PSh_{\infty}(\Or(G),\Sh(\Man)) \to \Sh(\Man)$. 

\begin{prp}\label{prp:levelwise.model}
    The category $\PSh_{\infty}(\Or(G), \Sh(\Man))$ has the `levelwise' model structure characterized by 
    \begin{itemize}
        \item the set of generating cofibrations given by $F|\cI|_{\mathrm{sm}} = \bigcup_{G/H \in \Or(G)} F_{G/H}  (|\cI|_{\mathrm{sm}})$, and
        \item the set of generating acyclic cofibrations given by $F|\cJ|_{\mathrm{sm}} = \bigcup_{G/H \in \Or(G)} F_{G/H}  (|\cJ|_{\mathrm{sm}})$. 
    \end{itemize}
    Here, $|\cI|_{\mathrm{sm}}$ and $|\cJ|_{\mathrm{sm}}$ be as in the proof of \cref{lem:Sh.model}. In this model structure, $\phi \colon \sF \to \sG$ is a weak equivalence (resp.\ fibration) if and only if $\phi_{G/H} \colon \sF(G/H) \to \sG(G/H)$ is a weak equivalence (resp.\ fibration) for any $G/H \in \Or(G)$. 
\end{prp}
\begin{proof}
    By \cref{lem:Sh.bicomplete} and \cite{lurieHigherToposTheory2009}*{Proposition 5.4.4.3}, the category $\PSh_\infty(\Or(G), \Sh(\Man))$ is accessible, and hence is locally presentable.
    By \cite{guillouEnrichedModelCategories2020}*{Theorem 4.32}, it suffices to show that, for any transfinite sequence $(\sF_\beta,\phi_\beta)$ of pushouts of morphisms in $F|\cJ|_{\mathrm{sm}}$, the morphism $\sF_0(G/H) \to \big( \colim \sF_\beta \big)(G/H)$ is a weak equivalence for any $G/H \in \Or(G)$. 
    In the same way as \cref{lem:Sh.model}, this follows from the compactness of $S^n$ by a transfinite induction. 
\end{proof}

\begin{prp}[{\cite{mayEquivariantHomotopyCohomology1996}*{Lemma V.3.1}}]\label{prp:Elmendorf}
    There is a pair of adjoint functors
    \begin{align*} 
\xymatrix{
     \Theta \colon \mathsf{PSh}_\infty(\mathsf{Or}(G)^{\op },\Sh(\Man))
     \ar@<1ex>[r]^{\hspace{8ex}  } \ar@{}[r]|{ \hspace{8ex} \bot} & 
    \Sh _G(\Man)\ar@<1ex>[l]^{\hspace{8ex}  } \colon \Phi  .
    }
    \end{align*}
    Moreover, there is a model structure on $\Sh_G(\Man)$ that makes $(\Phi , \Theta)$ a Quillen equivalence. 
\end{prp}
\begin{proof}
    The adjoint pair is defined by
    \begin{gather*}
        \Phi \sF (G/H) \coloneqq \sF^H, \quad  \Theta \sF \coloneqq \sfX(G).
    \end{gather*}

    As \cref{lem:Sh.model,prp:levelwise.model}, we use Kan's criterion for the right transfer model structure. It suffices to show that, for any transfinite sequence  $(\sF_\beta,\phi_\beta)$ of pushuots of morphisms in $\Theta F|\cJ|_{\mathrm{sm}}$, the transfinite composition $\sF_0^H \to \big( \colim \sF_\beta \big)^H $ is a weak equivalence for any $G/H \in \Or(G)$. 
    Since $\sF \mapsto \sF^H$ commutes with any small colimit, again in the same way as \cref{lem:Sh.model}, this follows from the compactness of $S^n$ by a transfinite induction. 
\end{proof}

For a $\Sh(\Man)$-enriched category $\cM$, we write $\Sing \cM$ for the $\sSet$-enriched category whose underlying category is the same as $\cM$ and the internal Hom set is given by $\Sing \sHom_{\cM} (X,Y)$. In particular, we write $\Or_\Delta (G) \coloneqq \Sing \Or(G)$.  
\begin{prp}\label{prp:base.change}
    The following adjoint pair gives a Quillen equivalence:
        \begin{align*} 
\xymatrix{
    |\blank|_{\mathrm{sm}} \colon  \mathsf{PSh}_\infty(\mathsf{Or}_{\Delta}(G),\sSet )\ar@<1ex>[r]^{} \ar@{}[r]|{\bot} & 
    \mathsf{PSh}_\infty(\mathsf{Or}(G),\Sh(\Man)) \ar@<1ex>[l]^{} \colon \Sing  .
    }
    \end{align*}
\end{prp}
\begin{proof}
    By \cref{lem:Sh.bicomplete,lem:Sh.closed.monoidal,lem:Sh.model,rmk:Sh.monoid,rmk:Sh.proper,rmk:Sh.monoidal.model}, the category $\Sh(\Man)$ satisfies all the assumption of \cite{guillouEnrichedModelCategories2020}*{Theorem 3.17}. 
\end{proof}

\begin{rmk}\label{rmk:base.change}
    According to \cite{guillouEnrichedModelCategories2020}*{Subsection 3.2}, the smooth realization functor $|\blank |_{\mathrm{sm}}$ in \cref{prp:base.change} is given by
\[
    |\sfX |_{\mathrm{sm}}(G/H) = \bigsqcup_{K \in \Or(G) } |\sfX(G/K)|_{\mathrm{sm}} \times \sHom_G(G/H,G/K) / \sim ,
\]
    where the equivalence relation $\sim $ is given by $(g^*\sfx , \phi) \sim (\sfx, g \psi)$ for $g \in |\Sing \Hom_G (G/K,G/L)|_{\mathrm{sm}}$. Here, $g$ acts on the manifold $\Hom _G(G/H,G/K)$ through the evaluation $|\Sing \Hom_G (G/K,G/L)|_{\mathrm{sm}} \to \Hom_G (G/K,G/L)$. 
    The adjunction counit $\epsilon_{\sX} \colon |\Sing \sX|_{\mathrm{sm}} \to \sX$ is given by 
    \begin{align*}
    \epsilon_{\sX}(G/H) \colon \bigg( \bigsqcup_{K \in \Or(G) } {}&{} |\Sing \sX (G/K)|_{\mathrm{sm}} \times \sHom_G(G/H,G/K) / \sim \bigg)\\
    {}&{}\xrightarrow{\bigsqcup \epsilon_{\sX(G/K)}} \bigg( \bigsqcup_{K \in \Or(G)} \sX(G/K) \times \sHom(G/H,G/K) /\sim \bigg) \cong \sX(G/H).
    \end{align*}
    Since $\sX$ is fibrant (\cref{rmk:fibrant.cofibrant} below), by \cite{hoveyModelCategories1999}*{Proposition 1.3.13}, the above $\epsilon_{\sX}$ is a weak equivalence. The same holds for topological analogues.    
\end{rmk}

\begin{rmk}\label{rmk:fibrant.cofibrant}
We remark on the model structures introduced here.
\begin{enumerate}
    \item  Recall that every object of $\kTop$ and $\Sh(\Man)$ are fibrant. By definition of the model structures of $\Sh_G(\Man)$, $\PSh_{\infty}(\Or(G),\Sh(\Man))$ given in \cref{prp:levelwise.model,prp:Elmendorf} and the analogous ones of $\kTop_G$, $\PSh_{\infty}(\Or_{\mathrm{top}}(G),\kTop)$, every object of these categories are fibrant.
    \item The $\Or(G)$-space $\Phi \sM \in \PSh_{\infty}(\Or_{\mathrm{top}}(G),\kTop)$ is cofibrant since it is obtained by attaching the $\Or(G)$-spaces of the form $F_{G/H}(\Delta_n)=\Phi (\Delta_n \times G/H)$. 
    \item Let us consider the smooth version $|\Sing_G \sM |_{G, \mathrm{sm}}$ of the $G$-space $|\Sing_G \sM|_G$ introduced in \cref{subsection:G.realization} as a $G$-sheaf, by replacing $\Delta_n(G,\bm{H})$ with $\Delta_n^e(G,\bm{H})$ and taking the quotient in \eqref{eqn:Gsimplicial.realization1} as sheaves. It is obtained by attaching the $G$-sheaves of the form $\Delta_n^e \times G/H$ iteratively, and hence is cofibrant. 
    What is proved in \cref{prp:Yoneda.Gsheaf} is that $|\Sing_G \sM |_{G, \mathrm{sm}}$ and $\sM$ are homotopy equivalent in the model category $\Sh_G(\Man)$ (cf.\ \cite{hirschhornModelCategoriesTheir2003}*{Definition 7.3.2}). Therefore, although we do not know whether $\sM$ is a cofibration, the morphism set of the homotopy category $  [\sM,\sF]_{\mathrm{Ho}\Sh_G(\Man)} $ is isomorphic to $\pi_0(\sHom (\sM,\sF)) \cong \sF[\sM]$.
\end{enumerate}
\end{rmk}

As is known as a version of the Elmendorf theorem (cf.\ \cite{dwyerSingularFunctorsRealization1984}), the same Quillen equivalence as \cref{prp:Elmendorf,prp:base.change} is valid if one replaces $\Sh(\Man)$ with $\kTop$ and $\Or(G)$ with the $\kTop$-enriched orbit category $\Or_{\mathrm{top}}(G)$. 
Finally, we obtain a zigzag of Quillen equivalences
\begin{align*}
    \Sh _G(\Man) \leftrightarrows {} & {} \mathsf{PSh}_\infty(\mathsf{Or}(G),\Sh(\Man)) 
    \leftrightarrows  \mathsf{PSh}_\infty(\mathsf{Or}_{\Delta}(G),\sSet) \\
    \leftrightarrows{} & {} \PSh (\Or_{\mathrm{top}}(G) , \kTop ) 
    \leftrightarrows  \kTop_G
    \end{align*}
that identifies the homotopy categories $\Sh_G(\Man)$ and $\kTop_G$ as desired.

\begin{rmk}\label{rmk:Elmendorf}
    The Elmendorf theorem \cite{elmendorfSystemsFixedPoint1983} (see also \cite{mayEquivariantHomotopyCohomology1996}*{Theorem V.3.2}) claims that there is a functor $\Psi \colon \PSh_{\infty}(\Or_{\mathrm{top}}(G),\kTop) \to \kTop_G$ and a natural transform $\epsilon \colon \Phi \Psi \to \id$ such that $\epsilon \colon (\Psi Y)^H \to Y^H$ is a weak equivalence. Such $(\Psi,\epsilon)$ gives an isomorphism $[X , \Psi Y]^G \cong [\Phi X, Y]^{\Or(G)}$ for any $G$-CW-complex $X$. 
    It is defined in \cite{elmendorfSystemsFixedPoint1983}*{Section 3} by the geometric realization of a simplicial $G$-space $B_{\bullet}(\sfX,G,J)$ given by the bar construction as
    \[
    B_{n}(\sfX,G,J)\coloneqq \bigsqcup_{H_0,\cdots,H_n}\sfX(G/H_0) \times \Hom_G(G/H_1,G/H_0) \times \cdots \times \Hom_G(G/H_{n},G/H_{n-1}) \times G/H_{n},
    \]
    and the natural transform $\epsilon$ is given by the evaluation. 
    Hence the same construction gives a functor $\Psi_{\mathrm{sm}} \colon \PSh_{\infty}(\Or(G),\Sh(\Man)) \to \Sh_G(\Man)$ such that the inclusion  of $G$-sheaves $\Psi_{\mathrm{sm}} |\sfX |_{\mathrm{sm}} \to C(\blank, \Psi |\sfX |)$ is a $G$-weak equivalence for any $\sfX \in \PSh_{\infty}(\Or_{\Delta}(G),\sSet)$. 
\end{rmk}

In conclusion, we get the isomorphism of equivariant homotopy sets
\begin{align*}
    \sF^G[\sM] \cong  [\sM,\sF]_{\mathrm{Ho}\Sh_G(\Man)} 
    \cong {}&{} [ |\Sing \Phi \sM| , |\Sing \Phi \sF|]_{\mathrm{Ho} \PSh_{\infty}(\Or_{\mathrm{top}}(G), \kTop)} \\
    \cong {}&{} [\Phi \sM,  |\Sing \Phi \sF|]_{\mathrm{Ho} \PSh_{\infty}(\Or_{\mathrm{top}}(G), \kTop)}\\
    \cong {}&{} [\Phi \sM, |\Sing \Phi \sF|]^{\Or (G)} \\
    \cong {}&{} [\sM, \Psi |\Sing \Phi \sF|]^G,
\end{align*}
that is, the $G$-space $\Psi |\Sing \Phi \sF|$ gives a $G$-realization for $\sF$.
Here, the first, the third and the fourth isomorphisms follow from \cref{rmk:fibrant.cofibrant}.

Although the $G$-space $|\Sing_G \sF|_G$ considered in \cref{subsection:G.realization} and the above $\Psi |\Sing \Phi \sF|$ are composed of similar parts, subtle differences make it not easy to construct a direct map. Here, we compare them as
    \begin{align*}
        |\Sing_G \sF|_G \xleftarrow{|\Sing_G \epsilon_{\sF}|_G} |\Sing_G (\Psi_{\mathrm{sm}}|\Sing \Phi \sF|_{\mathrm{sm}})|_G \to |\Sing_G (\Psi|\Sing \Phi \sF|)|_G \xrightarrow{\epsilon_{\Psi |\Sing \Phi \sF|}} \Psi|\Sing \Phi \sF|.
    \end{align*}
    By \cref{prp:Yoneda.Gsheaf,rmk:base.change,rmk:Elmendorf}, these $G$-maps are all $G$-weak equivalences.

\section{Stability of the spectral gap for non-interacting Hamiltonians}
In general, the spectral gap of the UAL Hamiltonian is not necessarily continuous under a perturbation that is continuous with respect to the most local norm. 
In particular, it is not clear whether it is stable under a small perturbation.
Bravyi--Hastings--Michalakis \cites{bravyiTopologicalQuantumOrder2010,bravyiShortProofStability2011} introduced a useful sufficient condition for the stability of the spectral gap against a small perturbation; the \emph{local topological quantum order} (LTQO) property. 
A general result in this direction is found in Nachtergaele--Sims--Yang \cite{nachtergaeleQuasilocalityBoundsQuantum2022}*{Corollary 7.5}. 
In these papers, it is proved that if $\sfH$ is a nice Hamiltonian with a spectral gap $1$, then there is a constant $s_0 >0$ such that $\sfH + s\sfV$ has a uniform spectral gap $1/2$ for any $s \in [0,s_{0}]$.
The constant $s_{0}$ depends on the input data in an elaborate way, and it is not easy to determine whether the theorem applies to our setting and when $s_{0} >0$ can be taken to be more than $1$. 
In this appendix, we focus on the special case that $\sfH$ is a non-interacting Hamiltonian, which was initially studied by Yarotsky~\cite{yarotskyPerturbationsGroundStates2004}, and give a simplified proof of the spectral gap stability in the way that an explicit lower bound of $s_{0}$ is given.  

We say that a UAL Hamiltonian $\sfH_0 =(\sfH_{0,\bm{x}})_{\bm{x} \in \Lambda} $ is a local Hamiltonian if each $\sfH_{0,\bm{x}}$ is a positive operator $\cA_{\bm{x}}$ such that $\Ker \sfH_{0,\bm{x}} = \bC \cdot \Omega_{\bm{x}}^{\sfH_0}$ for some unit vector $\Omega_{\bm{x}}^{\sfH_0}$. We remark that $\vvert \sfH_0 \vvert_f = \sup_{\bm{x} \in \Lambda} \| \sfH_{0,\bm{x}} \|$ independent of $f \in \cF$, which we write $\vvert \sfH _0\vvert$. 
Also, $\sfH_0$ has a spectral gap $1$ if and only if so does each $\sfH_{0,\bm{x}}$. 
\begin{lem}\label{lem:local.perturbation.gap}
    Let $\sfH_0=(\sfH_{0,\bm{x}})_{\bm{x} \in \Lambda}$ be a local Hamiltonian that has a spectral gap $1$ and let $\sfV_{\bm{x}}$ be a family of self-adjoint operators in $\cA_{\bm{x}}$ such that $\sfV_{\bm{x}} \Omega_{\bm{x}}^{\sfH_0} =0$. Then $\sfH(t)=\sfH_0 + t\sfV$ has a spectral gap $1/2$ if $t \in [0,(2\vvert \sfV \vvert)^{-1} ]$. 
\end{lem}
\begin{proof}
    Since the derivations $\Ad(\sfH)$ and $\Ad(\sfV)$ on $\cA_{\Lambda}$ preserves $\omega_{\sfH}$, they are represented by unbounded operators $H_{\omega}$ and $V_{\omega}$ on the GNS representation $(\sH_\omega, \pi_{\omega}, \Omega_{\omega})$ respectively. 
    The GNS representation of the unique ground state of $\sfH$ is the infinite tensor product of $(\sfH_{\bm{x}}, \Omega_{\bm{x}})$. It orthogonally decomposes into $\bigoplus_{Z \subset \Lambda} \sH_Z$ so that $\bigoplus_{Z \subset Y}\sH_Z = \pi_{\omega}(\cA_Y) \cdot \Omega_{\omega}$ for any $Y \subset \Lambda$, and both $H$ and $V$ acts diagonally. 
    Since $H|_{\sH_Z} \geq \# Z \cdot 1$ and $\| V|_{\sH_Z} \| \leq \# Z \cdot \sup_{\bm{x} \in \Lambda} \| \sfV_{\bm{x}} \|$, we have 
    \[ 
        (H_{\omega} + t V_{\omega})|_{Z} \geq \big( 1-t \cdot \sup_{\bm{x} \in \Lambda} \| \sfV_{\bm{x}} \| \big) \cdot \# Z \cdot 1 \geq \big( 1-t \cdot \sup_{\bm{x} \in \Lambda} \| \sfV_{\bm{x}} \| \big) \cdot 1. 
    \]
    This finishes the proof. 
\end{proof}

Hereafter, we fix an arbitrary choice of $\nu \geq l_{\Lambda} +1$ and $0<\mu<1$. 
\begin{lem}\label{lem:gap.stability.state.preserving.perturbation}
    Let $\Lambda$ be a weakly uniformly discrete subspace of $\bR^{l_{\Lambda}}$. Let $\sfH_0$ be a local Hamiltonian with a spectral gap $1$ and let $\sfV \in i\fDer ^{\al}_{\Lambda}$. Assume that $\omega_{\sfH_0}([\sfV, a]) =0$ for any $ a\in \cA_{\Lambda}^{\al}$. Then, the perturbation $\sfH (t) =\sfH_0 + t\sfV$ has a spectral gap $1/2$ if $t \in [0, (2^{5l_{\Lambda}+3}\kappa_{\Lambda}\kappa_{\bB} \vvert \sfV \vvert_{\nu_2,\mu_1})^{-1}]$.
\end{lem}
\begin{proof}
     Let $(B,\bB)$ be the associated brick of $\Lambda$ in the sense of \cref{defn:brick}. As in \cref{prp:brick}, set 
    \[
        \Phi_{\sfV}(B_{\rho}) \coloneqq \sum_{\bm{x} \in \Lambda} \sfV_{\bm{x}}^\rho  , \quad  \sfV_{\bm{x}}^\rho \coloneqq \sum_{\sigma \leq \rho} \mu_{\bB}(\rho,\sigma) \sfV_{\rho,\bm{x}}, \quad \sfV_{\bm{x},\rho} \coloneqq(\id_{\cA_{B_\rho}} \otimes \omega_{\sfH_0}) (\sfV_{\bm{x}}) 
    \]
    for $\rho \in \bB$,  
    where $\mu_{\bB}$ denotes the M\"obius function for $\bB$ (cf.\ \cref{subsubsection:brick}).
    For any $a \in \cA_{B_\rho}$ and $\sigma \leq \rho$, we have
    \[
        \sum_{\bm{x} \in \Lambda}\omega_{\sfH_0} ([\sfV_{\sigma, \bm{x}}, a]) = \sum_{\bm{x} \in \Lambda}\omega_{\sfH_0} \big( \big[\sfV_{\bm{x},\sigma}, (\id_{\cA_{B_\sigma}} \otimes \omega)(a) \big] \big) = \omega_{\sfH_0} ([\sfV, (\id_{\cA_{B_\sigma}} \otimes \omega)(a)]) =0, 
    \]
    and hence $\omega_{\sfH_0} ([\Phi_{\sfV}(B_\rho), a])=0$ for any $B_\rho \in \iota(\bB)$.  

    By replacing $\Lambda$ with $\Lambda \cup \bZ^{l_{\Lambda}}$ and $\sfH_0 + t \sfV$ with $\sfH_0 \boxtimes \sfh + t\sfV \otimes 1$ if necessary, without loss of generality we may assume that $\Lambda$ is relatively dense in $\bR^{l_{\Lambda}}$ and $\iota (\bB)=\bB$.  
    For $\bm{d} =(d_1,\cdots,d_{l_{\Lambda}}) \in \bZ_{\geq 1}^{l_{\Lambda}}$, we write $\bB_{\bm{d}}$ for the set of bricks of the form $D_{\bm{a},\bm{d}}\coloneqq \prod_{i=1}^{l_{\Lambda}}[a_i,a_i+d_i)$ for $\bm{a} = (a_1,\cdots,a_{l_{\Lambda}}) \in \bZ^{l_{\Lambda}}$. 
    We further decompose it into finite disjoint union $\bB_{\bm{d}}=\bigsqcup_{\fa}\bB_{\bm{d}, \fa}$ indexed by $\fa \in \bigoplus_{i=1}^{l_{\Lambda}} \bZ/d_i\bZ $ as $\bB_{\bm{d},\fa} = \{ D_{\bm{a},\bm{d}} \in \bB_{\bm{d}} \mid [\bm{a}] = \fa \} $. 
    We apply \cref{lem:local.perturbation.gap} to the perturbation $\sfH_{\bm{d},\rho}(t) \coloneqq \sfH_0+ t\sum_{\rho \in \bB_{\bm{d},\fa}} \sfV^\rho$, treating the family $\{ \sfV^{\rho} \}_{\rho \in \bB_{\bm{d},\fa}}$ as a local perturbation on the coarse-grained lattice in which each brick in $\bB_{\bm{d},\fa}$ is regarded as a single lattice point. 
    Then we obtain that $\sfH_{\bm{d},\fa}(t)$ has a spectral gap $1/2$ if $t \in [0,1/2v_{\bm{d},\fa}]$, where 
    \[ 
    v_{\bm{d},\fa} \coloneqq \sup_{\rho \in \bB_{\bm{d},\fa}} \| \Phi_\sfV(B_\rho )\| \leq 2^{3l_{\Lambda}+1} \kappa_{\Lambda} \cdot f_{\nu_1,\mu_1}(\|\bm{d}\|/2  -1) \cdot \vvert \sfV \vvert_{\nu_2,\mu_1}.
    \]
    Here, the inequality is given by \eqref{eqn:brick.decomposition} and $\mathrm{diam}(\rho)= \|\bm{d} \|$ for any $\rho \in \bB_{\bm{d},\fa}$. 
    In other words, $v_{\bm{d},\fa}\sfH + t \sum_{\rho \in \bB_{\bm{d},\fa}}\sfV^\rho$ has a spectral gap $v_{\bm{d},\fa}/2$ if $t \in [0,1/2]$. 

    By using $\# (\bigoplus_i \bZ/d_i \bZ)=\prod_i d_i \leq \| \bm{d}\|^{l_{\Lambda}}$,  let 
    \begin{align*}
        v \coloneqq  \sum_{\bm{d} \in \bZ_{\geq 1}^{l_{\Lambda}}} \sum_{\fa} v_{\bm{d},\fa} 
        \leq {}&{} 
        \sum_{\bm{d}} \|\bm{d}\|^{l_{\Lambda}}  \cdot   2^{3l_{\Lambda} +1} \kappa_{\Lambda} \cdot f_{\nu_1,\mu_1}(\|\bm{d}\|/2-1) \cdot \vvert \sfV \vvert_{\nu_2,\mu_1} \\
        \leq {}&{} 
        \sum_{\bm{d}} c\cdot 2^{3l_{\Lambda} +1} \kappa_{\Lambda} \cdot f_{\nu,\mu}(\|\bm{d}\|) \cdot \vvert \sfV \vvert_{\nu_2,\mu_1} \\
        \leq {}&{} c \cdot 2^{5l_{\Lambda}+2 } \kappa_{\Lambda}\kappa_{\bB} \cdot \vvert \sfV \vvert_{\nu_2,\mu_1},
    \end{align*}
    for some $c>0$ depending only on $\nu,\mu$. 
    Then, for $s \in [0,1/2]$, we have
    \[
        \sfH + \frac{s}{v} \sfV = \frac{1}{v} \sum_{\bm{d}} \sum_{\fa} \Big( v_{\bm{d},\fa}\sfH + \sum_{\rho \in \bB_{\bm{d},\fa}}s \sfV^\rho \Big) \geq \frac{1}{v} \sum_{\bm{d}} \sum_{\fa} \frac{v_{\bm{d},\fa}}{2}\big( 1 - \Omega_\omega \otimes \Omega_{\omega}^* \big) = \frac{1}{2}(1 - \Omega_\omega \otimes \Omega_\omega^*).
    \]
    That is, $\sfH + t\sfV$ has a spectral gap $1/2$ if $t \in [0,1/2v]$. 
\end{proof}

\begin{lem}\label{lem:appendix.Lieb.Robinson}
    Let $\sfV \in i\fDer ^{\al}_{\Lambda}$. 
    Assume that $\vvert \sfV \vvert_{\nu_8, \mu_4} \leq \vvert \sfH_0 \vvert$. Then, there are $M_{\nu,\mu} >0 $ and $N_{\sfH_0,\nu,\mu} >0$ such that the function $\Upsilon_{\sfG_{\sfH} , \nu,\mu}$ in \cref{prp:Lieb.Robinson} is bounded above by 
    \[
        \Upsilon_{\sfG_{\sfH},\nu,\mu}(t) \leq M_{\nu,\mu} + 2^{l_{\Lambda}+2} \kappa_{\Lambda}^2 \cdot f_{\nu,\mu}(2 \cdot (2C_{\nu_1}L_{\nu_1,\mu_1} N_{\nu_4,\mu_2} \vvert \sfV \vvert_{\nu_4,\mu_3} )^{1/\mu_1} \cdot t^{1/\mu_1})^{-1}.
    \]
\end{lem}
\begin{proof}
Let us recall the explicit form of $\Upsilon_{\sfG,\nu,\mu}$ given in \eqref{eqn:Upsilon}. That is, 
\begin{align*}
    \Upsilon_{\sfG_{\sfH} , \nu, \mu} (t) = M_{\nu,\mu} + 2^{l_{\Lambda} +2}  \kappa_{\Lambda}^2 \cdot f_{\nu,\mu}(2 \cdot (2C_{\nu_1}L_{\nu_1,\mu_1} \vvert \sfG_{\sfH} \vvert _{\nu_4,\mu_2})^{1/\mu_1}  \cdot t^{1/\mu_1})^{-1},  
\end{align*}
where $M_{\nu,\mu} = 2c_{\nu,\mu,2} + (\sup f_{\nu-l_{\Lambda},\mu_1 } / f_{\nu,\mu} ) \cdot f_{\mu_1}(r_0/2)$ is a constant depending only on $\nu,\mu$. 

By definition (cf.\ the proof of \cref{lem:adiabatic.well-defined}), $\vvert \sfG_{\sfH} \vvert _{\nu,\mu}$ has an upper bound given by 
\begin{align*}
    \vvert  \sfG_{\sfH} \vvert_{\nu,\mu} \leq {}&{} 2\int_0^\infty \bigg( \int_0^t \Upsilon_{\sfH(u),\nu,\mu}(u) \cdot \vvert \sfV \vvert_{\nu,\mu_1}du \bigg) w_v(t)dt \\
    \leq {}&{} \vvert \sfV \vvert_{\nu,\mu_1} \cdot 2\int_0^\infty t \Upsilon_{\sfH,\nu,\mu}(t) \cdot w_v(t)dt. 
\end{align*}
If $ \vvert \sfV \vvert_{\nu_4,\mu_2} \leq \vvert \sfH_0 \vvert$ holds, then we have $\vvert \sfH \vvert_{\nu_4,\mu_2} \leq 2\vvert \sfH_0\vvert$. Again by \eqref{eqn:Upsilon}, we have 
\[
    \Upsilon_{\sfH,\nu,\mu} (t) \leq M_{\nu,\mu} + 2^{l_{\Lambda} +2 } \kappa_{\Lambda}^2 \cdot f_{\nu,\mu}(2 \cdot (4 C_{\nu_1}L_{\nu_1,\mu_1} \vvert \sfH_0 \vvert )^{1/\mu_1} \cdot t^{1/\mu_1})^{-1}.
\]
We define the constant $N_{\sfH_0,\nu,\mu} >0$ as
\[ 
    N_{\sfH_0,\nu,\mu} \coloneqq  2\int_0^\infty t \cdot \Big( M_{\nu,\mu} + 2^{l_{\Lambda} +2 } \kappa_{\Lambda}^2 \cdot f_{\nu,\mu}(2 \cdot (4 C_{\nu_1}L_{\nu_1,\mu_1} \vvert \sfH_0 \vvert )^{1/\mu_1} \cdot t^{1/\mu_1})^{-1} \Big) \cdot w_v(t)dt. 
\]
Then, if $ \vvert \sfV \vvert_{\nu_4,\mu_2} \leq \vvert \sfH \vvert$, we get  $\vvert \sfG _{\sfH} \vvert_{\nu,\mu} \leq \vvert \sfV \vvert_{\nu,\mu_1} \cdot N_{\sfH_0,\nu,\mu}$. Apply this to the pair $(\nu_4,\mu_2)$ and substitute to the first equality. 
\end{proof}
Let 
\[ 
    I_{\nu,\mu} \coloneqq M_{\nu,\mu} +  2^{l_{\Lambda} +2}  \kappa_{\Lambda}^2 \cdot f_{\nu,\mu}(2 \cdot (2C_{\nu_1}L_{\nu_1,\mu_1} N_{\sfH,\nu,\mu})^{1/\mu_1})^{-1}.
\]
Then, if $\vvert \sfV \vvert_{\nu_4,\mu_3} \leq 1$ and $t \leq 1$, we have $\Upsilon_{\sfG_{\sfH},\nu,\mu}(t) \leq I_{\nu,\mu}$ for any $t \in [0,1]$. 

\begin{lem}\label{lem:gap.stability.inner}
    Let $\sfV$ be an inner perturbation, i.e., $\sup_{\bm{x} \in \Lambda} f(\rmd(\bm{x},\bm{x}_0))^{-1} \cdot \| \sfV_{\bm{x}} \| <\infty $ for any $f \in \cF$. Assume that $\vvert \sfV \vvert_{\nu_6,\mu_4} \leq 1$, $\vvert \sfV \vvert _{\nu_{10},\mu_5} \leq \vvert \sfH \vvert$ and 
    \begin{align*}
         2^{5l_{\Lambda}+3}\kappa_{\Lambda}\kappa_{\bB}  \cdot I_{\nu_2,\mu_1} \cdot (d_{\nu_2,\mu_2}\vvert \sfV \vvert_{\nu_3,\mu_5} \cdot N_{\sfH_0,\nu_3,\mu_4} \cdot (\vvert \sfH_0 \vvert +\vvert \sfV \vvert_{\nu_3,\mu_4}) + \vvert \sfV \vvert_{\nu_2,\mu_2}) \leq 1,
    \end{align*}
    where $d_{\nu,\mu}$ is the constant in \cref{lem:almost.local.continuity}. 
    Then $\sfH_0 + t \sfV$ has a spectral gap $1/2$ for $s \in [0,1]$. 
\end{lem}
\begin{proof}
    Let $t_0$ be the supremum of $t>0$ such that $\sfH + u\sfV$ has a spectral gap $1/2$ for any $u \in [0,t]$. As is observed in the proof of \cref{lem:cut.diffused}, the spectrum of $\sfH(u)=\sfH_0 + u\sfV$ is continuous, and hence we have $t_0>0$. Moreover, the ground state of $\sfH( u)$ is smooth. 
    We assume that $t_0 <1$ and show the contradiction.

    For $t \in [0,t_0]$, set
    \[
    \sfW(t) \coloneqq \alpha(\sfG_{\sfH} \,; t)^{-1} (\sfH_0 +t \sfV ) - \sfH_0.
    \]
    Then $\sfW(t)$ is an inner UAL derivation such that $\omega_{\sfH_0} ([\sfW,a])=0$ and, by \cref{lem:almost.local.continuity,prp:Lieb.Robinson,lem:appendix.Lieb.Robinson},  
    \begin{align*}
    \vvert \sfW \vvert_{\nu_2,\mu_1} \leq {}&{} \int_0^t \vvert \alpha(\sfG_{\sfH} \,; u)^{-1} ([\sfG_{\sfH} , \sfH]) \vvert_{\nu_2,\mu_1} + \vvert \alpha (\sfG_{\sfH} \,; u)^{-1} (\sfV) \vvert_{\nu_2,\mu_1} du \\
    \leq {}&{} t \cdot \Upsilon_{\sfG_{\sfH} , \nu_2,\mu_1}(t ) \cdot (  d_{\nu_2,\mu_2} \cdot \vvert \sfG_{\sfH} \vvert_{\nu_3,\mu_4} \cdot \vvert \sfH \vvert_{\nu_3,\mu_4} + \vvert \sfV \vvert_{\nu_2,\mu_2}) \\
    \leq {}&{} t \cdot I_{\nu_2,\mu_1} \cdot (d_{\nu_2,\mu_2} \cdot \vvert \sfV \vvert_{\nu_3,\mu_5} \cdot N_{\sfH_0,\nu_3,\mu_4} \cdot \vvert \sfH \vvert_{\nu_3,\mu_4} + \vvert \sfV \vvert_{\nu_2,\mu_2}). 
    \end{align*}
    By assumption, we get $2^{5l_{\Lambda}+3}\kappa_{\Lambda}\kappa_{\bB} \vvert \sfW \vvert_{\nu_2,\mu_1}  <1$.
    Therefore, by \cref{lem:gap.stability.state.preserving.perturbation}, $\alpha(\sfG_{\sfH} \,; t)(\sfH_0 + t\sfV) = \sfH_0 + \sfW(t)$ has a spectral gap $1/2$ for any $t \in [0,1]$. 
\end{proof}

\begin{thm}\label{thm:gap.stability}
    Let $\sfH_0 =(\sfH_{0,\bm{x}})_{\bm{x} \in \Lambda}$ be a local Hamiltonian with a spectral gap $1$ and let $\sfV =(\sfV_{\bm{x}})_{\bm{x} \in \Lambda} \in \fDer ^{\al}_{\Lambda}$. Assume that the almost local norms of $\sfV$ are sufficiently small to satisfy the assumption of \cref{lem:gap.stability.inner}. Then, for any $t \in [0,1]$, the perturbation $\sfH_0 + t\sfV$ is a smooth family of UAL Hamiltonians that has the uniform spectral gap $1/2$. 
\end{thm}
\begin{proof}
    We fix a reference point $\bm{x}_0 \in \Lambda$ and let $\sfV_{n}  \coloneqq \Pi_{\bm{x}_0,n}(\sfV)$ and $\sfH_{n} (t) \coloneqq \sfH_0 +t\sfV_n$. 
    Then, the adiabatic parallel transport $\alpha (\sfG_{\sfH_n} \,; t)$ restricts to $\cA_{B_n(\bm{x})}$, which is nothing else than the adiabatic parallel transport of the finite volume Hamiltonian $\Pi_{\bm{x_0},n}(\sfH + \sfV)$. By \cref{prp:Lieb.Robinson} and \cref{cor:Lieb.Robinson.approx.Ck} (2) (we also refer to \cite{nachtergaeleQuasilocalityBoundsQuantum2019}*{Theorem 6.14}), we have
    \[
     \alpha (\sfG_{\sfH_n} \,; t) (a)  \xrightarrow{n \to \infty} \alpha(\sfG_{\sfH} \,; t ) (a), 
    \]
    in the almost local topology, and hence 
    \[ 
    \omega_{n,t}(a) \coloneqq \omega_{\sfH_0} \circ \alpha (\sfG_{\sfH_n} \,; t)^{-1}(a) \xrightarrow{n \to \infty} \omega_{\sfH_0} \circ \alpha (\sfG_{\sfH} \,; t)^{-1}(a) =: \omega_t(a)
    \]
    for any  $a \in \cA_{\Lambda}^{\al}$. 
    The state $\omega_{n,t}$ is the distinguished ground state of $\sfH + t\sfV_n$.

    By \cref{lem:gap.stability.inner}, $\sfH_n(t)$ has a spectral gap $1/2$ for any $t \in [0,1]$. Let $a \in \cA_{B_{r}(\bm{x})}$ such that $\|a\|=1$ and $\omega_t (a)=0$. For $\varepsilon >0$, take a sufficiently large $n >0$ such that $|\omega_{n,t}(a)| \leq \varepsilon$, $|(\omega_t - \omega_{n,t})(a^*[\sfH+\sfV ,a])| \leq \varepsilon$, $|(\omega_t - \omega_{n,t})([\sfH+\sfV ,a])| \leq \varepsilon$, and $ \| [\sfV - \sfV_n , a] \| \leq \varepsilon$. Then 
    \begin{align*}
        |\omega_t([\sfH + \sfV,a])| \leq  |\omega_{n,t}([\sfH+\sfV,a]) | + \varepsilon  \leq  |\omega_{n,t}([\sfH+\sfV_n,a]) | + 2 \varepsilon =2\varepsilon ,
    \end{align*}
    and hence 
    \begin{align*}
        |\omega_{n,t}(a^*[\sfH + \sfV_n,a])| \geq {}&{} |\omega_{n,t}(a_n^* [\sfH + \sfV _n, a])| - |\omega_{n,t}(a)| \cdot \omega_{n,t}( [\sfH + \sfV_n , a])\\
        \geq  {}&{} |\omega_{n,t}(a_n^*[\sfH + \sfV_n , a_n]) | - 2 \varepsilon,
    \end{align*}
    where $a_n \coloneqq a - \omega_{n,t}(a) \cdot 1$ (note that $\|a_n \| \geq 1-\varepsilon$). Therefore, 
    \begin{align*}
        \omega_t(a^*[\sfH + \sfV,a]) \geq {}&{} \omega_{n,t}(a^*[\sfH+\sfV,a]) - \varepsilon \\
        \geq {}&{} \omega_{n,t}(a^*[\sfH+\sfV_n,a]) - 2\varepsilon \\
        \geq {}&{}  \omega_{n,t}(a_n^*[\sfH+\sfV_n,a_n]) - 4 \varepsilon \\
        \geq {}&{} 2^{-1} (1-\varepsilon)^2 -4\varepsilon. 
    \end{align*}
    Since $\varepsilon >0$ is arbitrary, this shows that $\omega_t$ is a non-degenerate ground state of $\sfH_0 + t\sfV$ with a spectral gap $1/2$.
\end{proof}

\end{document}